\documentclass[oneside]{scrbook}

\usepackage{color}
\usepackage{amsthm}
\usepackage{hyperref}
\usepackage{hyperxmp}
\hypersetup{
  pdftitle={Divisible Codes},
  pdfauthor={Sascha Kurz},
  pdfpublisher={arXiv},
  pdfsubject={A linear code over F\string_q with the Hamming metric is called Delta-divisible if the weights of all codewords are divisible by Delta. They have been introduced by Harold Ward a few decades ago. Applications include subspace codes, partial spreads, vector space partitions, and distance optimal codes. The determination of the possible lengths of projective divisible codes is an interesting and comprehensive challenge.},
  pdfkeywords={divisible codes, linear codes, projective systems, multisets of points, partial spreads, vector space partitions, subspace codes},
  pdfcopyright={Creative Commons Lizenz CC BY 4.0},
  pdflicenseurl={https://creativecommons.org}
}
\usepackage{xhfill}
\usepackage{amsmath}
\usepackage{amssymb}
\usepackage{tikz}
\usetikzlibrary{arrows.meta,positioning,calc,shapes.geometric,fit}
\addtokomafont{disposition}{\boldmath}

\usepackage{type1cm}        
\usepackage{makeidx}         
\usepackage{graphicx}        
\usepackage{multicol}        
\usepackage[bottom]{footmisc}

\usepackage{algorithm}
\definecolor{light}{gray}{0.85}

 \newenvironment{trailer}[1]%
    {
    {{{\color{light}\rule{5mm}{1mm}}\hspace{2mm}}\textbf{#1}
    \xhrulefill{light}{1mm}\em \\}%
 }

\newenvironment{warning}[1]%
    {
    {{{\color{light}\rule{5mm}{1mm}}\hspace{2mm}}\textbf{#1}
    \xhrulefill{light}{1mm}\em \\}%
 }

\newenvironment{important}[1]%
    {
    {{{\color{light}\rule{5mm}{1mm}}\hspace{2mm}}\textbf{#1}
    \xhrulefill{light}{1mm}\em \\}%
 }

 \newenvironment{question}[1]%
     {
     {{{\color{light}\rule{5mm}{1mm}}\hspace{2mm}}\textbf{#1}
     \xhrulefill{light}{1mm}\em \\}%
  }

\newenvironment{backgroundinformation}[1]%
    {
    {{{\color{light}\rule{5mm}{1mm}}\hspace{2mm}}\textbf{#1}
    \xhrulefill{light}{1mm}\em \\}%
 }

\newenvironment{programcode}[1]%
    {
    {{{\color{light}\rule{5mm}{1mm}}\hspace{2mm}}\textbf{#1}
    \xhrulefill{light}{1mm}\em \\}%
 }

\newtheorem{ntheorem}{Theorem}[chapter]
\newtheorem{nexercise}[ntheorem]{Exercise}
\newtheorem{nlemma}[ntheorem]{Lemma}
\newtheorem{ndefinition}[ntheorem]{Definition}
\newtheorem{nproposition}[ntheorem]{Proposition}
\newtheorem{nconjecture}[ntheorem]{Conjecture}
\newtheorem{ncorollary}[ntheorem]{Corollary}
\newtheorem{nexample}[ntheorem]{Example}
\newtheorem{nremark}[ntheorem]{Remark}

\newcommand{\RM}{\operatorname{RM}}
\newcommand{\GF}{\operatorname{GF}}
\newcommand{\F}{\mathbb{F}}
\newcommand{\N}{\mathbb{N}}
\newcommand{\Q}{\mathbb{Q}}
\newcommand{\R}{\mathbb{R}}
\newcommand{\Z}{\mathbb{Z}}
\newcommand{\wt}{\operatorname{wt}}
\newcommand{\dH}{\operatorname{d}_{\text{H}}}
\newcommand{\Aut}{\operatorname{Aut}}
\newcommand{\ba}{\mathbf{a}}  
\newcommand{\bb}{\mathbf{b}} 
\newcommand{\bc}{\mathbf{c}}
\newcommand{\bd}{\mathbf{d}}
\newcommand{\be}{\mathbf{e}}
\newcommand{\bg}{\mathbf{g}}
\newcommand{\bh}{\mathbf{h}}
\newcommand{\bu}{\mathbf{u}}
\newcommand{\bv}{\mathbf{v}}
\newcommand{\bx}{\mathbf{x}}
\newcommand{\by}{\mathbf{y}}
\newcommand{\bz}{\mathbf{z}}
\newcommand{\PG}{\operatorname{PG}}
\newcommand{\AG}{\operatorname{AG}}

\newcommand{\cB}{\mathcal{B}}
\newcommand{\cC}{\mathcal{C}}
\newcommand{\cD}{\mathcal{D}}

\newcommand{\cG}{\mathcal{G}}
\newcommand{\cH}{\mathcal{H}}  
\newcommand{\cK}{\mathcal{K}}
\newcommand{\cM}{\mathcal{M}}
\newcommand{\cN}{\mathcal{N}}
\newcommand{\cP}{\mathcal{P}}
\newcommand{\cS}{\mathcal{S}}
\newcommand{\cQ}{\mathcal{Q}}
\newcommand{\cT}{\mathcal{T}}
\newcommand{\cU}{\mathcal{U}}
\newcommand{\cV}{\mathcal{V}}

\newcommand{\qbin}[3]{\genfrac{[}{]}{0pt}{}{#1}{#2}_{#3}}
\newcommand{\spaces}[2]{\genfrac{[}{]}{0pt}{}{#1}{#2}}
\newcommand{\neff}{n_{\operatorname{eff}}} 
\newcommand{\snumb}[3]{s_{#3}(#1,#2)} 
\newcommand{\frobenius}[2]{\mathrm{F}_{#2}(#1)} 
\newcommand{\supp}{\operatorname{supp}}
\newcommand{\zv}{\mathbf{0}} 
\newcommand{\Res}{\operatorname{Res}}
\newcommand{\rem}{\operatorname{rem}}
\def\llceil{\lceil\kern-3.2pt\lceil}
\def\rrceil{\rceil\kern-3.2pt\rceil}
\def\llfloor{\lfloor\kern-3.2pt\lfloor}
\def\rrfloor{\rfloor\kern-3.2pt\rfloor}
\def\leftllceil{\left\lceil\kern-3.2pt\left\lceil}
\def\rightrrceil{\right\rceil\kern-3.2pt\right\rceil}
\def\leftllfloor{\left\lfloor\kern-3.2pt\left\lfloor}
\def\rightrrfloor{\right\rfloor\kern-3.2pt\right\rfloor}
\def\bigllfloor{\bigg\lfloor\kern-3.2pt\bigg\lfloor}
\def\bigrrfloor{\bigg\rfloor\kern-3.2pt\bigg\rfloor}  

\makeindex             

\begin{document}
\date{03.12.2026}
\title{Divisible Codes}
\author{Sascha Kurz\\\footnotesize sascha.kurz@uni-bayreuth.de}
\publishers{\footnotesize\begin{flushleft}\textbf{Abstract} A linear code over $\mathbb{F}_q$ with the Hamming metric is called  $\Delta$-divisible if the weights of all codewords are divisible by $\Delta$.  They have been introduced by 
Harold Ward a few decades ago \cite{ward1981divisible}. Applications include subspace codes, partial spreads, vector space partitions, and distance 
optimal codes. The determination of the possible 
 lengths of projective divisible codes is an interesting and comprehensive challenge.\end{flushleft}}

\maketitle
\pagenumbering{roman}
\tableofcontents
\pagebreak
\pagenumbering{arabic}

\chapter{Introduction}
\label{sec_introduction}
A \emph{linear code} $C$ of \emph{length} $n$ is a subspace of the vector space $\F_q^n$ of $n$-tuples with entries in the finite field $\F_q$, where the field size $q$ is a prime 
power $p^m$. The \emph{(Hamming) weight} $\wt(\bc)$ of each \emph{codeword} $\bc\in C$ is the number of non-zero coordinates of $\bc$, i.e., 
$\wt(\bc):=\# \left\{1\le i\le n\,:\, c_i\neq 0\right\}$. With this, the \emph{Hamming distance} between two codewords $\bc$ and $\bc'$ is given by $\dH(\bc,\bc')=\wt(\bc-\bc')$. In other 
words, the Hamming distance counts the number of coordinates that differ between two codewords. A linear code $C$ is called \emph{$\Delta$-divisible} iff the weights of all codewords 
are divisible by $\Delta$.\footnote{A non-linear code is called $\Delta$-divisible if every distance between a pair of codewords is divisible by $\Delta$. The study of divisible 
codes is also of interest in other metrics besides the Hamming metric, see e.g.\ \cite{polverino2022divisible} for rank metric codes.} Note that every linear code is $1$-divisible, 
so that one mostly considers the cases $\Delta>1$ only. If $\Delta=2$ or $\Delta=4$ we also speak of \emph{even} or \emph{doubly-even} codes, respectively.  

\begin{nexample}
  \label{ex_extended_hamming_8}
  The extended binary Hamming code of length $8$ is a $[8,4,4]_2$-code and given by the generator matrix
  $$
    \begin{pmatrix}
      11110000 \\ 
      00111100 \\ 
      00001111 \\ 
      01010101  
    \end{pmatrix}.
  $$
  It has weight enumerator $1+14x^4+1x^8$, i.e., the code is $4$-divisible (doubly-even). For the extended Hamming codes of 
  length $2^m$ with $m\ge 4$ this is no longer true. Indeed, the extended Hamming code of length $16$ has weight enumerator 
  $1+140x^4+448x^6+870x^8+448x^{10}+140x^{12}+1x^{16}$ and is just $2$-divisible. However, its dual code is a first order Reed--Muller code, 
  which has a much larger divisibility, as we will see next.
\end{nexample}

\begin{nexample}
  The first order binary (generalized) Reed--Muller code $\operatorname{RM}_2(4,1)$ of length $2^4=16$ given by the generator matrix
  $$
    \begin{pmatrix}
      1111111111111111 \\ 
      1111111100000000 \\ 
      1111000011110000 \\ 
      1100110011001100 \\ 
      1010101010101010  
    \end{pmatrix},
  $$
  has weight enumerator $1+30x^8+1x^{16}$, i.e., the code is $8$-divisible.
\end{nexample}
Table~\ref{table_applications_overview} in Chapter~\ref{sec_applications} gives an overview of applications of divisible codes.

\section{An introductory application}
\label{subsec_introductory_application}

Consider binary vectors of length $9$, i.e., elements of $\F_2^9$. The span $\left\langle v_1,\dots, v_r\right\rangle$ of a sequence of those vectors 
forms a subspace, i.e., a subset of $\F_2^9$ that is closed under addition and scalar multiplication. For the vectors 
\begin{eqnarray*}
  \bv^1 &=& (1,0,0,0,1,1,1,0,0),\\ 
  \bv^2 &=& (1,1,0,0,0,1,0,1,1),\\
  \bv^3 &=& (0,1,0,0,1,0,1,1,1),\text{ and} \\ 
  \bv^4 &=& (0,0,0,1,0,1,1,0,0) 
\end{eqnarray*}
the set
$$
  \left\langle \bv^1,\dots, \bv^4\right\rangle:=\left\{\sum_{i=1}^4 \lambda_i\bv^i \,:\, \lambda_i\in\F_2 \,\forall 1\le i\le 4\right\}
$$
consists of $8$ elements and is a $3$-dimensional subspace, i.e., it admits a basis of size $3$ and contains $2^3$ elements. Note that we are using 
row vectors for the elements of $\F_2^9$.
\begin{nexercise}
  Compute a basis of $\left\langle \bv^1,\dots, \bv^4\right\rangle$ using the Gaussian elimination algorithm (over $\F_2$).
\end{nexercise}
Note that each subspace $S$ (of $\F_2^9$) contains the all-zero vector $\zv$. Now we want to consider the following 
packing question: Do there exist $20$ four-dimensional and $30$ three-dimensional subspaces in $\F_2^9$ such that their pairwise intersection is trivial, 
i.e., the intersection consists just of the zero vector $\zv$?

In order to answer this question we first observe that each $k$-dimensional subspace of $\F_2^9$, where $0\le k\le 9$, consists of exactly 
$2^k-1$ non-zero vectors. So, the $20$ four-dimensional and the $30$ three-dimensional subspaces cover exactly 
$$
  20\cdot \left(2^4-1\right)+30\cdot\left(2^3-1\right)=510
$$
of the $511$ non-zero vectors in $\F_2^9$. In other words, there would be exactly one uncovered non-zero vector $\bu$. This does not yield a 
contradiction directly, but we may consider the set of covered non-zero vectors $\bv$ that satisfy $\bv\bh^\top=0$ for a fixed (row-) vector 
$\bh\in\F_2^9\backslash\{\zv\}$. For an arbitrary four-dimensional subspace $S$ and an arbitrary three-dimensional subspace $E$ we have
\begin{equation}
  \label{eq_hyp_intersection_1}
  \left|\left\{\bv\in S\,:\, \bv\in\F_2^9\backslash\{\zv\}, \bv\bh^\top=0 \right\}\right|\in\{7,15\}
\end{equation}
and
\begin{equation}
  \label{eq_hyp_intersection_2}
  \left|\left\{\bv\in E\,:\, \bv\in\F_2^9\backslash\{\zv\}, \bv\bh^\top=0 \right\}\right|\in\{3,7\}.
\end{equation}

\begin{nexercise}
  For $1\le k\le n$ let $S$ be a $k$-dimensional subspace of $\F_q^n$ and $\bh\in \F_q^n\backslash\{\zv\}$. Show that the 
  set $\left\{\bv\in S\,:\, \bv\in\F_q^n, \bv\bh^\top=0 \right\}$ is a subspace of dimension $k$ or $k-1$.
\end{nexercise}

 From (\ref{eq_hyp_intersection_1})  and (\ref{eq_hyp_intersection_2}) we can conclude that the number of non-zero vectors $v$ that satisfy $\bv\bh^\top=0$ and are 
 covered by one of the $20+30=50$ subspaces is congruent to $50\cdot 3\equiv 2$ modulo $4$. Thus, the total number of covered non-zero vectors satisfying $\bv\bh^\top=0$ is even, so  
 that, since $2^8-1=255$ non-zero vectors satisfy $\bv\bh^\top=0$, the number of uncovered non-zero vectors being perpendicular to $\bh$ is odd. Since $\bu$ is the unique non-zero vector that is not contained in 
 one of the $50$ subspaces, we have $\bu\bh^\top=0$ for all $\bh\in\F_2^9\backslash\{\zv\}$. This implies $\bu=\zv$, which is a contradiction. Thus, 
 no such $20$ four-dimensional and $30$ three-dimensional subspaces can exist in $\F_2^9$.

While our argument and example is rather ad-hoc, something more general is hiding behind the scenes. The problem is an existence question for so-called 
vector space partitions. The set of covered non-zero vectors can be associated with a linear code $C_0$ of (effective) length $510$ and the complement, i.e., 
the set of uncovered non-zero vectors, can be associated with a linear code $C_1$ of (effective) length $1$. As we will see in Lemma~\ref{lem:union_subspaces} and 
Lemma~\ref{lemma_t_complement}, both codes 
$C_0$ and $C_1$ have to be $4$-divisible. However, there is no $4$-divisible binary linear code of effective length $1$. In other words, the non-existence 
of $\Delta$-divisible codes with a certain effective length certifies the non-existence of a vector space partition of a certain type. For the details on 
vector space partitions we refer to Section~\ref{sec_vector_space_partitions} and for non-existence results for divisible codes we refer to 
Section~\ref{sec_lengths_of_divisible_codes} and Section~\ref{sec_nonexistence_projective_q_r}.   

\chapter{Preliminaries}
\label{sec_preliminiaries}
Let $C\subseteq\F_q^n$ be a linear code over $\F_q$. If $C$ is a $k$-dimensional subspace, we say that $C$ is an \emph{$[n,k]_q$-code}. The number $k$ is called 
the \emph{dimension} of $C$ and $n$ its \emph{length}. Note that $n\ge k$ and that we will assume $k\ge 1$ in the following. If $q=2$, $q=3$, or $q=4$, we speak of a \emph{binary}, 
a \emph{ternary}, or a \emph{quaternary} code, respectively. The \emph{support} $\supp(\bx)$ of a vector $\bx=\left(x_1,\dots,x_n\right)\in\F_q^n$ is the 
set of indices  of the non-zero coordinates, i.e., $\supp(\bx):=\left\{1\le i\le n\,:\, x_i\neq 0\right\}$. With this, we have $\wt(\bc)=\#\supp(\bc)$ for each codeword $\bc\in C$.  
The number $\#C$ of codewords of $C$ is given by $q^k$. Given a basis $\bg^1,\dots,\bg^k\in\F_q^n$ of an $[n,k]_q$-code $C$ we call the matrix 
$$
  G=\begin{pmatrix}
    \bg^1\\ 
    \vdots\\
    \bg^k
  \end{pmatrix}
  =
  \begin{pmatrix}
    g_1^1 & g_2^1 & \dots & g_n^1 \\ 
    \vdots & \vdots & \ddots & \vdots \\ 
    g_1^k & g_2^k & \dots & g_n^k
  \end{pmatrix}\in\F_q^{k\times n}  
$$
a \emph{generator matrix} of $C$, where $\bg^i=\left(g_1^i,\dots,g_n^i\right)\in\F_q^n$ for all $1\le i\le k$. An example is given by
\begin{equation}
  \label{eq_generator_matrix_example}
  G=\begin{pmatrix}
    2 & 1 & 0 & 1 & 0 & 2 \\ 
    1 & 0 & 0 & 2 & 0 & 1 \\
    0 & 1 & 0 & 0 & 1 & 0
  \end{pmatrix}\in\F_3^{3\times 6},    
\end{equation}
where we denote the elements of $\F_p\cong \Z/p\Z$ by $\{0,1,\dots,p-1\}$ if the field size equals a prime $p$. If $q=p^m$ with 
$m>1$, then for each irreducible polynomial $f$ of degree $m$ over $\F_p$ we have $\F_q\cong \F_p[x]/(f)$. As representatives we 
choose polynomials of degree at most $m-1$ with coefficients in $\{0,1,\dots,p-1\}$.  
\begin{nexercise}
  Verify that each $[n,k]_q$-code admits $\prod_{i=0}^{k-1} \left(q^k-q^i\right)$ different ordered bases, i.e., different generator matrices.
\end{nexercise}
Applying any sequence of row operations of the Gaussian elimination algorithm to $G$ gives another generator matrix of $C$. For 
our example in (\ref{eq_generator_matrix_example}) the Gaussian elimination algorithm gives the generator matrix
$$
 \begin{pmatrix}
    1 & 0 & 0 & 2 & 0 & 1 \\ 
    0 & 1 & 0 & 0 & 0 & 0 \\
    0 & 0 & 0 & 0 & 1 & 0
  \end{pmatrix}.
$$ 

Let $\Aut\!\left(\F_q^n\right)$ be the group of semilinear transformations of $\F_q^n$ that leave the Hamming distance invariant. 
For each transformation $\mu\in\Aut\!\left(\F_q^n\right)$ we can find a permutation $\pi$ of the set $\{1,\dots, n\}$, non-zero 
field elements $a_i\in\F_q\backslash\{0\}$, where $1\le i\le n$, and a field automorphism $\alpha$ of $\F_q$ such that
\begin{equation}
  \mu\!\left(\left(x_1,\dots,x_n\right)\right)=\left(\alpha\!\left(a_1x_{\pi(1)}\right)\!,\alpha\!\left(a_2x_{\pi(2)}\right)\!,\dots, 
  \alpha\!\left(a_nx_{\pi(n)}\right)\right)
\end{equation}
for all $\left(x_1,\dots,x_n\right)\in\F_q^n$. Two codes $C,C'\subseteq \F_q^n$ are said to be \emph{equivalent} or \emph{isomorphic} 
if a transformation $\mu\in\Aut\!\left(\F_q^n\right)$ exists such that $\mu(C)=C'$. The \emph{automorphism group} $\Aut(C)$ 
of a code $C\subseteq \F_q^n$ is the group
\begin{equation}
  \Aut(C):=\left\{\mu\in\Aut\!\left(\F_q^n\right)\,:\, \mu(C)=C\right\}. 
\end{equation}
Note that for the binary field we only have to consider permutations of the set $\{1,\dots, n\}$ of coordinate positions. By applying 
row operations and column permutations we can conclude that each $[n,k]_q$-code $C$ is equivalent to a code $C'$ 
with a generator matrix $G'$ whose leftmost part is a $k\times k$ \emph{unit-matrix} $I_k$. Such a matrix $G'$ is 
called \emph{systematic generator matrix}. In our example, generated by the matrix in 
Equation~(\ref{eq_generator_matrix_example}), a systematic generator matrix is given by
\begin{equation}
 G'=
  \begin{pmatrix}
    1 & 0 & 0 & 1 & 2 & 0 \\ 
    0 & 1 & 0 & 0 & 0 & 0 \\
    0 & 0 & 1 & 0 & 0 & 0
  \end{pmatrix}.\label{eq_generator_matrix_example_2}
\end{equation}

Note that the third column of the generator matrix $G$ in (\ref{eq_generator_matrix_example}), or the sixth column of the generator matrix $G'$ in 
(\ref{eq_generator_matrix_example_2}), is the zero vector $\zv$. The number $\neff$ of non-zero column vectors 
in a generator matrix $G$ of an $[n,k]_q$-code $C$ is called the \emph{effective length} of $C$. By $\supp(C):=\cup_{\bc\in C}\supp(\bc)$ we denote the support of a 
code $C$, so that $\#\supp(C)=\neff(C)$. If $\neff=n$, then $C$ is also called \emph{spanning} or of \emph{full length}. In our example we have effective length $\neff=5$.

The \emph{minimum (Hamming) distance} of a linear code $C$ is given by
\begin{equation}
  d(C)=\min \{\dH(\bc,\bc')\,:\, \bc,\bc'\in C,\bc\neq \bc'\}=\min\{\wt(\bc)\,:\, \bc\in C\backslash\{\zv\}\}.
\end{equation}
An $[n,k,d]_q$-code is an $[n,k]_q$-code with minimum distance $d$. If the weights of all non-zero codewords of an $[n,k]_q$-code $C$ are contained in 
$W=\left\{w_1,\dots,w_l\right\}$, then we speak of an $[n,k,W]_q$-code. By $A_w(C)\in\N_0$ we denote the number of codewords of weight $w$ in $C$, where $0\le w\le n$. So, we have 
$A_w(C)=0$ for all $0< w<d(C)$. The sequence of all weights can be summarized in the \emph{homogeneous weight enumerator}
\begin{equation}
  \label{eq_homogeneous_weight_enumerator}
    \overline{W}_{C}(x,y)=\sum_{w=0}^n A_w(C)x^wy^{n-w}
\end{equation}
of $C$. Setting $y=1$ we obtain the \emph{weight enumerator}
\begin{equation}
  \label{eq_weight_enumerator}
  W_{C}(x)=\sum_{w=0}^n A_w(C)x^w.
\end{equation}
\begin{nexercise}
  Let $C$ and $C'$ be isomorphic codes. Verify $W_{C}(x)=W_{C'}(x)$ and $\overline{W}_{C}(x,y)=\overline{W}_{C'}(x,y)$, 
  so that particularly we have $d(C)=d(C')$. 
\end{nexercise}
\begin{nexercise}
  \label{exercise_remove_zero_columns}
  Let $C$ be an $[n,k]_q$-code and $C'$ be an $[n',k]_q$-code with $n'\le n$ that arises from $C$ by removing some all-zero coordinates. 
  Verify $W_{C}(x)=W_{C'}(x)$ and $\overline{W}_{C}(x,y)=\overline{W}_{C'}(x,y)$, so that in particular we have $d(C)=d(C')$.   
\end{nexercise}

The orthogonal complement
\begin{equation}
  C^\perp :=\left\{\by\in\F_q^n \,:\, \langle \bx,\by\rangle=0 \text{ for all }\bx\in C\right\}
\end{equation}
of an $[n,k]_q$-code $C$, with respect to the standard inner product
\begin{equation}
  \langle \bx,\by\rangle:=\sum_{i=1}^n x_iy_i
\end{equation}
is called the \emph{dual code} of $C$. Note that $C^\perp$ is an $[n,n-k]_q$-code and $\Aut(C)\cong\Aut\!\left(C^\perp\right)$, where $\Aut(C)$ denotes the group of monomial automorphisms of $C$\footnote{For $q\le 3$ we even have equality. Over $\F_4=\{0,1,\alpha,\alpha^2\}$ this is not true in general: For $C=\langle(1,\alpha)\rangle$ the monomial map $(x,y)\mapsto(y,\alpha^2x)$ is an automorphism of $C$ but not of $C^\perp=\langle(1,\alpha^2)\rangle$.} 
via $\mu\mapsto\left(\mu^\top\right)^{-1}$, where $\left\langle \mu(\bx),\by\right\rangle=\left\langle\bx,\mu^\top(\by)\right\rangle$ for all $\bx,\by\in \F_q^n$ and all 
$\mu\in\Aut\!\left(\F_q^n\right)$. The \emph{dual minimum distance} $d^\perp$ is the minimum distance of the 
dual code. We call a linear code \emph{projective} iff $d^\perp\ge 3$. 
\begin{nexercise}
  Let $C$ be an $[n,k]_q$-code. Prove that $d^\perp(C)=1$ iff the effective length of $C$ is strictly smaller than $n$. Moreover, we have $d^\perp(C)=2$ if a generator 
  matrix $G$ of $C$ does not contain a zero column but two linearly dependent columns. 
\end{nexercise}
If $C\subseteq C^\perp$, then $C$ is called \emph{self-orthogonal} and \emph{self-dual} if $C=C^\perp$.
\begin{nexercise}
  Show that
  \begin{enumerate}
    \item[(a)] every binary self-orthogonal linear code is even;
    \item[(b)] every doubly-even binary linear code is self-orthogonal;
    \item[(c)] every self-dual ternary linear code is $3$-divisible.
  \end{enumerate}  
\end{nexercise}
The $[8,4,4]_2$-code of Example~\ref{ex_extended_hamming_8} is doubly-even, so that it is self-orthogonal by part (b). Since its dimension is 
$4=8/2$, it is even self-dual, i.e., it coincides with its dual code and has the same weight enumerator.

For larger fields divisibility still implies self-orthogonality in many cases. If $q$ is a power of an odd prime $p$, then every $p$-divisible code over $\F_q$ that contains 
the all-one vector is self-orthogonal \cite{li2024selforthogonal}. For even characteristic, every $4$-divisible code over $\F_{2^m}$ that contains the all-one vector is 
self-orthogonal \cite{heng2026euclidean}, and every $p$-divisible code over $\F_{q^2}$ with $q=p^m$ that contains the all-one vector is self-orthogonal 
with respect to the Hermitian inner product \cite{heng2026euclidean}.

\section{The MacWilliams Equations and the Linear Programming Method}
\label{subsec_lp_method}
The homogeneous weight enumerator $\overline{W}_{C}(x,y)$ of a linear code $C$ over $\F_q$ and the homogeneous weight enumerator $\overline{W}_{C^\perp}(x,y)$ 
of its dual code $C^\perp$ are related by the so-called \emph{MacWilliams identity} \cite{macwilliams1977theory} 
\begin{equation}
  \label{eq_macwilliams_pol_identity_homogeneous} 
  \overline{W}_{C^\perp}(x,y)= \left|C\right|^{-1}\cdot  \overline{W}_{C}(y-x,y+(q-1)x).
\end{equation}
So, given the complete \emph{weight distribution} $\left(A_i\right)$ of $C$, the 
weight distribution $\left(B_i\right)$ of the dual code $C^\perp$ with $B_i(C)=A_i(C^\perp)\in\N_0$ is uniquely determined.   
We have
\begin{equation}
  \label{eq_macwilliams}
  \sum_{j=0}^{n-i} {{n-j}\choose i} A_j=q^{k-i}\cdot \sum_{j=0}^i {{n-j}\choose{n-i}} B_j
\end{equation}
for all $0\le i\le n$, see e.g.\ \cite[Lemma 2.2]{macwilliams1963theorem}. If we restrict the range of $i$ to $0\le i<t$, then we 
speak of the \emph{first $t$ MacWilliams equations}. Solving the equation system for the $B_i$ gives:
\begin{ntheorem}{(MacWilliams Equations, see \cite{macwilliams1977theory})}\\
  For an $[n,k]_q$-code $C$ we have 
  \begin{equation}
    \label{eq_macwilliams_kr}
      \sum_{j=0}^n K_i(j)A_j(C)=q^kB_i(C)
  \end{equation}  
  for $0\le i\le n$, where
  \begin{equation}
    \label{eq_krawtchouck}
    K_i(j):=\sum_{s=0}^i (-1)^s {{n-j}\choose{i-s}}{j\choose s}(q-1)^{i-s}
  \end{equation}
  are the \emph{Krawtchouk polynomials} (here $j$ is considered as variable of a polynomial).
\end{ntheorem} 
There are lots of ways how to state the MacWilliams Equations. Another common representation are the so-called \emph{power moments} \cite{pless1963power}. 
For the binary case $q=2$ and the first five MacWilliams equations they are spelled out in:
\begin{nexercise}
  \label{exercise_power_moments_q_2_t_5}
  The weight distributions $\left(A_i\right)$ and $\left(B_i\right)$ of an $[n,k]_2$-code and its dual code satisfy
  \begin{eqnarray}
    \sum_{i=1}^n A_i &=& 2^k-1 \label{eq_pm_0}\\
    \sum_{i=1}^n iA_i &=& 2^{k-1}\left(n-B_1\right)\label{eq_pm_1}\\
    \sum_{i=1}^n i^2 A_i &=& 2^{k-1}\left(B_2-nB_1+n(n+1)/2\right) \label{eq_pm_2}
   \end{eqnarray}
   \begin{eqnarray}   
    \sum_{i=1}^n i^3 A_i &=& 2^{k-2}\left(3(B_2n\!-\!B_3)-(3n^2\!+\!3n\!-\!2)/2\cdot B_1 +n^2(n\!+\!3)/2\right)\label{eq_pm_3}\\
    \sum_{i=1}^n i^4 A_i &=& 2^{k\!-\!4}\big(4!(B_4\!-\!nB_3)\!+\!4(3n^2\!+\!3n\!-\!4)B_2\!-\!4(n^3\!+\!3n^2\!-\!2n)B_1\!\notag\\ 
               && +\!(n^4\!+\!6n^3\!+\!3n^2\!-\!2n)\big).\label{eq_pm_4}
  \end{eqnarray}  
\end{nexercise} 

In our context we have several additional conditions on the $A_i$ and $B_i$. First note that we have $A_0=B_0=1$ in general, $B_1=0$ iff the 
code is of full length, and $B_1=B_2=0$ iff the code is projective. $\Delta$-divisibility implies $A_i=0$ for all $i\in\N_{>0}$ with $i\not\equiv 0\pmod\Delta$. 
Via residual codes, see Lemma~\ref{lemma_heritable} and the discussion thereafter, and non-existence results for the effective lengths of divisible codes, 
see Section~\ref{sec_lengths_of_divisible_codes} and Section~\ref{sec_lengths_projective_q_r}, we can also exclude additional weights in many situations. 
Since the $A_i$ and $B_i$ are counts, they are integral. Moreover, the fact that scalar multiples of codewords are codewords again implies that also 
$A_i/(q-1)$ and $B_i/(q-1)$ are integers for all $i\ge 1$, cf.\ Exercise~\ref{exercise_weight_distribution_vs_spectrum}. More sophisticated extra conditions are discussed in 
Section~\ref{sec_improved_lp_method}.

\medskip

\begin{trailer}{(Integer) Linear programming method}By {\lq\lq}the{\rq\rq} \emph{linear programming method} (for linear codes) we understand the application 
of linear programs certifying the non-existence of linear codes, cf.~\cite{delsarte1973algebraic}. In general, a \emph{linear program} consists of a set of 
real \emph{variables} $x_i$, where some of them  
may be assumed to be non-negative, and a set of linear non-strict \emph{constraints}, i.e., {\lq\lq}$\le${\rq\rq}, {\lq\lq}$\ge${\rq\rq}, or {\lq\lq}$=${\rq\rq}. 
Additionally, there is a linear \emph{target function} that is either maximized or minimized. Specially structured forms are e.g.\ given by
$$
  \max\left\{\bc^\top \bx\,:\, A\bx\le \bb, \bx\ge 0\right\}
$$  
or
$$
  \max\left\{\bc^\top \bx\,:\, A\bx\le \bb, D\bx=\bd, \bx\ge 0\right\}.
$$
We remark that every linear program, LP for short, can be reformulated into e.g.\ the first form, possibly including a change of variables. Those linear programs can be solved 
efficiently in terms of the number of variables, the number of constraints, and the order of magnitude of the occurring coefficients. Choosing  
$\bc=\zv$ we can treat the question whether a linear inequality system admits a solution as an optimization problem. We say that an LP is \emph{infeasible} if there exists no solution 
satisfying all constraints. If some of the variables are assumed to be integral, 
we speak of an integer linear program (ILP) for short. While LPs can be solved in polynomial time, solving ILPs is NP-hard.

In our context we choose the $A_i,B_i$ as variables and the MacWilliams equations as constraints. Also the mentioned additional conditions can be formulated in this setting. 
Of course the length $n$, the dimension $k$, and the field size $q$ have to be specified. If such an LP does not admit a real-valued solution we say that the non-existence 
of a linear code with corresponding parameters is certified by the linear programming method. If we assume $A_i/(q-1)$ and $B_i/(q-1)$ to be integers for all $i\ge 1$, then we speak of the 
integer linear programming method (for linear codes). Of course, this setting allows a lot of variations, so that there is no precise definition of {\lq\lq}the{\rq\rq} 
(integer) linear programming method for linear codes.   
\end{trailer}
For more details on the application of linear programming in coding theory we refer to e.g.\ \cite{bierbrauer2016introduction}.

\medskip

\begin{warning}{Coefficients of LPs for linear codes can grow very quickly}Even for moderate parameters the coefficients of the Krawtchouk polynomials, 
see Equation~(\ref{eq_krawtchouck}), can grow very quickly. This causes severe numerical problems when computing with limited precision. Note that while the 
coefficients in e.g.\ Equation~(\ref{eq_macwilliams}) are a bit smaller, this advantage is quickly used up when a solution algorithm has performed some changes of basis. 
Some implementations with unbounded precision are available, see e.g.\ the computer algebra system \texttt{Maple} or the non-commercial solver \texttt{SCIP} for mixed 
integer programming. However, computation times significantly increase when using long number arithmetic.   
\end{warning}
In order to keep the number of constraints small and to partially avoid the mentioned numerical issues we will mainly use the first $t$ MacWilliams equations only, where 
$t$ is rather small. Based on experimental evidence we remark that choosing $t\in \{3,4,5\}$ gives the same implication on non-existence as larger values of $t$ in almost 
all cases. 
\begin{nexample}
  \label{example_no_8_div_52_10_2_code}
  No projective $8$-divisible $[52,10]_2$-code exists since solving the first four MacWilliams equations for $\left\{A_8,A_{16},A_{24},A_{32}\right\}$ gives 
  \begin{eqnarray*}
    A_8    &=& 10+A_{40}+4A_{48}+\frac{1}{4}B_3\\
    A_{16} &=& -28 -4A_{40}-15A_{48}-\frac{3}{4}B_3 \\ 
    A_{24} &=& 790+6A_{40}+20A_{48}+\frac{3}{4}B_3 \\ 
    A_{32} &=& 251 -4A_{40}-10A_{48}-\frac{1}{4}B_3,
  \end{eqnarray*}
  so that $A_{16}\le -28<0$, which is a contradiction.
\end{nexample}  
By Lemma~\ref{lemma_picture_q_2_r_2} no projective $4$-divisible binary linear codes of lengths $4$ or $12$ exist, so that we may additionally use $A_{48}=0$ and $A_{40}=0$. 
\begin{nexample}
  Let $C$ be a projective $[41,6]_2$-code whose non-zero weights are contained in $\{20,24,26,40\}$. Here, the first four MacWilliams equations imply 
  \begin{eqnarray*}
    B_3 &=& \frac{470}{3} - \frac{280}{3}A_{40}\\
    A_{20} &=& \frac{158}{3} - \frac{28}{3}A_{40} \\ 
    A_{24} &=& 5 + 35 A_{40} \\ 
    A_{26} &=& \frac{16}{3} - \frac{80}{3}A_{40}.
  \end{eqnarray*}
  However, $A_{26}\ge 0$ yields $A_{40}=0$, so that $A_{26}=\tfrac{16}{3}\notin\N_0$, which is a contradiction.
\end{nexample}  
The context of that example is that for field size $q=2$, dimension $k=6$, and minimum distance $d=20$ the \emph{Griesmer bound} is not attained, see e.g.\ \cite{baumert1973note}.
\begin{nexercise}
  Prove that an even $[41,6,20]_2$-code is projective and has non-zero weights in $\{20,24,26,40\}$ only.
\end{nexercise}   

In general we can determine lower and upper bounds for any linear combination of the $A_i$ and $B_i$ by using some subset of the MacWilliams equations. Adding integer 
rounding cuts sometimes gives tighter bounds:
\begin{nexample}
  \label{ex_even_13_5_6_2_code}
  In this example we want to show that each even $[13,5,6]_2$-code satisfies $B_1=0$, $B_2=0$, $2\le B_3\le 4$, $23 \le A_6\le 24$, $3\le A_8\le 6$, 
  $1\le A_{10}\le 4$, and $0\le A_{12}\le 1$. To this end we consider the following linear program 
  based on the first four MacWilliams equations:
  \begin{eqnarray*}
    \max B_1 & &\text{subject to} \\
    A_6 + A_8 + A_{10} + A_{12} &=& 31\\ 
    6A_6 + 8A_8 + 10A_{10} + 12A_{12} +16B_1 &=& 208 \\ 
    36A_6 + 64A_8 + 100A_{10} + 144A_{12} +208B_1-16B_2 &=& 1456 \\ 
    216A_6 + 512A_8 + 1000A_{10} + 1728A_{12} +2176B_1 -312B_2 + 24B_3 &=& 10816. 
  \end{eqnarray*}
  The (unique) optimal solution, computed with \texttt{Maple}, is given by
  $$
    B_1 = \frac{3}{8}, B_2 = 0, B_3 = 0, A_6 = \frac{109}{4}, A_8 = 0, A_{10} = \frac{13}{4}, A_{12} = \frac{1}{2}
  $$
  so that, in general, $B_1\le\left\lfloor\tfrac{3}{8}\right\rfloor=0$, i.e., we can assume $B_1=0$.  
  With this additional equation, maximizing $B_2$, $B_3$, $A_6$, $A_8$, $A_{10}$, and $A_{12}$ gives $B_2\le\left\lfloor\frac{18}{17}\right\rfloor=1$, 
  $B_3\le 4$, $A_6\le\left\lfloor\frac{437}{17}\right\rfloor=25$, $A_8\le 7$, $A_{10}\le\left\lfloor\frac{11}{2}\right\rfloor=5$, and 
  $A_{12}\le \left\lfloor\tfrac{20}{13}\right\rfloor=1$, respectively. Adding the tightened upper bounds, i.e., those for $B_2$, $A_6$, $A_{10}$, and $A_{12}$, 
  maximizing $B_2$ again yields $B_2\le \left\lfloor\frac{6}{7}\right\rfloor=0$, so that $B_2=0$. Another iteration yields $B_3\le 4$, $A_6\le 24$, $A_8\le 6$, 
  $A_{10}\le 4$, and $A_{12}\le 1$. Similarly we obtain $B_3\ge 2$, $A_6\ge 23$, $A_8\ge 3$, 
  $A_{10}\ge 1$, and $A_{12}\ge 0$ by minimizing the variables. All these final lower and upper bounds for the variables can indeed be attained  
  as shown in the subsequent example.     
\end{nexample}  

\begin{nexample}
  \label{example_even_13_5_6_2_code_enumeration}
  The non-negative integral solutions $\left(B_1,B_2,B_3,A_6,A_8,A_{10},A_{12}\right)$ of the first four MacWilliams equations 
  of an even $[13,5,6]_2$-code are given by
  $$
    \left(0, 0, 4, 24, 3, 4, 0\right)\text{ and }
    \left(0, 0, 2, 23, 6, 1, 1\right)\!.
  $$
  To this end we solve the four equations for $\left\{B_3,A_6,A_8,A_{10}\right\}$:
  \begin{eqnarray*}
    B_3 &=& 4 - 2A_{12} -8B_1 - 3B_2\\  
    A_6 &=& 24 - A_{12} + 10B_1 + 2B_2\\ 
    A_8 &=& 3 + 3A_{12} -12 B_1 - 4B_2\\ 
    A_{10} &=& 4 - 3A_{12} +2B_1 + 2B_2.
  \end{eqnarray*}
  From $B_3\ge 0$ we conclude $B_1=0$ and $B_2\in\{0,1\}$. If $B_2=1$, then $B_3\ge 0$ implies $A_{12}=0$, so that $A_8=-1<0$. Thus, we have $B_2=0$ and 
  $A_{10}\ge 0$ implies $A_{12}\in\{0,1\}$, which gives the two solutions stated above.  
  The MacWilliams transforms of the corresponding weight distributions $\left(A_i\right)_{0\le i\le 13}$ are given by
  $$
  \left(B_i\right)_{0\le i\le 13}=(1,0,0, 4, 30, 57, 36, 36, 57, 30, 4, 0, 0, 1)
  $$
  and
  $$
    \left(B_i\right)_{0\le i\le 13}=(1,0,0, 2, 40, 39, 46, 46, 39, 40, 2, 0, 0, 1).
  $$ 
  Of course, the latter does not show that both such codes exist, but it shows that we cannot conclude a contradiction using the linear programming method with all MacWilliams 
  equations.
\end{nexample}

While the above examples indicate that one eventually has to deal with a few details in the computations, we would like to remark that it is always possible to hide the linear  
programming computations in mathematical non-existence proofs:
\begin{nexercise}
  Use some arbitrary textbook on linear programming in order to show the following facts: 
  \begin{itemize}
    \item \emph{Farkas' lemma} or \emph{the Fourier–Motzkin elimination} algorithm yield a constructive certificate for the infeasibility of an LP or a linear 
          inequality system, respectively. 
    \item The \emph{LP duality theorem} and the solution of the \emph{dual linear program} can be used to compute multipliers for the constraints of the original LP  
          whose (scaled) sum gives a tight bound for the optimum value of the target function or shows infeasibility if a given feasibility problem is reformulated as the 
          minimization of the violation of the constraints.   
  \end{itemize}
\end{nexercise}
\begin{nexample}
  \label{example_multipliers}
  The first four MacWilliams equations for a projective $8$-divisible $[52,9]_2$-code are given by
  {\footnotesize
  \begin{eqnarray}
    A_8 + A_{16} + A_{24} + A_{32} + A_{40} + A_{48} &=& 511 \label{eq_mw_ex_1}\\ 
   44 A_8 + 36 A_{16} + 28 A_{24} + 20 A_{32} + 12 A_{40} + 4 A_{48} &=& 13260 \label{eq_mw_ex_2}\\
 946 A_8 + 630 A_{16} + 378 A_{24} + 190 A_{32} + 66 A_{40} + 6 A_{48} &=& 168402 \label{eq_mw_ex_3}\\ 
 13244 A_8 + 7140 A_{16} + 3276 A_{24} + 1140 A_{32} + 220 A_{40} + 4 A_{48} &=&  1392300 + 64 B_3, \label{eq_mw_ex_4}\\\notag
  \end{eqnarray}}
  so that a linear program for the minimization of the violation reads
  {\footnotesize\begin{eqnarray*}
    \min x & &\text{subject to} \\
   A_8 + A_{16} + A_{24} + A_{32} + A_{40} + A_{48} +x &\ge& 511 \\
   A_8 + A_{16} + A_{24} + A_{32} + A_{40} + A_{48} -x &\le& 511 \\
   44 A_8 + 36 A_{16} + 28 A_{24} + 20 A_{32} + 12 A_{40} + 4 A_{48} +x &\ge& 13260 \\
   44 A_8 + 36 A_{16} + 28 A_{24} + 20 A_{32} + 12 A_{40} + 4 A_{48} -x &\le& 13260 \\
 946 A_8 + 630 A_{16} + 378 A_{24} + 190 A_{32} + 66 A_{40} + 6 A_{48} +x &\ge& 168402 \\
 946 A_8 + 630 A_{16} + 378 A_{24} + 190 A_{32} + 66 A_{40} + 6 A_{48} -x &\le& 168402 \\
 -64 B_3+ 13244 A_8 + 7140 A_{16} + 3276 A_{24} + 1140 A_{32} + 220 A_{40} + 4 A_{48} +x &\ge& 1392300 \\
 -64 B_3+ 13244 A_8 + 7140 A_{16} + 3276 A_{24} + 1140 A_{32} + 220 A_{40} + 4 A_{48} -x &\le& 1392300. \\
  \end{eqnarray*}}
  Numbering the dual variables corresponding to the constraints of the above LP by $c_1,\dots,c_8$, the optimal solution of the corresponding dual LP 
  is given by $c_2=-0.919540$, $c_3=0.077176$, and $c_6=-0.003284$ with an optimal target value of $0.42036124795$. Using a suitable continued fractions approximation 
  we obtain the multipliers $m_1 := -\tfrac{80}{87}$, $m_2 := \tfrac{47}{609}$, and $m_3 := -\tfrac{2}{609}$ (as rational approximations for the floating-point values of 
  $c_2$, $c_3$, and $c_6$, respectively). With this, $m_1$ times Equation~(\ref{eq_mw_ex_1}) plus $m_2$ times Equation~(\ref{eq_mw_ex_2}) plus $m_3$ times Equation~(\ref{eq_mw_ex_3}) 
  gives
  $$
    -\frac{128}{203} A_8-\frac{128}{609} A_{16}-\frac{128}{609}A_{40}-\frac{128}{203}A_{48}=\tfrac{256}{609}>0,   
  $$ 
  which is a contradiction since $A_8,A_{16},A_{40},A_{48}\ge 0$.  
\end{nexample}
In a mathematical proof we may just state the multipliers without justification or details of their computation. This also allows us to give 
rigorous conclusions from numerical computations, i.e., compute multipliers with limited numerical precision, round them to some reasonably close rationals, 
and verify the final inequality with exact arithmetic, cf.\ Example~\ref{example_multipliers}.   

In Section~\ref{sec_nonexistence_projective_q_r} we will draw several analytical conclusions from the linear programming method that do not rely on floating-point computations 
at all.  

\section{Geometric description of linear codes}
\label{subsec_geometric_description}
Our next aim is to briefly describe linear codes from a geometric point of view. For further details we refer the interested reader to 
\cite{dodunekov1998codes}. So, let $V\simeq \F_q^v$ be a $v$-dimensional vector space 
over $\F_q$. We call each $i$-dimensional subspace of $V$ an \emph{$i$-space}. As a shorthand, we use the geometric terms \emph{points}, 
\emph{lines}, \emph{planes}, and \emph{hyperplanes} for $1$-, $2$-, $3$-, and $(v-1)$-spaces, respectively. A $(v-j)$-space is also called 
a (sub-)space of \emph{codimension $j$}, where $0\le j\le v$. In the special case of a space of codimension $2$, i.e., a $(v-2)$-space, we also 
speak of \emph{hyperlines}. Since two different $1$-dimensional subspaces generate a unique $2$-dimensional subspace, two different points 
are on exactly one common line, which partially motivates the use of the geometric language. Here we use the \emph{algebraic dimension} 
and not the \emph{geometric dimension}, which is one less.\footnote{Points are $0$-dimensional geometric objects and lines are $1$-dimensional 
geometric objects, while we prefer to say that $1$-spaces have (algebraic) dimension $1$ and $2$-spaces have (algebraic) dimension $2$.} The 
only exception is the notion of the $(v-1)$-dimensional \emph{projective geometry} $\PG(v-1,q)$ associated with $\F_q^v$. There are $v-1$ types of
 geometric objects ranging from points ($1$-spaces) to hyperplanes ($(v-1)$-spaces). By $\cP$ we denote the set of points and by $\cH$ we denote the 
 set of hyperplanes whenever the dimension $v$ of the ambient space and the field size $q$ are clear from the context. Each point $P\in \cP$ can be written as a
$1$-space
$$
  P=\left\langle \begin{pmatrix}x_1\\\vdots\\x_v\end{pmatrix}\right\rangle_q\!,
$$ 
where $\left(x_1,\dots,x_v\right)\in\F_q^v\backslash\{\zv\}$, or using projective coordinates $\left(x_1:x_2:\dots:x_v\right)$, where 
$$\left(tx_1:tx_2:\dots:tx_v\right)=\left(x_1:x_2:\dots:x_v\right)$$ for all $t\in\F_q\backslash\{0\}$. Since the orthogonal complement 
of a $(v-1)$-space is a $1$-space, we have similar notations for hyperplanes.

\begin{trailer}{Number of subspaces}By $\spaces{V}{k}$ we denote the set of all $k$-spaces in $V$ and by $\qbin{v}{k}{q}$ their cardinality $\#\spaces{V}{k}$. For integers 
$0\le k\le v$ we have
\begin{equation}
  \label{eq_qbin}
  \qbin{v}{k}{q}=\prod_{i=0}^{k-1} \frac{q^{v-i}-1}{q^{k-i}-1}.
\end{equation}
For other values of $k$ we set $\qbin{v}{k}{q}=0$ by convention.
\end{trailer}
\begin{nexercise}
  Prove Equation~(\ref{eq_qbin}) by counting ordered bases of subspaces.
\end{nexercise}
Using the notation $[v]_q:=\tfrac{q^v-1}{q-1}$ and $[v]_q!:=\prod_{i=1}^v [i]_q$  we can write 
\begin{equation}
  \label{eq_qbin_factorial}
  \qbin{v}{k}{q}=\frac{[v]_q!}{[k]_q!\cdot[v-k]_q!}, 
\end{equation}
which motivates that the numbers $\qbin{v}{k}{q}$ are also called \emph{$q$-binomial} or \emph{Gaussian binomial coefficients}. As they count the number of $k$-spaces contained in 
a $v$-space, they are a \emph{$q$-analog} of the binomial coefficients ${v\choose k}$ which count the number of $k$-sets 
contained in a $v$-set. Here, a $t$-set is a set of cardinality $t$ and we have $\lim\limits_{q\to 1} \qbin{v}{k}{q}={v\choose k}$. An important instance of Equation~(\ref{eq_qbin})
is given by 
\begin{equation}
  \#\cP=\qbin{v}{1}{q}=\qbin{v}{v-1}{q}=\#\cH=\frac{q^v-1}{q-1}=[v]_q.
\end{equation}
\begin{nexercise}
  Verify $\lim\limits_{q\to 1} \qbin{v}{k}{q}={v\choose k}$, 
  $$
    \qbin{v}{k}{q}=\qbin{v}{v-k}{q}
    \quad\text{and}\quad
    \qbin{v}{k}{q}=q^k \qbin{v-1}{k}{q}+\qbin{v-1}{k-1}{q}=\qbin{v-1}{k}{q}+q^{v-k}\qbin{v-1}{k-1}{q}.
  $$
\end{nexercise}

\begin{trailer}{Multisets of points}A \emph{multiset} $\cM$ of points in $\PG(v-1,q)$ is a mapping $\cM\colon\cP\to\N_0$. For each point $P\in\cP$ the integer $\cM(P)$ is called the 
\emph{multiplicity} of $P$ and it counts how often point $P$ is contained in the multiset. If $\cM(P)\in\{0,1\}$ for all $P\in\cP$ we also speak of a set instead 
of a multiset (of points). 
\end{trailer}
\begin{nexample}
  \label{ex_multiset}
  For the list of points
  $$
    \left\langle\begin{pmatrix}2\\1\\0\end{pmatrix}\right\rangle,
    \left\langle\begin{pmatrix}1\\0\\1\end{pmatrix}\right\rangle,
    \left\langle\begin{pmatrix}1\\2\\0\end{pmatrix}\right\rangle,
    \left\langle\begin{pmatrix}0\\0\\1\end{pmatrix}\right\rangle,
    \left\langle\begin{pmatrix}2\\1\\0\end{pmatrix}\right\rangle
  $$
  in $\PG(2,3)$ a representation as a multiset $\cM$ is given by $\cM(\langle(1,0,1)^\top\rangle)=1$, $\cM(\langle(1,2,0)^\top\rangle)=3$, $\cM(\langle(0,0,1)^\top\rangle)=1$, and 
  $\cM(P)=0$ for all other points $P$ in $\cP$. Note that $\langle(2,1,0)^\top\rangle=\langle(1,2,0)^\top\rangle$. 
\end{nexample}
  For the ease of a canonical representation of a point  
$\langle(x_1,\dots,x_v)^\top\rangle$ we will assume that the first non-zero value $x_i\in\F_q$ is equal to $1$. The \emph{cardinality} of the multiset is 
defined as
\begin{equation}
  \#\cM=\sum_{P\in\cP} \cM(P).
\end{equation}
More generally, we set $\cM(\cQ):=\sum_{P\in \cQ} \cM(P)$ for each subset $\cQ\subseteq \cP$, i.e., we extend the mapping $\cM$ additively. For each 
subspace $S$ in $\PG(v-1,q)$ we also use the notation $\cM(S)$ interpreting the points in $S$ as a subset of $\cP$. We also write $\cP\backslash S$ for the 
set of points that are not contained in a subspace $S$. Choosing $S=\PG(v-1,q)$ we have $\cM(S)=\#\cM$, i.e., another expression for the cardinality of $\cM$. 
We say that a subspace $S$ is \emph{empty} (with respect to $\cM$) if $\cM(S)=0$ and we say that $\cM$ is \emph{empty} if $\#\cM=0$.  
We also extend the notion of multiplicity from points to arbitrary subsets $\cQ\subseteq\cP$. For $i$-spaces $\cQ$ 
of multiplicity $m$ we speak of $m$-points, $m$-lines, $m$-planes, or $m$-hyperplanes in the cases where $i=1$, $i=2$, $i=3$, or 
$i=v-1$, respectively. The \emph{support} $\supp(\cM)$ of a multiset of points $\cM$ is the set of points of strictly positive multiplicity. We call $\cM$ 
\emph{spanning} if the $1$-spaces in $\supp(\cM)$ span $\F_q^v$. In other words, if no hyperplane has multiplicity $\#\cM$. In Example~\ref{ex_multiset} 
we have $\#\cP=13$, $\#\cM=5$, and the support of $\cM$ has cardinality $3$. Moreover, $\cM$ is spanning.

\begin{trailer}{Main correspondence between linear codes and multisets of points}Now we describe the main correspondence between $[n,k]_q$-codes $C$ 
with effective length $n$ and spanning multisets $\cM$ of 
points in $\PG(k-1,q)$. Let $G$ be an arbitrary generator matrix of $C$. Due to the condition on the effective length $\neff$ of $C$, the matrix $G$ 
does not contain a zero column. So, we can construct a multiset of points $P_1, \dots, P_n$ in $\PG(k-1,q)$ by assigning to each column $\bx\in\F_q^k$ 
of $G$ the point $\langle \bx\rangle_q\in \cP$. In the  other direction we can use a generator $\bx\in\F_q^k\backslash \{\zv\}$ of each point $\langle \bx\rangle_q$ 
of the multiset as a column, with the corresponding multiplicity, of a generator matrix $G$ of $C$. 
\end{trailer}

This geometric description allows us to read off the code parameters from the multiset $\cM$ of points 
in $\PG(k-1,q)$. The subsequent theorem shows how to determine the weight distribution of $C$ from $\cM$. To this end, 
we observe that the codewords of $C$ are the $\F_q$-linear combinations of the rows of a generator matrix $G$ of $C$. Let 
$\bg^i=\left(g^i_1,\dots,g^i_n\right)\in\F_q^n$ denote the $i$th row of $G$, so that each codeword $\bc\in C$ has the form
$\bc=h_1g^1+h_2g^2+\dots+h_kg^k$ and is uniquely determined by $\bh=\left(h_1,\dots,h_k\right)\in\F_q^k$. For a fixed coordinate 
$1\le j\le n$, corresponding to the point $P_j$, the vector $\bc$ has entry $0$ in coordinate $j$ exactly if
\begin{equation}
  c_j\,=\,h_1g^1_j+h_2g^2_j+\dots +h_kg^k_j \,=\, 0.
\end{equation} 
The coefficients $h_i$, collected in $\bh$, of this linear equation define a hyperplane $H\in \cH$. In other words, we 
have $c_j=0$ iff the point $P_j$ is contained in the hyperplane $H$. The above reasoning implies the following correspondence between linear codes and multisets of points:
\begin{ntheorem}
  \label{thm_correspondence_codes_multisets}
  Let $C$ be a spanning $[n,k]_q$-code, $G$ be a generator matrix of $C$, and $\cM$ be the corresponding 
  multiset of points in $\PG(k-1,q)$ (as described above). For each non-zero $\bh=\left(h_1,\dots,h_k\right)\in\F_q^k\backslash\{\zv\}$ 
  let $\bh^\perp$ characterize the hyperplane $H\in \cH$, which consists of all $\by=\left(y_1,\dots,y_k\right)$ with $\langle \bh,\by\rangle=0$. 
  Then, the weight of the codeword $\bc=\sum_{i=1}^k h_i\bg^i$ is given by
  \begin{equation}
    \label{eq_weight_multiplicity_relation}
    \wt(\bc)=\sum_{P\in\cP, P\notin H} \cM(P)=\cM(\cP\backslash H)=n-\cM(H).
  \end{equation}
  The minimum Hamming distance is given by
  \begin{equation}
    \label{eq_min_dist_hyperplane}
    d(C)=\min\{\cM(\cP\backslash H)\,:\, H\in\cH\}=n-\max\{\cM(H)\,:\,H\in\cH\}.
  \end{equation}
\end{ntheorem} 
In other words, the weight $\wt(\bc)$ of a codeword $\bc\in C$ equals the number of points of $\cM$ (counted with multiplicities) that are not contained 
in the hyperplane $H=\bh^\perp$ associated to $\bc$. We remark that if we start with a (non-empty) multiset $\cM$ of points in 
$\PG(k-1,q)$, then the corresponding code $C$ has dimension $k$ iff $\cM$ is spanning. The rank of the constructed matrix $G$ 
would be strictly smaller than $k$ otherwise.

We call two multisets of points isomorphic if the corresponding codes are. The set of automorphisms of $\PG(v-1,q)$ preserving the $\le$-ordering of 
subspaces is given by the natural action of $\operatorname{P\Gamma L} (\F_q^v)$ if $v\ge 3$. This famous result is called {\lq\lq}Fundamental Theorem of Projective Geometry{\rq\rq}, 
see e.g.\ \cite{artin2016geometric}.

Some codes have very nice descriptions using the geometric language. 
\begin{nexample}
  \label{ex_simplex_code}
   Let $\cM$ be the (multi-)set in $\PG(k-1,q)$ defined by $\cM(P)=1$ for all $P\in\cP$, where $k\ge 2$. It corresponds to the (projective) $\left[[k]_q,k,q^{k-1}\right]_q$ \emph{simplex code}. 
   The minimum distance follows from the fact that each hyperplane $H\in \cH$ contains $[k-1]_q=\tfrac{q^{k-1}-1}{q-1}$ points from $\cP$, so that
   $\cM(\cP\backslash H)=[k]_q-[k-1]_q=q^{k-1}$. Since the weights of all non-zero codewords are equal to $q^{k-1}$, the code is $q^{k-1}$-divisible. For $q=2$ and $k=3$ see Figure~\ref{fig_fano_plane}.
\end{nexample}
\begin{figure}[htbp]
  \centering
  \begin{tikzpicture}[scale=1.0,
      pt/.style={circle,fill=black,inner sep=1.7pt},
      hp/.style={circle,fill=blue!70!black,inner sep=2.3pt},
      lab/.style={font=\footnotesize,inner sep=1.5pt}]
    \coordinate (A) at (0,0);
    \coordinate (B) at (4,0);
    \coordinate (C) at (2,3.464);
    \coordinate (D) at (2,0);
    \coordinate (E) at (3,1.732);
    \coordinate (F) at (1,1.732);
    \coordinate (G) at (2,1.155);
    \draw[black!60] (A)--(B)--(C)--cycle;
    \draw[black!60] (A)--(E);
    \draw[black!60] (B)--(F);
    \draw[black!60] (G) circle (1.155);
    \draw[blue!70!black,very thick] (C)--(D);
    \foreach \x in {A,B,E,F} \node[pt] at (\x) {};
    \foreach \x in {C,D,G} \node[hp] at (\x) {};
    \node[lab,below left] at (A) {$(1,0,0)$};
    \node[lab,below right] at (B) {$(0,1,0)$};
    \node[lab,above] at (C) {$(0,0,1)$};
    \node[lab,below] at (D) {$(1,1,0)$};
    \node[lab,right=2pt] at (E) {$(0,1,1)$};
    \node[lab,left=2pt] at (F) {$(1,0,1)$};
    \node[lab,fill=white,fill opacity=0.85,text opacity=1,anchor=west] at ($(G)+(0.12,-0.28)$) {$(1,1,1)$};
    \node[lab,align=left,anchor=west] at (5.5,2.0) {hyperplane $H$:\\[1pt] $x_1+x_2=0$\\[1pt] $\cM(H)=3$\\[1pt] $\wt(\bc)=7-3=4$};
  \end{tikzpicture}
  \caption{The simplex code for $q=2$ and $k=3$: the Fano plane with a line $H$ highlighted.}
  \label{fig_fano_plane}
\end{figure}

\begin{trailer}{Divisibility for multisets of points}We call a multiset of points $\Delta$-divisible iff the corresponding linear code $C$ is 
$\Delta$-divisible. Note that this is equivalent to 
\begin{equation}
  \label{eq_divisible_multiset}
  \cM(H)\equiv \#\cM\pmod \Delta
\end{equation}
for all hyperplanes $H\in\cH$ if $\cM$ is in $\PG(v-1,q)$ with $v\ge 2$. If $v=1$, then $\dim(C)=1$ and the condition is equivalent to $\#\cM\equiv 0\pmod \Delta$.\footnote{Note 
that if a multiset of points $\cM$ with $\dim(\cM)=1$ is embedded in $\PG(v-1,q)$ with $v\ge 2$, then we have $\cM(H)\in\{0,\#\cM\}$ for all hyperplanes (counted with multiplicity) and there indeed exists a 
hyperplane with $\cM(H)=0$. See also Lemma~\ref{lem:ambient_space_unwichtig} stating that the dimension of the ambient space is irrelevant.}   
\end{trailer}

With respect to $[n,\le\!k,d]_q$-codes,\footnote{An $[n,\le k,d]_q$-code is an $[n,k',d]_q$-code where $1\le k'\le k$. We also use the {\lq\lq}$\le${\rq\rq}- or 
{\lq\lq}$\ge${\rq\rq}-notation for other parameters.} Equation~(\ref{eq_min_dist_hyperplane}) motivates the following geometric notion. A multiset $\cK$ of points 
in $\PG(k-1,q)$ is an \emph{$(n,s)$-arc} if $\cK(\cP)=n$, $\cK(H)\le s$ for every hyperplane $H\in\cH$, and there exists a hyperplane $H_0\in \cH$ with $\cK(H_0)=s$. 
If the last condition is skipped, we speak of an $(n,\le s)$-arc. The relation between $d$ and $s$ is given by $s=n-d$. The dimension of the subspace spanned by 
the points in the support of $\cK$ is called the dimension $\dim(\cK)$ of $\cK$. The corresponding linear code has dimension $\dim(\cK)$.

There is an analog of the weight distribution of linear codes for arcs.
\begin{ndefinition}
  \label{def_spectrum}
  Let $\cM$ be an $(n,\le s)$-arc in $\PG(k-1,q)$. The \emph{spectrum} of $\cM$ is the vector 
  $\ba=\left(a_0,\dots,a_s\right)\in\N_0^{s+1}$, where 
  \begin{equation}
    a_i=\#\left\{H\in\cH \,:\, \cM(H)=i\right\}
  \end{equation} 
  for $0\le i\le s$.
\end{ndefinition}
\begin{nexercise}
  \label{exercise_weight_distribution_vs_spectrum}
  Let $\cM$ be a $k$-dimensional multiset of $n$ points in $\PG(k-1,q)$ and $C$ be the corresponding $[n,k]_q$-code. Show that
  \begin{equation}
    A_i(C) = (q-1)\cdot a_{n-i}(\cM)
  \end{equation} 
  for all $1\le i\le n$, where $a_j:=0$ if no hyperplane has multiplicity $j$.
\end{nexercise}
In the case of a $\Delta$-divisible multiset of points $a_i>0$ implies $i\equiv n\pmod\Delta$. 
The analog of the first three MacWilliams equations are the so-called \emph{standard equations}:
\begin{nlemma}
  \label{lemma_standard_equations}
  The spectrum $\ba=\left(a_0,\dots,a_s\right)$ of an $(n,\le s)$-arc $\cM$ in $\PG(k-1,q)$, where $k\ge 2$, satisfies
  \begin{eqnarray}
    \sum_{i=0}^s a_i &=& [k]_q \label{eq_standard_equation_1}\\ 
    \sum_{i=0}^s ia_i &=& n\cdot [k-1]_q\label{eq_standard_equation_2}\\ 
    \sum_{i=0}^s {i \choose 2}a_i &=& {n\choose 2}\cdot[k-2]_q+q^{k-2}\cdot\sum_{i\ge 2} {i\choose 2}\lambda_i\label{eq_standard_equation_3},
  \end{eqnarray}
  where $\lambda_j$ denotes the number of points $P\in\cP$ with $\cM(P)=j$ for all $j\in \N_0$.
\end{nlemma}
\begin{nexercise}
  Prove Lemma~\ref{lemma_standard_equations} by double counting hyperplanes, incidences between points and hyperplanes, and incidences between pairs of 
  points and hyperplanes. Show that the three standard equations are indeed equivalent to the first three MacWilliams equations assuming a code of full length, i.e.\ 
  $B_1=0$.
\end{nexercise}
\begin{figure}[htbp]
  \centering
  \begin{tikzpicture}[node distance=6mm,
      box/.style={draw,rounded corners=2pt,align=center,font=\small,inner sep=4pt,minimum height=9mm,text width=30mm},
      arr/.style={-{Stealth[length=2.2mm]},thick}]
    \node[box] (M) {multiset $\cM$ of $n$ points in $\PG(k-1,q)$};
    \node[box,right=of M,text width=46mm] (a) {spectrum\\ $a_i=\#\{H\in\cH:\cM(H)=i\}$};
    \node[box,right=of a] (e) {standard equations in the $a_i$};
    \node[box,below=8mm of a,text width=46mm] (div) {$\Delta$-divisible: $a_i=0$ unless $i\equiv n\pmod\Delta$};
    \node[box,below=8mm of e,text width=30mm] (res) {no feasible $\ba\in\N_0^{s+1}$: $\cM$ does not exist};
    \draw[arr] (M)--(a);
    \draw[arr] (a)--(e);
    \draw[arr] (div)--(a);
    \draw[arr] (e)--(res);
    \draw[arr] (div)--(res);
  \end{tikzpicture}
  \caption{Use of the standard equations for divisible multisets.}
  \label{fig_standard_equations_scheme}
\end{figure}
\begin{nexercise}
  Let $\cM$ be a multiset of points in $\PG(k-1,q)$ and $\cM'$ be an embedding in $\PG(v-1,q)$ with $v>k$. Compute the spectrum $\ba'$ of $\cM'$ from 
  the spectrum $\ba$ of $\cM$.
\end{nexercise}

Define the sum of two multisets $\cK'$ and $\cK''$ in the same geometry $\PG(k-1,q)$ by $\left(\cK'+\cK''\right)(P)=\cK'(P)+
\cK''(P)$ for all points $P\in\cP$. With the aid of so-called characteristic functions we can describe more sophisticated constructions 
in a compact manner. So, given a set of points $\cQ\subseteq\cP$, we denote by $\chi_{\cQ}\colon\cP\to\{0,1\}$ the \emph{characteristic 
function} of $\cQ$, i.e., $\chi_{\cQ}(P)=1$ if $P\in\cQ$ and $\chi_{\cQ}(P)=0$ otherwise. If $J$ is a $j$-space in $\PG(k-1,q)$, where 
$1\le j\le k$, then we write $\chi_J$ for the characteristic function of the points contained in $J$.

\begin{trailer}{First-order Reed--Muller codes a.k.a.\ affine $k$-spaces}
\vspace*{-6mm}
\begin{nexample}
  \label{example_affine_space}
  Let $H$ be a hyperplane in $V=\PG(k-1,q)$, where $k\ge 2$. Then $\cK=\chi_{V}-\chi_H=\chi_{\cP}-\chi_H$ is a $\left(q^{k-1},q^{k-2}\right)$-arc 
  that corresponds to a $\left[q^{k-1},k,q^{k-1}-q^{k-2}\right]_q$-code. 
\end{nexample}
We remark that $\cK$ is an affine geometry $\AG(k-1,q)$ and that the corresponding code is a first-order 
Reed--Muller code $\RM_q(k-1,1)$ of length $q^{k-1}\!$. We also call the (multi-)set of points an \emph{affine $k$-space}. 
For $q=2$ and $k=4$ we obtain the $8$ points of $\AG(3,2)$ and the $[8,4,4]_2$-code of Example~\ref{ex_extended_hamming_8}. Indeed, each plane 
of $\PG(3,2)$ different from $H$ contains exactly $4$ of these points, while $H$ contains none, so that the weights are $8-4=4$ and $8-0=8$.
\end{trailer}
In general we have:
\begin{nlemma}
  \label{lemma_sum_of_characteristic_functions}
  Let $\cQ_1,\dots,\cQ_l$ be multisets of points and $m_1,\dots,m_l\in\Q$. If
  \begin{equation}
    \sum_{i=1}^l m_i\cdot{\cQ_i}(P)\in\N_0
  \end{equation}
  for each $P\in\cP$, then 
  \begin{equation}
    \cM=m_1\cdot \cQ_1+m_2\cdot \cQ_2+\dots+m_l\cdot\cQ_l
  \end{equation}
  defines a multiset of points in $\PG(k-1,q)$.
\end{nlemma}
\begin{nexercise}
  \label{exercise_q_linear_combination}
  Prove Lemma~\ref{lemma_sum_of_characteristic_functions} and show that the code defined in Example~\ref{example_affine_space} is $q^{k-2}$-divisible. 
\end{nexercise}
\begin{nexercise}
  \label{exercise_divisibility_sum_construction}
  Let $\cQ_1,\dots,\cQ_l$ be multisets of points that are $\Delta$-divisible. Show that the multiset 
  $\cM=\sum_{i=1}^l \cQ_i$ is $\Delta$-divisible and that the multiset $t\cdot \cQ_i$ is $t\Delta$-divisible for each integer $t\ge 1$. 
\end{nexercise}

If $\cM$ is a multiset of points in $\PG(k-1,q)$, then we can embed $\PG(k-1,q)$ in a $k$-space of $\PG(v-1,q)$ for each $v\ge k$ and naturally obtain 
a multiset $\cM'$ of points in $\PG(v-1,q)$. If $C$ is the linear code corresponding to $\cM$ and $C'$ the linear code corresponding to $\cM'$, then 
$C$ and $C'$ are isomorphic and the (effective) lengths of $C,C'$ equal $\#\cM=\#\cM'$ and the dimensions equal $\dim(\cM)=\dim(\cM')$. So, the union of the $20$ four-dimensional and $30$ three-dimensional subspaces 
considered in Subsection~\ref{subsec_introductory_application} yields a multiset $\cM$ of points in $\PG(8,2)$ that is $\min\!\left\{2^{4-1},2^{3-1}\right\}=4$-divisible, see Lemma~\ref{lem:union_subspaces}. 
Thus, also the corresponding binary linear code is $4$-divisible. We will consider the complementary multiset of points and its corresponding linear code in Subsection~\ref{subsec_application_q_r_div}.  

By $\gamma_0(\cM)$ we denote the \emph{maximum point multiplicity} of a given multiset of points $\cM$ in $\PG(v-1,q)$, i.e., we have $\cM(P)\le \gamma_0(\cM)$ for all 
$P\in \cP$ and there exists a point $Q\in\cP$ with $\cM(Q)=\gamma_0(\cM)$. If $\gamma_0(\cM)=1$, then we also speak of a set of points instead of a multiset of points. 
Clearly we have $\gamma_0(\cM)=0$ iff $\cM$ is empty, i.e., $\#\cM=0$. We say that $\cM$ is \emph{proper} iff there exists a point $P\in\cP$ with $\cM(P)=0$. Otherwise 
$\cM$ has full support. 
\begin{nexercise}
  Show that for a given multiset of points $\cM$ in $\PG(v-1,q)$ we have $\gamma_0(\cM)=1$ iff the corresponding linear code $C$ is projective.
\end{nexercise}
The analog of the point multiplicity $\cM(P)$ of a point $P$ for the corresponding linear code $C$ is the number of columns $g_i$ of a generator matrix with 
$\left\langle g_i\right\rangle=P$. Here we may also speak of the (maximum) column multiplicity.

\chapter{Basic results for $\Delta$-divisible multisets of points}
\label{sec_basis_results}
As already observed, each multiset of points $\cM$ in $\PG(v-1,q)$ can be embedded in a larger ambient space $\PG(v'-1,q)$, where $v'>v$. We can also embed 
in a smaller ambient space as long as the dimension is at least $\dim(\cM)$. As readily computed, the dimension of the ambient space is not really 
relevant for the notion of $\Delta$-divisibility.
\begin{nlemma}
	\label{lem:ambient_space_unwichtig}
	Let $V_1 < V_2$ be $\F_q$-vector spaces and $\cM$ a multiset of points in $V_1$.
	Then $\cM$ is $\Delta$-divisible in $V_1$ iff $\cM$ is $\Delta$-divisible in $V_2$ (using the natural continuation of the characteristic function 
	$\cM(P)=0$ for all $P\in \cP(V_2)\backslash \cP(V_1)$).
\end{nlemma}
So, we will also speak of a $\Delta$-divisible multiset of points $\cM$ over $\F_q$ without specifying the dimension of the ambient space. (Of course we have to assume 
that the dimension of the ambient space is at least $\dim(\cM)$.)

As observed by Harold Ward, it is not necessary to consider all positive integers $\Delta$ when studying $\Delta$-divisible codes.
\begin{ntheorem}
  \label{thm_delta_divides_q_power}
  (\cite[Theorem 1]{ward1981divisible}) Let $C$ be a $\Delta$-divisible $[n,k]_q$-code with $k\ge 1$ and $\gcd(\Delta,q)=1$. Then $C$ is equivalent to a code obtained by taking 
  a linear code $C'$ over $\F_q$, repeating each coordinate $\Delta$ times, and appending enough $0$ entries to make a code whose length is that of $C$. 
\end{ntheorem}
\begin{proof}
  The statement is clearly true for $k=1$, so that we assume $k\ge 2$. Let $\cM$ be the corresponding multiset of points in $\PG(k-1,q)$. If $k\ge 3$ consider an arbitrary 
  subspace $S$ of codimension $2$ and the $q+1$ hyperplanes $H_0,\dots,H_q$ containing $S$. We have 
  $$
    (q+1)\#\cM \equiv \sum_{i=0}^q \cM(H_i) =\#\cM+q\cdot\cM(S)\pmod\Delta,
  $$
  so that $\gcd(\Delta,q)=1$ implies that also the restriction $\cM|_{H_i}$ of $\cM$ to hyperplane $H_i$ is $\Delta$-divisible, i.e., we have $\cM(S)\equiv \cM(H_i)\equiv\#\cM\pmod\Delta$ 
  for every hyperplane $S$ of $H_i$. Thus, it suffices to consider the case $k=2$ where we have $\cM(P)\equiv \#\cM\pmod\Delta$. From $\#\cM=\sum_{P\in\cP} \cM(P)\equiv (q+1)\#\cM\pmod \Delta$ 
  and $\gcd(\Delta,q)=1$ we then conclude $\cM(P)\equiv\#\cM\equiv 0\pmod{\Delta}$, i.e., the multiplicity of every point is divisible by $\Delta$.
\end{proof}
\begin{important}{$s$ dividing $\Delta$ while $\gcd(s,q)=1$ implies $\mathbf{s}$-fold repetition}  
Given an arbitrary positive integer $\Delta$ and a field size $q=p^m$, we can uniquely write $\Delta=s\cdot t$, where $s,t\in\N_{>0}$, $\gcd(s,q)=1$, and 
$t$ divides a sufficiently large power of $p$, i.e., there exists an integer $e$ with $t=p^e$. From Theorem~\ref{thm_delta_divides_q_power} we conclude 
that each $\Delta$-divisible $[n,k]_q$ code $C$ arises from a $q^r$-divisible $[n',k]_q$-code $C'$, where $r=\tfrac{e}{m}$ and $n's\le n$, by repeating each coordinate 
in $C'$ exactly $s$ times and adding $n-n's$ zero entries. Thus, it is sufficient to study $q^r$-divisible codes over $\F_q$, where $q^r$ is a power of the characteristic $p$ 
of the finite field $\F_q$.  
\end{important}
\begin{trailer}{Ward's criterion for binary codes}Divisibility of a linear code can be checked on a generator matrix, which goes back to Ward \cite{ward1990weight}. 
For binary codes this is particularly transparent. Let $g_1,\dots,g_k\in\F_2^n$ be the rows of a generator matrix of $C$ and let $g_{i_1}\circ\dots\circ g_{i_j}$ denote the coordinatewise 
(Hadamard) product, so that $\wt\!\left(g_{i_1}\circ\dots\circ g_{i_j}\right)$ is the size of the intersection of the supports. Writing the parity of a sum of $t$ bits $x_1,\dots,x_t\in\{0,1\}$ as 
$\sum_{\emptyset\neq S\subseteq\{1,\dots,t\}}(-2)^{|S|-1}\prod_{i\in S}x_i$ and summing over all coordinates yields 
$$
  \wt\!\left(\sum_{i=1}^t v_i\right)=\sum_{i}\wt(v_i)-2\sum_{i<j}\wt(v_i\circ v_j)+4\sum_{i<j<h}\wt(v_i\circ v_j\circ v_h)-\dots
$$
for arbitrary $v_1,\dots,v_t\in\F_2^n$. Since every non-zero codeword is a sum of rows of $G$, a sufficient condition for $2^r$-divisibility is that the weight of each product of $j$ different rows is divisible by $2^{r-j+1}$ for 
$1\le j\le r+1$. The condition is also necessary, since the product weights can be recovered from the weights of the sums via 
$\wt\!\left(\circ_{i\in S}v_i\right)=2^{1-|S|}\sum_{\emptyset\neq T\subseteq S}(-1)^{|T|-1}\wt\!\left(\sum_{i\in T}v_i\right)$. For $r=1$ the criterion says that $C$ is even iff all rows of $G$ have even weight, and for $r=2$ that $C$ is doubly-even 
iff all rows have weight divisible by $4$ and all pairs of rows have an even number of common ones. As an example, the rows $(0001111)$, $(0110011)$, and $(1010101)$ of a generator matrix of the binary simplex code of Example~\ref{ex_simplex_code} 
for $k=3$ have weight $4$ and pairwise intersect in two coordinates, so that this code is $4$-divisible. An analogous criterion for ternary codes can be found in \cite{ward2001divisible}, and for general $q$ in \cite{ward1990weight}. For trace codes, 
i.e., codes given by a generalized generator matrix over an extension field, a divisibility criterion that extends Ward's criterion and yields lower bounds on the $p$-adic valuation of the number of solutions of certain equations over finite fields 
is given in \cite{huang2026trace}.
\end{trailer}
\begin{nexercise}
  Show that no projective $[54,6,\{24,27,30\}]_2$-code exists.
\end{nexercise}

Currently we can generalize Theorem~\ref{thm_delta_divides_q_power} to multisets of $h$-spaces under several technical assumptions only.
\begin{nproposition} 
  Let $\cM$ be a multiset of $h$-dimensional subspaces of $\PG(v-1,q)$ that pairwise either coincide or are disjoint, i.e.\ $\dim(U\cap V)\in\{0,h\}$ for all 
  $h$-spaces $U$, $V$ with $\cM(U)$, $\cM(V)>0$. If the number of $h$-spaces contained in any hyperplane $H$ is congruent to $\#\cM$ modulo $\Delta$, then we have $\cM(U)\equiv \#\cM\pmod\Delta$
  for every $h$-space $U$ with $\cM(U)>0$ if $\gcd\!\left(\Delta,q[h]_q\right)=1$.
\end{nproposition}
\begin{proof}
  Let $U$ be an $h$-space with $\cM(U)>0$ and let us denote the number of $h$-spaces contained in $H$ (counted with multiplicity) by $\widetilde{\cM}(H)$. With this, we have
  $$
    \sum_{U\le H\,:\, \dim(H)=v-1} \widetilde{\cM}(H) \equiv [v-h]_q\cdot\#\cM \pmod \Delta.
  $$
  Since each $h$-space $V$ that is disjoint to $U$ is counted $[v-2h]_q$ times in the sum on the left hand side, we have 
  $$
    [v-h]_q\cdot\#\cM \equiv [v-2h]_q\cdot\#\cM +[h]_q\cdot q^{v-2h}\cdot \cM(U) \pmod{\Delta},
  $$
  so that $\cM(U)\equiv \#\cM\pmod\Delta$.
\end{proof}
Note that the assumption that the $h$-spaces pairwise coincide or are disjoint cannot be dropped: The set $\cM$ of all lines in $\PG(3,q)$ satisfies $\widetilde{\cM}(H)\equiv\#\cM\equiv 0\pmod{[3]_q}$ for every hyperplane $H$ 
and $\gcd\!\left([3]_q,q[2]_q\right)=1$, but we have $\cM(U)=1\not\equiv 0\pmod{[3]_q}$ for every line $U$.

\begin{trailer}{Divisibility inherits}Assume that $q^r$ divides $\Delta$ and that $\cM$ is a $\Delta$-divisible multiset of points in $\PG(v-1,q)$ with $v\ge 3$. If $W$ is a subspace of codimension $2$, then 
there are $q+1$ hyperplanes $H_1,\dots,H_{q+1}$ through $W$, i.e., hyperplanes in $\PG(v-1,q)$ that fully contain the subspace $W$. Counting points yields
\begin{equation}
    \sum_{i=1}^{q+1} \cM(H_i)
    =q\cdot \cM(W)\,+\,\#\cM\equiv (q+1)\#\cM 
    \pmod{\Delta},
\end{equation} 
so that
$$
  q\cdot \cM(H_i)\equiv q\cdot \#\cM \equiv q\cdot \cM(W)\pmod{\Delta}, 
$$
which is equivalent to
\begin{equation}
  \cM(H_i)\equiv \#\cM \equiv \cM(W)\pmod{\Delta/q} 
\end{equation}
if $r\ge 1$, i.e., $\Delta/q\in\N_{>0}$. By induction over $j$ we can easily prove:
\begin{nlemma}
  \label{lemma_heritable}
  Let $\cM$ be a $\Delta$-divisible multiset of points in $V\simeq \PG(v-1,q)$, where $q^r$ divides $\Delta$, and $U\neq\langle\zv\rangle$ be a subspace of $V$ of codimension $0\le j\le r$.
  Then, the restriction $\cM|_U$ is a $q^{r-j}$-divisible multiset in the $(v-j)$-dimensional vector space $U$.
\end{nlemma}
\end{trailer}
So, e.g.\ when restricting a given multiset of points $\cM$ over $\F_q$ to a hyperplane $H$, the divisibility goes down by at most a factor $q$. Of course, if $\Delta$ is the 
maximum possible divisibility of $\cM$, then the maximum possible divisibility of $\cM|_H$ is a multiple of $\Delta/q$ (if $q\mid\Delta$) and is at least $1$. 
The converse of Lemma~\ref{lemma_heritable} is not true in general: 
\begin{nexample}
  Let $\cM$ in $\PG(5,q)$ be given by $q\cdot \chi_{\be_1}+q\cdot \chi_{\be_2}+q\cdot \chi_{\be_3}+q\cdot \chi_{\be_4}$, where $\be_i$ denotes the $i$th unit vector. 
  I.e., we consider a multiset $\cM$ given by four $q$-fold points whose span is $4$-dimensional. With this, we have $\#\cM=4q$ and $\cM$ is $q$-divisible. For any 
  hyperplane $H$ that contains $\be_1$, $\be_2$, and $\be_3$ but not $\be_4$ we have $\cM(H)=3q$. Thus, $\cM$ is not $\Delta$-divisible for any $\Delta>q$. However, we 
  even have $\cM(L)\equiv \#\cM\pmod q$ for any line $L$. 
\end{nexample}

\begin{nexercise}
  Let $\cM$ be a multiset of points in $\PG(v-1,p^h)$ and $1\le s\le v-1$ be an integer. Show that $\cM(S)\equiv\#\cM\pmod {p^r}$ for each $s$-space $S$ in 
  $\PG(v-1,p^h)$ implies $\cM(T)\equiv \#\cM\pmod {p^r}$ for each $t$-space $T$ in $\PG(v-1,p^h)$, where $s\le t\le v-1$ and $r$ are positive integers. 
\end{nexercise}

For the special case $j=1$, we can easily translate Lemma~\ref{lemma_heritable} to the language of linear codes. To this end we need a little more notation. Let $C$ 
be a $\Delta$-divisible $[n,k]_q$-code with $q\mid\Delta$. For an arbitrary index set $I\subseteq\{1,\dots,n\}$ and an arbitrary 
codeword $\bc\in C$ by $\bc_I$ we denote the $|I|$-tuple that consists of the entries $c_i$ with $i\in I$, which is also called a restricted codeword. For 
an arbitrary but fixed codeword $\tilde{\bc}\in C$ we set $I:=\{1,\dots,n\}\backslash \supp(\tilde{\bc})$ as abbreviation. With this we can define the so-called 
\emph{residual code} of $C$ with respect to $\tilde{\bc}$ by
$$
  \Res(C;\tilde{\bc}):=\left\{\bc_I\,:\,\bc\in C\right\}.
$$       
This code of length $|I|$ is $\Delta/q$-divisible as $\cM|_H$ is $\Delta/q$-divisible, where $\cM$ is the multiset of points corresponding to $C$ and $H$ is the hyperplane 
corresponding to the codeword $\tilde{\bc}$. In this latter and special form, Lemma~\ref{lemma_heritable} can be found in \cite[Lemma~13]{Ward-1998-JCTSA}.
\begin{nexercise}
  Let $C$ be a $q^r$-divisible $[n,k]_q$-code with $r>k-1$. Show that $C$ arises from a $q^{k-1}$-divisible $[n',k]_q$-code $C'$ by repeating each non-zero coordinate 
  $q^{r-k+1}$ times and adding a suitable number of zero coordinates. So, in particular we have that the effective length of $C$ is divisible by $q^{r-k+1}$, 
  cf.~\cite[Theorem 1.3]{ward1990weight}.
  
  To this end, let $c^1,\dots,c^k$ be arbitrary codewords of $C$ and $I^\bx\subseteq \{1,\dots,n\}$ be defined as the set of all indices $1\le i\le n$ such that $c^j_i=x_i$ for all 
  $1\le j\le k$, where $\bx\in\{0,1,\dots,q-1\}^k$ is arbitrary. By eventually considering an isomorphic code assume the following normalization criterion: For each $1\le j\le k$ and 
  each $i\in\{1,\dots,n\}\backslash\supp(\langle c^1,\dots,c^{j-1}\rangle)$ we have $c^j_i\in\{0,1\}$. With this, show that $\# I^\bx$ is divisible by $q^{r-k+1}$ if $\bx\neq \zv$. 
\end{nexercise}

\begin{trailer}{An easy averaging argument}Let $\cM$ be a non-empty multiset of points in $\PG(v-1,q)$, where $v\ge 2$. Since each point is contained in $\qbin{v-1}{v-2}{q}=\qbin{v-1}{1}{q}=[v-1]_q$ hyperplanes, see 
Exercise~\ref{exercise_counting_oversubspaces}, and there are $\qbin{v}{v-1}{q}=\qbin{v}{1}{q}=[v]_q$ hyperplanes in total, the average number of points per hyperplane is given by
\begin{equation}
  \frac{\sum_{H\in\cH} \cM(H)}{|\cH|}=\frac{\#\cM \cdot [v-1]_q}{[v]_q}
  = 
  \frac{\#\cM \cdot [v-1]_q}{q [v-1]_q + 1}
  = \frac{\#\cM}{q+\frac{1}{[v-1]_q}}
  <\frac{\#\cM}{q}.
\end{equation}
Choosing a hyperplane $H\in\cH$ that minimizes $\cM(H)$ we obtain:
\begin{nlemma}
  \label{lemma_average}
  Let $\cM$ be a non-empty multiset of points in $\PG(v-1,q)$. If $v\ge 2$, then there exists a hyperplane $H\in\cH$ with $\cM(H)< \frac{\#\cM}{q}$. 
\end{nlemma}
\begin{nexercise}
  \label{exercise_upper_bound_proper_multiset}
  Let $\cM$ be a proper multiset of points in $\PG(v-1,q)$ with $v\ge 2$. Show $\min\{\cM(H)\,:\,H\in\cH\}\le \frac{[v-2]_q}{[v-1]_q} \cdot \#\cM$.  
\end{nexercise} 
The non-geometric coding counterpart of Lemma~\ref{lemma_average} is the well-known existence of a codeword of weight $>\tfrac{q-1}{q} \cdot \neff$. 
For a refinement we refer to the \emph{Hamada bound} in Theorem~\ref{theorem_hamada_bound}.
\end{trailer}
\begin{nexample}
  \label{ex_no_4_div_card_9}
  From Lemma~\ref{lemma_average} we can directly conclude that there is no $2$-divisible multiset of points of cardinality $1$ over $\F_2$. Note 
  that due to $1\not\equiv 0\pmod 2$, there cannot be such a multiset in $\PG(0,2)$. Now assume that 
  $\cM$ is a $4$-divisible multiset of points of cardinality $9$ over $\F_2$. Since $9\not\equiv 0\pmod 4$, we conclude that the dimension $v$ of the 
  ambient space of $\cM$ is at least $2$. Lemma~\ref{lemma_average} guarantees the existence of a hyperplane $H$ with 
  $\cM(H)<\#\cM/q=9/2$. Since $\cM(H)\equiv \#\cM\pmod 4$, we have $0<\cM(H)\le 1$. So, due to Lemma~\ref{lemma_heritable} we have that the 
  restricted multiset $\cM|_H$ is $2$-divisible and has cardinality $1$. Thus, there is no $4$-divisible multiset of points of cardinality $9$ over 
  $\F_2$.  
\end{nexample}
In contrast, cardinality $8$ is attained by the affine solid $\AG(3,2)$ of Example~\ref{ex_extended_hamming_8}, which is $4$-divisible.
\begin{nexercise}
  \label{exercise_counting_oversubspaces}
  Let $S$ be an $s$-space in $\PG(v-1,q)$ and $t$ be an integer with $s\le t\le v$. Show that the number of $t$-spaces that contain 
  $S$ is given by $\qbin{v-s}{t-s}{q}$.
\end{nexercise}
 
Example~\ref{ex_simplex_code} and Exercises~\ref{exercise_q_linear_combination} and \ref{exercise_divisibility_sum_construction} directly give: 
\begin{nlemma}
\label{lem:union_subspaces}
  Let $\cU$ be a multiset of subspaces of~$\,\PG(v\!-\!1,q)$ and $\cM = \biguplus_{U\in\mathcal{U}} \spaces{U}{1}$ the \emph{associated multiset 
  of points}. (In the expression $\biguplus_{U\in\cU}$, the subspace $U$ is repeated according to its multiplicity in the multiset $\cU$.) 
  Let $k$ be the smallest dimension among the subspaces in $\cU$. If $k \geq 1$, then the multiset $\cM$ is $q^{k-1}$-divisible.
\end{nlemma} 
So the multiset of points given by the points of the $20$ solids and $30$ planes from Subsection~\ref{subsec_introductory_application} is $\Delta$-divisible for $\Delta=2^{3-1}=4$.     

\bigskip

\begin{trailer}{The kernel of the incidence matrix between points and $\mathbf{k}$-spaces}Let $A\in\{0,1\}^{ [v]_q\times \qbin{v}{k}{q}}$ be the incidence 
matrix between the $[v]_q$ points $P$ and the $\qbin{v}{k}{q}$ $k$-spaces $K$ in $\PG(v-1,q)$, i.e., the entries of $A$ are given by $a_{P,K}=1$ iff $P\le K$. 
For any multiset $\cK$ of $k$-spaces in $\PG(v-1,q)$ let $\bx\in\N_0^{\qbin{v}{k}{q}}$ be the corresponding counting vector. With this, the vector 
$A\bx\in\N_0^{[v]_q}$ is in one-to-one correspondence to the multiset $\cM$ of points associated to $\cK$. Let $\by\in\R^{[v]_{q}}$ be an element of the 
cokernel of $A$, which is the kernel of $A^\top$, i.e., $A^\top \by=\zv$. So, we have $\by^\top A=\zv$, $\by^\top (A \bx) =\zv$, and the necessary condition $\by^\top \bz=\zv$ for 
any candidate vector $\bz\in\N_0^{[v]_q}$ for $\cM$. In other words, the kernel of the incidence matrix $A^\top$ gives necessary conditions whether a 
given multiset $\cM$ of points can be decomposed into a multiset of $k$-spaces. (As in Lemma~\ref{lem:union_subspaces}, we can also consider the situation 
of subspaces $K$ with $\dim(K)\ge k$ instead of $\dim(K)=k$.) 

Now observe that we have some freedom over which domain we compute the kernel and interpret the corresponding necessary conditions $\by^\top\bz=\zv$. Over the 
real or rational numbers the kernel of $A^\top$ is trivial for $1\le k\le v-1$ so that we don't get any non-trivial conditions $\by^\top\bz=\zv$. For $k=v-1$, i.e.\ hyperplanes, this observation 
follows from Exercise~\ref{exercise_dual_multiset_coefficients} and we can use Exercise~\ref{exercise_chi_subspace} for the general case. If $p$ is the characteristic of 
$\F_q$, then the domain $\F_p\cong \Z/p\Z$ is a natural choice. To this end we remark that the $p$-rank of the incidence matrix between points and $k$-spaces can be explicitly 
computed using the famous \emph{Hamada formula} \cite{hamada1973p}, see also \cite{macwilliams1968p,smith1969p} for ancestors. Here we want to start with a more general 
perspective first. We are interested in the set of multisets of points that can be decomposed into a multiset of $k$-spaces, i.e., non-negative integer combinations of the 
rows of $A^\top\!$. As said, over $\Q$ the row space of $A^\top$ is $\Q^{[v]_q}$. Over $\Z$ we can use the structure theorem for finitely generated modules over a principal 
ideal domain $R$ -- the \emph{invariant factor decomposition} to be more precise.  
For every non-zero matrix $B\in R^{m\times n}$ there exist invertible matrices $S\in R^{m\times m}$ and $T\in R^{n\times n}$ such that
$D=SBT\in R^{m\times n}$ has zero entries outside of its main diagonal $d_1,\dots,d_{\min(m,n)}$, where we additionally have 
$d_i\,|\, d_{i+1}$ for all $1\le i<\min(m,n)$.   
This is the \emph{Smith normal form} of $B$ and the $d_i$ are called \emph{invariant factors}, \emph{elementary divisors}, or \emph{invariants}. 
The matrices $S$ and $T$ can be algorithmically obtained by recursively applying invertible row and column operations to $B$ till it reaches the desired diagonal form. 
Now we choose $R=\Z$ and set $M=\left\{x^\top B\,:\, x\in\Z^m\right\}$, $M'=\left\{x^\top D\,:\, x\in\Z^m\right\}$. Setting $d_i:=0$ for $\min\{m,n\}<i\le n$ we have 
$M'=d_1\Z\times d_2\Z\times\dots\times d_n\Z\subseteq \Z^{1\times n}$. Using the convention that $i\equiv j \pmod 0$ iff $i=j$ for all $i,j\in\Z$, we can also write 
\begin{equation}
  M'=\left\{z' \in \Z^{1\times n}\,:\, z'\be_i\equiv 0\pmod {d_i}\,\,\forall 1\le i\le n\right\}\!,
\end{equation}  
where $\be_i$ denote the unit vectors. Using that $x^\top B=z$ is equivalent to $\left(x^\top S^{-1}\right)D=zT$ we obtain
\begin{equation}
  M=\left\{z \in \Z^{1\times n}\,:\, z\left(T\be_i\right)\equiv 0\pmod {d_i}\,\,\forall 1\le i\le n\right\}\!,
\end{equation}
i.e., we can read off the conditions for elements in $M$ from the columns of $T$. For a brief description of the Smith normal form in our context we refer to \cite[Section 4]{kurz2020additiveGriesmer}. 
\end{trailer}

\begin{nexample}
  Let $B=\begin{pmatrix}
    1 & 1 & 5 & 7 \\ 
    2 & 8 & 10 & 20 \\ 
    3 & 3 & 45 & 51  
  \end{pmatrix}$
  and $M$ be the $\Z$-module generated by the rows of $B$. The Smith normal form of $B$ is given by
  $${\small
    D:=SBT=
    \begin{pmatrix}
    1 & 0 &  0 & 0 \\ 
    0 & 6 &  0 & 0 \\ 
    0 & 0 & 30 & 0
    \end{pmatrix}
  \quad\text{with}\quad
    S=\begin{pmatrix}
       1 & 0 & 0 \\
      -2 & 1 & 0 \\ 
      -3 & 0 & 1 
    \end{pmatrix}
    \quad\text{and}\quad
    T=\begin{pmatrix}
      1 & -1 & -6 &  1 \\ 
      0 &  1 & -1 &  1 \\  
      0 &  0 &  0 &  1 \\ 
      0 &  0 &  1 & -1
    \end{pmatrix}.
  }$$
  We have 
  \begin{eqnarray*}
    M' 
    &=& \left\{ \left(z_1,z_2,z_3,z_4\right)\in \Z^{1\times 4}\,:\, z_2\equiv 0\pmod 6,z_3\equiv 0\pmod {30}, z_4=0\right\},
  \end{eqnarray*}
  for the $\Z$-module $M'$ generated by the rows of $D$
  and   
  {\small\begin{eqnarray*}
    M&=&\big\{ \left(z_1,z_2,z_3,z_4\right)\in \Z^{1\times 4}\,:\, -z_1+z_2\equiv 0\pmod 6, -6z_1-z_2+z_4\equiv 0 \pmod {30},\\ 
    &&  z_1+z_2+z_3-z_4=0\big\}.
  \end{eqnarray*}}
  The conditions can as well be split up to the different occurring prime powers and we can modify the coefficients according to the respective moduli:
  \begin{eqnarray*}
     z_1+z_2\equiv 0 \pmod 2                && z_1+z_2+z_3+2z_4 \equiv 0\pmod 3 \\  
     z_2+z_4 \equiv 0\pmod 2                && 4z_1+4z_2+z_4 \equiv 0\pmod 5    \\
     z_1+z_2+z_3+z_4 \equiv 0\pmod 2        && z_1+z_2+z_3+4z_4 \equiv 0\pmod 5 \\ 
     2z_1+z_2\equiv 0\pmod 3                && z_1+z_2+z_3-z_4 =0               \\
     2z_2+z_4\equiv 0\pmod 3                &&
    \end{eqnarray*}
    The last condition $\!\!{\pmod p}$ for $p\in\{2,3,5\}$ is superfluous since it is implied by $z_1+z_2+z_3-z_4 =0$. We can easily generate more 
    superfluous conditions like e.g.\ $z_1+z_3\equiv0\pmod 2$ or $z_3\equiv 0\pmod 5$ by considering the span of our constraints. Further irredundant descriptions of $M$, i.e.\ 
    a description as a tuple of integers where we cannot drop one of the conditions, are e.g.\ given by 
    {\small\begin{eqnarray*}
    M&=&\big\{ \left(z_1,z_2,z_3,z_4\right)\in \Z^{1\times 4}\,:\, z_1+z_2\equiv 0\!\!\!\pmod 2, z_1+z_3\equiv 0 \!\!\!\pmod 2, z_3\equiv0\!\!\!\pmod 5,\\ 
    && z_2+z_3\equiv 0 \!\!\!\pmod 3,z_1+z_2+z_4\equiv 0\!\!\!\pmod 3, z_1+z_2+z_3-z_4=0\big\}
    \end{eqnarray*}} 
    or
    {\small\begin{eqnarray*}
    M&=&\big\{ \left(z_1,z_2,z_3,z_4\right)\in \Z^{1\times 4}\,:\, 5z_1+5z_2+2z_3\equiv 0\!\!\!\pmod {10}, 3z_1+2z_2+5z_3\equiv 0 \!\!\!\pmod 6,\\ 
    && z_1+z_2+z_4\equiv 0\!\!\!\pmod 3, z_1+z_2+z_3-z_4=0\big\}.
    \end{eqnarray*}}
    The span of our original set of constraints as well as the \emph{Chinese remainder theorem} allow a lot of freedom in the choice of a suitable (irredundant) description of $M$. 
    We only need to ensure that the span of the used constraints coincides with the original one. 
\end{nexample}

\begin{ntheorem}(E.g.~\cite[Theorem 3.1]{chandler2006invariant}.)
  \label{theorem_invariant_factors}
  Let $A$ be the incidence matrix of points and $k$-spaces in $\PG(v-1,q)$ and let $p$ be the characteristic of $\F_q$. Then, the invariant factors of $A^\top$ are all $p$-powers 
  except the last, which is a $p$-power times $[k]_q$.
\end{ntheorem}

\begin{nexample}
  \label{example_incidences_lines_points_pg_2_3}
  Let $B$ be the incidence matrix between lines and points in $\PG(2,3)$ and $T$ be  
  the column transformation matrix of the Smith normal form of $B$, i.e.\  
  $$
  B=
  {\tiny\arraycolsep=0.3\arraycolsep\ensuremath{\left(\begin{array}{rrrrrrrrrrrrr}
  1 & 0 & 0 & 1 & 0 & 0 & 1 & 0 & 0 & 1 & 0 & 0 & 0 \\
  1 & 0 & 0 & 0 & 1 & 0 & 0 & 0 & 1 & 0 & 1 & 0 & 0 \\
  1 & 0 & 0 & 0 & 0 & 1 & 0 & 1 & 0 & 0 & 0 & 1 & 0 \\
  0 & 1 & 0 & 0 & 1 & 0 & 0 & 1 & 0 & 1 & 0 & 0 & 0 \\
  0 & 1 & 0 & 0 & 0 & 1 & 1 & 0 & 0 & 0 & 1 & 0 & 0 \\
  0 & 1 & 0 & 1 & 0 & 0 & 0 & 0 & 1 & 0 & 0 & 1 & 0 \\ 
  0 & 0 & 1 & 0 & 0 & 1 & 0 & 0 & 1 & 1 & 0 & 0 & 0 \\
  0 & 0 & 1 & 1 & 0 & 0 & 0 & 1 & 0 & 0 & 1 & 0 & 0 \\
  0 & 0 & 1 & 0 & 1 & 0 & 1 & 0 & 0 & 0 & 0 & 1 & 0 \\ 
  1 & 1 & 1 & 0 & 0 & 0 & 0 & 0 & 0 & 0 & 0 & 0 & 1 \\
  0 & 0 & 0 & 1 & 1 & 1 & 0 & 0 & 0 & 0 & 0 & 0 & 1 \\
  0 & 0 & 0 & 0 & 0 & 0 & 1 & 1 & 1 & 0 & 0 & 0 & 1 \\ 
  0 & 0 & 0 & 0 & 0 & 0 & 0 & 0 & 0 & 1 & 1 & 1 & 1
  \end{array}\right)}}
  \quad\text{and}\quad
   T=
    {\tiny\arraycolsep=0.3\arraycolsep\ensuremath{\left(\begin{array}{rrrrrrrrrrrrr}
    1 & 0 & 0 & 0 & 0 & 0 & 0 &  0 & -2 & -1 & -1 & -2 &  -3 \\ 
    0 & 1 & 0 & 0 & 0 & 0 & 0 & -1 & -1 & -1 & -2 & -1 & -11 \\ 
    0 & 0 & 1 & 0 & 0 & 0 & 0 & -2 &  0 & -1 &  0 &  0 & -11 \\  
    0 & 0 & 0 & 1 & 0 & 0 & 0 & -2 &  0 & -2 & -1 &  0 &  -7 \\ 
    0 & 0 & 0 & 0 & 1 & 0 & 0 &  0 & -2 & -2 &  0 & -1 &  -3 \\ 
    0 & 0 & 0 & 0 & 0 & 1 & 0 & -1 & -1 & -2 & -2 & -2 &  -3 \\ 
    0 & 0 & 0 & 0 & 0 & 1 & 1 & -1 & -1 & -3 & -3 & -3 &  -7 \\ 
    0 & 0 & 0 & 0 & 0 & 0 & 0 &  1 &  0 &  0 &  0 &  0 &  -3 \\ 
    0 & 0 & 0 & 0 & 0 & 0 & 0 &  0 &  1 &  0 &  0 &  0 &  -3 \\ 
    0 & 0 & 0 & 0 & 0 & 0 & 1 &  0 &  0 &  0 & -1 & -1 &  -7 \\ 
    0 & 0 & 0 & 0 & 0 & 0 & 0 &  0 &  0 &  0 &  1 &  0 &  -3 \\ 
    0 & 0 & 0 & 0 & 0 & 0 & 0 &  0 &  0 &  0 &  0 &  1 &  -3 \\ 
    0 & 0 & 0 & 0 & 0 & 0 & 0 &  0 &  0 &  0 &  0 &  0 &   1 
    \end{array}\right)}}.
  $$
  Here five invariant factors of $B$ equal $3$, the last equals $12$, and the first seven equal $1$.  
  So, for $z=\left(z_1,\dots,z_{13}\right)\in\Z^{1\times [3]_3}$ we have  $z\in M:=\left\{x^\top B\,:\, x\in \Z^{[3]_3}\right\}$ iff 
  {\footnotesize\begin{eqnarray*} 
    2z_2+z_3+z_4+2z_6+2z_7+z_8 & \equiv& 0\pmod 3,\\	   	      	
    z_1+2z_2+z_5+2z_6+2z_7+z_9 & \equiv& 0\pmod 3,\\     	
    2z_1+2z_2+2z_3+z_4+z_5+z_6 & \equiv& 0\pmod 3,\\	   	   	   	      	
    2z_1+z_2+2z_4+z_6+2z_{10}+z_{11} & \equiv& 0\pmod 3,\\
    z_1+2z_2+2z_5+z_6+2z_{10}+z_{12} & \equiv& 0\pmod 3,\\
    z_2+z_3+2z_4+2z_7+2z_{10}+z_{13} & \equiv& 0\pmod 3,\\
    \sum_{i=1}^{13} z_i &\equiv& 0\pmod 4,
  \end{eqnarray*}}
  where we have broken up the condition of the last column of $T$ modulo $12$ into two conditions modulo $3$ and modulo $4$, respectively.
\end{nexample}

Since $[k]_q$ is coprime to $p=\operatorname{char}(\F_q)$, Theorem~\ref{theorem_invariant_factors} suggests splitting the implied conditions into several $\!\!\!\mod p^l$ 
conditions (for possibly different values of $l$) and a single $\!\!\!\mod [k]_q$ condition. The latter just reflects the fact that we have $\#\cM\equiv 0\pmod {[k]_q}$ 
for a multiset of points arising as a union of $k$-spaces in $\PG(v-1,q)$, noting that this isn't true any more if we consider multisets of subspaces with dimension at least $k$ 
and assume $v>k$. The $\!\!\!\mod 3$ conditions in Example~\ref{example_incidences_lines_points_pg_2_3} may be a bit harder to guess. We observe that the sum of the first two 
conditions is equivalent to $\sum_{i=1}^9 z_i\equiv 0\pmod 3$, so that it might be worthwhile to look for different (irredundant) descriptions of $M$. Also in the generic case
of an arbitrary integer matrix $B$ the Chinese remainder theorem allows us to restrict our considerations to $\!\!\!\mod p^l$ conditions. We now formalize the underlying
objects in terms of codes over integer residue rings, see e.g.~\cite{blake1975codes}.

\begin{ndefinition}
  Let $p$ be a prime, $l$ be a positive integer, and $B\in R^{m\times n}$ a matrix where $R=\Z$ or $R=\Z/p^l\Z$. The \emph{(linear) $\Z_{p^l}$-code} $C$ of $B$ is given by the row span 
  of $B$ w.r.t.\  $\Z/p^l\Z$. The matrix $B$ is called a \emph{generator matrix} of $C$. The \emph{dual code} $C^\perp$ consists of all row vectors that are orthogonal to all elements 
  in $C$ (w.r.t.\ $\Z/p^l\Z$). We also call $C^\perp$ the \emph{$\Z_{p^l}$-kernel} of $B$.
\end{ndefinition}  

\begin{nexample}
  Let $C$ be the $\Z_3$-code of the incidence matrix $B$ between lines and points in $\PG(2,3)$ as in Example~\ref{example_incidences_lines_points_pg_2_3}. Generator matrices 
  of $C$ and its dual code $C^\perp$ are e.g.\ given by
  $$
  {\tiny\arraycolsep=0.3\arraycolsep\ensuremath{\left(\begin{array}{rrrrrrrrrrrrr}
  1 & 0 & 0 & 0 & 0 & 1 & 0 & 1 & 0 & 0 & 0 & 1 & 0 \\
  0 & 1 & 0 & 0 & 0 & 1 & 0 & 2 & 2 & 0 & 1 & 0 & 2 \\
  0 & 0 & 1 & 0 & 0 & 1 & 0 & 0 & 1 & 0 & 2 & 2 & 2 \\
  0 & 0 & 0 & 1 & 0 & 2 & 0 & 1 & 2 & 0 & 2 & 1 & 1 \\
  0 & 0 & 0 & 0 & 1 & 2 & 0 & 2 & 1 & 0 & 1 & 2 & 0 \\
  0 & 0 & 0 & 0 & 0 & 0 & 1 & 1 & 1 & 0 & 0 & 0 & 1 \\
  0 & 0 & 0 & 0 & 0 & 0 & 0 & 0 & 0 & 1 & 1 & 1 & 1 \\
  \end{array}\right)}}
  \quad\text{and}\quad 
  {\tiny\arraycolsep=0.3\arraycolsep\ensuremath{\left(\begin{array}{rrrrrrrrrrrrr}
  1 & 0 & 0 & 0 & 0 & 1 & 0 & 1 & 0 & 2 & 2 & 0 & 2 \\
  0 & 1 & 0 & 0 & 0 & 1 & 0 & 2 & 2 & 0 & 1 & 0 & 2 \\
  0 & 0 & 1 & 0 & 0 & 1 & 0 & 0 & 1 & 0 & 2 & 2 & 2 \\
  0 & 0 & 0 & 1 & 0 & 2 & 0 & 1 & 2 & 2 & 1 & 0 & 0 \\
  0 & 0 & 0 & 0 & 1 & 2 & 0 & 2 & 1 & 0 & 1 & 2 & 0 \\
  0 & 0 & 0 & 0 & 0 & 0 & 1 & 1 & 1 & 2 & 2 & 2 & 0
  \end{array}\right)}},
  $$ 
  respectively. Clearly, we have $\dim(C)+\dim(C^\perp)=[3]_3=13$ since $C$ and $C^\perp$ are vector spaces. Now consider 
  the incidence matrix 
  $$
  B'=
  {\tiny\arraycolsep=0.3\arraycolsep\ensuremath{\left(\begin{array}{rrrrrrrrrrrrr}
  0 & 1 & 1 & 0 & 1 & 1 & 0 & 1 & 1 & 0 & 1 & 1 & 1 \\
  0 & 1 & 1 & 1 & 0 & 1 & 1 & 1 & 0 & 1 & 0 & 1 & 1 \\
  0 & 1 & 1 & 1 & 1 & 0 & 1 & 0 & 1 & 1 & 1 & 0 & 1 \\ 
  1 & 0 & 1 & 1 & 0 & 1 & 1 & 0 & 1 & 0 & 1 & 1 & 1 \\
  1 & 0 & 1 & 1 & 1 & 0 & 0 & 1 & 1 & 1 & 0 & 1 & 1 \\
  1 & 0 & 1 & 0 & 1 & 1 & 1 & 1 & 0 & 1 & 1 & 0 & 1 \\ 
  1 & 1 & 0 & 1 & 1 & 0 & 1 & 1 & 0 & 0 & 1 & 1 & 1 \\
  1 & 1 & 0 & 0 & 1 & 1 & 1 & 0 & 1 & 1 & 0 & 1 & 1 \\
  1 & 1 & 0 & 1 & 0 & 1 & 0 & 1 & 1 & 1 & 1 & 0 & 1 \\ 
  0 & 0 & 0 & 1 & 1 & 1 & 1 & 1 & 1 & 1 & 1 & 1 & 0 \\
  1 & 1 & 1 & 0 & 0 & 0 & 1 & 1 & 1 & 1 & 1 & 1 & 0 \\
  1 & 1 & 1 & 1 & 1 & 1 & 0 & 0 & 0 & 1 & 1 & 1 & 0 \\ 
  1 & 1 & 1 & 1 & 1 & 1 & 1 & 1 & 1 & 0 & 0 & 0 & 0 \\
  \end{array}\right)}}  
  $$
  between affine planes and points in $\PG(2,3)$. By $C'$ we denote the corresponding $\Z_3$-code and observe that every row of $B'$ is orthogonal to every row of $B$ w.r.t.\ $\Z/3\Z$, 
  i.e.\ $C'\subseteq C^\perp$. By e.g.\ computing the Hermite normal form of $B'$ we can verify $\dim(C')=6$, so that indeed $C'=C^\perp$. In the context of 
  Example~\ref{example_incidences_lines_points_pg_2_3} this means that we can replace the six $\!\!\mod 3$-conditions by $B'z\equiv 0\pmod 3$, which corresponds to thirteen 
  single $\!\!\mod 3$-conditions. Of course we can also select six of these such that the corresponding rows of $B'$ generate $C'$.  
\end{nexample}

\begin{nlemma}
  \label{lemma_full_affine_space_in_kernel}
  For positive integers $k,k',v$ with $l:=k+k'-v-1\ge 1$ the incidence vector of an affine $k'$-space is contained in the $\Z_{q^l}$-kernel of the incidence matrix of 
  $k$-spaces and points in $\PG(v-1,q)$. 
\end{nlemma}
\begin{proof}
  We describe an affine $k'$-space $A'$ by a $k'$-space $S'$ and a $(k'-1)$-space $L'$ with $L'\le S'$, i.e.\ $\chi_{A'}=\chi_{S'}-\chi_{L'}$ or $A'=S'\backslash L'$. Now let $S$ be an 
  arbitrary $k$-space and $I:=S'\cap S$, so that $i:=\dim(I)\ge k+k'-v$. If $i\le 0$ or $I\le L'$, then we have $\#(A'\cap S)=0$ and $\#(A'\cap S)=q^{i-1}$ otherwise. Thus, we have 
  $\#(A'\cap S)\equiv 0\pmod {q^l}$ in all cases.
\end{proof}
Note that the statement for $l=0$ is trivially true.

\begin{nexample}
    \label{example_incidences_planes_points_pg_3_2}
  Let $B$ be the incidence matrix between planes and points in $\PG(3,2)$ and $T$ be  
  the column transformation matrix of the Smith normal form of $B$, i.e.\
  $$
  B=
  {\tiny\arraycolsep=0.3\arraycolsep\ensuremath{\left(\begin{array}{rrrrrrrrrrrrrrr}
  1 & 0 & 1 & 0 & 1 & 0 & 1 & 0 & 1 & 0 & 1 & 0 & 1 & 0 & 0 \\
  1 & 0 & 0 & 1 & 1 & 0 & 0 & 1 & 1 & 0 & 0 & 1 & 0 & 1 & 0 \\
  1 & 0 & 1 & 0 & 0 & 1 & 0 & 1 & 0 & 1 & 0 & 1 & 1 & 0 & 0 \\ 
  1 & 0 & 0 & 1 & 0 & 1 & 1 & 0 & 0 & 1 & 1 & 0 & 0 & 1 & 0 \\
  0 & 1 & 0 & 1 & 0 & 1 & 0 & 1 & 1 & 0 & 1 & 0 & 1 & 0 & 0 \\
  0 & 1 & 1 & 0 & 0 & 1 & 1 & 0 & 1 & 0 & 0 & 1 & 0 & 1 & 0 \\ 
  0 & 1 & 0 & 1 & 1 & 0 & 1 & 0 & 0 & 1 & 0 & 1 & 1 & 0 & 0 \\
  0 & 1 & 1 & 0 & 1 & 0 & 0 & 1 & 0 & 1 & 1 & 0 & 0 & 1 & 0 \\
  1 & 1 & 0 & 0 & 1 & 1 & 0 & 0 & 1 & 1 & 0 & 0 & 0 & 0 & 1 \\ 
  1 & 1 & 0 & 0 & 0 & 0 & 1 & 1 & 0 & 0 & 1 & 1 & 0 & 0 & 1 \\
  0 & 0 & 1 & 1 & 0 & 0 & 1 & 1 & 1 & 1 & 0 & 0 & 0 & 0 & 1 \\
  0 & 0 & 1 & 1 & 1 & 1 & 0 & 0 & 0 & 0 & 1 & 1 & 0 & 0 & 1 \\
  1 & 1 & 1 & 1 & 0 & 0 & 0 & 0 & 0 & 0 & 0 & 0 & 1 & 1 & 1 \\
  0 & 0 & 0 & 0 & 1 & 1 & 1 & 1 & 0 & 0 & 0 & 0 & 1 & 1 & 1 \\
  0 & 0 & 0 & 0 & 0 & 0 & 0 & 0 & 1 & 1 & 1 & 1 & 1 & 1 & 1 \\
  \end{array}\right)}}
  \quad\text{and}\quad
  T=
  {\tiny\arraycolsep=0.3\arraycolsep\ensuremath{\left(\begin{array}{rrrrrrrrrrrrrrr}
  1 & 0 & 0 & 0 & 0 & 0 & 0 & 0 & -1 & 0 & 0 & -1 & -3 & -3 & -13 \\  
  0 & 1 & 0 & 0 & 0 & 0 & -1 & -1 & -1 & -1 & -1 & -3 & -1 & -1 & -6 \\ 
  0 & 0 & 1 & 0 & 0 & -1 & 0 & -1 & 0 & -1 & -1 & -3 & -1 & -2 & -20 \\  
  0 & 0 & 0 & 1 & 0 & -1 & -1 & 0 & 0 & 0 & -1 & -1 & -3 & -2 & -27 \\ 
  0 & 0 & 0 & 1 & 1 & -1 & -1 & -1 & 0 & 0 & -1 & -4 & -5 & -2 & -41 \\  
  0 & 0 & 0 & 0 & 0 & 1 & 0 & 0 & 0 & 0 & 0 & 0 & -1 & 0 & -6 \\ 
  0 & 0 & 0 & 0 & 0 & 0 & 1 & 0 & 0 & 0 & 0 & 0 & -1 & -1 & -6 \\  
  0 & 0 & 0 & 0 & 0 & 0 & 0 & 1 & 0 & 0 & 0 & 0 & -1 & -1 & -13 \\ 
  0 & 0 & 0 & 0 & 1 & 0 & 0 & 0 & 1 & 0 & 0 & -3 & -4 & -1 & -34 \\  
  0 & 0 & 0 & 0 & 0 & 0 & 0 & 0 & 1 & 1 & 0 & -1 & -2 & -1 & -13 \\ 
  0 & 0 & 0 & 0 & 0 & 0 & 0 & 0 & 0 & 1 & 1 & -1 & -2 & -2 & -13 \\  
  0 & 0 & 0 & 0 & 0 & 0 & 0 & 0 & 0 & 0 & 0 & 1 & 0 & 0 & -6 \\ 
  0 & 0 & 0 & 0 & 0 & 0 & 0 & 0 & 0 & 0 & 1 & 0 & 0 & -1 & -13 \\  
  0 & 0 & 0 & 0 & 0 & 0 & 0 & 0 & 0 & 0 & 0 & 0 & 0 & 1 & -6 \\  
  0 & 0 & 0 & 0 & 0 & 0 & 0 & 0 & 0 & 0 & 0 & 0 & 0 & 0 & 1
  \end{array}\right)}}\!. 
  $$
  Here six invariant factors of $B$ equal $2$, three equal $4$, the last equals $28$, and the first five equal $1$. Generator matrices 
  for the $\Z_4$-code $C$ of $B$ and its dual code $C^\perp$ are e.g.\ given by
  $$
  {\tiny\arraycolsep=0.3\arraycolsep\ensuremath{\left(\begin{array}{rrrrrrrrrrrrrrr}
  \textbf{1} & 0 & 0 & 1 & 0 & 1 & 1 & 0 & 0 & 1 & 1 & 0 & 0 & 1 & 0 \\
  0 & \textbf{1} & 0 & 1 & 0 & 1 & 0 & 1 & 0 & 1 & 0 & 1 & 0 & 1 & 1 \\
  0 & 0 & \textbf{1} & 1 & 0 & 0 & 1 & 1 & 0 & 0 & 1 & 1 & 1 & 1 & 0 \\
  0 & 0 & 0 & \textbf{2} & 0 & 0 & 0 & 2 & 0 & 0 & 0 & 2 & 0 & 2 & 0 \\
  0 & 0 & 0 & 0 & \textbf{1} & 1 & 1 & 1 & 0 & 0 & 0 & 0 & 1 & 1 & 1 \\
  0 & 0 & 0 & 0 & 0 & \textbf{2} & 0 & 2 & 0 & 0 & 0 & 0 & 0 & 2 & 2 \\
  0 & 0 & 0 & 0 & 0 & 0 & \textbf{2} & 2 & 0 & 0 & 0 & 0 & 0 & 0 & 2 \\
  0 & 0 & 0 & 0 & 0 & 0 & 0 & 0 & 0 & 0 & 0 & 0 & 0 & 0 & 0 \\
  0 & 0 & 0 & 0 & 0 & 0 & 0 & 0 & \textbf{1} & 1 & 1 & 1 & 1 & 1 & 1 \\
  0 & 0 & 0 & 0 & 0 & 0 & 0 & 0 & 0 & \textbf{2} & 0 & 2 & 0 & 2 & 2 \\
  0 & 0 & 0 & 0 & 0 & 0 & 0 & 0 & 0 & 0 & \textbf{2} & 2 & 0 & 0 & 2 \\
  0 & 0 & 0 & 0 & 0 & 0 & 0 & 0 & 0 & 0 & 0 & 0 & 0 & 0 & 0 \\
  0 & 0 & 0 & 0 & 0 & 0 & 0 & 0 & 0 & 0 & 0 & 0 & \textbf{2} & 2 & 2 \\
  0 & 0 & 0 & 0 & 0 & 0 & 0 & 0 & 0 & 0 & 0 & 0 & 0 & 0 & 0 \\
  0 & 0 & 0 & 0 & 0 & 0 & 0 & 0 & 0 & 0 & 0 & 0 & 0 & 0 & 0 \\
  \end{array}\right)}} 
  \quad\text{and}\quad
  {\tiny\arraycolsep=0.3\arraycolsep\ensuremath{\left(\begin{array}{rrrrrrrrrrrrrrr}
  \textbf{1} & 0 & 0 & 1 & 0 & 1 & 1 & 0 & 1 & 0 & 0 & 1 & 1 & 0 & 1 \\
  0 & \textbf{1} & 0 & 1 & 0 & 1 & 0 & 1 & 0 & 1 & 0 & 1 & 0 & 1 & 1 \\
  0 & 0 & \textbf{1} & 1 & 0 & 0 & 1 & 1 & 0 & 0 & 1 & 1 & 1 & 1 & 0 \\
  0 & 0 & 0 & \textbf{2} & 0 & 0 & 0 & 2 & 0 & 0 & 0 & 2 & 0 & 2 & 0 \\
  0 & 0 & 0 & 0 & \textbf{1} & 1 & 1 & 1 & 1 & 1 & 1 & 1 & 0 & 0 & 0 \\
  0 & 0 & 0 & 0 & 0 & \textbf{2} & 0 & 2 & 0 & 0 & 0 & 0 & 0 & 2 & 2 \\
  0 & 0 & 0 & 0 & 0 & 0 & \textbf{2} & 2 & 0 & 0 & 0 & 0 & 2 & 2 & 0 \\
  0 & 0 & 0 & 0 & 0 & 0 & 0 & 0 & 0 & 0 & 0 & 0 & 0 & 0 & 0 \\
  0 & 0 & 0 & 0 & 0 & 0 & 0 & 0 & \textbf{2} & 0 & 0 & 2 & 2 & 0 & 2 \\
  0 & 0 & 0 & 0 & 0 & 0 & 0 & 0 & 0 & \textbf{2} & 0 & 2 & 0 & 2 & 2 \\
  0 & 0 & 0 & 0 & 0 & 0 & 0 & 0 & 0 & 0 & \textbf{2} & 2 & 2 & 2 & 0 \\
  0 & 0 & 0 & 0 & 0 & 0 & 0 & 0 & 0 & 0 & 0 & 0 & 0 & 0 & 0 \\
  0 & 0 & 0 & 0 & 0 & 0 & 0 & 0 & 0 & 0 & 0 & 0 & 0 & 0 & 0 \\
  0 & 0 & 0 & 0 & 0 & 0 & 0 & 0 & 0 & 0 & 0 & 0 & 0 & 0 & 0 \\
  0 & 0 & 0 & 0 & 0 & 0 & 0 & 0 & 0 & 0 & 0 & 0 & 0 & 0 & 0 \\
  \end{array}\right)}}
  $$ 
  with invariant factors
  $$[1,1,1,1,1,2,2,2,2,2,2,0,0,0,0]\quad\text{and}\quad[1,1,1,1,2,2,2,2,2,2,0,0,0,0,0],$$
  which we abbreviate as $1^52^6 0^4$ and $1^42^60^5$. 
  The rows of the stated generator matrix for $C^\perp$ correspond to ten $\!\!\!\mod 4$-conditions, where six can be rewritten to $\!\!\!\mod 2$-conditions. We remark 
  that the four rows with a $1$ as leading coefficient are incidence vectors of affine solids and the six rows with a $2$ as leading coefficient are twice 
  the incidence vector of an affine plane.
\end{nexample}

\begin{ndefinition}
  Let $B\in R^{m\times n}$ be a non-empty matrix and $C$ be the generated $\Z_{p^l}$-code, where $p$ is a prime, $l$ a positive integer, and 
  $R=\Z$ or $R=\Z/p^l\Z$. Let $d_1,\dots,d_n$ be the invariant factors of $B$ w.r.t.\ $\Z/p^l\Z$, $k_i:=\#\left\{1\le j\le n\,:\, d_j=p^i\right\}$ for $0\le i\le l-1$, 
  and $k_l:=\#\left\{1\le j\le n\,:\, d_j=0\right\}$. Then $[k_0,\dots,k_l]$ is the \emph{type} (or, more precisely, the $p^l$-type) of $C$. 
\end{ndefinition}

With the help of the Smith normal form we directly see that for a $\Z_{p^l}$-code $C$ with type $[k_0,\dots,k_l]$ the type of the dual code 
$C^\perp$ is just the reversal $[k_l,\dots,k_0]$, cf.\ \cite{calderbank1995modular}. So, the notion of the type of a $\Z_{p^l}$-code generalizes the notion 
of the dimension of a $\Z_p$-code as an invariant. In our context we can e.g.\ use it to deduce that the incidence vectors of the affine solids in $\PG(3,2)$ span 
the $\Z_4$-kernel of the incidence matrix between planes and points in $\PG(3,2)$ by verifying that the $\Z_4$-code of the affine solids in $\PG(3,2)$ has type 
$[4,6,5]$.   


\begin{table}[htp]
\begin{center} 
  \begin{tabular}{|c|ccccccc|}
    \hline
     & $\begin{pmatrix}0\\0\\1\end{pmatrix}$ & $\begin{pmatrix}0\\1\\0\end{pmatrix}$ & $\begin{pmatrix}0\\1\\1\end{pmatrix}$ & $\begin{pmatrix}1\\0\\0\end{pmatrix}$ & 
       $\begin{pmatrix}1\\0\\1\end{pmatrix}$ & $\begin{pmatrix}1\\1\\0\end{pmatrix}$ & $\begin{pmatrix}1\\1\\1\end{pmatrix}$\\ 
    \hline
    $\begin{pmatrix}0&0&1\end{pmatrix}$ & 0 & 1 & 0 & 1 & 0 & 1 & 0 \\
    $\begin{pmatrix}0&1&0\end{pmatrix}$ & 1 & 0 & 0 & 1 & 1 & 0 & 0 \\
    $\begin{pmatrix}0&1&1\end{pmatrix}$ & 0 & 0 & 1 & 1 & 0 & 0 & 1 \\ 
    $\begin{pmatrix}1&0&0\end{pmatrix}$ & 1 & 1 & 1 & 0 & 0 & 0 & 0 \\
    $\begin{pmatrix}1&0&1\end{pmatrix}$ & 0 & 1 & 0 & 0 & 1 & 0 & 1 \\ 
    $\begin{pmatrix}1&1&0\end{pmatrix}$ & 1 & 0 & 0 & 0 & 0 & 1 & 1 \\
    $\begin{pmatrix}1&1&1\end{pmatrix}$ & 0 & 0 & 1 & 0 & 1 & 1 & 0 \\
    \hline
  \end{tabular} 
  \quad\quad
  \begin{tabular}{|rrrrrrr|}
    \hline
    0 & \textbf 1 & 0 & 0 & 1 & 0 & 1 \\
    \textbf 1 & 0 & 0 & 0 & 0 & -1 & 3 \\
    0 & 0 & \textbf 1 & 0 & 1 & 1 & 0 \\
    0 & 0 & 0 & 0 & 0 & 0 & \textbf 6 \\
    0 & 0 & 0 & \textbf 1 & -1 & 1 & -1 \\
    0 & 0 & 0 & 0 & \textbf 2 & 0 & -2 \\
    0 & 0 & 0 & 0 & 0 & \textbf 2 & -2 \\
    \hline
  \end{tabular} 
  \quad\quad
  \begin{tabular}{|rrrrrrr|}
    \hline
    0 & \textbf 1 & 0 & 0 & 1 & 0 & 1 \\
    \textbf 1 & 0 & 0 & 0 & 0 & 1 & 1 \\
    0 & 0 & \textbf 1 & 0 & 1 & 1 & 0 \\
    0 & 0 & 0 & 0 & 0 & 0 & 0 \\
    0 & 0 & 0 & \textbf 1 & 1 & 1 & 1 \\
    0 & 0 & 0 & 0 & 0 & 0 & 0 \\
    0 & 0 & 0 & 0 & 0 & 0 & 0 \\
    \hline
  \end{tabular}
  \caption{The incidence matrix between points and lines in $\PG(2,2)$ and reductions of its transpose.}
  \label{table_incidences_between_points_and_lines}
\end{center}  
\end{table}

\begin{nexample}
\label{example_kernel_gauss}
If we label the points by generating row vectors and the hyperplanes by orthogonal column vectors, then the incidence matrix $A$ between points and hyperplanes in 
$\PG(2,2)$ is given on the left hand side of Table~\ref{table_incidences_between_points_and_lines}. We now apply the Gaussian elimination algorithm to the transposed 
matrix $A^\top$ without swapping rows or columns. In order to make results applicable for different domains, we perform all computations over $\Z$. More precisely, 
we only allow multiplications or divisions by the units $\{-1,1\}$ in $\Z$ and adding the $\lambda$-fold of a row to another row is only permitted for $\lambda\in\Z$. 
The result is displayed in the middle of Table~\ref{table_incidences_between_points_and_lines}. With this we can conclude that the $\R$-rank of $A^{\top}$ is $7$ and the 
corresponding kernel of $A^\top$ has dimension zero. Reducing modulo $2$ gives the result for the computations over $\F_2$, see the matrix on the right hand side of 
Table~\ref{table_incidences_between_points_and_lines}. I.e., the $2$-rank of $A^{\top}$ is four and the corresponding kernel of $A^{\top}$ has dimension three. 
The associated $2^3=8$ necessary conditions $\by^\top \bz=\zv$ are given by:
{\small\begin{eqnarray*}
  \cM(\left\langle\begin{pmatrix}0&0&1\end{pmatrix}\right\rangle)+\cM(\left\langle\begin{pmatrix}0&1&0\end{pmatrix}\right\rangle)+\cM(\left\langle\begin{pmatrix}1&0&0\end{pmatrix}\right\rangle)+\cM(\left\langle\begin{pmatrix}1&1&1\end{pmatrix}\right\rangle) &\equiv& 0 \pmod 2\\ 
  \cM(\left\langle\begin{pmatrix}0&0&1\end{pmatrix}\right\rangle)+\cM(\left\langle\begin{pmatrix}0&1&1\end{pmatrix}\right\rangle)+\cM(\left\langle\begin{pmatrix}1&0&0\end{pmatrix}\right\rangle)+\cM(\left\langle\begin{pmatrix}1&1&0\end{pmatrix}\right\rangle) &\equiv& 0 \pmod 2\\ 
  \cM(\left\langle\begin{pmatrix}0&1&0\end{pmatrix}\right\rangle)+\cM(\left\langle\begin{pmatrix}0&1&1\end{pmatrix}\right\rangle)+\cM(\left\langle\begin{pmatrix}1&1&0\end{pmatrix}\right\rangle)+\cM(\left\langle\begin{pmatrix}1&1&1\end{pmatrix}\right\rangle) &\equiv& 0 \pmod 2\\ 
  \cM(\left\langle\begin{pmatrix}0&1&0\end{pmatrix}\right\rangle)+\cM(\left\langle\begin{pmatrix}0&1&1\end{pmatrix}\right\rangle)+\cM(\left\langle\begin{pmatrix}1&0&0\end{pmatrix}\right\rangle)+\cM(\left\langle\begin{pmatrix}1&0&1\end{pmatrix}\right\rangle) &\equiv& 0 \pmod 2\\ 
  \cM(\left\langle\begin{pmatrix}0&0&1\end{pmatrix}\right\rangle)+\cM(\left\langle\begin{pmatrix}0&1&1\end{pmatrix}\right\rangle)+\cM(\left\langle\begin{pmatrix}1&0&1\end{pmatrix}\right\rangle)+\cM(\left\langle\begin{pmatrix}1&1&1\end{pmatrix}\right\rangle) &\equiv& 0 \pmod 2\\ 
  \cM(\left\langle\begin{pmatrix}0&0&1\end{pmatrix}\right\rangle)+\cM(\left\langle\begin{pmatrix}0&1&0\end{pmatrix}\right\rangle)+\cM(\left\langle\begin{pmatrix}1&0&1\end{pmatrix}\right\rangle)+\cM(\left\langle\begin{pmatrix}1&1&0\end{pmatrix}\right\rangle) &\equiv& 0 \pmod 2\\ 
  \cM(\left\langle\begin{pmatrix}1&0&0\end{pmatrix}\right\rangle)+\cM(\left\langle\begin{pmatrix}1&0&1\end{pmatrix}\right\rangle)+\cM(\left\langle\begin{pmatrix}1&1&0\end{pmatrix}\right\rangle)+\cM(\left\langle\begin{pmatrix}1&1&1\end{pmatrix}\right\rangle) &\equiv& 0 \pmod 2\\ 
  0 &\equiv& 0 \pmod 2\\ 
\end{eqnarray*}}
Note that the four points occurring in one of the first seven equations form an affine plane in each case, i.e., the complement is one of the $\qbin{3}{2}{2}=7$ lines in $\PG(2,2)$. This 
is not a coincidence as we will see in the subsequent remark.        
\end{nexample}
     
\begin{nremark}
  \label{remark_kernel_subfield_subcodes}
  For a prime $p$ the incidence matrix between the points and the $k$-spaces in $\PG(v-1,p)$ is the generator matrix of 
  a projective generalized Reed--Muller code, see e.g.\ \cite[Theorem 5.41]{assmus1998polynomial}. A few simplified formulas for the Hamada $p$-rank formula 
  for special cases and explicit bases can e.g.\ be found in \cite[Section 5.9]{assmus1998polynomial}, see also the survey~\cite{ceccherini1992dimension}. For 
  a prime power $q$ these so-called \emph{geometric codes} admit a representation by polynomial functions, which is a rather natural description for generalized 
  Reed--Muller codes, see \cite{glynn1995classification} for the details.\footnote{A similar statement also applies to e.g.\ codes obtained from Hermitian varieties, 
  see e.g.\ \cite{key1991hermitian} or \cite[Theorem 2.30]{bartoli2012constructions}.}
  
  For a direct entry observe that the intersection of each $k$-space $K$ with an arbitrary subspace $S$ of codimension at most $k-1$ consists of $[i]_q$ points for 
  some integer $1\le i\le k$ and that these numbers all are congruent to $1$ modulo $q$. Going over to complements we end up with numbers of points that are congruent 
  to zero modulo $q$. If $q=p^2$ we can apply the same idea using Baer subspaces and in general we have to consider subfield subcodes, see e.g.\ 
  \cite[Section 5.8]{assmus1998polynomial}. To compute the dimension of the resulting span and to compare it with the Hamada formula or one of its 
  simplifications in special cases then is the, of course unavoidable, technical part if an exhaustive classification is desired. 
\end{nremark}     
     
\begin{nexample}
  Consider the incidences between points and planes in $\PG(3,2)$. The incidence matrix $A$ and the resulting matrix after applying the Gaussian elimination algorithm 
  over $\Z$, see Example~\ref{example_kernel_gauss} for the technical details, are given by
  $${\footnotesize
  \left(\begin{array}{rrrrrrrrrrrrrrr}
    0 & 1 & 0 & 1 & 0 & 1 & 0 & 1 & 0 & 1 & 0 & 1 & 0 & 1 & 0 \\
    1 & 0 & 0 & 1 & 1 & 0 & 0 & 1 & 1 & 0 & 0 & 1 & 1 & 0 & 0 \\
    0 & 0 & 1 & 1 & 0 & 0 & 1 & 1 & 0 & 0 & 1 & 1 & 0 & 0 & 1 \\
    1 & 1 & 1 & 0 & 0 & 0 & 0 & 1 & 1 & 1 & 1 & 0 & 0 & 0 & 0 \\
    0 & 1 & 0 & 0 & 1 & 0 & 1 & 1 & 0 & 1 & 0 & 0 & 1 & 0 & 1 \\
    1 & 0 & 0 & 0 & 0 & 1 & 1 & 1 & 1 & 0 & 0 & 0 & 0 & 1 & 1 \\
    0 & 0 & 1 & 0 & 1 & 1 & 0 & 1 & 0 & 0 & 1 & 0 & 1 & 1 & 0 \\
    1 & 1 & 1 & 1 & 1 & 1 & 1 & 0 & 0 & 0 & 0 & 0 & 0 & 0 & 0 \\
    0 & 1 & 0 & 1 & 0 & 1 & 0 & 0 & 1 & 0 & 1 & 0 & 1 & 0 & 1 \\
    1 & 0 & 0 & 1 & 1 & 0 & 0 & 0 & 0 & 1 & 1 & 0 & 0 & 1 & 1 \\
    0 & 0 & 1 & 1 & 0 & 0 & 1 & 0 & 1 & 1 & 0 & 0 & 1 & 1 & 0 \\
    1 & 1 & 1 & 0 & 0 & 0 & 0 & 0 & 0 & 0 & 0 & 1 & 1 & 1 & 1 \\
    0 & 1 & 0 & 0 & 1 & 0 & 1 & 0 & 1 & 0 & 1 & 1 & 0 & 1 & 0 \\
    1 & 0 & 0 & 0 & 0 & 1 & 1 & 0 & 0 & 1 & 1 & 1 & 1 & 0 & 0 \\
    0 & 0 & 1 & 0 & 1 & 1 & 0 & 0 & 1 & 1 & 0 & 1 & 0 & 0 & 1 \\
  \end{array}\right)}$$and$${\footnotesize
  \left(\begin{array}{rrrrrrrrrrrrrrr}
    0 & \textbf 1 & 0 & 0 & 1 & 0 & 1 & 0 & 1 & 0 & 1 & -1 & 2 & -1 & 2 \\
    \textbf 1 & 0 & 0 & 0 & 0 & 1 & 1 & 0 & 0 & 1 & 1 & -1 & -1 & 2 & 2 \\
    0 & 0 & \textbf 1 & 0 & 1 & -1 & 0 & 0 & -3 & 1 & -2 & 1 & -2 & 0 & -3 \\
    0 & 0 & 0 & 0 & 0 & 0 & 0 & 0 & 0 & 0 & 16 & 0 & 0 & 0 & 12 \\
    0 & 0 & 0 & \textbf 1 & -1 & 1 & -1 & 0 & 0 & 0 & 0 & 1 & -1 & 1 & -1 \\
    0 & 0 & 0 & 0 & \textbf 2 & 0 & 0 & 0 & 0 & 0 & -6 & 0 & 0 & -2 & -8 \\
    0 & 0 & 0 & 0 & 0 & \textbf 2 & 0 & 0 & 0 & 0 & -6 & 0 & -2 & 0 & -8 \\
    0 & 0 & 0 & 0 & 0 & 0 & \textbf 2 & 0 & 0 & 0 & -6 & 0 & -2 & -2 & -6 \\
    0 & 0 & 0 & 0 & 0 & 0 & 0 & \textbf 1 & -1 & 1 & -1 & 1 & -1 & 1 & -1 \\
    0 & 0 & 0 & 0 & 0 & 0 & 0 & 0 & \textbf 2 & 0 & -2 & 0 & 2 & 0 & -2 \\
    0 & 0 & 0 & 0 & 0 & 0 & 0 & 0 & 0 & \textbf 2 & -2 & 0 & 0 & 2 & -2 \\
    0 & 0 & 0 & 0 & 0 & 0 & 0 & 0 & 0 & 0 & -12 & \textbf 2 & -2 & -2 & -14 \\
    0 & 0 & 0 & 0 & 0 & 0 & 0 & 0 & 0 & 0 & 12 & 0 & 0 & \textbf 4 & 12 \\
    0 & 0 & 0 & 0 & 0 & 0 & 0 & 0 & 0 & 0 & 12 & 0 & \textbf 4 & 0 & 12 \\
    0 & 0 & 0 & 0 & 0 & 0 & 0 & 0 & 0 & 0 & 12 & 0 & 0 & 0 & 16 \\
  \end{array}\right)}.
  $$
  So, the $2$-rank of $A$ equals $5$ and the corresponding kernel of $A^\top$ has dimension $10$. 
  The associated necessary conditions are sums over point multiplicities that are congruent to zero modulo $2$. 
  The conditions $\cM(H)-\#\cM\equiv 0\pmod 4$ from Lemma~\ref{lem:union_subspaces} for each hyperplane $H$ in $\PG(3,2)$ 
  can be concluded from the kernel approach if we compute modulo $4$. Here the matrix $A$ has the invariant factors $1^5\, 2^6\, 4^3\, 28$, so that the kernel of $A^\top$ over $\Z/4\Z$ 
  has type $[4,6,5]$ and is spanned by the incidence vectors of the affine solids, cf.\ Example~\ref{example_incidences_planes_points_pg_3_2}.  
\end{nexample}     

\medskip
  
Note that Exercise~\ref{exercise_counting_oversubspaces} can be used to deduce
\begin{equation}
  \label{eq_subspace_multiplicity}
  \cM(S) =\frac{1}{q^{v-s-1}}\cdot\left(\sum_{H\in\cH\,:\, S\le H} \cM(H)\,-\,[v-s-1]_q\cdot \#\cM\right)
\end{equation}
for a multiset of points $\cM$ in $\PG(v-1,q)$, where $S$ is an $s$-dimensional subspace with $1\le s\le v-1$. 
\begin{nexercise}
  \label{exercise_chi_subspace}
  Show
  \begin{equation}
    \chi_S=\frac{1}{q^{v-s-1}}\cdot \sum_{H\in\cH\,:\, S\le H} \chi_H\,-\, \frac{[v-s-1]_q}{q^{v-s-1}}\cdot \chi_V
  \end{equation}
  for an $s$-dimensional subspace $S$ of $V=\PG(v-1,q)$ and deduce that $\chi_S$ is $q^{s-1}$-divisible from the $q^{v-1}$-divisibility of $\chi_V$ and the 
  $q^{v-2}$-divisibility of $\chi_H$ for every $H\in\cH$.
\end{nexercise}
Since $\sum_{H\in\cH} \cM(H)=[v-1]_q \cdot \#\cM$ the point multiplicities $\cM(P)$ can be computed from the hyperplane multiplicities $\cM(H)$ and vice versa. 
Interchanging the roles of the points $P\in\cP$ and the hyperplanes $H\in\cH$ yields the so-called \emph{dual} multiset.  
\begin{nexercise}
  \label{exercise_dual_multiset_coefficients}
  Show that each multiset of points $\cM$ in $\PG(v-1,q)$ with $v\ge 2$ can be uniquely written as $\cM=\sum_{H\in\cH} \alpha_H\cdot\chi_H$ where 
  \begin{equation}
    \alpha_H=\frac{1}{q^{v-2}}\cdot\left(\cM(H)-\frac{[v-2]_q}{[v-1]_q}\cdot \#\cM\right) \in\mathbb{Q}
  \end{equation}
  for every $H\in\cH$.
\end{nexercise}
Note that $\alpha_H\ge 0$ iff $[v-1]_q\cdot\cM(H)\ge [v-2]_q\cdot \#\cM$. If $\cM$ is proper and $\alpha_H\ge 0$ for all $H\in\cH$, then we can use 
Exercise~\ref{exercise_upper_bound_proper_multiset} to deduce, using $[v-1]_q=q[v-2]_q+1$, that there exists an integer $x$ such 
that $\#\cM=x[v-1]_q$ and $\min\{\cM(H)\,:\, H\in\cH\}=x[v-2]_q$, i.e.\ $\cM$ is a $(x[v-1]_q,x[v-2]_q;v,q)$-minihyper, see Section~\ref{subsec_minihypers}. 
With this, Exercise~\ref{exercise_dual_multiset_coefficients} gives $q^{v-2}\cdot \alpha_H\in\N_0$ for every hyperplane $H$, while we only have 
$q^{v-2}[v-1]_q\cdot \alpha_H\in\mathbb{Z}$ in general. See e.g.\ \cite[Section 2]{LandjevVandendriessche2012} for more details.

\chapter{Lengths of divisible codes}
\label{sec_lengths_of_divisible_codes}
In this chapter we will consider the possible effective lengths of $\Delta$-divisible linear codes over $\F_q$. Due to Theorem~\ref{thm_delta_divides_q_power} it is 
sufficient to consider $\Delta$-divisible codes where $\Delta=q^r$ with $r\in \Q$ such that $mr\in \N_0$ for field sizes $q=p^m$. We will first consider 
the restricted case $r\in \N_0$, see \cite{kiermaier2020lengths}, and then consider the general situation $r\in\Q$, see \cite{kurz2023lengths}. Since adding zero 
coordinates to codewords does not change the divisibility, see Exercise~\ref{exercise_remove_zero_columns}, we focus on the effective lengths and not the lengths of 
$q^r$-divisible linear codes over $\F_q$. We remark that we will mostly use the geometric language, i.e., consider the possible cardinalities of $q^r$-divisible multisets 
of points in $\PG(v-1,q)$. 


There are a few very basic constructions for $q^r$-divisible multisets of points, see Example~\ref{ex_simplex_code} and Exercises~\ref{exercise_q_linear_combination} and \ref{exercise_divisibility_sum_construction}:   
\begin{nlemma}
	\label{lem:qr-div-basic}
	(\cite[Lemma 2]{kiermaier2020lengths})\\[-5mm]
	\begin{enumerate}
		\item[(i)]\label{lem:qr-div-basic:subspace} Let $U$ be an $\F_q$-vector space of dimension $k \geq 1$. The set $\spaces{U}{1}$ of $[k]_q$ points contained in $U$ is $q^{k-1}$-divisible.
		\item[(ii)]\label{lem:qr-div-basic:union} For $q^r$-divisible multisets $\cM$ and $\cM'$ in $V$, the sum (or multiset union) $\cM+\cM'$ is $q^r$-divisible.
		\item[(iii)]\label{lem:qr-div-basic:qfold} The $q$-fold repetition of a $q^r$-divisible multiset $\cM$ is $q^{r+1}$-divisible.
	\end{enumerate}
\end{nlemma}

Note that for a multiset of points $\cM_1$ in $V_1$ and a multiset of points $\cM_2$ in $V_2$ we can consider their embeddings $\cM_1',\cM_2'$ in $V_1\times V_2$ and 
consider the sum $\cM_1'+\cM_2'$ in the ambient space $V_1\times V_2$. By applying Lemma~\ref{lem:qr-div-basic} we obtain:  
\begin{nlemma}
  \label{lemma_sum}
  The set of possible cardinalities of $q^r$-divisible multisets of points over $\F_{q}$ is closed under addition.
\end{nlemma}
For each integer $r$ and each dimension $1\le i\le r+1$ the $q^{r+1-i}$-fold repetition of an $i$-space in $\PG(v-1,q)$ is a $q^r$-divisible multiset of points of 
cardinality $q^{r+1-i}\cdot [i]_q$. So, for a fixed prime power $q$, a non-negative integer $r$, and $i\in\{0,\ldots,r\}$, we define
\begin{equation}
	    \snumb{r}{i}{q}
	    := q^i\cdot[r-i+1]_q
	    = \frac{q^{r+1}-q^i}{q-1}
	    =\sum_{j=i}^r q^{j}
	    =q^i + q^{i+1} + \ldots + q^r
\end{equation}
and state:
\begin{nlemma}
	\label{lemma:snumb_card}
	For each $r\in\N_0$ and each $i\in\{0,\ldots,r\}$ there is a $q^r$-divisible multiset of points of cardinality 
	$\snumb{r}{i}{q}$.
\end{nlemma}
As a consequence of Lemma~\ref{lemma_sum} and Lemma~\ref{lemma:snumb_card} all integers $n = \sum_{i=0}^r a_i \snumb{r}{i}{q}$ with $a_i\in\mathbb{N}_0$ are realizable 
cardinalities of $q^r$-divisible multisets of points. Later on we will state in Theorem~\ref{thm_characterization_div} that these integers are indeed the only possibilities. 
E.g.\ for $q=2$ and $r=2$ the possible cardinalities are given by $\{4,6,7,8\}\cup\N_{\ge 10}$. The impossibility of cardinality $9$ was shown in Example~\ref{ex_no_4_div_card_9}.

\begin{backgroundinformation}{Frobenius coin problem}The \emph{Frobenius coin problem} \cite{brauer1942problem}, named after the German mathematician Ferdinand Georg Frobenius (1849--1917), 
asks for the largest monetary amount $F(a_1,\dots,a_r)$ that cannot be 
obtained using only coins of specified denominations in $\left\{a_1,\dots,a_r\right\}$. If $\gcd(a_1,\dots,a_r)=1$ the number $F(a_1,\dots,a_r)$ is always finite and 
we have $F(a_1,a_2)=(a_1-1)(a_2-1)-1$ and exactly $(a_1-1)(a_2-1)/2$ numbers cannot be obtained in such a way. For $r\ge 3$ no general formula is known. 
\end{backgroundinformation}
In analogy to the Frobenius coin problem we define $\frobenius{r}{q}$ as the 
smallest integer such that a $q^r$-divisible multiset of cardinality $n$ exists for all integers $n>\frobenius{r}{q}$ in $\PG(v-1,q)$ provided that the dimension $v$ is sufficiently large.  
In other words, $\frobenius{r}{q}$ is the largest integer which is not realizable as the size of a $q^r$-divisible multiset of points over $\F_q$. If all non-negative  
integers are realizable then $\frobenius{r}{q} = -1$, which is the case for $r = 0$. We have $\frobenius{2}{2}=9$ and will state a general formula for $\frobenius{r}{q}$ in Proposition~\ref{prop_frobenius}. 
For the moment we just remark that for $r\ge 1$ the numbers $\snumb{r}{r}{q} = q^r$ and $\snumb{r}{0}{q} = 1 + q + q^2 + \ldots + q^r$ are coprime, so that $\frobenius{r}{q}$ 
is indeed finite and there is only a finite set of cardinalities which is not realizable as a $q^r$-divisible multiset for every choice of $q$ and $r$. We remark that the classical 
Frobenius number is e.g.\ applied in \cite{beutelspacher1978partitions} to the existence problem of vector space partitions.

Note that the number $\snumb{r}{i}{q}$ is divisible by $q^i$, but not by $q^{i+1}$. This property allows us to create kind of a positional 
system upon the sequence of base numbers
\[
	S_q(r) := (\snumb{r}{0}{q}, \snumb{r}{1}{q},\ldots, \snumb{r}{r}{q})\text{.}
\]
Our next aim is to show that each integer $n$ has a unique \emph{$S_q(r)$-adic expansion}
\begin{equation}
	\label{eq:sqadic}
	n = \sum_{i=0}^r a_i \snumb{r}{i}{q}
\end{equation}
with $a_0,\ldots,a_{r-1}\in\{0,\ldots,q-1\}$ and \emph{leading coefficient} $a_r\in\mathbb{Z}$.
The idea is to consider Equation~\eqref{eq:sqadic} modulo $q, q^2,\ldots,q^r$ which gradually determines $a_0, a_1,\ldots,a_{r-1}\in\{0,\ldots,q-1\}$, using that 
$\snumb{r}{i}{q}$ is divisible by $q^i$, but not by $q^{i+1}$. For the existence part, we give an algorithm that computes the $S_q(r)$-adic expansion:

\begin{programcode}{Algorithm}
\noindent
$\!$\textbf{Input:} $n\in\mathbb{Z}$, field size $q$, exponent $r\in\N_0$\\
\textbf{Output:} representation $n=\sum\limits_{i=0}^r a_i \snumb{r}{i}{q}$ with $a_0,\ldots,a_{r-1}\in\{0,\ldots,q-1\}$ and $a_r\in\mathbb{Z}$\\
$m\gets n$\\
For {$i\gets 0$ To $r-1$}\\
{
\hspace*{0.6cm}$a_i\gets m\bmod q$\\
\hspace*{0.6cm}$m\gets \frac{m-a_i\cdot[r-i+1]_q}{q}$\\
}
$a_r\gets m$\\
\end{programcode}
 Here $m\bmod q$ denotes the remainder of the division of $m$ by $q$.
\begin{nexercise}
  Let $n\in\Z$ and $r\in\N_0$. Show that the above algorithm 
  computes the unique $S_q(r)$-adic expansion of $n$.
\end{nexercise}
The $S_2(2)$-adic expansion of $n=11$ is given by $11=1\cdot 7+0\cdot 6+1\cdot 4$ and the $S_2(2)$-adic expansion of $n=9$ is given by $1\cdot 7+1\cdot 6-1\cdot 4$, 
i.e., the leading coefficient is $-1$. 
\begin{nexercise}
  \label{exercise_s_q_r_adic_expansion_137}
  Compute the $S_3(3)$-adic expansion of $n=137$ and determine the leading coefficient.
\end{nexercise}
In Example~\ref{ex_no_4_div_card_9} we have shown the non-existence of $4$-divisible multisets of cardinality $9$ over $\F_2$. Using the same tools, i.e., Lemma~\ref{lemma_heritable} 
and Lemma~\ref{lemma_average}, the following characterization on the lengths of $q^r$-divisible codes and multisets can be obtained by induction:

\begin{important}{Characterization of lengths of divisible codes}
\vspace*{-6mm}
 \begin{ntheorem}{(\cite[Theorem 1]{kiermaier2020lengths})}
  \label{thm_characterization_div}
  For $n\in\Z$ and $r\in\N_0$ the following statements are equivalent:
  \begin{enumerate}
  \item[(i)]\label{thm:characterization_div:card_multiset} There exists a $q^r$-divisible multiset of points of cardinality $n$ over $\F_q$.   
  \item[(ii)]\label{thm:characterization_div:card_code} There exists a full-length $q^r$-divisible linear code of length $n$ over $\F_q$.
  \item[(iii)]\label{thm:characterization_div:n_strong} The leading coefficient of the $S_q(r)$-adic expansion of $n$ is non-negative.
  \end{enumerate}
\end{ntheorem}
\end{important}
So, the $S_q(r)$-adic expansion of $n$ provides a certificate not only for the existence, but remarkably also for  
the non-existence of a $q^r$-divisible multiset of size $n$. As computed in Exercise~\ref{exercise_s_q_r_adic_expansion_137}, the leading coefficient of 
the $S_3(3)$-adic expansion of $n=137$ is $-2$, so that there is no $27$-divisible ternary linear code of effective length $137$. 

Theorem~\ref{thm_characterization_div} allows us also to compute the Frobenius-coin-problem-like number $\frobenius{r}{q}$ as the largest integer $n$ whose $S_q(r)$-adic expansion  
$n = \sum_{i=0}^{r-1}a_i\snumb{r}{i}{q} + a_r q^r$ has leading coefficient $a_r < 0$. Clearly, this $n$ is attained by choosing $a_0 = \ldots = a_{r-1} = q-1$ and $a_r = -1$.

\begin{trailer}{Frobenius number for lengths of divisible codes}
\vspace*{-6mm}
\begin{nproposition}{(\cite[Proposition 1]{kiermaier2020lengths})}
  \label{prop_frobenius}
  For every prime power $q$ and $r\in\N_0$ we have
  \[
	  \frobenius{r}{q}= r\cdot q^{r+1} - [r+1]_q = rq^{r+1} - q^r - q^{r-1} - \ldots - 1\text{.}
  \]
\end{nproposition}
\end{trailer}

Just for the ease of a direct usage, we spell out a few implications of Theorem~\ref{thm_characterization_div} in the following.
\begin{nlemma}
  Let $n$ be the effective length of a non-trivial $2^1$-divisible code over $\F_2$. Then, we have $n\ge 2$.
\end{nlemma}
\begin{nlemma}
  Let $n$ be the effective length of a non-trivial $2^2$-divisible code over $\F_2$. Then, we have $n\in\{4,6,7,8\}$ or $n\ge 10$.
\end{nlemma}
\begin{nlemma}
  Let $n$ be the effective length of a non-trivial $2^3$-divisible code over $\F_2$. Then, we have $n\in\{8,12,14,15,16,20,22,23,24,26,27,28,29,30,31,32\}$ or $n\ge 34$.
\end{nlemma}
\begin{nlemma}
  Let $n$ be the effective length of a non-trivial $3^1$-divisible code over $\F_3$. Then, we have $n\in\{3,4\}$ or $n\ge 6$.
\end{nlemma}
\begin{nexercise}
  Show that the effective length $n$ of a non-trivial $q^r$-divisible code over $\F_q$ satisfies $n\ge q^r$ and describe the unique example where equality is attained.  
\end{nexercise}

We remark that for the cases when the field size is a proper prime power $q=p^m$ Theorem~\ref{thm_delta_divides_q_power} and Theorem~\ref{thm_characterization_div} are not sufficient 
to determine the possible lengths of $p^r$-divisible codes over $\F_q$. 
Due to Theorem~\ref{thm_delta_divides_q_power} it suffices to consider $\Delta$-divisible codes over $\F_q$ where $\Delta$ is a power of $p$. More 
concretely, we will use the parameterization $\Delta=p^{am-b}=q^{a-b/m}$ where $a,b\in\N_0$ with $a\ge 1$ and $b\le m-1$. For non-negative integers $a,b$ with $a\ge 1$, $b\le m-1$, and 
$i\in\{0,\dots,a\}$ we define
\begin{equation}
  s_q(a,b,i):=[a+1]_q
\end{equation}
if $i=0$ and
\begin{equation}
  s_q(a,b,i):=q^i\cdot [a-i+1]_q/p^b= p^{im-b}\cdot [a-i+1]_q=p^{m-b}\cdot\left(q^{i-1}+q^i+\dots+q^{a-1}\right)
\end{equation}
for $1\le i\le a$. Note that for $i\ge 1$ the number $s_q(a,b,i)$ is divisible by $p^{im-b}$ but not by $p^{im-b+1}$, where $im-b\ge 1$, and $s_q(a,b,0)$ is coprime to $p$. This property allows 
us to create kind of a positional system upon the sequence of base numbers
$$
  S_q(a,b) := \big(s_q(a,b,0),s_q(a,b,1),\dots,s_q(a,b,a)\big).
$$  
As it can be easily shown, each integer $n$ has a unique \emph{$S_q(a,b)$-adic expansion}
\begin{equation}
  n =\sum_{i=0}^a c_i \cdot s_q(a,b,i)
\end{equation}
with $c_0 \in \left\{0,\dots,p^{m-b}-1\right\}$,  $c_1,\dots,c_{a-1} \in \{0,\dots,q-1\}$ and \emph{leading coefficient} $c_a\in \Z$. 
\begin{ntheorem}{(\cite[Theorem 2]{kurz2023lengths})}
  \label{thm_characterization_div_fraction}
  Let $q=p^m$, $n \in \Z$, and $a,b\in \N_0$ with $a\ge 1$, $b\le m-1$. The following statements are equivalent:
  \begin{itemize}
    \item[(i)] There exists a $p^{am-b}$-divisible linear code of effective length $n$ over $\F_q$.
    \item[(ii)] The leading coefficient $c_a$ of the $S_q(a,b)$-adic expansion of $n$ is non-negative.
  \end{itemize}
\end{ntheorem}

\begin{nexample}
  \label{example_8_div_q_4}
  For $q=4$ and $r=\tfrac{3}{2}$ the multisets of points of an $8$-fold point, a $2$-fold line, and a plane are $4^r$-divisible of cardinalities $8$, $10$, and 
  $21$, respectively. The set of all positive integers that cannot be written as sums of $8$s, $10$s, and $21$s is given by $E_1\cup E_2$, where
  $$
    E_1=\{1,3,5,7,9,11,13,15,17,19,23,25,27,33,35,43\}
  $$
  and
  $$
    E_2=\{2,4,6,12,14,22\}.
  $$    
  Thus, $4^{3/2}$-divisible multisets of points of cardinality $n$ over $\F_4$ exist for all $n\in\N_0\backslash(E_1\cup E_2)$. 
\end{nexample}
The used constructions in 
Example~\ref{example_8_div_q_4} are rather straightforward generalizations of the situation 
of $q^r$-divisible multisets of points when $r$ is an integer. More precisely, we consider $i$-spaces $S_i$ with $1\le i\le \lceil r\rceil+1$ in order to 
construct the $q^r$-divisible multisets of points $\lceil q^{r-i+1}\rceil\cdot\chi_{S_i}$ having cardinality $\lceil q^{r-i+1}\rceil\cdot[i]_q$. Note that $q^{r-i+1}=p^{m(r-i+1)}$ is an integer for $i\le \lfloor r\rfloor+1$, while $\lceil q^{r-i+1}\rceil=1$ for $i=\lceil r\rceil+1\ge r+1$. In Theorem~\ref{thm_characterization_div_fraction} 
it turns out that the possible cardinalities of $q^r$-divisible multisets of points over $\F_q$ can always be attained by taking suitable unions of the basic constructions 
mentioned before. 


Of course, similar questions also make sense for codes over rings instead of over finite fields $\F_q$.

\section{Applications}
\label{subsec_application_q_r_div}
Now we are ready to treat the example from Subsection~\ref{subsec_introductory_application} from a more general point of view. First we need a notion of a complementary multiset of points.
\begin{ndefinition}
  Let $\cM$ be a multiset of points in $\PG(v-1,q)$ with maximum point multiplicity at most $\lambda$, i.e., $\cM(P)\le \lambda$ for all points $P\in\cP$. The \emph{$\lambda$-complement} 
  $\cM^{\complement_\lambda}$ of $\cM$ is the multiset of points in $\PG(v-1,q)$ defined by $\cM^{\complement_\lambda}(P)=\lambda-\cM(P)$ for all $P\in\cP$. 
\end{ndefinition}
If $\cM$ is the multiset of points in $\PG(9-1,2)$ corresponding to the points of $20$ solids and $30$ planes with pairwise trivial intersection, then the maximum point 
multiplicity of $\cM$ is $1$. Here we have $\#\cM=510$ and the $1$-complement $\cM^{\complement_1}$ has cardinality $1$ and also a maximum point multiplicity of $1$.

For a given ambient space $\PG(v-1,q)$ and a positive integer $\lambda$ let $\cV$ be the multiset $\lambda\cdot \cP$ defined by $\cV(P)=\lambda$ for all $P\in\cP$. Since 
$\cV$ is $\lambda q^{v-1}$-divisible, the equation $\cM+\cM^{\complement_\lambda}=\cV$ implies: 
\begin{nlemma}
  \label{lemma_t_complement}
  Let $\lambda\in\N_{>0}$ and $\cM$ a multiset of points in $\PG(v-1,q)$ of maximum point multiplicity at most $\lambda$, $q=p^m$, and $e$ the largest 
  integer such that $p^e$ divides $\lambda$. Let $r\in\Q_{\ge 0}$ with $mr\in\N_0$ and $r\le \tfrac{e}{m}+(v-1)$. Then, $\cM$ is $q^r$-divisible iff  
  its $\lambda$-complement $\cM^{\complement_\lambda}$ is.
\end{nlemma}
In the above example we have $v=9$ and $\lambda=1$, so that $\cM^{\complement_1}$ is $4$-divisible since $\cM$ is $4$-divisible due to Lemma~\ref{lem:union_subspaces}. 
Since there is no $4$-divisible multiset of points of cardinality $1$ over $\F_2$, no configuration of $20$ solids and $30$ planes with pairwise trivial intersection can 
exist in $\PG(9-1,2)$.
\begin{nexercise}
  Determine the maximum integer $f$ such that there exists a non-empty $p^f$-divisible multiset of points in $\PG(v-1,q)$ with maximum point multiplicity $\lambda$, where $q=p^m$ 
  and $v,p,m,\lambda$ are arbitrary but fixed.
\end{nexercise}

For the case of multisets of subspaces of the same dimension we can state rather explicit results using sharpened rounding operators.

\medskip

\begin{trailer}{Sharpened rounding}
\vspace*{-6mm}
\begin{ndefinition}
	\label{def:divisible_gauss_bracket}
  For $a\in\Z$ and $b\in\N_{>0}$ let $\llfloor a/b \rrfloor_{q^r}$ be the maximal $n\in\Z$ such that there exists a $q^r$-divisible $\F_q$-linear code of effective length $a-nb$.
  If no such code exists for any $n$, we set $\llfloor a/b \rrfloor_{q^r} = -\infty$.
  Similarly, let $\llceil a/b\rrceil_{q^r}$ denote the minimal $n\in\Z$ such that there exists a $q^r$-divisible $\F_q$-linear code of effective length $nb-a$. 
  If no such code exists for any $n$, we set $\llceil a/b\rrceil_{q^r} = \infty$.
\end{ndefinition}
\end{trailer}
Note that the symbols $\llfloor a/b \rrfloor_{q^r}$ and $\llceil a/b \rrceil_{q^r}$ encode the four values $a$, $b$, $q$ and $r$. Thus, the fraction $a/b$ is a formal fraction and the power 
$q^r$ is a formal power, i.e.\ we assume $1530/14\neq 765/7$ and $2^2\neq 4^1$ in this context.
\begin{nexercise}
  Compute $\llfloor 765/7 \rrfloor_{2^2}$ and $\llfloor 1530/14 \rrfloor_{4^1}$. Verify
  \[
		    \llfloor 0/b\rrfloor_{q^r} = \llceil 0/b\rrceil_{q^r} = 0
		\]
		and
		\begin{eqnarray*}
		  && \ldots \leq \llfloor a/b\rrfloor_{q^2} \leq \llfloor a/b\rrfloor_{q^1} \leq \llfloor a/b \rrfloor_{q^0} = \left\lfloor \tfrac{a}{b} \right\rfloor \\
		  &&  \leq a/b \leq \lceil a/b\rceil = \llceil a/b \rrceil_{q^0} \leq \llceil a/b\rrceil_{q^1} \leq \llceil a/b\rrceil_{q^2} \leq \ldots
		\end{eqnarray*}
\end{nexercise}
\begin{nexercise}
  Develop an algorithm for the computation of $\llfloor a/b \rrfloor_{q^r}$ and $\llceil a/b \rrceil_{q^r}$. Minimize its necessary complexity. 
\end{nexercise}

Having the notion of the sharpened rounding of Definition~\ref{def:divisible_gauss_bracket} at hand, we can state:
\begin{nlemma}
  \label{lem:pack_cover}
  Let $k \in \N_{>0}$ and $\mathcal{U}$ be a multiset of $k$-spaces in $\PG(v-1,q)$.
  \begin{enumerate}
    \item[(i)]\label{lem:pack_cover:pack} If every point in $\cP$ is covered by at most $\lambda$ elements of $\mathcal{U}$, then 
    \[
	\#\cU\le \llfloor\lambda[v]_q/[k]_q\rrfloor_{q^{k-1}}\text{.}
    \]
    \item[(ii)]\label{lem:pack_cover:cover} If every point in $\cP$ is covered by at least $\lambda$ elements in $\mathcal{U}$, then 
    \[
	\#\cU\ge \llceil\lambda[v]_q/[k]_q\rrceil_{q^{k-1}}\text{.}
    \]
  \end{enumerate}  
\end{nlemma}
\begin{nexercise}
  Prove Lemma~\ref{lem:pack_cover} using the multisets of points $\cM^{\complement_\lambda}$ and $\cM' = \cM - \lambda \cdot \PG(v-1,q)$, i.e., 
  $\cM'(P)=\cM(P)-\lambda$ for all $P\in\cP$.
\end{nexercise}   
\begin{nexample}
  \label{ex_plane_cover_pg_7_2_lambda_3}
  What is the maximum number of planes in $\PG(7,2)$ such that every point is covered at most three times? Counting points gives
  $$
    \left\lfloor \frac{3\cdot [8]_2}{[3]_2}\right\rfloor=\left\lfloor 109+\tfrac{2}{7}\right\rfloor=109
  $$   
  as an upper bound, while Lemma~\ref{lem:pack_cover} gives the upper bound
  $$
    \leftllfloor \frac{3\cdot[8]_2}{[3]_2}\rightrrfloor_{2^{2}}=107,
  $$
  since no $2^2$-divisible code of length $2$ or $9$ exists over $\F_2$. This bound is indeed tight, see e.g.\ \cite{etzion2020subspace} where also more general packings 
  of $k$-spaces are studied.
\end{nexample} 

In some cases the sharpened rounding can even be computed when the input data is parametric.

\begin{trailer}{Asymptotic maximum size of $\lambda$-fold partial spreads}
\vspace*{-6mm}
\begin{nexample}
  \label{ex_gen_partial_spread_asymptotic_bound} 
  Let $v=tk+r$ with $t\in\N_{>0}$ and $r\in\{1,\dots,k-1\}$ and $\cU$ be a multiset of $k$-spaces in $\PG(v-1,q)$ such that every point is covered at most $\lambda\in\N_0$ times.
  We will show
  \begin{equation}
    \label{ie_gen_partial_spread_asymptotic_bound}
    \#\cU\le \lambda\cdot\left(1+\sum_{i=1}^{t-1} q^{ik+r} \right)  = 
    \lambda\cdot\left(\frac{ q^v-q^{k+r}}{q^k-1}+1\right)
    <\lambda \frac{[v]_q}{[k]_q} 
  \end{equation}
  for $k>\lambda [r]_q$.
  
  First we deduce
  \begin{eqnarray*}
    \lambda\!\left(\!1\!+\!\sum_{i=1}^{t-1} q^{ik+r}\!\right) &=&\lambda q^{k+r} \cdot\frac{q^{k(t\!-\!1)}\!-\!1}{q^k\!-\!1}+\lambda  
    = \lambda\cdot \frac{ q^v\!-\!q^{k+r}+q^k\,-\,1}{q^k\!-\!1} \\
    &=&\lambda \frac{[v]_q-[k+r]_q+[k]_q}{[k]_q}<\lambda \frac{[v]_q}{[k]_q},
  \end{eqnarray*}
  from the geometric series, so that we assume to the contrary 
  $$
    \#\cU=\lambda\cdot\left(1+\sum_{i=1}^{t-1} q^{ik+r} \right)+1=\lambda \cdot \frac{[v]_q-[k+r]_q+[k]_q}{[k]_q}+1
  $$
  after removing elements of $\cU$ if necessary. Let $\cM$ be the multiset of points associated with $\cU$. Then $\cM$ is $q^{k-1}$-divisible by Lemma~\ref{lem:union_subspaces} and has maximum point multiplicity 
  at most $\lambda$, so that its $\lambda$-complement $\cM^{\complement_\lambda}$ is $q^{k-1}$-divisible by Lemma~\ref{lemma_t_complement}, with 
  $\#\cM^{\complement_\lambda}=\lambda[v]_q-[k]_q\cdot\#\cU=\lambda q^k[r]_q-[k]_q$. From
  \begin{eqnarray*}
    (q-1)\sum_{i=0}^{k-2} \snumb{k-1}{i}{q} &=& (q-1)\sum_{i=0}^{k-2} q^i\cdot [k-i]_q=(q-1)\sum_{i=0}^{k-2} \frac{q^k-q^i}{q-1} \\ 
    &=&(k-1)q^k-[k-1]_q
  \end{eqnarray*}
  we conclude the $S_q(k-1)$-adic expansion
  $$
    \#\cM^{\complement_\lambda}=\left(q\left(\lambda[r]_q-k+1\right)-1\right)\snumb{k-1}{k-1}{q}+ \sum_{i=0}^{k-2} (q-1)\cdot\snumb{k-1}{i}{q}.
  $$
  By Theorem~\ref{thm_characterization_div} the leading coefficient is non-negative, i.e., $\lambda[r]_q\ge k$, which contradicts $k>\lambda [r]_q$.
\end{nexample} 
We remark that one can easily give a matching construction, i.e., the stated upper bound in Inequality~(\ref{ie_gen_partial_spread_asymptotic_bound}) is tight. The 
special case $\lambda=1$ is the main theorem of \cite{nastase2016maximum}.   
While the proof is a bit technical, we have actually just applied Lemma~\ref{lem:pack_cover} and evaluated the sharpened rounding analytically  
(for special parameters). 
\end{trailer}

\begin{nexercise}
  \label{ex_covering_points}
  Let $\cU$ be a multiset of $k$-spaces in $\PG(v-1,q)$ that covers each point at least once. Show 
  $$
    \#\cU\ge \left\lceil \frac{[v]_q}{[k]_q} \right\rceil
  $$
  and determine for the case of equality the geometric structure of the (multi-)set of points that are covered more than once.
\end{nexercise}
 
 \begin{nexercise}
  \label{ex_solid_cover_pg_6_2}
  Let $\cU$ be a multiset of $4$-spaces in $\PG(6,2)$ that cover every $2$-space at least once. Show $\#\cU>77$.\\ 
  \textit{Hint:} First show that a $2^3$-divisible multiset of points $\cM$ of cardinality $12$ over $\F_2$ is a $4$-fold line, 
  i.e., $\cM=4\cdot\chi_L$ for some line $L$.
\end{nexercise}
We remark that the best known published lower bound for the number of solids in $\PG(6,2)$ that cover every line at least once is $77$ and an example of  
$93$ such solids is known, see \cite{etzion2014covering}. Without proof we state that the lower bound can be improved to $86$ and the upper bound to $91$.

\begin{question}{Research problem}Apply similar techniques to improve further lower bounds from \cite{etzion2014covering}. 
\end{question}

\chapter[Constructions for projective $q^r$-divisible codes]{Constructions for projective $q^r$-divisible codes or multisets of points with bounded maximum point multiplicity}
\label{sec_constructions_projective}
In Section~\ref{sec_lengths_of_divisible_codes} we have completely characterized the possible cardinalities of $q^r$-divisible multisets of points over $\F_q$, where $r$ is 
an arbitrary positive integer. As a refinement we now consider $q^r$-divisible multisets of points over $\F_q$ whose maximum point multiplicity is upper bounded by some positive 
integer $\lambda$, see e.g.\ \cite{korner2023lengths}. In the extreme case $\lambda=1$ the corresponding linear codes are projective. A first observation is that we can combine a $q^r$-divisible multiset $\cM_1$ in 
an $\F_q$-vector space $V_1$ and another $q^r$-divisible multiset $\cM_2$ in an $\F_q$-vector space $V_2$ to a $q^r$-divisible multiset $\cM$ in $V_1\times V_2$ by considering 
$V_1$ and $V_2$ as subspaces of $V_1\times V_2$. Here we have $\#\cM=\#\cM_1+\#\cM_2$ and $\gamma_0(\cM)=\max\!\left\{\gamma_0(\cM_1),\gamma_0(\cM_2)\right\}$, so that:
\begin{nlemma}
  \label{lemma_sum_mult}
  The set of possible cardinalities of $q^r$-divisible multisets of points over $\F_{q}$ with maximum point multiplicity at most $\lambda$ is closed under addition.
\end{nlemma}

Let us start to consider constructions for multisets of points with maximum point multiplicity $1$, i.e., sets of points.

\begin{trailer}{Combinations of Simplex and first order Reed-Muller codes}In Example~\ref{ex_simplex_code} and Example~\ref{example_affine_space} we have seen the first two basic constructions of $q^r$-divisible sets.
\begin{nlemma}
  Let $u$ be an arbitrary positive integer and $U$ be an arbitrary $u$-space in $\PG(v-1,q)$, where $v\ge u$. Then $\chi_U$ is a $q^{u-1}$-divisible set of 
  cardinality $[u]_q$ and dimension $u$. If $u\ge 2$ and $H$ is a hyperplane of $U$, i.e., a $(u-1)$-space that is contained in $U$, then 
  $\chi_U-\chi_H$ is a $q^{u-2}$-divisible set of cardinality $q^{u-1}$ and dimension $u$. 
\end{nlemma} 
\end{trailer}
The small cardinalities of $q^r$-divisible sets over $\F_q$ that cannot be attained by combinations of $(r+1)$-spaces and affine $(r+2)$-spaces can be determined easily:
\begin{nexercise}
  \label{exercise_upt_to_r_qrp1}
  Let $1\le n\le rq^{r+1}$ such that no $u,v\in\N_0$ with $u\cdot [r+1]_q+v\cdot q^{r+1}=n$ exist. Then, there exist $a,b\in \N_0$ with $a\le r-1$, $b\le q-2$, and
  $$
    (a(q-1)+b)[r+1]_q+a+1\le n\le (a(q-1)+b+1)[r+1]_q-1.
  $$     
  Moreover, there are no $u,v\in\N_0$ with $u\cdot [r+1]_q+v\cdot q^{r+1}=rq^{r+1}+1$.
\end{nexercise}

Up to the bound $rq^{r+1}$ the attainable cardinalities of $q^r$-divisible sets of points over $\F_q$ using the first two basic constructions only are given by
\begin{itemize}
  \item $\{3,4\}$ for $q=2$ and $r=1$;
  \item $\{7,8,14,15,16\}$ for $q=2$ and $r=2$;
  \item $\{15,16,30,31,32,45,46,47,48\}$ for $q=2$ and $r=3$;
  \item $\{4,8,9\}$ for $q=3$ and $r=1$;
  \item $\{13,26,27,39,40,52,53,54\}$ for $q=3$ and $r=2$;
  \item $\{5,10,15,16\}$ for $q=4$ and $r=1$;
  \item $\{21,42,63,64,84,85,105,106,126,127,128\}$ for $q=4$ and $r=2$.
\end{itemize}  

\begin{nexample}
  \label{example_ovoid}
  An \emph{ovoid} in $\PG(3,q)$ is a set $\cM$ of $q^2+1$ points, no three collinear, such that every hyperplane 
  contains $1$ or $q+1$ points, i.e., $\cM$ is $q$-divisible. Ovoids exist for all $q$, see e.g.\ \cite{o1996ovoids}.
\end{nexample}
\begin{nexample}
  \label{ex_ovoid_pg33}
  For $q=3$ an ovoid in $\PG(3,3)$ consists of $10$ points, e.g.\ the points $(x_0:x_1:x_2:x_3)$ with $x_0x_1+x_2^2+x_3^2=0$. The corresponding code is 
  a projective $[10,4,6]_3$-code with weight enumerator $1+60x^6+20x^9$, i.e., it is $3$-divisible and has exactly two non-zero weights.
\end{nexample}
For the binary field a $2$-divisible set of cardinality $2^2+1$ is contained in a different parametric family.
\begin{ndefinition}
  \label{def_projective_base}
  A set of $k+1$ points in $\PG(k-1,q)$, where $k\ge 2$, such that any subset of $k$ points spans the full space is called a \emph{$k$-dimensional projective base}. 
\end{ndefinition}  
\begin{nexercise}
  \label{exercise_projective_base}
  Show that the binary $k$-dimensional projective base is $2$-divisible and has cardinality $k+1$ if $k\ge 2$. 
  Moreover, show that a representation is given by the points $\left\langle \be_1\right\rangle,\dots, \left\langle \be_k\right\rangle$ and $\left\langle \be_1+\dots+\be_k\right\rangle$, 
  where $\be_1,\dots,\be_k$ denote the unit vectors in $\F_2^k$.  
\end{nexercise}
\begin{trailer}{A cone construction}
\vspace*{-3mm}
\begin{ndefinition}
  Let $X$, $Y$ be complementary subspaces of $\PG(v-1,q)$ and $\cB$ be a set of points in $\PG(Y)$. The \emph{cone with vertex $X$ and base $\cB$} is 
  the multiset of points $\cM$ given by $\cM=\sum_{B\in\cB} \chi_{\left\langle B,X\right\rangle}$.
\end{ndefinition}
If $\dim(X)=s$, then the set of points of $\left\langle P,X\right\rangle$ is $q^s$-divisible for every point $P$. If $\cB$ is $q^r$-divisible and $q^r$ divides $\#\cB$ then we can easily check that 
the cone $\cM$ with vertex $X$ and base $\cB$ is $q^{r+s}$-divisible and all points outside of $X$ have multiplicity at most $1$ while the points in $X$ have multiplicity $\#\cB$. 
Clearly we can subtract $(\#\cB-1)\cdot\chi_X$ or $\#\cB\cdot \chi_X$ from $\cM$ in order to obtain a set of points.
\begin{nexercise}
  \label{exercise_cone}
  Let $X$, $Y$ be complementary subspaces of $\PG(v-1,q)$, $s=\dim(X)$, and $\cB$ be a $q^r$-divisible set of points in $\PG(Y)$. Show that
  \begin{equation}
    \label{cone_construction_1}
    \sum_{B\in\cB} \chi_{\left\langle B,X\right\rangle\backslash X}
  \end{equation}
  is $q^{r+s}$-divisible of cardinality $\#\cB\cdot q^s$ if $\#\cB\equiv 0\pmod {q^{r+1}}$ and
  \begin{equation}
    \label{cone_construction_2}
    \sum_{B\in\cB} \chi_{\left\langle B,X\right\rangle\backslash X}\,+\,\chi_X
  \end{equation}
  is $q^{r+s}$-divisible of cardinality $\#\cB\cdot q^s+[s]_q$ if $\#\cB(q-1)\equiv -1\pmod {q^{r+1}}$.   
\end{nexercise} 
\end{trailer}
\begin{nexample}
  \label{example_cone_constructions} 
  For a $6$-dimensional projective base $\cB$ over $\F_2$ and an $s$-space $X$, where $s\ge 1$, (\ref{cone_construction_2}) yields a $2^{s+1}$-divisible set of $2^{s+3}-1$ points over $\F_2$. 
  Similarly, for a $7$-dimensional projective base $\cB$ over $\F_2$ and an $s$-space $X$, where $s\ge 1$, (\ref{cone_construction_1}) yields a $2^{s+1}$-divisible set of $2^{s+3}$ points over $\F_2$.
\end{nexample}
\begin{backgroundinformation}{Parity check bits}We remark that adding a so-called parity (check) bit to the codewords of a binary linear code yields a $2$-divisible linear code whose length is increased by one. Since 
a binary $4$-dimensional projective base gives a $2$-divisible set of cardinality $5$ over $\F_2$, there are $2$-divisible sets of points of cardinality $n$ over 
$\F_2$ for all $n\ge 3$.    
\end{backgroundinformation} 

In a certain sense we can generalize the idea of parity check bits to construct binary codes with higher divisibility. To this end, assume that we are given a $2^r$-divisible $[n,k]_2$-code 
$C$ that contains a $2^{r+1}$-divisible $[n',k-1]_2$-code $C'$ such that all codewords in $C\backslash C'$ have weight congruent to $2^r$ modulo $2^{r+1}$. Geometrically, $C$ corresponds to a $2^r$-divisible multiset of points $\cM$ in $\PG(k-1,2)$ and $C'$ corresponds to a 
$2^{r+1}$-divisible multiset of points $\cM'$ in $\PG(k-2,2)$. Moreover, there exists a point $P$ in $\PG(k-1,2)$ such that $\cM'$ arises from $\cM$ by projection through $P$. Especially, 
we have $\cM(P)=n-n'$ and for every hyperplane $H$ of $\PG(k-1,2)$ we have $\cM(H)\equiv \#\cM \pmod {2^{r+1}}$ if $P\le H$ and $\cM(H)\equiv \#\cM+2^r \pmod {2^{r+1}}$ otherwise. 
Defining the multiset of points $\widetilde{\cM}$ in $\PG(k-1,2)$ by $\widetilde{\cM}(P)=\cM(P)+2^r$ and $\widetilde{\cM}(Q)=\cM(Q)$ for all other points $Q\neq P$, we obtain 
a $2^{r+1}$-divisible multiset of points in $\PG(k-1,2)$. For the special case $r=0$ we note that the set of codewords of even weights of an arbitrary $[n,k]_2$-code $C$ forms 
a subcode of dimension $k-1$ or $k$. So, if $C$ is not even itself, then the above geometric construction corresponds to adding a parity check bit.      
\begin{nexample}
  \label{example_complement_of_parabolic_quadric}
  Consider the matrix $G$ consisting of the $16$ column vectors in $\F_2^6$ that have Hamming weight $2$ or $6$ and let $\cM$ be the corresponding set of $16$ points 
  in $\PG(5,2)$.\footnote{This is a specific embedding of the complement of the parabolic quadric $Q(4,2)$, see e.g.\ \cite{hirschfeld1998projective}, in $\PG(5,2)$. The subsequent 
  point $N$ is the nucleus of the quadric, i.e., every line through $N$ contains exactly one point of the quadric.}   
  By $N$ we denote the unique point with  Hamming weight $6$. Let us describe the hyperplanes by the set of points being perpendicular to a vector 
  $\bv\in\F_2^6\backslash \zv$, i.e., we write $H(\bv)$. We can easily check that $\cM(H(\bv))=10$ if $\wt(\bv)\in\{1,5\}$, $\cM(H(\bv))=6$ if $\wt(\bv)=3$, 
  $\cM(H(\bv))=8$ if $\wt(\bv)\in\{2,4\}$, and $\cM(H(\bv))=16$ if $\wt(\bv)=6$. So, $G$ spans a $[16,5]_2$-code $C$ with non-zero weights in $\{6,8,10\}$. The codewords 
  of weight $8$ correspond to the hyperplanes containing $N$ (and not being equal to $H(N)$). Increasing the multiplicity of $N$ by $2$ yields a $2^2$-divisible multiset 
  of cardinality $18$ in $\PG(5,2)$, with dimension $5$ and $N$ is the unique point with multiplicity larger than $1$.
\end{nexample}
By a little trick we can turn the above example into a $2^2$-divisible set of $21$ points in $\PG(6,2)$. Instead of increasing the multiplicity of $N$ by $2$, we decrease it by 
$2$, so that it becomes $-1$. Adding a suitable plane $\pi$ containing $N$ gives the desired multiset $\cM+\chi_\pi-2\chi_N$. Since $\dim(\cM)=5$ and $\dim(\chi_\pi)=3$ an 
embedding in $\PG(6,2)$ as a set of points is possible. 
  
  
\begin{trailer}{A switching construction}
\vspace*{-5mm}
\begin{nlemma}
  \label{lemma_switching_construction}
  Let $\cM$ be a $q^r$-divisible set of points in $\PG(v-1,q)$ such that there exists an $r$-space $S$ with $\cM(P)=1$ for all points $P$ in $S$, i.e., 
  $S$ is contained in the support $\supp(\cM)$ of $\cM$. Then, there exists a $q^r$-divisible set of points $\cM'$ with cardinality $\#\cM'=\#\cM + q^{r+1} -[r+1]_q$.
\end{nlemma}  
\begin{proof}
  Let $\widetilde{\cM}$ be the embedding of $\cM$ in $\PG(v'-1,q)$ for sufficiently large $v'\ge v$ (chosen later on) and $T_1,\dots, T_{q-1}$ be $(r+1)$-spaces containing $S$. With this 
  the multiset of points 
  $$
    \cM':=\widetilde{\cM}+\sum_{i=1}^{q-1} \chi_{T_i}\,-q\cdot \chi_S  
  $$
  is $q^r$-divisible and has cardinality $$\#\cM'=\#\cM+(q-1)[r+1]_q-q[r]_q=\#\cM + q^{r+1} -[r+1]_q.$$ 
  If $v'$ is sufficiently large then the $T_i$ can clearly be chosen in such a way that their pairwise intersection as well as their intersection with  
  $\supp(\cM)$ equals $S$, so that $\gamma_0(\cM')=1$. 
\end{proof}   
The construction is called \emph{switching construction} since an $r$-space is switched for $q-1$ affine $(r+1)$-spaces.
\end{trailer} 
Starting from an $(r+1)$-space over $\F_q$ we can construct 
several non-isomorphic $q^r$-divisible sets of points over $\F_q$ with cardinality $q^{r+1}$ if $q>2$.
\begin{nexercise}
  Show that for each $r\ge 1$ and each $r+2\le k\le r+q$ there exists a $q^r$-divisible set of points over $\F_q$ with cardinality $q^{r+1}$ and dimension $k$. 
\end{nexercise}
The switching construction from Lemma~\ref{lemma_switching_construction} can be used to construct projective $2^r$-divisible codes for an entire interval of effective 
lengths:
\begin{ncorollary}
  \label{cor_spread_switching}
  For each integer $r\ge 1$ and each $2^{2r}-1\le n\le 2^{2r}+2^r$ there exists a $2^r$-divisible set of points over $\F_2$ with cardinality $n$.    
\end{ncorollary}
\begin{proof}
  Let $\cS$ be an $r$-spread, i.e. a partition of $\PG(2r-1,2)$ into $2^r+1$ pairwise disjoint $r$-spaces. The corresponding set of points is $2^{2r-1}$-divisible, where 
  $2r-1\ge r$. For $0\le j\le 2^r+1$ of these $r$-spaces we can apply the switching construction from Lemma~\ref{lemma_switching_construction}. 
\end{proof}
For $r=2$ we obtain $4$-divisible sets of points over $\F_2$ with cardinalities between $15$ and $20$. Together with the examples for cardinalities $7$, $8$, and $14$, we 
obtain examples for all cardinalities $n\ge 14$. Using the same construction for general field sizes $q$ we obtain a sequence of possible cardinalities:
\begin{ncorollary}
  \label{cor_spread_switching_q}
  For each integer $r\ge 1$ and each $0\le j\le q^r+1$ there exists a $q^r$-divisible set of points over $\F_q$ with cardinality $n=[2r]_q+j\cdot\left(q^{r+1}-[r+1]_q\right)$.  
\end{ncorollary}

For $q=2$ we can use the switching construction to construct a $2^r$-divisible set of points over $\F_2$ of cardinality $r\cdot 2^{r+1}+1$ for all $r\ge 2$.
\begin{nexercise}
  \label{ex_affine_spaces_switching}
  For an integer $r\ge 1$ consider the multisets of points $\cM_i$ over $\F_2$ consisting of the $2^{r+1}$ points generated by the $2^{r+1}$ binary vectors in 
  $\be_{i+r+1}+\left\langle \be_i,\be_{i+1},\dots,\be_{i+r}\right\rangle$ for all $1\le i\le r$. Show that the $\cM_i$ are affine $(r+2)$-spaces and that 
  $\cM=\sum_{i=1}^{r} \cM_i$ is a $2^{r}$-divisible set of points over $\F_2$ of cardinality $r\cdot 2^{r+1}$ and dimension $2r+1$ whose support 
  contains the $r$-space $S=\left\langle e_{r+2},\dots,e_{2r+1}\right\rangle$. Use the switching construction to obtain a $2^r$-divisible set 
  of $r\cdot 2^{r+1}+1$ points over $\F_2$.  
\end{nexercise}
Using $r=3$, $r=4$, and $r=5$ in the construction of Exercise~\ref{ex_affine_spaces_switching} we obtain an $8$-divisible set of $49$, a $16$-divisible set of 
$129$, and a $32$-divisible set of $321$ points over $\F_2$.

It is also possible to extend the switching construction to a more general setting. To this end let $\Delta$ and $\Delta'$ be two integers such that 
$\rho:=\tfrac{\Delta}{\Delta'}\in\N_{>0}$. For a given prime power $q$ let $\cM$ be a $\Delta$-divisible set of points over $\F_q$, $\cD$ be a $\Delta'$-divisible set 
of points over $\F_q$, $D=\left\langle\cD\right\rangle$ be the subspace spanned by $\cD$, and let $\cM_1,\dots,\cM_{\rho-1}$ be $\Delta$-divisible sets of points 
over $\mathbb{F}_q$ such that $\cM_i(P)\ge \cD(P)$ for all $1\le i\le \rho-1$ and $\cM(P)\ge \cD(P)$, where $P$ ranges over all points in $D$. In other words, the 
sets of points $\cM_i$ and $\cM$ all contain the set of points $\cD$ as a subset. With this, the multiset of points given by 
\begin{equation}
  \label{eq_generalized_switching_construction}
  \cM+\sum_{i=1}^{\rho-1} \cM_i-\rho\cdot \cD
\end{equation} 
is $\Delta$-divisible over $\F_q$. If $\cM$ and the $\cM_i$ can be embedded in suitable subspaces such that their pairwise intersection is given by the points in $\supp(\cD)$, which is e.g.\ 
possible if $\cM|_D=\cD|_D$ and $\cM_i|_D=\cD|_D$ for all $1\le i\le \rho-1$, then the multiset of points given by Equation~(\ref{eq_generalized_switching_construction}) 
is indeed a set of points.
\begin{nexercise}
  \label{exercise_generalized_switching_construction}
  Let $\Delta, \Delta'\in\N_{>0}$ such that $\rho:=\tfrac{\Delta}{\Delta'}\in\N_{>0}$, $\cM$ be a $\Delta$-divisible set of points over $\F_q$, $\cD$ be a $\Delta'$-divisible 
  set of points over $\F_q$, and let $\cM_1,\dots,\cM_{\rho-1}$ be $\Delta$-divisible sets of points over $\mathbb{F}_q$. Further assume that $\cM$ and the $\cM_i$ are embedded in suitable subspaces such that their pairwise intersection is given by $\cD$, e.g.\ because $\cM|_D=\cD|_D$ and
  $\cM_i|_D=\cD|_D$ for all $1\le i\le \rho-1$, where $D:=\left\langle\cD\right\rangle$. Show that there exists a $\Delta$-divisible set of points over $\F_q$ of 
  cardinality $\#\cM+\sum_{i=1}^{\rho-1} \# \cM_i-\rho\cdot \#\cD$ and dimension 
  $\dim(\cM)+\sum_{i=1}^{\rho-1} \dim(\cM_i)-\rho\cdot \dim(\cD) \le k\le \dim(\cM)+\sum_{i=1}^{\rho-1} \dim(\cM_i)-(\rho-1)\cdot \dim(\cD)$.   
\end{nexercise}
For an application with a specific choice of the $\cM_i$ and $\cD$ we refer to Exercise~\ref{ex_baer_switching}.

\section{Constructions using subfields}
\label{subsec_subfield_constructions}
Considering $\F_{q^l}$ as an extension field of $\F_q$ we can assume $\F_q\subset \F_{q^l}$ for each integer $l\ge 2$. So, the field $\F_{q^l}$ can be also seen 
as an $l$-dimensional vector space over $\F_q$. If $\bv$ is a $k$-dimensional vector over $\F_4\simeq \F_2[x]/(x^2+x+1)\F_2[x]$, we can represent 
each entry $v_i\in \F_4$ by $a_1x+a_0$ with $a_0,a_1\in \F_2$ and replace it by the vector $\left(a_0,a_1\right)^\top \in\F_2^2$. This yields a representation of $\bv$ 
as an element in $\F_2^{2k}$ instead of $\F_4^k$. So, starting from a multiset of points in $\PG(v-1,q^l)$ we can construct a multiset of points in $\PG(vl-1,q)$. However, 
we have to be a bit careful when using the relation between vectors and points, i.e., $1$-dimensional subspaces, for different field sizes. For a given vector 
$\bv\in\F_{q^l}^k$ the point $\left\langle\bv\right\rangle_{\F_{q^l}}$ admits $q^l-1$  representations $\left\langle\bv\right\rangle_{\F_{q^l}}
=\left\langle\alpha\cdot\bv\right\rangle_{\F_{q^l}}$, where $\alpha\in\F_{q^l}\!\backslash 0$. If $\bv'\in\F_q^{kl}$ is a representation of $\bv$ over $\F_q$, then 
the point $\left\langle \bv'\right\rangle_{\F_q}$ admits only $q-1$ representations $\left\langle\bv'\right\rangle_{\F_{q}}
=\left\langle\alpha\cdot\bv'\right\rangle_{\F_q}$, where $\alpha\in\F_{q}\!\backslash 0$. So, if we want that all non-zero vectors are covered by the points, we have 
to replace a single point in $\PG(v-1,q^l)$ by $\tfrac{q^l-1}{q-1}=[l]_q$ points in $\PG(vl-1,q)$. In terms of linear codes this can be described by concatenation (with 
an $l$-dimensional simplex code). In the other direction, starting from a point $\left\langle\bv'\right\rangle_{\F_q}$ 
in $\PG(v-1,q)$ we may also replace the point by the point $\left\langle \bv'\right\rangle_{\F_{q^l}}$ in $\PG(v-1,q^l)$ using $\F_q\subset \F_{q^l}$. Note that $\alpha \bv'$, 
where $\alpha\in\F_q\backslash 0$, leads to the same point in $\PG(v-1,q^l)$. The analog for linear codes is the interpretation of a given generator matrix of an $[n,k]_q$-code 
over $\F_{q^l}$. It remains to study how divisibility properties are transferred by these two constructions.      

\begin{trailer}{Concatenated codes}\emph{Concatenation} was introduced by George David Forney Jr.\ in his 
PhD thesis \cite{forney1965concatenated}. Here we are given an \emph{outer} $[N,K,D]_{q^l}$-code $C_{out}$ and an \emph{inner} $[n,l,d]_q$-code $C_{in}$, where 
we note that the dimension of the inner code $C_{in}$ equals the degree $[\F_{q^l}:\F_q]$ of the field extension. Each vector in $\F_q^{lK}$ can be associated with 
a vector in $\F_{q^l}^K$ and then mapped via the outer code $C_{out}$ to $\F_{q^l}^N$. Then each field element in $\F_{q^l}$ can be associated with an element in $\F_q^l$ 
and then mapped via the inner code $C_{in}$ to $\F_q^n$. Putting everything together, the concatenation of $C_{out}$ and $C_{in}$ gives an $[nN,lK,\ge dD]_q$-code $C$, where 
the minimum distance of $C$ may also be strictly larger than $dD$, see e.g.\ \cite[Theorem 5.9]{bierbrauer2016introduction}. For more details, including an example of the 
computation of a generator matrix of the concatenated code $C$, we refer to \cite[Section 5.2]{bierbrauer2016introduction}. We remark that decomposing a given linear code 
over $\F_q$ as a concatenated code, if possible, is an interesting algorithmic problem, see e.g.\ \cite{sendrier1998concatenated}. While the determination of the weight 
distribution of a concatenated code often requires some extra work, see e.g.\ \cite{weng1977concatenated}, the situation becomes much easier when $C_{in}$ is an  
$l$-dimensional simplex code.
\end{trailer}
\begin{nexercise}
  \label{exercise_concatenation_with_simplex_code}
  Let $C$ be a projective $\Delta$-divisible $[n,k]_{q^l}$-code. Show that the concatenation of $C$ with an $l$-dimensional simplex code over $\F_q$ yields a 
  projective $\Delta q^{l-1}$-divisible $\left[n\cdot [l]_q,kl\right]_q$-code.
\end{nexercise} 
\begin{nexample}
  \label{ex_8_div_51}
  Let $C$ be the projective $4$-divisible $[17,4]_4$-code corresponding to an ovoid in $\PG(3,4)$, see Example~\ref{example_ovoid}. Concatenation with the projective $2$-divisible 
  $[3,2]_2$-simplex-code yields a projective $8$-divisible $[51,8]_2$-code. Note that $C$ as well as the concatenated code are two-weight codes. By construction, the corresponding 
  set of points can be partitioned into $17$ lines.
\end{nexample}
We remark that, up to isomorphism, there is a unique $8$-divisible set of points of cardinality $51$ over $\F_2$ \cite[Lemma 24]{honold2018partial}. By shortening the $[51,8]_2$-code 
from Example~\ref{ex_8_div_51} we obtain $8$-divisible $[50,7]_2$-codes, which however are not projective. Nevertheless $8$-divisible sets of $50$ points over $\F_2$ indeed exist. 
To this end we have enumerated all projective $8$-divisible binary codes with length at most $51$ using the software package \texttt{LinCode} \cite{bouyukliev2020computer,kurz2019lincode}.  
Observe that a projective $8$-divisible binary code with an effective length $49\le \neff\le 51$ does not contain codewords of weights 
$40$ or $48$ since the corresponding residual code would be a projective $4$-divisible binary code with an effective length in $\{1,2,3,9,10,11\}$, which does not exist as we will see 
in Lemma~\ref{lemma_picture_q_2_r_2}. We have tabulated the corresponding counts of projective $8$-divisible binary codes in Table~\ref{table_8div_enumeration}. 
\begin{table}[htp]
  \begin{center}
    \begin{tabular}{ccccccccc}
      \hline
       n / k & 8 & 9 & 10 & 11 & 12 & 13 & $\sum$ \\
       \hline
       49 & 9 & 38 & 44 & 21 & 7 & 1 & 120 \\ 
       50 & 1 &  0 &  0 &  0 & 0 & 0 &   1 \\
       51 & 1 &  0 &  0 &  0 & 0 & 0 &   1 \\ 
       \hline 
    \end{tabular}
    \caption{Number of projective $8$-divisible binary codes with $49\le \neff\le 51$ per dimension.}
    \label{table_8div_enumeration}
  \end{center}  
\end{table}

\medskip

\begin{backgroundinformation}{There is a unique projective $8$-divisible binary linear code of length $50$.}As shown by the above exhaustive enumeration, each projective  
$8$-divisible binary code of length $50$ has dimension $8$ and is indeed unique up to isomorphism. A generator matrix is given by
$$\begin{pmatrix}
  11111111111111111111111111111110000000000010000000\\
  00000000000000011111111111111111111111000001000000\\
  00000001111111100000000111111110001111111000100000\\
  00011110000111100001111000011110110011011100010000\\
  01100110011001100110011001100111010101101100001000\\
  00111010101010101000111010101011100011011100000100\\
  00101101010010010010001010010010100100001000000010\\
  11111101000000101011010111111110000000110100000001
\end{pmatrix}.$$
The automorphism group of the code has order $3840$ and the weight enumerator is given by 
$W_C(x)=1+5x^{16}+210x^{24}+40x^{32}$.  
\end{backgroundinformation}
We remark that the $8$-divisible binary codes with length up to $48$ have been enumerated in \cite{betsumiya2012triply} and the counts of the corresponding subset of projective 
codes were stated in \cite{projective_divisible_binary_codes}. 

\begin{question}{Research problem}
\vspace*{-5mm}
\begin{itemize}
  \item Find a parametric family of projective $q^r$-divisible linear codes containing the projective $8$-divisible $[50,8]_2$-code.
  \item Give a computer-free proof of the uniqueness of a projective $8$-divisible binary linear code of length $50$.
\end{itemize}  
\end{question}

\begin{trailer}{Generator matrices interpreted over extension fields}Let $G$ be a generator matrix of an $[n,k]_q$-code $C$. Since $\F_q\subset\F_{q^l}$ 
for each integer $l\ge 2$ we can also interpret $G$ as a generator matrix of an $[n,k]_{q^l}$-code $C'$. Let $\cM$ and $\cM'$ be the corresponding multisets of 
points in $\PG(k-1,q)$ and $\PG(k-1,q^l)$, respectively, and assume that $C$ is $q^r$-divisible. Each hyperplane $H$ in $\PG(k-1,q^l)$ corresponds to a subspace $S$ 
in $\PG(k-1,q)$ with dimension $k-l\le\dim(S)\le k-1$, so that we can use Lemma~\ref{lemma_heritable} to conclude that $C'$ is $q^{r-l+1}$-divisible. Here we have $\cM'(H)=\cM(S)$ and the dimension bound $\dim(S)\ge k-l$ is responsible for the loss of the factor $q^{l-1}$ in the divisibility.

\begin{nexercise}Let $G$ be a generator matrix of a $q^r$-divisible $[n,k]_q$-code $C$ and $C'$ be the row span over $\F_{q^l}$. Show that $C'$ is a $q^{r-l+1}$-divisible 
$[n,k]_{q^l}$-code that has the same maximum point multiplicity as $C$.
\end{nexercise}
\end{trailer}

\begin{nexample}
  \label{example_baer_plane}
  Let
  $$
    G=\begin{pmatrix}
    1111000\\
    0011110\\
    0101011
    \end{pmatrix}\in \F_2^{3\times 7}
    \quad\text{and}\quad 
    G'=\begin{pmatrix}
    1111000\\
    0011110\\
    0101011
    \end{pmatrix}\in \F_4^{3\times 7}.
  $$
  The code $C$ spanned by $G$ is a $3$-dimensional simplex code over $\F_2$, i.e., a $[7,3,\{4\}]_2$-code, the code $C'$ spanned by $G'$ is a $[7,3]_4$-code, and 
  let $\cM'$ be the corresponding multiset of points in $\PG(2,4)$. Let us represent the hyperplanes in $\PG(2,4)$ by perpendicular points and the elements of $\F_4$ 
  by linear polynomials over $\F_2$. If $H$ is a hyperplane represented by a point $\left\langle a_1,a_2,a_3\right\rangle$ with $a_1,a_2,a_3\in \F_2$ and 
  $\left(a_1,a_2,a_3\right)\neq \zv$, then $\cM'(H)=3$. In all other cases we have $\cM'(H)=1$, so that $\cM'$ and $C'$ are $2$-divisible.    
\end{nexample}
As another example we consider the projective $32$-divisible binary linear code of length $321$ obtained in Exercise~\ref{ex_affine_spaces_switching}. Over $\F_4$ 
we obtain a projective $16$-divisible linear code of length $321$.\footnote{This example does not occur in the proof of Lemma~\ref{lemma_picture_q_4_r_2} since 
$321=(85+4\cdot 43)+64$ allows a different construction using the codes of Example~\ref{example_affine_space} and Corollary~\ref{cor_spread_switching_q}.}

\begin{trailer}{Baer subspaces}If $G$ is the generator matrix of an $l$-dimensional simplex code over $\F_q$, cf.\ Example~\ref{example_baer_plane}, and $\cM'$ 
is the multiset of points in $\PG(v-1,q^2)$, where $v\ge l$, corresponding to the linear 
code over $\F_{q^2}$ spanned by $G$, then we call $\cM'$ an $l$-dimensional \emph{Baer subspace}. If $l=3$, then we speak of a \emph{Baer plane}. Note that 
$l$-dimensional Baer subspaces are $q^{l-2}$-divisible, where $l\ge 2$. If $S$ is an $(l+1)$-dimensional Baer subspace and $T$ an $l$-dimensional Baer subspace that is contained 
in $S$, then $\chi_S-\chi_T$ is called \emph{$(l+1)$-dimensional affine Baer subspace} and is $q^{l-2}$-divisible for $l\ge 2$.  
\end{trailer}
We remark that Baer subspaces yield two-weight codes, cf.\ \cite[Example RT1]{calderbank1986geometry}. Affine Baer subspaces only give {\lq\lq}few{\rq\rq} weight 
codes.

As mentioned e.g.\ in \cite{hirschfeld1998projective}, a partition of $\PG(v-1,q^l)$ into subgeometries $\PG(v-1,q)$ exists 
iff $\gcd(v,l)=1$. In particular, $\PG(v-1,q^2)$ can be partitioned into subgeometries $\PG(v-1,q)$, i.e.\ Baer subspaces, 
precisely when $v$ is odd. Such a partition is called a \emph{Baer subgeometry partition} (BSP). Using a Singer cycle, BSPs for 
$\PG(2,q^2)$ were constructed by Bruck \cite{bruck1960quadratic}. While this technique can be generalized to other 
parameters, also other constructions are known, see e.g.\ \cite{baker2000baer}. Since the set of all points $\PG(2,q^2)$ 
is $q^4$-divisible and can be partitioned into $q^2-q+1$ Baer planes, we can generalize the switching construction from 
Corollary~\ref{cor_spread_switching_q} by switching Baer planes to affine Baer solids, i.e., we apply Exercise~\ref{exercise_generalized_switching_construction} 
choosing the $\cM_i$ as Baer solids and $\cD$ as a common Baer plane.
\begin{nexercise}
  \label{ex_baer_switching}
  Construct projective $q^2$-divisible codes of length $$n=\left(q^4+q^2+1\right)+j\cdot\left(q^4-[4]_q\right)$$ over 
  $\F_{q^2}$ for all $0\le j\le q^2-q+1$. 
\end{nexercise}        

\section{Computer searches}
\label{subsec_computer_searches}
We have already reported that there is a unique projective $8$-divisible binary linear code of length $50$, see Table~\ref{table_8div_enumeration}. This example 
was found using exhaustive generation of linear codes (with restrictions on the set of allowed weights). Suitable software packages are e.g.\ 
\texttt{QextNewEdition}, or its predecessor \texttt{Q-Extension} \cite{bouyukliev2007q},  and \texttt{LinCode}, see \cite{bouyukliev2020computer}. Further classifications 
for linear codes have e.g.\ been presented in \cite{betten2006error} and \cite{ostergaard2002classifying}, see also \cite[Section 7.3]{kaski2006classification}.

The search problem for projective $q^r$-divisible codes can easily be formulated as an integer linear programming (ILP) problem using binary 
characteristic variables $x_P$ for all points $P$ of $\PG(v-1,q)$, i.e., $x_P$ encodes the multiplicity of $P$. Prescribing the desired cardinality 
$n=\sum_{P\in\cP} x_P$ and the dimension $k$, it remains to convert the restrictions induced by $q^r$-divisibility, see Equation~(\ref{eq_divisible_multiset}), into linear constraints:
\begin{equation}
  \sum_{P\le H} x_P = n- z_H\cdot q^r
\end{equation}   
for each hyperplane $H\in\cH$, where $z_H\in \N_0$ and $z_H\le \left\lfloor n/q^r\right\rfloor$. Although Lemma~\ref{lemma_heritable} allows to include \textit{modulo-constraints} 
on the number of holes for subspaces other than hyperplanes, ILP solvers seem not to benefit from these extra constraints. If the desired divisible codes do not need to be projective, 
we can use integer variables $x_P\in\N_0$ (with a possible upper bound $x_P\le \lambda$ in case of maximum point multiplicity $\lambda$). Of course, we may prescribe 
$x_{\left\langle\be_i\right\rangle}=1$ (or $x_{\left\langle\be_i\right\rangle}\ge 1$) for all $1\le i\le k$.  

Since larger instances can not be successfully treated directly by customary ILP solvers, we have additionally prescribed some 
symmetry to find examples. This general approach is called the Kramer--Mesner method \cite{KramerMesner:76}. Given a group $G$ acting on the set of points $\cP$ 
and the set of hyperplanes $\cH$ we additionally assume $x_{P^g}=x_P$ and $z_{H^g}=z_H$ for each $P\in\cP$, each $H\in\cH$, and each $g\in G$, where $P^g\in\cP$ and $H^g\in\cH$ 
denote the image of the group operation of $g$ applied to $P\in\cP$ and $H\in\cH$, respectively. This rather general method was e.g.\ applied to general linear codes 
\cite{braun2005optimal} and two-weight codes \cite{kohnert2007constructing_twoweight}. For an exemplary application to the construction of constant-dimension codes 
we refer e.g.\ to \cite{paper_axel}. 

As an example consider $q=7$, $r=1$, and $k=4$. Let $g$ be the permutation matrix of the cyclic shift $\be_j\mapsto \be_{j+1}$, where the indices are taken modulo $4$, 
and let $G=\langle g\rangle$ be the cyclic group of order $4$ generated by $g$. The orbits of $G$ on the $400$ points of $\PG(3,7)$ have lengths $1$, $2$, or $4$. 
Instead of one variable per point and one variable per hyperplane we use one binary variable per point orbit and one integer variable $z_H$ per hyperplane orbit. 
Each hyperplane orbit contributes the same constraint $\sum_{P\le H} x_P=n-7\cdot z_H$ for all of its elements, so that the number of variables and constraints drops by 
roughly a factor of $4$. For $n=141$ a solution consists of $39$ orbits, namely one orbit of length $1$, six orbits of length $2$, and $32$ orbits of length $4$, 
which gives $1+6\cdot 2+32\cdot 4=141$. Representatives of these orbits, written as vectors in $\F_7^4$, are 
\begin{eqnarray*}
  &&0001,\,0012,\,0016,\,0101,\,0103,\,0106,\,0112,\,0114,\,0122,\,0125,\,0126,\,0134,\\
  &&0141,\,0142,\,0145,\,0162,\,0166,\,1114,\,1115,\,1116,\,1123,\,1132,\,1134,\,1135,\\
  &&1142,\,1146,\,1165,\,1166,\,1234,\,1235,\,1245,\,1252,\,1263,\,1265,\,1313,\,1343,\\
  &&1352,\,1463,\,1616,
\end{eqnarray*}
where the vectors $0101$, $0106$, $1166$, $1265$, $1313$, and $1463$ have orbits of length $2$ and $1616$ has an orbit of length $1$. 
The resulting projective $7$-divisible $[141,4]_7$ code is the code with weight enumerator 
$W_C(x)=1+30x^{112}+1692x^{119}+672x^{126}+6x^{133}$ that is listed below. Its projective automorphism group has order $4$, so that the code is not cyclic. 

Prescribing cyclic groups in our application, we found the following generator matrices:
\begin{itemize}
  \item $q=2$, $r=3$, $n=50$, $k=8$, $W_C(x)=1+5x^{16}+210x^{24}+40x^{32}$, $$\left(\begin{smallmatrix}
00000000000000000000000000111111111111111111111111\\
00000000000000111111111111000000000000111111111111\\
00000000111111000000111111000000111111000000111111\\
00000011000011000011000011001111001111001111001111\\
00001100000101001111011101000011010001110011010111\\
00110000001111000101101110000101110110010100100111\\
01010101011101011011100111011101011011110101111001\\
10010101011011011101010111011011110011011101111010\\
\end{smallmatrix}\right)$$
 \item $q=2$, $r=3$, $n=74$, $k=12$, $W_C(x)=1+3x^8+ 60x^{24}+ 1423x^{32}+2585x^{40}+ 24x^{48}$, $$\left(\begin{smallmatrix}
00001101111111111001000010001001110111111100101001100011011000011011000111\\
11100001101111111000110110001110110001001111100101100011100100110011110000\\
11001100101011011010100010001010100111000000000010000110000011011111110001\\
00001110110001000010010100100111100011101010100110010000000011110001011101\\
01000010111000110011111010110101010111110000101100000001011011111111110001\\
00100111100011111010011100011110010000000011111101001110111110001000110111\\
11001011110011111011010111100110101001001001001111001110000010101001101010\\
01111000111001010010110000011011111111000101011110010000011101101101011100\\
11110110101010111011010111011011110000000110001110010111000101001010011010\\
10101001110000001100111011011110011111000101010010110110011001011101011100\\
00100001111110101101001110111101000111010011011110101000001001011000110111\\
11111110000110110001010001011111000100110110110100011110001110011010011010\\
\end{smallmatrix}\right)$$ 
 \item $q=2$, $r=4$, $n=161$, $k=10$, $W_C(x)=1+50x^{64}+886x^{80}+ 87x^{96}$, 
 \begin{eqnarray*}
 && \left(\begin{smallmatrix}
0000000000000000000000000000000000000000000000000000000000000000000000000000000\\ 
0000000000000000000000000000000000000000011111111111111111111111111111111111111\\ 
0000000000000000000000000111111111111111100000000000000000000000011111111111111\\ 
0000000000000001111111111000000000011111100000000001111111111111100000011111111\\ 
0000000001111110000001111000011111100111100000011110000000011111100111100001111\\ 
0000011110001110111110011000100001111000100001100010001111100001101001111110001\\ 
0001101110110110000111101001100110101111100010000110110001100110100011100010010\\ 
0010100110001010011010111010111010101001000101101001000110101011110011100110100\\ 
0100001010010110101010100001000101101001001110101010010010000110000000001011110\\ 
1000100110110000010110100011101000000011001101011101001010110111011010100101110\\ 
\end{smallmatrix}\right.\\ 
&& \left.\begin{smallmatrix} 
0011111111111111111111111111111111111111111111111111111111111111111111111111111111\\
1100000000000000000000000000000000000000001111111111111111111111111111111111111111\\
1100000000000000001111111111111111111111110000000000000000111111111111111111111111\\
1100000000001111110000000000000011111111110000001111111111000000000011111111111111\\
1100000011110000110000001111111100001111110000110000001111000011111100000011111111\\
1100001100010001110001110000111100110111110011010001111111000100001100001100011111\\
0100110011100010000010000111000101011000110101010010111111011100111100010100101111\\
0101111101110110110000010001011000110001011010000001000001001101011101111100110111\\
0101110010111011010110101011001110111010011111110000110011110000100100111101111011\\
1100010100110100011100100010011000000111011011010111010101111100010111000100111101\\
\end{smallmatrix}\right)
\end{eqnarray*}
  \item $q=2$, $r=4$, $n=162$, $k=10$, $W_C(x)=1+x^{32}+30x^{64}+890x^{80}+102x^{96}$, 
  \begin{eqnarray*} &&\left(\begin{smallmatrix}
000000000000000000000000000000000000000000000000000000000000000000000000000000000\\ 
000000000000000000000000000000000000000000111111111111111111111111111111111111111\\
000000000000000000000011111111111111111111000000000000000000001111111111111111111\\ 
000000000011111111111100000000111111111111000000001111111111110000000011111111111\\ 
000000111100000011111100001111000000111111000000110000000011110000001100000000111\\ 
000011001100011100011100000000000001000001000111010000111100110001110100001111001\\ 
000100010100000000100000110001000111000110001001100011000101010111111100110011000\\ 
001001000001100101100101010110011111011010011011000101011110111010111111011111011\\ 
010000011000101011101101111010101011101111001000000010001011011100010001010101111\\ 
100001001101111010110101000100110010001101010100111000011111110000101001001001101\\ 
\end{smallmatrix}\right.\\ 
&& \left.\begin{smallmatrix}
011111111111111111111111111111111111111111111111111111111111111111111111111111111\\
100000000000000000000000000000000000000001111111111111111111111111111111111111111\\
100000000000000000000111111111111111111110000000000000000000011111111111111111111\\
100000000000011111111000000000000111111110000000000001111111100000000000011111111\\
100001111111100111111000011111111001111110000001111110000111100000011111100001111\\
100110000111101000111001100001111010001110111110111111111111100011100011100110011\\
101110011001100000001010101110011100110111001110001110111001111101111111101011101\\
100010000010000001010001000010101110010011010010000011001010101100101100111110110\\
001001100010111011100100001111011111110110000100010100011010000100010101110101101\\
100100101110110101111000000101101001011100100111011001101110101001000101000111010\\
\end{smallmatrix}\right)
\end{eqnarray*}
  \item $q=2$, $r=4$, $n=195$, $k=10$, $W_C(x)=1+33x^{80}+855x^{96}+135x^{112}$,
  {\scriptsize\begin{eqnarray*}&&\left(\begin{smallmatrix}
00000000000000000000000000000000000000000000000000000000000000000000000000000000000000000000000000\\ 
00000000000000000000000000000000000000000000000000011111111111111111111111111111111111111111111111\\ 
00000000000000000000000000011111111111111111111111100000000000000000000000011111111111111111111111\\ 
00000000000000011111111111100000000000011111111111100000000000011111111111100000000000011111111111\\ 
00000000011111100000011111100000011111100000011111100000011111100000011111100000011111100000011111\\ 
00000001100011100011100001100001100011100011100001100001100011100011100001100001100011100011100001\\ 
00000110101100000101100110000010101100100001100110101110100101101101100011101110001101101100100111\\ 
00011010000100100000100010101101110100101110111011100111001110100110101100000010100110111100101010\\ 
01101000000100001110101101100111001101000101011010101001101011010110110010000010010001110111000111\\ 
10100000001101101101000011101010101110000100101101100010111010110101000110100101110100110001011011\\ 
\end{smallmatrix}\right.\\ 
&& \left.\begin{smallmatrix}
0111111111111111111111111111111111111111111111111111111111111111111111111111111111111111111111111\\
1000000000000000000000000000000000000000000000000111111111111111111111111111111111111111111111111\\
1000000000000000000000000111111111111111111111111000000000000000000000000111111111111111111111111\\
1000000000000111111111111000000000000111111111111000000000000111111111111000000000000111111111111\\
1000000111111000000111111000000111111000000111111000000111111000000111111000000111111000000111111\\
1001111000111000111001111001111000111000111001111001111000111000111001111001111000111000111001111\\
1000011001011001001010011010011001001011001010001010011011011011001110011010111001011011011010011\\
1001100011100011110010101110001011011100011000011000011000001101000010101111011110101001101110101\\
1000101000101101010111011100111010100001011100110010100111001100011010011010001010100111011101100\\
0010100110011010110101001010110110010000101110011010001001111011010011010101001001101001010111010\\
\end{smallmatrix}\right)
 \end{eqnarray*}}
 \item $q=2$, $r=4$, $n=197$, $k=10$, $W_C(x)=1+10x^{80}+837x^{96}+176x^{112}$,
  {\scriptsize\begin{eqnarray*}&&\left(\begin{smallmatrix}
00000000000000000000000000000000000000000000000000000000000000000000000000000000000000000000000000\\ 
00000000000000000000000000000000000000000000000000000111111111111111111111111111111111111111111111\\ 
00000000000000000000000000000111111111111111111111111000000000000000000000000111111111111111111111\\ 
00000000000000000111111111111000000000000111111111111000000000000111111111111000000000000111111111\\ 
00000000001111111000000011111000000011111000001111111000000011111000001111111000001111111000000011\\ 
00000001110000111000001100011000011100111000110000011000001100011001110000111000110000011000011100\\ 
00000110110001001000110100100000100101001001000011111001110100101010110111011011110011101001100101\\ 
00011010010011000000011000011011000111111011010100101010010101000101010011001001010101110010101111\\ 
01101000000010001011011101100001001001011000110101111010101010010110101001011000110110100101100110\\ 
10100000010111010001101000101000110010011001101100111010100100110101101010110001011010111011000011\\ 
\end{smallmatrix}\right.\\
&& \left.\begin{smallmatrix}
000111111111111111111111111111111111111111111111111111111111111111111111111111111111111111111111111\\
111000000000000000000000000000000000000000000000000111111111111111111111111111111111111111111111111\\
111000000000000000000000000111111111111111111111111000000000000000000000000111111111111111111111111\\
111000000000000111111111111000000000000111111111111000000000000111111111111000000000000111111111111\\
111000000011111000001111111000001111111000000011111000001111111000000011111000000011111000001111111\\
111000111100011001110011111000110001111001111100111001110011111000111100011001111100111000110001111\\
011001000100101010110100011001010110011010001100001010011101111001001101101010011101001111110010011\\
101001001001010100011000111011100010101110110101011000101110011001010110000101111100011011001110001\\
101001011000101001111001101101001100100110011010110000010100111010010111110011100100100111010101100\\
010010010010110000110101111110011011001010110001100010101101101100100010101000101100110101101100110\\
\end{smallmatrix}\right)
\end{eqnarray*}}
  \item $q=5$, $r=1$, $n=41$, $k=5$, $W_C(x)=1+4x^{25}+1360x^{30}+1760x^{35}$, 
$$\left(\begin{smallmatrix}
00000000000111111111111111111111111111111\\
00011111111000000001111111122233344444444\\
00100112344012223330011122403303300111224\\
01003013314040131332201203012211324033021\\
10003222331041124030121012410144041303440\\
\end{smallmatrix}\right)
$$
  \item $q=7$, $r=1$, $n=141$, $k=4$, i.e., a projective $7$-divisible $[141,4]_7$ code, $W_C(x)=1+30x^{112}+1692x^{119}+672x^{126}+6x^{133}$,  
\begin{eqnarray*}&&\left(\begin{smallmatrix}
00000000000000000000001111111111111111111111111111111111111111111111111\\ 
00001111111111111111110000000000000000001111111111111111112222222222222\\ 
01110000011222234446660001122333444456660011122333444456660002223334455\\ 
10260135624025641250260460426056123400154645603245012611560120261455602\\ 
\end{smallmatrix}\right.\\ 
&& \left.\begin{smallmatrix}
1111111111111111111111111111111111111111111111111111111111111111111111\\
2222233333333333333333344444444444444444455555555555666666666666666666\\
6666600122333444555566600111222344555566601112234566000111122335556666\\
0345624314234013124524512016012525035613561562424324045145606261350136\\
\end{smallmatrix}\right)
\end{eqnarray*}
  \item $q=4$, $r=2$, $n=150$, $k=6$, $W_C(x)=1+15x^{64}+45x^{96}+3810x^{112}+225x^{128}$, where $\F_4=\{0,1,2,3\}$ with $2=\alpha$ and $3=\alpha^2=\alpha+1$, 
{\footnotesize\begin{eqnarray*}&&\left(\begin{smallmatrix}
0000000000000000000000000000000000000011111111111111111111111111111111111111111111111111111111111\\
0000000000111111111111111111111111111100000000000000000000000000001111111111111111111111111111222\\
0011111111000000001111222222223333333300000000111122222222333333330000111111112222222233333333000\\
1101111223000023330233002222330012223300112222123300001123000222330133011233330011223301111223001\\
2330122221001300132302000112121301231100130011021201131203123013021223211301232333132231123121231\\
2132312120023310200232022012012100331213300103301310100123033131012321102320020223100300112023021\\
\end{smallmatrix}\right.\\
&& \left.\begin{smallmatrix}
11111111111111111111111111111111111111111111111111111\\
22222222222222222222222223333333333333333333333333333\\
00000111111112222222233330000000011111111222233333333\\
22223000122330012223300110112233300112233001311112333\\
02330123200033300231301013000201203231222021001230123\\
12023331201101203322121123122100302312312323120033232\\
\end{smallmatrix}\right)
\end{eqnarray*}}
  \item $q=4$, $r=2$, $n=193$, $k=6$, $W_C(x)=1+3x^{32}+6x^{112}+51x^{128}+3750x^{144}+285x^{160}$, 
{\footnotesize\begin{eqnarray*}&&\left(\begin{smallmatrix}
0000000000000000000000000000000000000000000000000111111111111111111111111111111111111111111111111\\
0000000000000111111111111111111111111111111111111000000000000000000000000000000000000111111111111\\
0000011111111000000001111111111112222222233333333000000001111111111112222222233333333000000000000\\
0111101112233001112330111122233330011233300112233001222230001111222230000112200111123000111222233\\
1011231330303020120011022301300130113211213022223231122320230133012320233031302001200012001123301\\
2313000030132212323021013123301300212202130212302022012322031303110231202232103231012122120231331\\
\end{smallmatrix}\right.\\
&& \left.\begin{smallmatrix}
111111111111111111111111111111111111111111111111111111111111111111111111111111111111111111111111\\
111111111111111111111111222222222222222222222222222222222222333333333333333333333333333333333333\\
111111112222222233333333000000001111111122222222333333333333000000001111111122222222222233333333\\
002222330011123300001223000112230111123301222233001111222233001111330000112300011112223300122233\\
011123132323330300123010023230221112311230023313030123112312230123011123121102301132230201000202\\
020320030230131002010310011021213011321031301221330201021332131330021302230313131220202113203330\\
\end{smallmatrix}\right)
\end{eqnarray*}}
  \item $q=4$, $r=2$, $n=259$, $k=6$, $W_C(x)=1+63x^{176}+3381x^{192}+651x^{208}$, 
{\footnotesize\begin{eqnarray*}&&\left(\begin{smallmatrix}
0000000000000000000000000000000000000000000000000000000000000000000111111111111111111111111111111\\
0000000000000001111111111111111111111111111111111111111111111111111000000000000000000000000000000\\
0011111111111110000000000000000011111111122222222222222222333333333000001111111112222222222222333\\
1100111122333330000001111122223300012233300001112222233333111122333011220011112230001111112223000\\
2323013313002330222331223301130302200213300330120112302333122323022202330201231210130001231232122\\
2000330230233032023132011312303020321030213013313230332023302012023032123233213031200131202301102\\
\end{smallmatrix}\right.\\
&& \left.\begin{smallmatrix}
1111111111111111111111111111111111111111111111111111111111111111111111111111111111111111111111111\\
0000001111111111111111111111111111111111111111111111111111222222222222222222222222222222222222222\\
3333330000000000000000011111111111112222222223333333333333000000000000011111111122222222222223333\\
1223330000011222223333300001122233331112223330000112233333000111113333301111233300000123333330000\\
3120230023312012330012212330201211231131132231133230111333003011220122331123011300223230002231233\\
2113002321222230021331232123133101010101220320101322323123123203130303010202223102130300120120001\\
\end{smallmatrix}\right.\\
&& \left.\begin{smallmatrix}
11111111111111111111111111111111111111111111111111111111111111111\\
22222222222223333333333333333333333333333333333333333333333333333\\
33333333333330000000000000111111111111111112222222222222333333333\\
11111112223330000111122333000000111222333330000011122333001112233\\
01122230221130122012333223012233012013112330112301301123121331313\\
01201201121320101313103232320202202321031032131333031110330033101\\
\end{smallmatrix}\right)
\end{eqnarray*}}
  \item $q=4$, $r=2$, $n=280$, $k=6$, $W_C(x)=1+105x^{192}+3360x^{208}+630x^{224}$, 
{\footnotesize\begin{eqnarray*}&&\left(\begin{smallmatrix}
0000000000000000000000000000000000000000000000000000000000000000000000001111111111111111111111111\\
0000000000000000111111111111111111111111111111111111111111111111111111110000000000000000000000000\\
0000000111111111000000000000011111111111111111222222222222233333333333330000000000000111111111111\\
0011111001122233001111222233300001112222223333000111122223300011111223330001111223333000111222222\\
1100123022211213021122112312211233330112331122113001201221100302333022230230123330122002013011123\\
2313023112313321332302231000203030230031230101232133110030102103023100211021121123202032022212311\\
\end{smallmatrix}\right.\\
&& \left.\begin{smallmatrix}
1111111111111111111111111111111111111111111111111111111111111111111111111111111111111111111111111\\
0000000000000000000000000000000111111111111111111111111111111111111111111111111111111112222222222\\
1111122222222222223333333333333000000000000000001111111111111222222222222233333333333330000000001\\
3333300000011223330000111223333000011222222333330011122222233000000122233300001123333330001122230\\
0023300122302230020022113032223113302012233000230201200113301001123012301101222300223332231222300\\
0131313302022231202312033120120010220300201012212200003030332230113221230223031210010132303312003\\
\end{smallmatrix}\right.\\
&& \left.\begin{smallmatrix}
11111111111111111111111111111111111111111111111111111111111111111111111111111111111111\\
22222222222222222222222222222233333333333333333333333333333333333333333333333333333333\\
11111111222222222333333333333300000000000001111111111111222222222222222223333333333333\\
01122223001122233000001122333300111122222330111122223333000000111222233330001222333333\\
21101332011212302011330223002312012200023133012202220012001233001011202331230122001123\\
10303030012031323312231310023313210201201133211230121301033103133123220011103113231302\\
\end{smallmatrix}\right)
\end{eqnarray*}}
  \item $q=7$, $r=1$, $n=83$, $k=5$, $W_C(x)=1+6x^{7}+270x^{63}+13452x^{70}+3078x^{77}$, 
$$\left(\begin{smallmatrix}
00000011111111111111111111111111111111111111111111111111111111111111111111111111111\\
11111100000000000001111111111111222222333333333333344444444444445555555555555666666\\
12444500011333344450002233344466023566001111334456611123444556660001122334566011466\\
54035323634025601442441403613416363525450034450130101600356551561450135460256316613\\
64455116020034131250361260103163446134162564310546442603031063205643153145416004634\\
\end{smallmatrix}\right)$$
  \item $q=7$, $r=1$, $n=91$, $k=4$, $W_C(x)=1+6x^{49}+2016x^{77}+378x^{84}$, 
$${\footnotesize\left(\begin{smallmatrix}
0000000111111111111111111111111111111111111111111111111111111111111111111111111111111111111\\
1111111000000000000001111111222222233333333333333444444444444445555555555555566666666666666\\
0234456011122234556660011226112345600023334455566011122333455660001222334555600011233355666\\
6121345125623645150123456354021234534501231213545404535014123351364256140124212623423503146\\
\end{smallmatrix}\right)}$$
  \item $q=7$, $r=1$, $n=133$, $k=4$, $W_C(x)=1+18x^{105}+1662x^{112}+720x^{119}$, 
{\footnotesize\begin{eqnarray*}&&\left(\begin{smallmatrix}
0000000000000000000000000000111111111111111111111111111111111111111111111111111111111111111111111\\
0000111111111111111111111111000000000000000001111111111111111122222222222222222333333333344444444\\
0111001111222334445555555666000011233344455560122333444455666601122234445555666111223356600112333\\
1456460135234240350123456016134604623405602646046126024614134631223510142345346012132310113140124\\
\end{smallmatrix}\right.\\
&& \left.\begin{smallmatrix}
111111111111111111111111111111111111\\
444444444555555555555555556666666666\\
455566666000122334444455563333445566\\
613612356256523230245635641346253425\\
\end{smallmatrix}\right)
\end{eqnarray*}}
  \item $q=8$, $r=1$, $n=120$, $k=4$, $W_C(x)=1+105x^{96}+3360x^{104}+630x^{112}$, where the entry $b_0+2b_1+4b_2$ with $b_i\in\{0,1\}$ represents $b_0+b_1\alpha+b_2\alpha^2\in\F_8$ and $\alpha^3=\alpha+1$, 
{\footnotesize\begin{eqnarray*}&&\left(\begin{smallmatrix}
0000000000000000000000001111111111111111111111111111111111111111111111111111111111111111111111111\\
0000001111111111111111110000000000111111111122222222223333333333444444444455555555555555555566666\\
1111111122222233445566772223336777000333477723336667770333444777111333577700112244556677777713335\\
0123462401345713250547015670134023147036614704560124562356036127256047712556120347031201234731251\\
\end{smallmatrix}\right.\\
&& \left.\begin{smallmatrix}
11111111111111111111111\\
66666777777777777777777\\
55777001122333333445566\\
25347560347034567471256\\
\end{smallmatrix}\right)
\end{eqnarray*}}
\end{itemize}

In our ILP model the integer variable $z_H$ may be replaced by several binary variables $y_{H,n'}$, which are equal to $1$ iff hyperplane $H$ contains exactly $n'$ selected 
points, i.e., has multiplicity $n'$. This way, it is possible to exclude some specific multiplicities for hyperplanes or to 
count (and incorporate given bounds on) the number of hyperplanes with a given multiplicity. Restrictions for $n$ and the $n'$ 
are given by our exclusion results for $q^r$-divisible and $q^{r-1}$-divisible sets, respectively, see 
Section~\ref{sec_lengths_of_divisible_codes} or Section~\ref{sec_lengths_projective_q_r}. Prescribing a specific solution of the MacWilliams identities directly  
translates to equations for the number of hyperplanes with a given multiplicity. 

A Diophantine linear equation system, in the same vein as our ILP model, together with the prescription of a subgroup of the automorphism group of the code was used in 
\cite{kohnert2007constructing_twoweight} in order to construct two-weight codes with previously unknown parameters. In \cite[Lemma 7]{bouyukliev2020computer} the ILP 
approach was adjusted to the situation where a given $\Delta$-divisible $[n,k]_q$-code should be extended to a $\Delta$-divisible $[n',k+1]_q$-code.

\begin{question}{Research problem}
Use the ILP approach and some carefully selected candidates for subgroups of the automorphism group of a potential projective $q^r$-divisible code over $\F_q$ whose 
length is currently unknown to exist, in order to construct such a code, see Section~\ref{sec_lengths_projective_q_r}.
\end{question}

\section{Two-weight codes}
\label{subsec_two_weight_codes}
A linear $[n,k]_q$ code $C$ is called a \emph{two-weight code} if the non-zero codewords of $C$ attain just two possible weights, i.e., if it is an 
$[n,k,\{w_1,w_2\}]_q$ code for $w_1\neq w_2\in\N_{>0}$. An online-table for known two-weight codes is at \url{http://www.tec.hkr.se/~chen/research/2-weight-codes} 
and an exhaustive survey was given by Calderbank and Kantor \cite{calderbank1986geometry}. As observed by Delsarte, a projective two-weight code typically 
has a large divisibility. For non-projective two-weight codes the two non-zero weights can still be written as $up^t$ and $(u+1)p^t$ under mild extra conditions, 
see \cite{kurz2024nonprojective}.

\medskip

\begin{trailer}{Projective two-weight codes are divisible}
\vspace*{-6mm}
\begin{nlemma}(\cite[Corollary 2]{delsarte1972weights})\\
\label{lemma_delsarte}
Let $C$ be a projective two-weight code over $\mathbb{F}_q$, where $q=p^e$ for some prime $p$. Then there exist suitable integers $u$ and $t$ with
  $u\ge 1$, $t\ge 0$ such that the weights are given by $w_1=up^t$ and $w_2=(u+1)p^t$.
\end{nlemma}
\end{trailer}
For the $[10,4,6]_3$-code of Example~\ref{ex_ovoid_pg33} we have $p=3$, $u=2$, $t=1$, i.e., $w_1=2\cdot 3=6$ and $w_2=3\cdot 3=9$.
We remark that first order Reed--Muller codes or affine spaces, see Example~\ref{example_affine_space}, are examples of $\left[q^{r+1},r+2,\left\{ q^{r+1}-q^r,q^{r+1}\right\}\right]_q$ 
two-weight codes for all prime powers $q$ and all $r\in\N_0$. (Repeated) simplex codes are the unique possibility for one-weight codes and Baer subspaces, see 
Subsection~\ref{subsec_subfield_constructions}, yield two-weight codes. Solving Diophantine linear equation systems, similar to those discussed in Subsection~\ref{subsec_computer_searches}, 
leads to many examples of two-weight codes in \cite{kohnert2007constructing_twoweight}. Using Bose--Chaudhuri--Hocquenghem (BCH) codes a 
parametric family of two-weight codes was constructed in \cite{bierbrauer1997family}:

\begin{ntheorem}{(Cf.~\cite[Theorem 4]{bierbrauer1997family})}
  \label{thm_BCH_div}
\noindent For every prime-power $q$ and all $m,n'\in\N_{>0}$ with $m\le n'$ there exists a projective $q^{n'+m-1}$-divisible 
  $\left[q^m\cdot [n']_q\cdot\left(q^{n'}-q^{n'-m}+1\right),3n'\right]_q$-code.
\end{ntheorem} 
In some cases these codes can be obtained by concatenation with a suitable simplex code. We refer to this theorem as the Bierbrauer--Edel (BE) construction. 

Due to their omnipresence a lot of research has been done on two-weight codes and many examples are available. Although the topic has been studied for decades, new parametric 
families are still found, see e.g.\ \cite{heng2021family}. In Table~\ref{table_two_weight_codes} we 
list those parameters that we will use in Section~\ref{sec_lengths_projective_q_r} as examples.

\begin{table}[htp]
  \begin{center}
    \begin{tabular}{rrcrcp{7.6cm}}
      \hline
        $n$ & $k$ & $\{w_1,w_2\}$ & $\Delta$ & $q$ & description \\
      \hline
        51 &  8 & $\{24,32\}$   &  8 & 2 & \cite[Example CY1]{calderbank1986geometry}, Example~\ref{ex_8_div_51}\\
        73 &  9 & $\{32,40\}$   &  8 & 2 & \cite{fiedler1998strongly}, optimal code\\
       196 &  9 & $\{96,112\}$  & 16 & 2 & BE construction in Theorem~\ref{thm_BCH_div}\\
       198 & 10 & $\{96,112\}$  & 16 & 2 & computer search with prescribed automorphisms \cite{kohnert2007constructing_twoweight}, optimal code\\       
       231 & 10 & $\{112,128\}$ & 16 & 2 & \cite{dissett2000combinatorial}, optimal code\\
       234 & 12 & $\{112,128\}$ & 16 & 2 & \cite[Theorem 6.1]{calderbank1986geometry} applied to Example FE3 over $\mathbb{F}_4$ \\
       273 & 12 & $\{128,144\}$ & 16 & 2 & quasi-cyclic code \cite{chen2006constructions}, cyclotomic construction \cite{delange1995cyclotomic}\\
       276 & 11 & $\{128,144\}$ & 16 & 2 & \cite[Example $\text{RT5}^\text{d}$]{calderbank1986geometry}\\
       455 & 12 & $\{224,256\}$ & 32 & 2 & \cite[Example CY1]{calderbank1986geometry}\\ 
       780 & 12 & $\{384,416\}$ & 32 & 2 & BE construction in Theorem~\ref{thm_BCH_div}\\  
       845 & 12 & $\{416,448\}$ & 32 & 2 & computer search with prescribed automorphisms \cite{kohnert2007constructing_twoweight}\\
       975 & 12 & $\{480,512\}$ & 32 & 2 & listed in \cite{kohnert2007constructing_twoweight}\\
      1105 & 12 & $\{544,576\}$ & 32 & 2 & computer search with prescribed automorphisms \cite{kohnert2007constructing_twoweight}\\
      1170 & 12 & $\{576,608\}$ & 32 & 2 & concatenation of the $78$-set of type $(2,6)$ in $\PG(2,16)$ from \cite{deresmini1986}\\
        10 &  4 & $\{6,9\}$     &  3 & 3 & \cite[Example CY1 and RT2]{calderbank1986geometry}, Example~\ref{ex_ovoid_pg33}\\ 
        11 &  5 & $\{6,9\}$     &  3 & 3 & \cite[Example RT6]{calderbank1986geometry}, dual of the ternary Golay code \cite{barg1993dawn,golay,football_pool}\\
        55 &  5 & $\{36,45\}$   &  9 & 3 & optimal code \cite{gulliver1996two}, quasi-cyclic code \cite{chen2006constructions} \\     
        56 &  6 & $\{36,45\}$   &  9 & 3 & \cite[Example FE2]{calderbank1986geometry}, Hill cap \cite{hill1978caps}, optimal code\\
        84 &  6 & $\{54,63\}$   &  9 & 3 & BE construction in Theorem~\ref{thm_BCH_div}, optimal code\\
        98 &  6 & $\{63,72\}$   &  9 & 3 & optimal code \cite{gulliver1996two}, quasi-cyclic code \cite{chen2006constructions}, \cite{kohnert2007constructing_twoweight}\\
       126 &  6 & $\{81,90\}$   &  9 & 3 & points $P$ with $Q(P)=1$ for $\mathrm{Q}^-(5,3)$, \cite[Example FE1]{calderbank1986geometry}, \cite{martis2014enumeration}, optimal code\\
       260 &  6 & $\{192,208\}$ & 16 & 4 & BE construction in Theorem~\ref{thm_BCH_div}\\ 
       315 &  6 & $\{224,240\}$ & 16 & 4 & \cite[Example CY2]{calderbank1986geometry}, \cite[Theorem 3.1]{cossidente2008geometry}\\
        325 &  6 & $\{240,256\}$ & 16 & 4 & elliptic quadric $\mathrm{Q}^-(5,4)$, \cite[Example RT2]{calderbank1986geometry}\\
        39 &  4 & $\{30,35\}$   &  5 & 5 & optimal code \cite{dissett2000combinatorial}\\ 
      125 &  4 & $\{105,112\}$ &  7 & 7 & quasi-twisted code \cite{chen2027new}\\
      175 &  4 & $\{147,154\}$ &  7 & 7 & \cite[Example FE1]{calderbank1986geometry}\\
       205 &  4 & $\{180,189\}$ &  9 & 9 & complement of \cite[Example CY2]{calderbank1986geometry}\\      
      \hline
    \end{tabular}
    \caption{Parameters of a few selected two-weight codes.}
    \label{table_two_weight_codes}  
  \end{center}
\end{table}  

\enlargethispage{2\baselineskip}
We remark that the example of a projective binary $32$-divisible code of length $780$ can be obtained by concatenation of the example of a 
projective quaternary $16$-divisible code of length $260$. For more results on field changes we refer to \cite[Section 6]{calderbank1986geometry} for 
two-weight codes and Subsection~\ref{subsec_subfield_constructions} for divisible codes.

\chapter{Non-existence results for projective $q^r$-divisible codes}
\label{sec_nonexistence_projective_q_r}

The aim of this section is to draw some parametric conclusions from the linear programming method for projective $q^r$-divisible codes. However, we 
will mainly use the geometric reformulation, i.e., the standard equations in Lemma~\ref{lemma_standard_equations}. For parametric conclusions of the 
linear programming method for distance optimal linear codes we refer e.g.\ to \cite[Section 15.3]{bierbrauer2016introduction} and \cite{bierbrauer2007direct}. 
Our first example is an alternative version of Lemma~\ref{lemma_average}. Given a multiset of points $\cM$ in $\PG(v-1,q)$ let 
$\cT(\cM):=\left\{0\le i\le \#\cM\,:\, a_i>0\right\}$ denote the set of attained hyperplane multiplicities, where $a_i$ is the number of hyperplanes 
$H\in\cH$ with $\cM(H)=i$. 

\begin{trailer}{A {\lq\lq}linear{\rq\rq} condition}
\vspace*{-6mm}
\begin{nlemma}
  \label{lemma_hyperplane_types_arithmetic_progression_2}
  For integers $u\in\Z$, $m\ge 0$ and $\Delta\ge 1$ let $\cM$ be a $\Delta$-divisible multiset of points in $\PG(v-1,q)$  of cardinality
  $n=u+m\Delta\ge 0$.  Then, we have
  \begin{equation}
   \label{eq_hyperplane_types_arithmetic_progression_2}
  (q-1)\cdot\sum_{h\in\mathbb{Z},h\le m}
  ha_{u+h\Delta}=\left(u+m\Delta-uq\right)\cdot
  \frac{q^{v-1}}{\Delta}-m,
 \end{equation}
 where we set $a_{u+h\Delta}=0$ if $u+h\Delta<0$.
\end{nlemma}
\begin{nexercise}
  Use the standard equations from Lemma~\ref{lemma_standard_equations} to verify Equation~(\ref{eq_hyperplane_types_arithmetic_progression_2}). 
\end{nexercise}

\begin{ncorollary}
  \label{cor_nonexistence_arithmetic_progression_2}
  For integers $u,m\ge 0$ and $\Delta\ge 1$ let the multiset of points $\cM$ in $\PG(v-1,q)$  satisfy 
  $\#\cM=u+m\Delta$ and $\cT(\cM)\subseteq\{u,u+\Delta,\dots,u+m\Delta\}$. 
  Then, $u<\frac{m\Delta}{q-1}$ or $u=m=0$.
\end{ncorollary} 
\end{trailer}

\begin{nexample}
  \label{ex_proj_8_div_card_33} 
  Applying Corollary~\ref{cor_nonexistence_arithmetic_progression_2} with $q=2$, $\Delta=2$, $u=1$, and $m=0$ yields that no $2$-divisible multiset of 
  points over $\F_2$ of cardinality $1$ exists. With this we can choose $q=2$, $\Delta=4$, $u=5$, and $m=1$ in Corollary~\ref{cor_nonexistence_arithmetic_progression_2} 
  to conclude that no $4$-divisible multiset of points over $\F_2$ of cardinality $9$ exists. Using this and the non-existence of a $4$-divisible multiset of points over $\F_2$  
  of cardinality $1$, we can choose $q=2$, $\Delta=8$, $u=17$, and $m=2$ in Corollary~\ref{cor_nonexistence_arithmetic_progression_2} to conclude that no $8$-divisible multiset 
  of points over $\F_2$ of cardinality $33$ exists.  
\end{nexample}  

Of course, the non-existence of an $8$-divisible $[33,k]_2$ full-length code also follows from the methods presented in Section~\ref{sec_lengths_of_divisible_codes}, which 
are essentially based on the averaging argument in Lemma~\ref{lemma_average} and a suitable induction. Arguably Lemma~\ref{lemma_average} has some advantages over 
Corollary~\ref{cor_nonexistence_arithmetic_progression_2} since we can directly start with an $8$-divisible multiset $\cM$ of points over $\F_2$ of cardinality $33$ and conclude 
the existence of a hyperplane $H$ with $\cM(H)\in\{1,9\}$. The example of a potential $8$-divisible $[33,k]_2$ full-length code is also interesting when using the linear programming 
method directly. First note that we will have to prescribe some suitable values for $k$. If we allow all weights in $\{8,16,24,32\}$, then the MacWilliams equations admit a non-negative 
rational solution while the weights $32$ and $24$ might also be excluded with a separate linear programming computation. Thus, it definitely is useful to tabulate the possible 
lengths of (projective) $q^r$-divisible codes as we do in Section~\ref{sec_lengths_projective_q_r}.     

\begin{trailer}{A {\lq\lq}quadratic{\rq\rq} condition}
\vspace*{-6mm}
\begin{nlemma}
  \label{lemma_hyperplane_types_arithmetic_progression}
  For integers $u\in\mathbb{Z}$, $m\ge 0$, and $\Delta\ge 1$ let $\cM$ be a $\Delta$-divisible set of points in $\PG(v-1,q)$ of cardinality $n=u+m\Delta\ge 0$. 
  Then, we have 
  \begin{equation}
    (q-1)\cdot\sum_{h\in\mathbb{Z},h\le m} h(h-1)a_{u+h\Delta}=\tau_q(u,\Delta,m)\cdot \frac{q^{v-2}}{\Delta^2}-m(m-1),
  \end{equation}   
  where we set $\tau_q(u,\Delta,m)=$
  \begin{equation}
    m(m-q)\Delta^2+\left(q^2u-2mqu+mq+2mu-qu-m\right)\Delta+(q-1)^2u^2+(q-1)u
  \end{equation}
  and $a_{u+h\Delta}=0$ if $u+h\Delta<0$.
\end{nlemma}
\begin{proof}
  Rewriting the standard equations from Lemma~\ref{lemma_standard_equations} yields 
  \begin{eqnarray*}
    (q-1)\cdot\sum_{h\in\mathbb{Z},h\le m} a_{u+h\Delta} &=& q^2\cdot q^{v-2}-1,\\
    (q-1)\cdot\sum_{h\in\mathbb{Z},h\le m} (u+h\Delta)a_{u+h\Delta} &=& (u+m\Delta)(q\cdot q^{v-2}-1),\\
    (q-1)\cdot\!\!\!\!\!\sum_{h\in\mathbb{Z},h\le m}\!\!\!\!\! (u+h\Delta)(u+h\Delta-1)a_{u+h\Delta} &=& (u+m\Delta)(u+m\Delta-1)(q^{v-2}-1).
  \end{eqnarray*}
  $u(u+\Delta)$ times the first equation minus $(2u+\Delta-1)$ times the second equation plus the third equation gives 
  $\Delta^2$ times the stated equation. 
\end{proof}
\end{trailer}

The multipliers used in the proof of Lemma~\ref{lemma_hyperplane_types_arithmetic_progression} can be directly read off from the following observation.
\begin{nlemma}
  \label{lemma_binomial_matrix_3}
  For pairwise different non-zero 
  numbers $a,b,c$ the inverse matrix of 
  $$
    \begin{pmatrix}
      1     & 1     & 1     \\
      a     & b     & c     \\
      a^2-a & b^2-b & c^2-c \\
    \end{pmatrix}
  $$
  is given by 
  $$
    \begin{pmatrix}
       bc(c-b) & -(c+b-1)(c-b) &  (c-b) \\
      -ac(c-a) &  (c+a-1)(c-a) & -(c-a) \\
       ab(b-a) & -(b+a-1)(b-a) &  (b-a) \\
    \end{pmatrix}\cdot \big((c-a)(c-b)(b-a)\big)^{-1}
  $$
\end{nlemma}

Similar as for the {\lq\lq}linear condition{\rq\rq} we can conclude explicit non-existence criteria from Lemma~\ref{lemma_hyperplane_types_arithmetic_progression}: 
\begin{ncorollary}
  \label{cor_nonexistence_arithmetic_progression}
  For integers $u\in\mathbb{Z}$ and $\Delta,m\ge 1$ let $\cK$ be a $\Delta$-divisible set of points of cardinality $n=u+m\Delta\ge 0$ in $\PG(v-1,q)$. 
  If one of the following conditions hold, then $(q-1)\cdot \sum_{h\in\mathbb{Z},h\le m} h(h-1)a_{u+h\Delta}\notin\N_0$, which is impossible.
  \begin{enumerate}
    \item[(a)] $\tau_q(u,\Delta,m)<0$;
    \item[(b)] $\tau_q(u,\Delta,m)\cdot q^{v-2}$ is not divisible by $\Delta^2$;
    \item[(c)] $m\ge 2$ and $\tau_q(u,\Delta,m)=0$.
  \end{enumerate}
  We have the following special cases:
  \begin{eqnarray*}
    \tau_q(u,q^r,m)&=&\left(m(m-q)q^r-2mqu+q^2u+mq+2mu-qu-m\right)\cdot q^r\\
     && +\left(q^2u^2-2qu^2+qu+u^2-u\right),\\
    \tau_2(u,2^r,m)&=&\left(m(m-2)2^r -2mu+m+2u\right)\cdot 2^r +\left(u^2+u\right).\\
  \end{eqnarray*}
\end{ncorollary}
\begin{nexercise}
  \label{ex_first_gaps_4_div_and_8_div}
  Conclude the non-existence of projective $4$-divisible $[n,k]_2$-codes for all $n\in \{1,\dots,6\}\cup\{9,\dots,13\}$ and the 
  non-existence of projective $8$-divisible $[n,k]_2$-codes for all $n\in \{1,\dots,14\}\cup\{17,\dots,29\}\cup\{33,\dots,44\}$ from 
  Corollary~\ref{cor_nonexistence_arithmetic_progression_2} and Corollary~\ref{cor_nonexistence_arithmetic_progression}.  
\end{nexercise}
Note that in order to apply Lemma~\ref{lemma_hyperplane_types_arithmetic_progression}, we have to choose a parameter $m\in\N_0$. Given $m$,  
we can easily analyze when $\tau_q(u,\Delta,m)$ is non-positive:
\begin{nlemma}
  \label{lemma_negative_tau}    
  Given a positive integer $m$, we have $\tau_q(u,\Delta,m)\le 0$ iff 
  \begin{eqnarray}
    \label{ie_forbidden_interval}
    &&(q-1)u-(m-q/2)\Delta+\frac{1}{2}\notag\\ 
    &\in& \left[-\frac{1}{2}\sqrt{q^2\Delta^2-4qm\Delta+2q\Delta+1},
    \frac{1}{2}\sqrt{q^2\Delta^2-4qm\Delta+2q\Delta+1}\right].
  \end{eqnarray}
  The last interval is non-empty, i.e., the radicand is non-negative, iff $1\le m\le \left\lfloor(q\Delta+2)/4\right\rfloor$.
  We have $\tau_q(u,\Delta,1)=0$ iff $u=(\Delta-1)/(q-1)$ or $u=-\Delta$.
\end{nlemma}

\begin{backgroundinformation}{Quadratic functions that are non-negative over the integers}We remark that \cite[Theorem 1.B]{bose1952orthogonal} is quite similar to 
Lemma~\ref{lemma_hyperplane_types_arithmetic_progression} and its implications. Actually, their analysis grounds on \cite{plackett1946design} and is strongly related 
to the classical second-order Bonferroni Inequality \cite{bonferroni1936teoria,galambos1977bonferroni,galambos1996bonferroni} in Probability Theory. In simple words, 
the trick of Lemma~\ref{lemma_hyperplane_types_arithmetic_progression} is that $h(h-1)=h^2-h$ is non-negative for every integer $h$. Note that $f(x)=x^2-x$ attains its 
minimum at $x=\tfrac{1}{2}$ with function value $-\tfrac{1}{4}$. So, in some sense we perform a (quadratic) integer rounding cut.
\end{backgroundinformation}

We can also use Corollary~\ref{cor_nonexistence_arithmetic_progression_2} and Corollary~\ref{cor_nonexistence_arithmetic_progression} to show that for all cardinalities 
$n\le rq^{r+1}$ the attainable lengths of $q^r$-divisible sets of points over $\F_q$ are those that are attained by combinations of $(r+1)$-spaces and affine $(r+2)$-spaces, 
cf.\ Exercise~\ref{exercise_upt_to_r_qrp1}. 
\begin{ntheorem}{\cite[Theorem 11]{honold2018partial}}
  \label{thm_exclusion_r_1_to_ovoid}
  Let $\cM$ be a $q^1$-divisible set of points in $\PG(v-1,q)$ with cardinality $n$. If $1\le n\le q^2$, 
  then either $n=q^2$ or $q+1$ divides $n$. Additionally, the non-excluded cases can be realized. 
\end{ntheorem}
\begin{ntheorem}{\cite[Theorem 12]{honold2018partial}}
  \label{thm_exclusion_q_r}
  For the cardinality $n$ of a $q^r$-divisible set of points in $\PG(v-1,q)$, where $r\in\N_{>0}$, we have
  $$
    n\notin\left[(a(q-1)+b)\qbin{r+1}{1}{q}+a+1,(a(q-1)+b+1)\qbin{r+1}{1}{q}-1\right],
  $$
  where $a,b\in\mathbb{N}_0$ with $b\le q-2$ and $a\le r-1$. If $n\le rq^{r+1}$, then all other cases can be realized.
\end{ntheorem}

Similar as the conditions based on a linear and a quadratic polynomial in Lemma~\ref{lemma_hyperplane_types_arithmetic_progression_2} and 
Lemma~\ref{lemma_hyperplane_types_arithmetic_progression}, we can also conclude a condition based on a cubic polynomial. To this end 
we consider an explicit example first. 
\begin{nlemma}
  \label{lemma_no_8_div_52}
  No $2^3$-divisible set of points in $\PG(v-1,2)$ of cardinality $52$ exists.
\end{nlemma}
\begin{proof}
  Using the abbreviation $y=2^{v-3}$ the first four MacWilliams identities, see Equation~(\ref{eq_macwilliams}), are given by
  {\footnotesize\begin{eqnarray*}
    A_0+A_8+A_{16}+A_{24}+A_{32} +A_{40}+A_{48} &=& 8y \\
    {52\choose1}+{44\choose1}A_8+{36\choose1}A_{16}+{28\choose1}A_{24}+{20\choose1}A_{32} +{12\choose1} A_{40} +{4\choose1} A_{48} &=& 4y\cdot 52 \\
    {52\choose2} +{44\choose2} A_8+{36\choose2} A_{16}+{28\choose2} A_{24}
    +{20\choose2} A_{32} +{12\choose2} A_{40} +{4\choose2} A_{48} &=& 2y\cdot {52\choose2}  \\
    {52\choose3} +{44\choose3} A_8+{36\choose3} A_{16}
    +{28\choose3} A_{24}+{20\choose3} A_{32} +{12\choose3} A_{40} +{4\choose3} A_{48}&=& y\cdot\left({52\choose3} +B_3\right)
  \end{eqnarray*}} 
  Substituting $x=y\cdot B_3$ and rearranging yields
  \begin{eqnarray*}
    A_8 &=& -4+A_{40}+4A_{48}+\frac{1}{512} x+\frac{7}{64}y\\
    A_{16} &=& 6-4A_{40}-15A_{48}-\frac{3}{512} x-\frac{17}{64}y\\
    A_{24} &=& -4+6A_{40}+20A_{48}+\frac{3}{512} x+\frac{397}{64}y\\
    A_{32} &=& 1-4A_{40}-10A_{48}- \frac{1}{512} x+\frac{125}{64}y.
  \end{eqnarray*}
  With this we compute
  $$
    A_{16}+\frac{31}{20}A_8 = -\frac{1}{5} - \frac{49}{20} A_{40} -\frac{44}{5}A_{48} -\frac{123}{1280}y - \frac{29}{10240}x,
  $$
  which contradicts $A_8,A_{16},A_{40},A_{48},x,y\ge 0$.
\end{proof}
We remark that Lemma~\ref{lemma_no_8_div_52} generalizes Example~\ref{example_no_8_div_52_10_2_code} and Example~\ref{example_multipliers} dealing 
with all dimensions $v$, encoded in $y=2^{v-3}$, simultaneously. To this end we have replaced the non-linear $y\cdot B_3$ by a new variable $x$, which 
relaxes the problem on the one hand but turns the problem into a linear one on the other hand. 

\begin{nremark}
  The non-existence of a $2^3$-divisible set of cardinality $n=52$ implies several upper bounds for partial spreads, see Section~\ref{sec_partial_spreads} 
  and in particular Lemma~\ref{lemma_partial_spread_div_bound}. More precisely, we e.g.\ have $129\le A_2(11,8;4)\le 132$, $2177\le A_2(15,8;4) \le 2180$, 
  and $34945\le A_2(19,8;4) \le 34948$. 
\end{nremark}

The underlying idea of the proof of Lemma~\ref{lemma_no_8_div_52} can be generalized. Choosing a suitable basis for the first four MacWilliams equations, the 
multiplication with the inverse of a suitable $4\times 4$-matrix, cf.~Lemma~\ref{lemma_binomial_matrix_3} yields:

\medskip

\begin{trailer}{A {\lq\lq}cubic{\rq\rq} condition}
\vspace*{-6mm}
\begin{nlemma}
  \label{lemma_implication_fourth_mac_williams}
  Let $t\in\mathbb{Z}$ be an integer and $\cK$ be a $\Delta$-divisible set of points of cardinality $n>0$ in $\PG(v-1,q)$. Then, we have
  $$
    \sum_{i\ge 1} \Delta^2(i-t)(i-t-1)\cdot (g_1\cdot i+g_0)\cdot A_{i\Delta}\,\,+qhx 
    = n(q-1)\,t(t+1)\Delta^2 g_2,
  $$
  where $x=6q^{v-3}B_3\in\R_{\ge 0}$ for the number $B_3$ of codewords of weight $3$ in the dual code, $g_1=\Delta qh$, $g_0=-n(q\!-\!1)g_2$, $g_2=h-\left(2\Delta qt\!+\!\Delta q\!-\!2nq\!+\!2n\!+\!q\!-\!2\right)$ and 
  $$
    h= \Delta^2q^2t^2+\Delta^2q^2t-2\Delta nq^2t-\Delta nq^2+2\Delta nqt+n^2q^2+\Delta nq-2n^2q+n^2+nq-n.  
  $$ 
\end{nlemma}
\begin{ncorollary}
  \label{cor_implication_fourth_mac_williams}
  Using the notation of Lemma~\ref{lemma_implication_fourth_mac_williams}, if $t\notin\{-1,0\}$, $h\ge 0$, and $g_2<0$, 
  then there exists no $\Delta$-divisible set of points $\cK$ of cardinality $n$ in $\PG(v-1,q)$.
\end{ncorollary}
\begin{proof}
  First we observe $(i-t)(i-t-1)\ge 0$, $t(t+1)> 0$, and $g_1\ge 0$. Since $g_2<0$, we have $g_0\ge 0$ 
  so that $g_1i+g_0\ge 0$. Thus, the entire left hand side is non-negative and the right hand side is negative -- a contradiction.
\end{proof}
\end{trailer}
Applying Corollary~\ref{cor_implication_fourth_mac_williams} with $t=3$ gives Lemma~\ref{lemma_no_8_div_52}. Note that in Example~\ref{example_multipliers} 
we have only used the first three MacWilliams equations. As a further example we consider the parameters $q=2$, $\Delta=2^4=16$, and $n=235$. The condition $h\ge 0$ is satisfied for all integers $t$ since the excluded interval $(6.700,6.987)$ contains 
no integer. The condition $g_2<0$ just allows to choose $t=7$. 

We can perform a closer analysis in order to develop computational cheap checks. We have $g_2<0$ iff
\begin{equation}
  n\in\left( \frac{\Delta qt+\frac{\Delta q}{2}-\frac{3}{2}-\frac{1}{2}\cdot\sqrt{\omega}}{q-1}, \frac{\Delta qt+\frac{\Delta q}{2}-\frac{3}{2}+\frac{1}{2}\cdot\sqrt{\omega}}{q-1}\right),
\end{equation}
where $\omega=\Delta^2q^2-4qt\Delta-2\Delta q+4q+1$. Thus, $\omega>0$, i.e., we have $$t\le \left\lfloor\frac{q\Delta-2}{4}+\frac{1}{\Delta}+\frac{1}{4q\Delta}\right\rfloor.$$ 
We have $h\ge 0$ iff
\begin{equation}
  n\notin\left( \frac{\Delta qt+\frac{\Delta q}{2}-\frac{1}{2}-\frac{1}{2}\cdot\sqrt{\omega-4q}}{q-1}, \frac{\Delta qt+\frac{\Delta q}{2}-\frac{1}{2}+\frac{1}{2}\cdot\sqrt{\omega-4q}}{q-1}\right).
\end{equation}
The most promising possibility, if not the only at all, seems to be 
\begin{equation}
  n\in \Big( \frac{\Delta qt+\frac{\Delta q}{2}-\frac{3}{2}-\frac{1}{2}\cdot\sqrt{\omega}}{q-1},\frac{\Delta qt+\frac{\Delta q}{2}-\frac{1}{2}-\frac{1}{2}\cdot\sqrt{\omega-4q}}{q-1}\Big],
\end{equation}
which allows the choice of at most one integer $n$. In our example $q=2$, $\Delta=2^4=16$ the possible $n$ for $t=1,\dots,7$ correspond to 
$33,66,99,132,166,200,235$, respectively. The two other conditions are automatically satisfied.

\begin{nexercise}
  \label{ex_32_div}
  Show that no projective $2^5$-divisible $[n,k]_2$-code with
  $$
    n \in \left\{325, 390, 456, 521, 587, 652, 718, 784, 850, 917, 985\right\}
  $$
  exists.  
\end{nexercise}

\begin{nlemma}
  \label{lemma_exclusion_n_89_q_3_Delta_9}
  No $3^2$-divisible set of points in $\PG(k-1,3)$ of cardinality $89$ exists.
\end{nlemma}
\begin{proof}
  We set $x=3^{k-4}$, $y=3^{k-4}\cdot B_3$, and $z=3^{k-4}\cdot B_4$. Solving the first five MacWilliams equations 
  for $A_{9}$, $A_{54}$, $A_{63}$, $x$, and $y$ yields the equation
  \begin{eqnarray*}
    99630A_9+121905A_{18}+99873A_{27}+60021A_{36}&&\\ 
    +22275A_{45}+22518A_{72}+61236A_{81}+z&=&0,
  \end{eqnarray*}
  so that $A_9=A_{18}=A_{27}=A_{36}=A_{45}=A_{72}=A_{81}=z=0$. With that, the equation system has the unique solution 
  $x=189$, $y=33642$, $A_{54}=6230$, and $A_{63}=9078$. However, $189$ is not a power of three, but $x=3^{k-4}$.
\end{proof}
We remark that for the parameters of Lemma~\ref{lemma_exclusion_n_89_q_3_Delta_9} the first four MacWilliams equations 
permit non-negative rational solutions for all dimensions $9\le k\le 89$. When adding the fifth MacWilliams equation, the 
corresponding polyhedron gets empty.

\begin{nexercise}
  \label{ex_implement_exclusion_lemmas}
  Implement the non-existence criteria for lengths of projective $q^r$-divisible codes over $\F_q$ presented in this section, cf.\ 
  Lemma~\ref{lemma_partial_picture_q_2_r_6}, Lemma~\ref{lemma_partial_picture_q_3_r_3}, and Lemma~\ref{lemma_partial_picture_q_5_r_2}.
\end{nexercise}

\begin{question}{Research problem}Conclude a general {\lq\lq}quartic condition{\rq\rq} from the linear programming method covering the parameters 
of Lemma~\ref{lemma_exclusion_n_89_q_3_Delta_9}.   
\end{question}

\chapter{Lengths of projective $q^r$-divisible codes}
\label{sec_lengths_projective_q_r}
The aim of this section is to summarize the current knowledge on the possible lengths of projective $q^r$-divisible codes. Even for small parameters 
there are several lengths where the existence of a corresponding code still remains undecided. This leaves plenty of space for own research, i.e., new constructions, 
cf.\ Section~\ref{sec_constructions_projective}, and more sophisticated techniques for non-existence proofs, cf.\ Section~\ref{sec_nonexistence_projective_q_r}, are needed. 

We will give brief proofs for our subsequent results. All of them are constructed in the same manner. On the constructive side we list some {\lq\lq}base examples{\rq\rq}, i.e., 
examples for some small cardinalities/lengths. Specific parametric series are mentioned explicitly, for more details on the used two-weight codes we refer to 
Subsection~\ref{subsec_two_weight_codes} and Table~\ref{table_two_weight_codes}, and for optimal linear codes we refer to Subsection~\ref{subsec_optimal_codes}  
and Table~\ref{table_optimal_codes}. Explicit generator matrices obtained by computer searches are listed in Subsection~\ref{subsec_computer_searches}. 
Without explicitly stating, we then invoke Lemma~\ref{lemma_sum_mult}, i.e., we use the fact that the set of attainable lengths 
is closed under addition. For the non-existence results we list the utilized results from Section~\ref{sec_nonexistence_projective_q_r}. In the statements 
we explicitly list those cardinalities/lengths where no non-existence result is mentioned and which are not implied by combinations of the base examples. Stating all 
details becomes a bit extensive when the parameters are not rather small. So, for a few medium sized parameters we only state the ranges of excluded cardinalities obtained 
via the methods outlined in Section~\ref{sec_nonexistence_projective_q_r} together with Lemma~\ref{lemma_average}, cf.\ Exercise~\ref{ex_implement_exclusion_lemmas}. Here $[a,b]$ denotes the list of integers $a,a+1,\dots,b$.

For the binary field the smallest open case is length $130$ for projective $16$-divisible codes. For $q=3$ and $\Delta=9$ the smallest open lengths are $70$ and $77$. If 
$q\ge 5$, then there are even open cases for projective $q$-divisible codes over $\F_q$, e.g., length $40$ for $q=5$.

\begin{nlemma}
  \label{lemma_picture_q_2_r_1}
  Let $\cM$ be a $2^1$-divisible set of $n$ points in $\PG(v-1,2)$, then $n\ge 3$ and all cases can be realized. 
\end{nlemma}
\begin{proof}
  The values $n\in\{1,2\}$ are excluded by Theorem~\ref{thm_exclusion_r_1_to_ovoid}. The \textit{base examples} of cardinalities $3$, $4$, and $5$ are given by 
  Example~\ref{ex_simplex_code}, Example~\ref{example_affine_space}, and Exercise~\ref{exercise_projective_base}, respectively.
\end{proof}

\begin{nlemma}
  \label{lemma_picture_q_2_r_2}
  Let $\cM$ be a $2^2$-divisible set of $n$ points in $\PG(v-1,2)$, then $n\in\{7,8\}$ or $n\ge 14$ and all mentioned cases can be realized. 
\end{nlemma}
\begin{proof}
  The cases $1\le n\le 6$ and $9\le n\le 13$ are excluded by Theorem~\ref{thm_exclusion_q_r}. \textit{Base examples} for cardinalities $7$ and $8$ are given by 
  Example~\ref{ex_simplex_code} and Example~\ref{example_affine_space}. For the range $15\le n\le 20$ we refer to Corollary~\ref{cor_spread_switching}.
\end{proof}

\begin{nlemma}
  \label{lemma_picture_q_2_r_3}
  Let $\cM$ be a $2^3$-divisible set of $n$ points in $\PG(v-1,2)$, then 
  $$n\in\{15,16,30,31,32,45,46,47,48,49,50,51\}$$ 
  or $n\ge 60$ and all cases can be realized. 
\end{nlemma}
\begin{proof}
  The cases $1\le n\le 14$, $17\le n\le 29$, and $33\le n\le 44$ are excluded by Theorem~\ref{thm_exclusion_q_r}.
  The case $n=52$ is excluded by Corollary~\ref{cor_implication_fourth_mac_williams} with $t=3$, see also Lemma~\ref{lemma_no_8_div_52}.  
  The cases $53\le n\le 58$ are excluded by Lemma~\ref{lemma_hyperplane_types_arithmetic_progression} using $m=4$. 
  The special case $n=59$ is treated in \cite{honold2019lengths}. 
  
  \textit{Base examples} for cardinalities $15$, $16$, and $49$ are given by Example~\ref{ex_simplex_code}, Example~\ref{example_affine_space}, and 
  Exercise~\ref{ex_affine_spaces_switching}, respectively. The range $63\le n\le 72$ is covered by Corollary~\ref{cor_spread_switching}. There are 
  two-weight codes for cardinalities $n\in\{51,73\}$ and sporadic examples found by computer searches for cardinalities $n\in\{50,74\}$. 
\end{proof}

\begin{nlemma}
  \label{lemma_picture_q_2_r_4}
  Let $\cM$ be a $2^4$-divisible set of $n$ points in $\PG(v-1,2)$, then 
  \begin{eqnarray*}
    n&\in&\{31,32,62,63,64,93,\dots,96,124,\dots,130,155,\dots,165,185,\dots,199,\\&&215,\dots,234,244,\dots,309\}
  \end{eqnarray*} 
  or $n\ge 310$ and all cases, possibly except 
  \begin{eqnarray*}
    n &\in& \{130,163,164,165,185,215,216,232,233,244,245,246,247,\\
               && 274,275,277,278,306,309\},
  \end{eqnarray*}
  can be realized. 
\end{nlemma}
\begin{proof}
  The cases $1 \le n \le 30$, $33 \le n \le 61$, $65 \le n \le 92$, and $97 \le n \le 123$ are excluded by 
  Theorem~\ref{thm_exclusion_q_r}. The cases $133 \le n \le 154$, $167 \le n \le 184$, 
  $201 \le n \le 214$, and $236 \le n \le 243$ are excluded by Lemma~\ref{lemma_hyperplane_types_arithmetic_progression} 
  using $m=5$, $m=6$, $m=7$, and $m=8$, respectively. The cases $n\in\{132,166,200,235\}$ are excluded by 
  Corollary~\ref{cor_implication_fourth_mac_williams} with $t=4,\dots,7$, respectively. The special case $n=131$ was 
  treated in \cite{kurz2020no131}. 
  
  \textit{Base examples} for cardinalities $31$, $32$, and $129$ are given by Example~\ref{ex_simplex_code}, Example~\ref{example_affine_space}, and  
  Exercise~\ref{ex_affine_spaces_switching} respectively. The range $255\le n\le 272$ is covered by Corollary~\ref{cor_spread_switching}. There are 
  two-weight codes for cardinalities $n\in\{196,198,231,234,273,276\}$. Additionally, we have a distance-optimal code for $n=199$ and sporadic examples found by 
  computer searches for cardinalities $n\in\{161,162,195,197\}$.
\end{proof}
Figure~\ref{fig_numberline_2_4} summarizes the statement of Lemma~\ref{lemma_picture_q_2_r_4}.

\begin{figure}[htbp]
  \centering
  \begin{tikzpicture}
    \fill[black!8] (0.019,0) rectangle (12.559,0.45);
    \fill[blue!70!black] (1.159,0) rectangle (1.235,0.45);
    \fill[blue!70!black] (2.337,0) rectangle (2.451,0.45);
    \fill[blue!70!black] (3.515,0) rectangle (3.667,0.45);
    \fill[blue!70!black] (4.693,0) rectangle (4.921,0.45);
    \fill[blue!70!black] (5.871,0) rectangle (6.175,0.45);
    \fill[blue!70!black] (7.049,0) rectangle (7.581,0.45);
    \fill[blue!70!black] (8.227,0) rectangle (8.797,0.45);
    \fill[blue!70!black] (8.873,0) rectangle (8.911,0.45);
    \fill[blue!70!black] (9.405,0) rectangle (10.393,0.45);
    \fill[blue!70!black] (10.469,0) rectangle (10.507,0.45);
    \fill[blue!70!black] (10.583,0) rectangle (11.609,0.45);
    \fill[blue!70!black] (11.647,0) rectangle (11.723,0.45);
    \fill[blue!70!black] (11.761,0) rectangle (12.559,0.45);
    \fill[orange!90!black] (4.921,0) rectangle (4.959,0.45);
    \fill[orange!90!black] (6.175,0) rectangle (6.289,0.45);
    \fill[orange!90!black] (7.011,0) rectangle (7.049,0.45);
    \fill[orange!90!black] (8.151,0) rectangle (8.227,0.45);
    \fill[orange!90!black] (8.797,0) rectangle (8.873,0.45);
    \fill[orange!90!black] (9.253,0) rectangle (9.405,0.45);
    \fill[orange!90!black] (10.393,0) rectangle (10.469,0.45);
    \fill[orange!90!black] (10.507,0) rectangle (10.583,0.45);
    \fill[orange!90!black] (11.609,0) rectangle (11.647,0.45);
    \fill[orange!90!black] (11.723,0) rectangle (11.761,0.45);
    \draw (0,0)--(12.578,0);
    \draw (0.000,0)--(0.000,-0.1) node[below,font=\footnotesize]{$0$};
    \draw (1.900,0)--(1.900,-0.1) node[below,font=\footnotesize]{$50$};
    \draw (3.800,0)--(3.800,-0.1) node[below,font=\footnotesize]{$100$};
    \draw (5.700,0)--(5.700,-0.1) node[below,font=\footnotesize]{$150$};
    \draw (7.600,0)--(7.600,-0.1) node[below,font=\footnotesize]{$200$};
    \draw (9.500,0)--(9.500,-0.1) node[below,font=\footnotesize]{$250$};
    \draw (11.400,0)--(11.400,-0.1) node[below,font=\footnotesize]{$300$};
    \fill[blue!70!black] (0.0,-1.15) rectangle (0.3,-0.9); \node[anchor=west,font=\footnotesize] at (0.35,-1.02) {exists};
    \fill[orange!90!black] (2.4,-1.15) rectangle (2.7,-0.9); \node[anchor=west,font=\footnotesize] at (2.75,-1.02) {open};
    \fill[black!8] (4.4,-1.15) rectangle (4.7,-0.9); \draw[black!40] (4.4,-1.15) rectangle (4.7,-0.9); \node[anchor=west,font=\footnotesize] at (4.75,-1.02) {excluded};
  \end{tikzpicture}
  \caption{Cardinalities $n\le 330$ of $2^4$-divisible sets of points over $\F_2$.}
  \label{fig_numberline_2_4}
\end{figure}

\begin{question}{Research problem}Decide whether a projective $16$-divisible binary linear code of length $130$ exists.
\end{question}

\begin{nremark}
  Due to Lemma~\ref{lemma_picture_q_2_r_3} a projective $16$-divisible binary linear code $C$ of length $165$ 
  has codewords with weight at most $96$. From the first three MacWilliams equations we conclude
  \begin{equation}
     -5120 A_{16} -3072 A_{32} -1536 A_{48}-512 A_{64} =30\cdot\left(256-2^{k-2}\right),
  \end{equation}
  so that $k\ge 10$. For $k=10$ we have $A_{16}=A_{32}=A_{48}=A_{64}=0$, i.e., $C$ is a projective two-weight code with weights $80$ and $96$. 
  However, the residual code of a codeword of weight $96$ is a $[69, 9, 32]_2$-code, which does not exist, see \cite[Theorem 2]{bouyukliev2000some}.
\end{nremark}

\begin{nlemma}
  \label{lemma_picture_q_2_r_5}
  Let $\cM$ be a $2^5$-divisible set of $n$ points in $\PG(v-1,2)$, then 
  \begin{eqnarray*}
    n&\in&\{63,64,126,127,128,189,\dots,192,252,\dots,256,315,\dots,322,378,\dots,389,\\ 
    &&441,\dots,455,503,\dots,520,566,\dots,586,628,\dots,651,691,\dots,717,\\ 
    &&753,\dots,783,815,\dots,849,877,\dots,916,938,\dots,984\}
  \end{eqnarray*} 
  or $n\ge 998$ and all cases, possibly except
  \begin{eqnarray*}
    n &\in& \{322,386,\dots,389,450,\dots,454,503,514,\dots,517,520,566,578,\dots,580,\\ 
    && 584,\dots,586,628,629,643,648,\dots,651,691,692,712,\dots,717,753,\dots,755,\\ 
    && 777,\dots,779,781,\dots,783,815,\dots,818,841,842,846,\dots,849,877,\dots,881,\\ 
    && 905,911,\dots,916,938,\dots,944,976,\dots,984,998,\dots,1007,1057,\dots,1070,\\ 
    && 1121,\dots,1133,1185,\dots,1196,1249,\dots,1259,1313,\dots,1322,\\
    && 1378,\dots,1385,1442,\dots,1448\},
  \end{eqnarray*}
  can be realized. 
\end{nlemma}
\begin{proof}
  The cases $1 \le n \le 62$, $65 \le n \le 125$, $129 \le n \le 188$, $193\le n\le 251$, and $257 \le n \le 314$ are excluded by 
  Theorem~\ref{thm_exclusion_q_r}. The cases $326 \le n \le 377$, $391 \le n \le 440$, 
  $457 \le n \le 502$, $522\le n\le 565$, $588\le n\le 627$, $653\le n\le 690$, $719\le n\le 752$, $785\le n\le 814$, $851\le n\le 876$, 
  $918\le n\le 937$, and $986 \le n \le 997$ are excluded by Lemma~\ref{lemma_hyperplane_types_arithmetic_progression} 
  using $m=6,\dots,16$, respectively. The cases $n\in\{325,390,456,521,587,652,718,784,850,917,985\}$ are excluded by 
  Corollary~\ref{cor_implication_fourth_mac_williams} with $t=5,\dots,15$, respectively, cf.~Exercise~\ref{ex_32_div}. 
  For $n\in\{323,324\}$ we have $v\ge 2$ since $n\ge 2$, so that Lemma~\ref{lemma_average} yields a hyperplane $H$ with $\cM(H)<n/2$. Since $\cM(H)\equiv n\pmod{32}$, we have 
  $\cM(H)\in\{3,35,67,99,131\}$ for $n=323$ and $\cM(H)\in\{4,36,68,100,132\}$ for $n=324$. Due to Lemma~\ref{lemma_heritable} the 
  restriction $\cM|_H$ is $2^4$-divisible, which contradicts Lemma~\ref{lemma_picture_q_2_r_4}.
  
  \textit{Base examples} for cardinalities $63$, $64$, and $321$ are given by Example~\ref{ex_simplex_code}, Example~\ref{example_affine_space}, and  
  Exercise~\ref{ex_affine_spaces_switching}, respectively. The range $1023\le n\le 1056$ is covered by Corollary~\ref{cor_spread_switching}. 
  There are two-weight codes for cardinalities $n\in\{455,780,845,975,1105,$ $1170\}$. Sums of these base examples realize e.g.\ the cardinalities 
  $385=321+64$, $642=2\cdot 321$, $776=455+321$, and $844=780+64$, as well as all cardinalities $n\ge 1449$.
\end{proof}

\begin{nlemma}
  \label{lemma_partial_picture_q_2_r_6}
  Let $\cM$ be a $2^6$-divisible set of $n$ points in $\PG(v-1,2)$, then $n$ is not contained in any of the intervals $[1,126]$, $[129,253]$, $[257,380]$, 
  $[385,507]$, $[513,634]$,\,\,\, $[641,761]$,\,\,\, $[772,888]$,\,\,\, $[902,1015]$,\,\,\, $[1032,1142]$,\,\,\, $[1161,1269]$,\,\,\, 
  $[1291,1395]$, $[1420,1522]$, $[1549,1649]$, $[1678,1776]$, 
  $[1808,1902]$, $[1937,2029]$, $[2066,2156]$, $[2196,2282]$, $[2325,2409]$, $[2455,2535]$, $[2585,2661]$, $[2714,2788]$, $[2844,2914]$, $[2974,3040]$, 
  $[3104,3166]$, $[3234,3292]$, $[3364,3418]$, $[3495,3543]$, $[3626,3668]$, $[3757,3793]$, $[3889,3917]$, and $[4023,4039]$. 
\end{nlemma}
  
\begin{nlemma}
  \label{lemma_picture_q_3_r_1}
  Let $\cM$ be a $3^1$-divisible set of $n$ points in $\PG(v-1,3)$, then $n=4$ or $n\ge 8$ and all 
  cases can be realized. 
\end{nlemma}
\begin{proof}
  The values $1\le n\le 3$ and $5\le n\le 7$ are excluded by Theorem~\ref{thm_exclusion_r_1_to_ovoid}. 
  
  \textit{Base examples} for cardinalities $4$, $9$, and $10$ are given by Example~\ref{ex_simplex_code}, Example~\ref{example_affine_space}, and Example~\ref{example_ovoid}, 
  respectively. Additionally, there exists a two-weight code of cardinality $n=11$. 
\end{proof}

\begin{nlemma}
  \label{lemma_picture_q_3_r_2}
  Let $\cM$ be a $3^2$-divisible set of $n$ points in $\PG(v-1,3)$, then 
  \begin{eqnarray*}
    n&\in&\{13,26,27,39,40,52,\dots,56,65,\dots,70,77,\dots,85,90,\dots,128\}
  \end{eqnarray*} 
  or $n\ge 129$ and all cases, possibly except 
  \begin{eqnarray*}
    n &\in& \{70,77,99, 100,101,102,113,114,115,128\},
  \end{eqnarray*}
  can be realized. 
\end{nlemma}
\begin{proof}
  The cases $1 \le n \le 12$, $14 \le n \le 25$, $28 \le n \le 38$, and $41 \le n \le 51$ are excluded by Theorem~\ref{thm_exclusion_q_r}.  
  The cases $57 \le n \le 64$, $72 \le n \le 76$, and $87 \le n \le 88$ are excluded by Lemma~\ref{lemma_hyperplane_types_arithmetic_progression} 
  using $m=5,\dots,7$, respectively. The cases $n\in\{71,86\}$ are excluded by 
  Corollary~\ref{cor_implication_fourth_mac_williams} with $t\in\{5,6\}$, respectively.  
  The case $n=89$ is excluded in Lemma~\ref{lemma_exclusion_n_89_q_3_Delta_9}. 
  
  \textit{Base examples} for cardinalities $13$ and $27$ are given by Example~\ref{ex_simplex_code} and Example~\ref{example_affine_space}. 
  There are two-weight codes for cardinalities $n\in\{55,56,84,98,126\}$ and optimal codes for cardinalities $n\in\{85,90,127,141\}$.   
\end{proof}

\begin{nlemma}
  \label{lemma_partial_picture_q_3_r_3}
  Let $\cM$ be a $3^3$-divisible set of $n$ points in $\PG(v-1,3)$, then $n$ is not contained in any of the intervals $[1,39]$, $[41,79]$, $[82,119]$, $[122,159]$, 
  $[163,199]$, $[203,239]$, $[246,279]$, $[287,319]$, $[329,359]$, $[370,399]$, $[411,439]$, $[452,478]$, $[493,518]$, $[535,558]$, $[576,597]$, $[618,637]$, $[659,676]$, 
  $[701,715]$, $[743,754]$, and $[786,793]$. 
\end{nlemma}

\begin{nlemma}
  \label{lemma_picture_q_4_r_1}
  Let $\cM$ be a $4^1$-divisible set of $n$ points in $\PG(v-1,4)$, then 
  $$n\in\{5,10,15,16,17\}$$ or $n\ge 20$ and all cases can be realized. 
\end{nlemma}
\begin{proof}
  The values $1\le n\le 4$, $6\le n\le 9$, and $11\le n\le 14$ are excluded by Theorem~\ref{thm_exclusion_r_1_to_ovoid}.
  The cases $n\in\{18,19\}$ are excluded by Lemma~\ref{lemma_hyperplane_types_arithmetic_progression} using $m=4$. 
  
  \textit{Base examples} for cardinalities $5$, $16$, and $17$ are given by Example~\ref{ex_simplex_code}, Example~\ref{example_affine_space}, and  
  Example~\ref{example_ovoid}, respectively. The cases $21\le n\le 24$ are covered by Exercise~\ref{ex_baer_switching}.
\end{proof}

\begin{nlemma}
  \label{lemma_picture_q_4_r_2}
  Let $\cM$ be a $4^2$-divisible set of $n$ points in $\PG(v-1,4)$, then 
  \begin{eqnarray*}
    n&\in&\{21,42,63,64,84,85,105,106,126,\dots,129,147,\dots,151,168,\dots,173,\\ 
    && 189,\dots,195,210,\dots,217,231,\dots,239,251,\dots,261,272,\dots,283,293,\\ 
    &&\dots,305,313,\dots,328\}
  \end{eqnarray*} 
  or $n\ge 333$ and all cases, possibly except
    \begin{eqnarray*}
    n&\in&\{129,151,172,173,194,195,215,216,217,236,\dots,239,251,258,\\ 
    && 261,272,279,282,283,293,304,305,313,314,326,327,328,\\ 
    && 333,334,335,353,\dots,356,374,\dots,377,395,\dots,398,417,418,419,\\ 
    && 438,439,440,459,460,461,481,482,503,524\}
  \end{eqnarray*}
  can be realized. 
\end{nlemma}
\begin{proof}
  The cases $1 \le n \le 20$, $22 \le n \le 41$, $44 \le n \le 62$, $66 \le n \le 83$, $87 \le n \le 104$, $109 \le n \le 125$, 
  $131\le n\le 146$, $153\le n\le 167$, $174\le n\le 188$, $196\le n\le 209$, $218\le n\le 230$, $240\le n\le 250$, $262\le n\le 271$, 
  $284\le n\le 292$, $306\le n\le 312$, and $329 \le n \le 332$ are excluded by Lemma~\ref{lemma_hyperplane_types_arithmetic_progression} 
  using $m=1,\dots,16$, respectively. 
  The cases $n\in\{65,130,152\}$ are excluded by 
  Corollary~\ref{cor_implication_fourth_mac_williams} with $t\in\{3,6,7\}$, respectively.  
  Applying Corollary~\ref{cor_nonexistence_arithmetic_progression_2} with $m\in\{1,4,5\}$ gives the non-existence for $n\in\{43,86,107,108\}$. 
   
  \textit{Base examples} for cardinalities $21$ and $64$ are given by Example~\ref{ex_simplex_code} and Example~\ref{example_affine_space}. 
  Additionally, there exist two-weight codes with $n\in\{260,325\}$ (see Table~\ref{table_two_weight_codes}),\footnote{The cardinality $325=21+304$ is realized independently of the open case $n=304$.} and Example~\ref{ex_4_2_div_303} gives 
  a $4^2$-divisible set of cardinality $303$.\footnote{The existence of a $4^2$-divisible set of cardinality $304$ is open. It is excluded in $\PG(v-1,4)$ for $v\le 6$, i.e., for projective codes of dimension at most $6$, 
  by the first three standard equations, and computer searches for $v=7$ with prescribed automorphism groups were unsuccessful.}
  For the sequence $n=85+43\cdot j$, where $0\le j\le 17$, we refer to 
  Corollary~\ref{cor_spread_switching_q}.
  Further base examples for the cardinalities $n\in\{150,193,259,280\}$ were found by computer searches with prescribed cyclic automorphism groups generated by powers of a companion matrix, see 
  Subsection~\ref{subsec_computer_searches}. The point sets consist of $14$, $15$, $13$, and $8$ orbits, respectively, and are $4^2$-divisible in $\PG(5,4)$.
  The Hermitian curve in $\PG(2,16)$ consists of $65$ points and every line contains either $1$ or $5$ of them, i.e., it is $4$-divisible over $\F_{16}$. 
  Concatenation with a two-dimensional simplex code over $\F_4$, see Exercise~\ref{exercise_concatenation_with_simplex_code}, gives a $4^2$-divisible set of 
  $325$ points over $\F_4$, the elliptic quadric $\mathrm{Q}^-(5,4)$.
  Finally, $\PG(4,4)$ is a $4^4$-divisible set of $341$ points that can be partitioned into $11$ Baer subgeometries $\PG(4,2)$, see 
  Subsection~\ref{subsec_subfield_constructions}. Each of these Baer subgeometries is $2^3$-divisible and contained in a $2^4$-divisible Baer subgeometry 
  $\PG(5,2)$ that meets the hyperplane $\PG(4,4)$ of $\PG(5,4)$ precisely in the given Baer subgeometry $\PG(4,2)$. Applying the construction 
  in (\ref{eq_generalized_switching_construction}) with $\Delta=16$, $\Delta'=8$, and $\rho=2$ to $j$ of the $11$ Baer subgeometries $\PG(4,2)$, where the corresponding Baer subgeometries $\PG(5,2)$ are placed in pairwise different additional dimensions, 
  cf.~Exercise~\ref{ex_baer_switching}, gives $4^2$-divisible sets of cardinality $341+j\cdot(63-2\cdot 31)=341+j$ for all $0\le j\le 11$.
\end{proof}

\begin{nexample}
  \label{ex_4_2_div_303}
  Let $\F_4=\{0,1,\alpha,\alpha^2\}$ with $\alpha^2=\alpha+1$ and let $\sigma\colon x\mapsto x^2$ denote the Frobenius automorphism of $\F_4$. 
  Consider the matrices 
  $$
  A_1=\begin{pmatrix}\alpha&0&1\\ \alpha^2&0&\alpha^2\\ 1&\alpha^2&\alpha\end{pmatrix}\quad\text{and}\quad  
  A_2=\begin{pmatrix}0&\alpha&1\\ 1&1&1\\ 0&\alpha^2&1\end{pmatrix}
  $$ 
  in $\mathrm{GL}(3,4)$ and let $G\le \mathrm{GL}(6,4)$ be the group generated by $\operatorname{diag}\!\left(A_1,A_1^\sigma\right)$ and 
  $\operatorname{diag}\!\left(A_2,A_2^\sigma\right)$, where $A^\sigma$ arises from $A$ by applying $\sigma$ entrywise. The group $G$ has order $36$ and 
  is isomorphic to $C_3\times A_4$. It acts faithfully on $\PG(5,4)$ with $59$ orbits. The union of the $13$ orbits, which contain the points 
  represented by the vectors  
  $$
  \begin{array}{c}
  (1,0,1,1,0,1),\;(0,0,0,0,1,1),\;(0,1,1,0,0,0),\;(0,0,1,1,0,1),\;(0,1,1,0,1,1),\\
  (0,1,1,1,1,0),\;(1,0,1,0,0,1),\;(0,0,1,0,1,1),\;(0,0,1,0,1,\alpha),\;(0,0,1,0,1,\alpha^2),\\
  (0,1,1,0,0,1),\;(0,1,\alpha,0,1,0),\;(0,1,\alpha^2,0,1,0)
  \end{array}
  $$ 
  respectively, is a set $\cM$ of $3+2\cdot 6+4\cdot 18+6\cdot 36=303$ points. The hyperplanes of $\PG(5,4)$ meet $\cM$ in $63$, $79$, or $95$ points, 
  where $288$, $1071$, and $6$ hyperplanes, respectively, have this intersection size. Thus, $\cM$ is a $4^2$-divisible set of cardinality $303$ 
  with the corresponding code having weight enumerator $1+18x^{208}+3213x^{224}+864x^{240}$. A generator matrix, whose columns are the points of 
  $\cM$, is given by the three blocks below, where we write $2$ for $\alpha$ and $3$ for $\alpha^2$. 
  
\begin{center}
{\ttfamily\tiny
00000000000000000000000000000000000000000000000000000000000000000000000000000001111111111111111111111 \\
00000000000000000000000000011111111111111111111111111111111111111111111111111110000000000000000000000 \\
00000011111111111111111111100000000000000000000011111111111112222222223333333330000000000000000000001 \\
01111100000000011112222333311111112222222333333300000001122330001122330001122330001111112222223333330 \\
10012311122233301110222033300011230001223000123300001230102031230203011230301021230122330112330112230 \\
12301112312312310232013301212323231231133123121201231230000000000301020000203011231002032010033010201
}
\end{center}

\begin{center}
{\ttfamily\tiny
11111111111111111111111111111111111111111111111111111111111111111111111111111111111111111111111111111 \\
00000000000000000000000000000000000000000000001111111111111111111111111111111111111111111111111111222 \\
11111111111111111111222222222222233333333333330000000000000111111111111111111111222222222333333333000 \\
00000111112222233333000011122233300001112223330000000112233000000111112222233333111222333111222333000 \\
00123001230012300123012300100200301230010020030000123010203112233000230001300012012023013013012023000 \\
23000011320232103213000023113212300002311321230123123000000231312123001230012300013021023012032013123
}
\end{center}

\begin{center}
{\ttfamily\tiny
11111111111111111111111111111111111111111111111111111111111111111111111111111111111111111111111111111 \\
22222222222222222222222222222222222222222222222223333333333333333333333333333333333333333333333333333 \\
00000011111111111112222222223333333333333333333330000000001111111111111222222222222222222222333333333 \\
00012301111222233330001122330001111112222223333330000001230111122223333000111111222222333333000112233 \\
12323102233113311220001223131230123330111230122230001233120223311331122123012223012333011123000131223 \\
00031201312231223131231321233121310132012123302321230002310131223122313231120121223023301331123123213
}
\end{center}
\end{nexample}
Figure~\ref{fig_weight_distributions} compares its weight distribution with that of the elliptic quadric $\mathrm{Q}^-(5,4)$.

\begin{figure}[htbp]
  \centering
  \begin{tikzpicture}
    \draw (-0.2,0)--(10.2,0);
    \fill[blue!70!black] (0.20,0) rectangle (1.10,0.015);
    \node[font=\scriptsize,above] at (0.65,0.015) {$18$};
    \node[font=\scriptsize,below] at (0.65,0) {$208$};
    \fill[blue!70!black] (1.70,0) rectangle (2.60,2.600);
    \node[font=\scriptsize,above] at (2.15,2.600) {$3213$};
    \node[font=\scriptsize,below] at (2.15,0) {$224$};
    \fill[blue!70!black] (3.20,0) rectangle (4.10,0.699);
    \node[font=\scriptsize,above] at (3.65,0.699) {$864$};
    \node[font=\scriptsize,below] at (3.65,0) {$240$};
    \node[font=\footnotesize] at (2.15,3.3) {three-weight code, $n=303$};
    \fill[blue!70!black] (5.70,0) rectangle (6.60,2.525);
    \node[font=\scriptsize,above] at (6.15,2.525) {$3120$};
    \node[font=\scriptsize,below] at (6.15,0) {$240$};
    \fill[blue!70!black] (7.20,0) rectangle (8.10,0.789);
    \node[font=\scriptsize,above] at (7.65,0.789) {$975$};
    \node[font=\scriptsize,below] at (7.65,0) {$256$};
    \node[font=\footnotesize] at (6.90,3.3) {two-weight code, $n=325$};
    \node[font=\scriptsize,anchor=east] at (-0.3,-0.55) {weight};
    \node[font=\scriptsize,anchor=east] at (-0.3,2.6) {codewords};
  \end{tikzpicture}
  \caption{Weight distributions of two $4^2$-divisible $[n,6]_4$-codes.}
  \label{fig_weight_distributions}
\end{figure}

\begin{nlemma}
  \label{lemma_picture_q_5_r_1}
  Let $\cM$ be a $5^1$-divisible set of $n$ points in $\PG(v-1,5)$, then 
  $$
    n\in\{6,12,18,24,25,26,30,31,32\}
  $$
  or $n\ge 36$ and all cases, possibly except $n=40$, can be realized. 
\end{nlemma}
\begin{proof}
  The values $1\le n\le 5$, $7\le n\le 11$, $13\le n\le 17$, and $19\le n\le 23$ are excluded by Theorem~\ref{thm_exclusion_r_1_to_ovoid}.
  The cases $27\le n\le 29$ and $34\le n\le 35$ are excluded by Lemma~\ref{lemma_hyperplane_types_arithmetic_progression} 
  using $m=5$ and $m=6$, respectively. 
  The case $n=33$ is excluded by Corollary~\ref{cor_implication_fourth_mac_williams} with $t=5$.
  
  \textit{Base examples} for cardinalities $6$, $25$, and $26$ are given by Example~\ref{ex_simplex_code}, Example~\ref{example_affine_space}, and Example~\ref{example_ovoid}, 
  respectively. Additionally, there exists a two-weight code of cardinality $n=39$ and two sporadic examples found by computer searches for cardinalities $n\in\{41,46\}$.  
\end{proof}

\begin{nlemma}
  \label{lemma_partial_picture_q_5_r_2}
  Let $\cM$ be a $5^2$-divisible set of $n$ points in $\PG(v-1,5)$, then $n$ is not contained in any of the intervals $[1,30]$, $[32,61]$, $[63,92]$, $[94,123]$, 
  $[126,154]$, $[157,185]$, $[188,216]$, $[219,247]$, $[252,278]$, $[283,309]$, $[316,340]$, $[347,371]$, $[379,402]$, $[410,433]$, $[442,464]$, $[473,495]$, $[505,526]$, 
  $[537,557]$, $[568,587]$, $[600,618]$, $[632,649]$, $[663,680]$, $[695,711]$, $[727,742]$, $[758,772]$, $[790,803]$, $[822,834]$, $[854,864]$, $[886,895]$, $[918,925]$, 
  and $[951,955]$.
 \end{nlemma}

\enlargethispage{2\baselineskip}
\begin{question}{Research problem}Resolve one of the following open cardinalities of $q$-divisible sets in $\PG(v-1,q)$. For $q=7$ the cardinality $141$ is not open, since it is attained by the $[141,4]_7$-code 
of Subsection~\ref{subsec_computer_searches}, where also codes for $n\in\{83,91,133\}$ are listed. The cardinality $125$ is attained by a two-weight code, see Table~\ref{table_two_weight_codes}.\footnote{Among all open cardinalities in this chapter (for $q\in\{3,4,7,8,9\}$ and $k\le 16$), the standard equations admit a projective two-weight code only for the parameters $(q,n,k)\in\{(3,70,6),(4,129,7),(4,172,7),(4,195,6),(4,215,7),(4,258,7),(7,75,4),(8,151,5),(9,123,4)\}$.} For $q=8$, the cardinality $120$ is attained by a code listed there, while $185=120+65$ and $250=120+65+65$ are obtained by adding disjoint ovoids, see Example~\ref{example_ovoid}.
\begin{itemize}\setlength{\itemsep}{0pt}\setlength{\parskip}{0pt}
  \item { $q=7$: $\{75,\allowbreak 92,\allowbreak 95,\allowbreak 101,\allowbreak 102,\allowbreak 103,\allowbreak 109,\allowbreak 110,\allowbreak 111,\allowbreak 117,\allowbreak 118,\allowbreak 119,\allowbreak 126,\allowbreak 127$, $134,\allowbreak 135,\allowbreak 142,\allowbreak 143,\allowbreak 151,\allowbreak 159,\allowbreak 167\}$;}
  \item { $q=8$: $\{93,\allowbreak 102,\allowbreak 111,\allowbreak 121,\allowbreak 134,\allowbreak 140,\allowbreak 143,\allowbreak 149,\allowbreak 150,\allowbreak 151,\allowbreak 152,\allowbreak 158,\allowbreak 159,\allowbreak 160,\allowbreak 161,\allowbreak 167$, $168,\allowbreak 169,\allowbreak 170,\allowbreak 176,\allowbreak 177,\allowbreak 178,\allowbreak 179,\allowbreak 186,\allowbreak 187,\allowbreak 188,\allowbreak 196,\allowbreak 197,\allowbreak 205,\allowbreak 206,\allowbreak 214,\allowbreak 215,\allowbreak 223,\allowbreak 224$, $232,\allowbreak 233,\allowbreak 241,\allowbreak 242,\allowbreak 251\}$;}
  \item { $q=9$: $\{123,\allowbreak 133,\allowbreak 143,\allowbreak 153,\allowbreak 154,\allowbreak 175,\allowbreak 179,\allowbreak 185,\allowbreak 189,\allowbreak 195,\allowbreak 196,\allowbreak 199,\allowbreak 206,\allowbreak 207,\allowbreak 208,\allowbreak 209,\allowbreak 216$, $217,\allowbreak 218,\allowbreak 219,\allowbreak 226,\allowbreak 227,\allowbreak 228,\allowbreak 229,\allowbreak 236,\allowbreak 237,\allowbreak 238,\allowbreak 239,\allowbreak 247,\allowbreak 248,\allowbreak 249,\allowbreak 257,\allowbreak 258,\allowbreak 259,\allowbreak 267,\allowbreak 268$, 
  $269,\allowbreak 277,\allowbreak 278,\allowbreak 279,\allowbreak 288,\allowbreak 289,\allowbreak 298,\allowbreak 299,\allowbreak 308,\allowbreak 309,\allowbreak 318,\allowbreak 319,\allowbreak 329,\allowbreak 339,\allowbreak 349,\allowbreak 359\}$.}
\end{itemize} 
\end{question}

\chapter{Applications}
\label{sec_applications}
In Theorem~\ref{thm_delta_divides_q_power} we have seen that in order to study $\Delta$-divisible codes it is sufficient to study $q^r$-divisible codes, 
where $r\in\Q$, $m\cdot r\in\N_0$, and $q=p^m$. The equivalence between $q^r$-divisible codes and $q^r$-divisible multisets of points in projective geometries 
have been discussed in Subsection~\ref{subsec_geometric_description}. Besides that there are several relations to other combinatorial structures, which is the 
topic of this section. In the subsequent subsections we give brief descriptions and pointers to the literature, while we devote entire sections to the relations 
to partial spreads and vector space partitions, see sections \ref{sec_partial_spreads} and \ref{sec_vector_space_partitions}. Our list is very far from being 
exhaustive. For applications in quantum computation we refer to \cite{hu2022divisible}. The unique minimal\footnote{A linear $[n,k]_q$-code $C$ is called minimal if 
there do not exist two non-zero codewords $c_1,c_2\in C$ with $\supp\!\left(c_1\right)\subsetneq\supp\!\left(c_2\right)$. A major problem in this area is the 
determination of the minimum possible length $n=m_q(k)$ of a minimal $[n,k]_q$-code. We indeed have $m_3(5)=19$.} linear $[19,5]_3$-code is $3$-divisible
\cite{kurz2023trifferent}. The relation between minimality and divisibility is more extensively studied in \cite{kurz2023divisibleminimal}.   

Figure~\ref{fig_applications_overview} gives a graphical overview.

\begin{figure}[htbp]
  \centering
  \resizebox{\textwidth}{!}{%
  \begin{tikzpicture}[
      box/.style={draw,rounded corners=2pt,align=center,font=\small,inner sep=3pt,minimum height=9mm,minimum width=24mm},
      cen/.style={draw,ellipse,thick,align=center,font=\small\bfseries,inner sep=4pt},
      arr/.style={-{Stealth[length=2mm]},black!60}]
    \node[cen] (c) at (0,0) {$q^r$-divisible\\ codes and sets};
    \node[box] (n0) at (0.000,3.400) {Vector space\\partitions\\{\scriptsize\ref{sec_vector_space_partitions}}};
    \draw[arr] (c) -- (n0);
    \node[box] (n1) at (3.468,2.751) {Partial spreads\\{\scriptsize\ref{sec_partial_spreads}}};
    \draw[arr] (c) -- (n1);
    \node[box] (n2) at (5.611,1.051) {Subspace codes,\\packings, coverings\\{\scriptsize\ref{subsec_subspace_codes}}};
    \draw[arr] (c) -- (n2);
    \node[box] (n3) at (5.611,-1.051) {Minihypers\\{\scriptsize\ref{subsec_minihypers}}};
    \draw[arr] (c) -- (n3);
    \node[box] (n4) at (3.468,-2.751) {Few-weight codes\\{\scriptsize\ref{subsec_few_weight_codes}}};
    \draw[arr] (c) -- (n4);
    \node[box] (n5) at (0.000,-3.400) {Distance-optimal\\codes\\{\scriptsize\ref{subsec_optimal_codes}}};
    \draw[arr] (c) -- (n5);
    \node[box] (n6) at (-3.468,-2.751) {$q$-analogs of group\\divisible designs\\{\scriptsize\ref{subsec_qgdd}}};
    \draw[arr] (c) -- (n6);
    \node[box] (n7) at (-5.611,-1.051) {Dual hyperovals\\{\scriptsize\ref{subsec_dual_hyperovals}}};
    \draw[arr] (c) -- (n7);
    \node[box] (n8) at (-5.611,1.051) {Nodal surfaces\\{\scriptsize\ref{subsec_nodal_surfaces}}};
    \draw[arr] (c) -- (n8);
    \node[box] (n9) at (-3.468,2.751) {Additive and\\quantum codes\\{\scriptsize\ref{subsec_additive_codes}}};
    \draw[arr] (c) -- (n9);
  \end{tikzpicture}}
  \caption{Structures in which divisibility is used, with the sections that treat them.}
  \label{fig_applications_overview}
\end{figure}
Table~\ref{table_applications_overview} summarizes how $q^r$-divisibility enters in the structures treated in this chapter and in the 
other chapters.
\begin{table}[htp]
  \begin{center}
    \renewcommand{\arraystretch}{1.15}
    \begin{tabular}{p{3.4cm}p{8.3cm}r}
      \hline
      Structure & How $q^r$-divisibility is used & Section\\
      \hline
      Vector space partitions & The points of a multiset of subspaces of dimension at least $k$ and their complement are $q^{k-1}$-divisible, so that 
      non-existence of divisible multisets of a given cardinality excludes the partition & \ref{subsec_introductory_application}, \ref{sec_vector_space_partitions}\\
      Partial spreads & The holes of a partial $k$-spread form a $q^{k-1}$-divisible set, so that the possible numbers of holes are restricted & \ref{sec_partial_spreads}\\
      Subspace codes & Sharpened rounding in the recursive Johnson bound for $A_q(v,d;k)$ & \ref{subsec_subspace_codes}\\
      Subspace packings and coverings & Sharpened rounding in the corresponding recursive bounds & \ref{subsec_subspace_packings_coverings}\\
      Minihypers & Divisibility of $\left(x[t]_q,x[t-1]_q\right)$-minihypers, and the holes of partial spreads as minihypers & \ref{subsec_minihypers}\\
      Extendability results & Classification of the small divisible sets that can occur as hole sets & \ref{sec_classification_results}, \ref{sec_extendability_results}\\
      Few-weight codes & Projective two-weight codes are divisible, and many few-weight codes have a large divisibility & \ref{subsec_two_weight_codes}, \ref{subsec_few_weight_codes}\\
      Distance-optimal codes & Griesmer codes over prime fields are $p^e$-divisible if $p^e\mid d$ \cite{Ward-1998-JCTSA}, and prescribing divisibility reduces search spaces & \ref{subsec_optimal_codes}\\
      $q$-analogs of group divisible designs & Necessary condition: $q^{k-g}$ divides $\lambda$ & \ref{subsec_qgdd}\\
      $k$-dimensional dual hyperovals & The associated multiset of points is $q^{k-1}$-divisible & \ref{subsec_dual_hyperovals}\\
      Codes of nodal surfaces & The associated binary codes are doubly-even or triply even & \ref{subsec_nodal_surfaces}\\
      Additive codes & The points on the subspaces of an additive code form a $q^{h-1}$-divisible multiset, which yields coding bounds & \ref{subsec_additive_codes}\\
      Quantum codes & Doubly-even binary codes are self-orthogonal, so that they yield quantum codes via the CSS construction & \ref{subsec_quantum_codes}\\
      \hline
    \end{tabular}
    \caption{Overview of applications of divisible codes. The numbers in the last column refer to the corresponding sections.}
    \label{table_applications_overview}
  \end{center}
\end{table}

\section{Subspace codes}
\label{subsec_subspace_codes}

For two subspaces $U$ and $U'$ of $\PG(v-1,q)$ the \emph{subspace distance} is given by $d_S(U,U') = \dim (U+U') - \dim (U\cap U')$. 
A set $\mathcal{C}$ of subspaces in $\PG(v-1,q)$, called codewords, with minimum subspace distance $d$ is called a \emph{subspace code}. 
Its maximal possible cardinality is denoted by $A_q(v,d)$, see e.g.\ 
\cite{honold2016constructions}. If all codewords have the same dimension, 
say $k$, then we speak of a \emph{constant dimension code} and denote the corresponding maximum possible cardinality by $A_q(v,d;k)$, see 
e.g.~\cite{etzionsurvey}. For known bounds, we refer to \url{http://subspacecodes.uni-bayreuth.de} \cite{TableSubspacecodes} containing 
also the generalization to subspace codes of mixed dimension. For $2k\le v$ the cardinality $A_q(v,2k;k)$ is the maximum 
size of a partial $k$-spread, see Section~\ref{sec_partial_spreads}. For $d<2k$ the recursive Johnson bound 
$$A_q(v,d;k)\le \left\lfloor \qbin{v}{1}{q} \cdot A_q(v-1,d;k-1)/\qbin{k}{1}{q}\right\rfloor,$$  
see \cite{xia2009johnson}, recurs on this situation. The involved rounding can be slightly sharpened using the non-existence of $q^r$-divisible 
multisets of a certain cardinality, see \cite[Lemma 13]{kiermaier2020lengths} and Lemma~\ref{lem:pack_cover}:
\begin{equation}
  A_q(v,d;k)\le \llfloor A_q(v-1,d;k-1) \cdot [v]_q/[k]_q\rrfloor_{q^{k-1}}\text{.} 
\end{equation}
For $d<2k$ this gives the tightest known upper bound for $A_q(v,d;k)$ 
except $A_2(6,4;3)=77<81$ \cite{honold2015optimal} and $A_2(8,6;4)=257< 289$ \cite{heinlein2019classifying}. For general subspace codes the 
underlying idea of the Johnson bound in combination with $q^r$-divisible multisets has been generalized in \cite{honold2019johnson}. Constructions and bounds for subspace codes are treated in more detail in the companion chapter \cite{kurz2021constructions}. 

\section{Subspace packings and coverings}
\label{subsec_subspace_packings_coverings}
A constant-dimension code consisting of $k$-dimensional codewords in $\PG(v-1,q)$ has minimum subspace distance at least $d$ iff each 
$(k-\tfrac{d}{2}+1)$-dimensional subspace is contained in at most one codeword. If we relax the condition a bit and require that 
for a multiset $\cU$ of $k$-spaces each $(k-\tfrac{d}{2}+1)$-dimensional subspace is contained in at most $\lambda$ codewords, then we have the definition of a 
\emph{subspace packing}. Of course, similar to constant-dimension codes, $q^r$-divisible multisets can be used to obtain 
upper bounds on the cardinality of a subspace packing, see \cite{etzion2020subspace}. Indeed, 
\cite[Lemma 13]{kiermaier2020lengths} and Lemma~\ref{lem:pack_cover} cover that case, i.e.,  
\begin{equation}
  A^\lambda_q(v,d;k)\le \llfloor A_q^\lambda(v-1,d;k-1) [v]_q/[k]_q\rrfloor_{q^{k-1}}
\end{equation}
for $k\ge 2$, where $A_q^\lambda(v,d;k)$ denotes the maximum cardinality of a multiset $\cU$ of $k$-spaces in $\PG(v-1,q)$ such that each 
$(k-\tfrac{d}{2}+1)$-dimensional subspace is covered at most $\lambda$ times.

If we replace {\lq\lq}contained in at most $\lambda$ codewords{\rq\rq} by {\lq\lq}contained in at least $\lambda$ codewords{\rq\rq} 
we obtain so-called \emph{subspace coverings}. For the special case of $\lambda=1$ we refer e.g.\ to \cite{etzion2014covering,etzion2011q}. 
Again, \cite[Lemma 13]{kiermaier2020lengths} and Lemma~\ref{lem:pack_cover} cover this situation and relate it to $q^r$-divisible multisets, i.e., 
\begin{equation}
  B^\lambda_q(v,d;k)\ge \llceil  B^\lambda_q(v-1,d;k-1) [v]_q/[k]_q\rrceil_{q^{k-1}}
\end{equation}
for $k\ge 2$, where $B_q^\lambda(v,d;k)$ denotes the minimum cardinality of a multiset $\cU$ of $k$-spaces in $\PG(v-1,q)$ such that each 
$(k-\tfrac{d}{2}+1)$-dimensional subspace is covered at least $\lambda$ times.

\section{Orthogonal arrays}
\label{subsec_orthogonal_arrays}
A $t-(v,k,\lambda)$ orthogonal array, where $t\le k$, is a $\lambda v^t \times k$ array whose entries are chosen from a set 
$X$ with $v$ points such that in every subset of $t$ columns of the array, every $t$-tuple of points of $X$ appears in exactly $\lambda$ rows. 
Here, $t$ is called the strength of the orthogonal array. For a survey see e.g.\ \cite{hedayat2012orthogonal}. 
A library of orthogonal arrays can be found at \url{http://neilsloane.com/oadir/}. A variant of the linear programming method for orthogonal 
arrays with mixed levels was presented in \cite{sloane1996linear}, see also \cite{bierbrauer1997note}. Orthogonal arrays can be regarded as natural 
generalizations of orthogonal Latin squares\cite{keedwell2015latin}, cf.~\cite{bose1952orthogonal}. Linear orthogonal arrays are ultimately linked to 
linear codes, see e.g.\ \cite[Section 4.3]{hedayat2012orthogonal}, via:
\begin{ntheorem}
  Suppose that $C$ is an $[n,k]_q$-code. Then $\dH(C)\ge d$ iff $C^\perp$ is a linear $\operatorname{OA}_\lambda(d-1,n,q)$, where 
  $\lambda=q^{n-k-d+1}$. 
\end{ntheorem}

\section{$(s,r,\mu)$-nets}
\label{subsec_nets}

\begin{ndefinition}(\cite[Definition 2]{nets_and_spreads})\\
  Let $J$ be an incidence structure. Define $B\parallel G$ for blocks $B$, $G$ of $J$ to mean that either $B = G$ or $[B,G]=0$. Then $J$ is called an $(s,r,\mu)$-net provided: 
  \begin{enumerate}
    \item[(i)] $||$ is a parallelism; 
    \item[(ii)] $G\not\parallel H$ implies $[G,H] = \mu$; 
    \item[(iii)] there is at least one point, some parallel class has $s\ge 2$ blocks, and there are $r\ge 3$ parallel classes. 
  \end{enumerate}
\end{ndefinition} 
We note that the existence of an $(s,r,\mu)$-net is equivalent to the existence of an orthogonal array of strength two, see Subsection~\ref{subsec_orthogonal_arrays}. 
From partial spreads $(s,r,\mu)$-nets can be constructed, see \cite{nets_and_spreads}. Additionally, there is a connection between $3$-nets and Latin squares, see 
e.g.~\cite[Section 8.1]{keedwell2015latin}. 

Nets can be seen as a relaxation of a finite projective plane, see e.g.~\cite{ostrom1968vector}. For the famous existence question of finite projective planes of 
small order we refer to \cite{lam1991search,perrott2016existence}.

\section{Minihypers}
\label{subsec_minihypers}

An \emph{$(f,m;v,q)$-minihyper} is a pair $(F,w)$, where $F$ is a subset of the point set of $\PG(v-1,q)$ and $w$ is a weight function 
$w\colon \PG(v-1,q)\to\mathbb{N}$, $x\mapsto w(x)$, satisfying
\begin{enumerate}
  \item[(1)] $w(x)>0$ $\Longrightarrow$ $x\in F$, 
  \item[(2)] $\sum_{x\in F} w(x)=f$, and
  \item[(3)] $\min\{\sum\limits_{x\in H} w(x) \mid H\in\mathcal{H}\}=m$, where $\mathcal{H}$ is the set of hyperplanes of $\PG(v-1,q)$.
\end{enumerate}
We also say that a multiset $\cM$ of points $\cM$ in $\PG(v-1,q)$ is an \emph{$(f,m)$-minihyper} if $\#\cM=f$, $\cM(H)\ge m$ for all $H\in\cH$, and $\min\{\cM(H)\,:\,H\in\cH\}=m$. 
For a positive integer $e$ (and field size $q$) write $f=\sum_{i=1}^e f_i[i]_q$ where $f_e=\left\lfloor f/[e]_q\right\rfloor$ and $f_j=\left\lfloor \left(f-\sum_{i=j+1}^e f_i[i]_q\right)/[j]_q\right\rfloor$
for $j=e-1,\dots,1$. By $\left[f_e,\dots,f_1\right]$ we denote the \emph{$[e]_q$-expansion} of $f$. The expansion has the properties $f_e\ge 0$ and $0\le f_j\le q$ for $1\le j\le e-1$. 
Moreover, $f_j=q$ for some $1\le j<e$ implies $f_i=0$ for all $1\le i<j$. With this, the mapping $f\mapsto \left[f_e,\dots,f_1\right]$ is a bijection from $\N_0$ onto the set of $e$-element lists 
with the mentioned properties.  

\medskip

\begin{trailer}{The Hamada bound}
\vspace*{-6mm}
\begin{ntheorem} 
  \label{theorem_hamada_bound}
  Let $\cM$ be an $(f,m)$-minihyper in $\PG(v,q)$ and the $[v-1]_q$-expansion of $m$, defined as for $f$ above, be $\left[m_{v-1},\dots,m_1\right]$. Then,
  \begin{equation}
    f\ge \left[m_{v-1},\dots,m_1,0\right]=qm+\sum_{i=1}^{v-1} m_i.
  \end{equation}   
\end{ntheorem}

For a proof we refer e.g.\ to \cite[Theorem 4.1]{ward2007arcs}. We remark that there exist several variants of this result where e.g.\ $\gamma_0(\cM)=1$ or 
$0\le f_i\le q-1$ is assumed. The latter assumption also allows to drop the parameter $e$ from the expansion. For a survey on minihypers we refer to \cite{hamada1993characterization}.
\end{trailer}

\medskip

A distinguished class of minihypers is given by $\left(x[t]_q,x[t-1]_q\right)$-minihypers in $\PG(t,q)$, see the discussion after Exercise~\ref{exercise_dual_multiset_coefficients}. 
Here we have rather strong divisibility properties:
\begin{ntheorem}\cite[Theorem 3.1]{LandjevVandendriessche2012} 
  Let $\cM$ be an $\left(x[t]_q,x[t-1]_q\right)$-minihyper in $\PG(t,q)$, where $x\le q-p^g$ for some non-negative integer $g$ and the characteristic $p$ of $\F_q$. 
  Then, $\cM$ is $p^{g+1}q^{t-2}$-divisible.  
\end{ntheorem}  

For the other direction we consider $\cM=\chi_E+3\cdot\chi_L+9\cdot\chi_P$ in $\PG(v-1,3)$ for a point $P$, a line $L$, and a plane $E$, where $v\ge 3$ and 
$P\le L\le E$. We can easily check that $\cM$ is $9$-divisible and has $[3]_3$-expansion $[2,2,0]$. However, $\cM$ is not a $(34,10)$-minihyper, as one might 
hope in view of Theorem~\ref{theorem_hamada_bound}, but only a $(34,7)$-minihyper. 
\begin{nexercise}
  Let $\cM$ be a $q^t$-divisible multiset of points of cardinality $x[t+1]_q$ in $\PG(v-1,q)$ with $1\le t< v$. Show that 
  $\cM$ is an $\left(x[t+1]_q,x[t]_q\right)$-minihyper if $x\le q-1$. Assuming $\gamma_0(\cM)=1$, show that 
  $\cM$ is an $\left(x[t+1]_q,x[t]_q\right)$-minihyper if $x\le t(q-1)$.\\ 
  (\textit{Hint}: Use Theorem~\ref{thm_exclusion_q_r}.)    
\end{nexercise}
\begin{nexample}
  \label{example_minihyper_q_times_subspace}
  Let $K$ be a $(t+1)$-space in $\PG(v-1,q)$ with $v\ge t+1$, where $t\ge 2$. For a point $P\le K$ the set of points $\chi_K-\chi_P$ is 
  an $\left(q[t]_q,q[t-1]_q\right)$-minihyper that is not $q$-divisible. If $S\le K$ is an $s$-space with $2\le s\le t$ and $S'\le S$ is an $(s-1)$-space, 
  then $\cM:=\chi_K-\chi_S+q\cdot \chi_{S'}$ is a $\left(q[t]_q,q[t-1]_q\right)$-minihyper that is not $q^s$-divisible. 
\end{nexample}
\begin{nexercise}
  Show that each $\left(2[t]_2,2[t-1]_2\right)$-minihyper in $\PG(v-1,2)$ is either a union of two $t$-spaces or covered by the construction in Example~\ref{example_minihyper_q_times_subspace}.
\end{nexercise}

\medskip

  
Minihypers have e.g.\ been used to prove extendability results for partial spreads, see e.g.~\cite{ferret2003results,govaerts2002particular,govaerts2003particular} 
and Section~\ref{sec_extendability_results}. If $\cP$ is the set of holes of a partial $k$-spread, then the partial spread is extendible iff $\cP$ contains all points of 
a $k$-dimensional subspace. As an example, in \cite{honold2019classification} the possible hole configurations of partial $3$-spreads in $\PG(6,2)$ 
of cardinality $16$, i.e., with $15$ holes, were classified. In four cases the partial spread is extensible and in one case it is not, cf.\ Example~\ref{example_cone_constructions}.

A close relation between divisible sets and minihypers can be found in \cite{landjev2019divisible,landjev2016extendability}. Arcs whose hyperplane multiplicities satisfy congruence conditions modulo $q$, 
so-called $(t \bmod q)$-arcs, are studied geometrically in \cite{kurz2025geometry}. To this end an $(n, w)$-arc in $\PG(k-1,q)$  
is called $t$-quasidivisible iff every hyperplane has a multiplicity congruent to $n+i \pmod q$, where $i\in \{0,1,...,t\}$. With this, every $t$-quasidivisible 
arc associated with a linear code meeting the Griesmer bound, and satisfying an additional numerical condition, is $t$ times extendable. For more 
papers using minihypers to study codes meeting the Griesmer bound see e.g.~\cite{hamada1993characterization,hill2007geometric}.  
 
\begin{question}{Research problem}Can some results obtained using minihypers be improved by using the properties of divisible codes?\end{question} 

The use of classification results for projective $q^{k-1}$-divisible codes, see e.g.\ Section~\ref{sec_classification_results}, for extendability results 
for partial $k$-spreads can be generalized to extendability results for constant-dimension codes, see \cite{nakic2016extendability} and 
Section~\ref{sec_extendability_results}. 

\section{Few-weight codes}
\label{subsec_few_weight_codes}

A linear $[n,k]_q$ code $C$ is called an \emph{$s$-weight code} if the non-zero codewords of $C$ attain (at most) $s$ possible weights. For $s=1$ 
repetitions of simplex codes give the only examples. The case $s=2$ is discussed in Subsection~\ref{subsec_two_weight_codes}. For 
projective two-weight codes there is a strong relation to $q^r$-divisible codes, see Lemma~\ref{lemma_delsarte}. While we do not have  
such a strong relation for $s\ge 3$, it turns out that many examples of codes with relatively few weights are $q^r$-divisible, where $r$ is relatively 
large. For some literature on three-weight codes, see e.g.\ \cite{calderbank1984three,ding2015class,heng2015several,kiermaier2022three,wang2021binary,yang2021weight,zhou2014class}. 
Few-weight codes, i.e., $s$-weight codes with $s\ge 4$ but $s$ still being relatively small, are e.g.\ treated in \cite{he2023several,li2021binary,wu2020linear}.

\section[$k$-dimensional dual hyperovals]{$\mathbf{k}$-dimensional dual hyperovals}
\label{subsec_dual_hyperovals}
A set $\cK$ of $k$-spaces in $\PG(v-1,q)$ with $\#\cK\ge 3$ such that 
\begin{itemize}
\item[--] $\dim(X\cap Y)=1$ for any distinct $X,Y\in\cK$;
\item[--] $\dim(X\cap Y\cap Z)=0$ for any distinct $X,Y,Z\in\cK$; and 
\item[--] the points in the elements of $\cK$ generate $\PG(v-1,q)$
\end{itemize}
is called \emph{$k$-dimensional dual arc} (in $\PG(v-1,q)$). The associated multiset $\cM$ of points 
is $q^{k-1}$-divisible. Each point of an arbitrary element $X\in \cK$ is contained in at most one further 
element of $\cK$ so that $\#\cK\le [k]_q+1$, see e.g.\ \cite[Lemma 2.2]{yoshiara2008dimensional}. In the case 
of equality $\cK$ is called \emph{($k$-dimensional) dual hyperoval}. Here we have $\cM(P)\in\{0,2\}$ for all 
points $P\in\cP$, so that $\tfrac{1}{2}\cM$ is a $q^{k-1}$-divisible set of $\tfrac{1}{2}[k]_q([k]_q+1)$ points 
if $q$ is odd, which we assume in the following. Note that two distinct elements of $\cK$ span a $2k$-dimensional subspace, so that $v\ge 2k$. For $v=2k$ it was shown in 
\cite[Proposition 2.11]{del2000d} that each hyperplane contains either $0$ or $2$ elements from $\cK$, i.e., 
$\tfrac{1}{2}\cM$ corresponds to a two-weight code, cf.~Example SU2 in \cite{calderbank1986geometry}. If $\tfrac{1}{2}\cM$ 
can be the set of double points of a $k$-dimensional dual hyperoval seems to be a rather hard question. A few 
more necessary conditions are known, see e.g.\ \cite{del2000d,yoshiara2008dimensional}.

\section{$q$-analogs of group divisible designs}
\label{subsec_qgdd}
Let $K$ and $G$ be sets of positive integers. A \emph{$q$-analog of a group divisible design} of index $\lambda$  and
order $v$ is a triple $(\cV, \cG, \cB)$, where
\begin{itemize}
\item[--] $\cV$ is a vector space over $\mathbb{F}_q$ of dimension $v$,
\item[--] $\cG$ is a vector space partition whose dimensions lie in $G$, and
\item[--] $\cB$ is a family  of subspaces (blocks) of $\cV$,
\end{itemize}
that satisfies
\begin{enumerate}
\item $\#  \mathcal{G} > 1$,
\item if $B\in \cB$ then $\dim B \in K$,
\item every $2$-dimensional subspace of $\cV$ occurs in exactly  $\lambda$ blocks or one group, but not both.
\end{enumerate}

This notion was introduced in \cite{buratti2019qgdd} and generalizes the classical definition of a group divisible design in the 
set case, see e.g.~\cite{brouwer1977group}. If $K=\{k\}$ and $G=\{g\}$, then we speak of a $(v, g, k,\lambda)_q$-GDD. All necessary 
existence conditions of the set case can be easily transferred to the 
$q$-analog case. Moreover, there is an additional necessary existence condition whose proof is based on $q^r$-divisible multisets:
\begin{nlemma}{(\cite[Lemma 5]{buratti2019qgdd})} 
  Let $(\cV, \cG, \cB)$ be a $(v, g, k,\lambda)_q$-GDD and
  $2\leq g \leq k$, then $q^{k-g}$ divides $\lambda$.
\end{nlemma}
Note that in the set case the divisibility by $1^{k-g}$ is trivially satisfied.

\section{Codes of nodal surfaces}
\label{subsec_nodal_surfaces}

In algebraic geometry, a nodal surface is a surface in a (usually complex) projective space whose only singularities are nodes, i.e., a very simple type 
of a singularity. A major problem about them is to find the maximum number of nodes of a nodal surface of given degree. In \cite{beauville1979nombre} to each 
such nodal surface is assigned a linear code with a certain divisibility and the problem was solved for quintic surfaces. Using the link to linear codes it was shown  
in \cite{jaffe1997sextic} that a sextic surface 
can have at most $65$ nodes. In \cite[Theorem 5.5.9]{PhdPettersen} a unique irreducible $3$-parameter family of $65$-nodal sextics, containing the famous Barth sextic 
\cite{barth1996two}, was determined. The uniqueness of the associated $8$-divisible binary linear code was established in \cite{kurz2020classification}. In general, 
the binary codes associated to nodal surfaces are either doubly-even or triply even, depending on whether the degree of the surface is odd or even, see \cite{catanese1981babbage}. For a broader treatment of nodal surfaces and their relation to coding theory we refer to the book \cite{catanese2026varieties}.

For another type of singularities, so-called cusps, we end up with $3$-divisible codes over $\F_3$, see e.g.\ \cite{barth2007cusps}.

\section{Distance-optimal codes}
\label{subsec_optimal_codes}
Given a field size $q$ the possible parameters $n$, $k$, and $d$ of an $[n,k,d]_q$-code allow different optimizations, i.e., we can fix two parameters and optimize 
the third. The codes attaining the maximum possible value for the minimum distance $d$, given length $n$ and dimension $k$, are called \emph{distance-optimal} codes.  
Among the distance-optimal codes, there are quite some $q^r$-divisible codes with a relatively large value of $r$. E.g.\ ten of the eleven {\lq\lq}base examples{\rq\rq} 
used in the proof of Lemma~\ref{lemma_picture_q_3_r_2} are distance-optimal and the remaining one is the best known $[90,8,54]_3$-code. This phenomenon can partially be explained by our search technique screening 
the lists of available optimal linear codes and checking them for divisibility. Our sources were \url{http://www.codetables.de} maintained by Markus Grassl, 
\url{http://mint.sbg.ac.at} maintained at the university of Salzburg, and the database of \textit{best known linear codes} implemented in \texttt{Magma}. 
In Table~\ref{table_optimal_codes} we list the parameters and references of those cases that appear as {\lq\lq}base examples{\rq\rq} in the proofs of 
Section~\ref{sec_lengths_projective_q_r}, but are neither two-weight codes, see Subsection~\ref{subsec_two_weight_codes}, nor have an explicit construction in Section~\ref{sec_constructions_projective}. 
\begin{table}[htp]
  \begin{center}
    \begin{tabular}{rrrrcl}
      \hline
        $n$ & $k$ & $d$ & $\Delta$ & $q$ & reference \\
      \hline
      199 & 11 & 96 & 16 & 2 & twisted BCH code extended with a parity check bit \cite{edel1997twisted}\\  
       85 &  7 & 54 & 9 & 3 & \cite{bierbrauer1997new} \\
       90 &  8 & 54 & 9 & 3 & \cite{maruta2008constructing} \\ 
      127 &  7 & 81 & 9 & 3 & \cite{grassl2005new}\\
      141 &  7 & 90 & 9 & 3 & \cite{grassl2005new}\\
      \hline
    \end{tabular}
    \caption{Parameters of a few selected optimal or best known codes.}
    \label{table_optimal_codes}  
  \end{center}
\end{table}  
We remark that there are way more possible lengths of distance-optimal projective $q^r$-divisible linear codes. However, in many cases the corresponding lengths 
can be obtained as the sum of smaller base examples. Note that it is unknown whether a $[90,8,55]_3$-code exists. 

For some cases it can be shown that distance-optimal codes have to admit a certain divisibility. To this end we have to mention the \emph{Griesmer bound}, see 
\cite{griesmer1960bound}, stating that each $[n,k,\ge d]_q$-code $C$ satisfies
\begin{equation}
  \label{ie_griesmer}
    n \,\,\ge\,\, \sum_{i=0}^{k-1} \left\lceil\frac{d}{q^i}\right\rceil.
\end{equation} 
Codes attaining Inequality~(\ref{ie_griesmer}) with equality are called \emph{Griesmer codes} or codes meeting the Griesmer bound. Those codes have a 
high divisibility, at least if the field size is a prime:
\begin{ntheorem}{(\cite[Theorem 1]{Ward-1998-JCTSA})}
  Let $C$ be an $[n,k,d]_p$-code, where $p$ is a prime, meeting the Griesmer bound. If $p^e$ divides $d$, where $e\in\N_0$, then $C$ is $p^e$-divisible. 
\end{ntheorem}   
As an example, the $[10,4,6]_3$-code of Example~\ref{ex_ovoid_pg33} meets the Griesmer bound, since $6+2+1+1=10$. As $3$ divides $d=6$, 
it is $3$-divisible, which is consistent with its weight enumerator.
For field sizes $q=p^f$ that are not prime the situation is more subtle, e.g., the $[6,3,4]_4$ hexacode meets the Griesmer bound, $4\mid d$, 
and has non-zero weights $4$ and $6$, so that it is not $4$-divisible. Ward conjectured that each Griesmer $[n,k,d]_q$-code is $p^{\max\{0,e-f+1\}}$-divisible, where $e$ 
is the $p$-adic valuation of $d$, which coincides with the stated theorem for $f=1$. Partial results are contained in \cite{deng2026divisibility}, 
and proofs of the conjecture were recently posted on arXiv in \cite{deng2026proof,zhang2026hyperplane}. The known divisibility results for Griesmer codes also carry over 
to the related anti-Griesmer codes, see \cite{deng2026generalized}. Further constructions of Griesmer and optimal linear codes via affine Solomon--Stiffler codes 
are given in \cite{chen2025griesmer}.
Similar results also hold for distance-optimal non-Griesmer codes, see e.g.\ \cite{ball2001q}. An interesting example is given by the 
$[46,9,20]_2$-code found in \cite{kurz202146}. It is optimal, unique, and does not have any non-trivial automorphism. So, heuristic 
searches prescribing automorphisms had to be unsuccessful for this example. Like prescribing automorphisms, prescribing $\Delta$-divisibility 
might help to reduce search spaces to a more manageable size while still permitting solutions. 
Divisibility is also useful for additive codes, see Section~\ref{subsec_additive_codes}.
\begin{question}{Research problem}Try to improve the best known lower bounds for distance-optimal codes for a few parameters 
by assuming $q^r$-divisibility for the largest possible $r$ so that the minimum distance $d$ is divisible by $q^r$. 
\end{question}

\section{Quantum codes}
\label{subsec_quantum_codes}
Quantum error-correcting codes can be obtained from classical codes by the CSS construction of Calderbank and Shor \cite{calderbank1996good} and Steane \cite{steane1996multiple}. 
If $C$ is a binary $[n,k]_2$-code with $C\subseteq C^\perp$, then there is a quantum code with parameters $[[n,n-2k,d]]$, where $d$ is at least the minimum distance of $C^\perp$. 
Since every doubly-even binary linear code is self-orthogonal, a $4$-divisible code $C$ suffices, so that divisibility is a convenient way to ensure the required inclusion. 
\begin{nexample}
  For each $k\ge 3$ the binary simplex code of Example~\ref{ex_simplex_code} is $2^{k-1}$-divisible and thus doubly-even, so that it is self-orthogonal. Its dual code is the 
  binary Hamming code with minimum distance $3$, which yields a quantum code with parameters $\left[\left[2^k-1,2^k-1-2k,3\right]\right]$. For $k=3$ we obtain the 
  Steane code $[[7,1,3]]$.
\end{nexample}
More generally, stabilizer codes correspond to additive codes over $\F_4$ that are self-orthogonal with respect to a trace inner product, see \cite{calderbank1998quantum}. Over larger fields the criteria mentioned in Chapter~\ref{sec_preliminiaries} again let divisibility guarantee the required self-orthogonality, 
and \cite{heng2026euclidean} uses them to obtain quantum codes with good parameters.

\section{Additive codes}
\label{subsec_additive_codes}
Additive codes are a natural generalization of linear codes, and divisibility turns out to be a useful tool for them as well. 
An \emph{additive code} $C$ over $\F_{q^h}$ of length $n$ is an $\F_q$-linear subspace $C\subseteq \F_{q^h}^n$. If $C$ has $\F_q$-dimension $k$, then 
we speak of an additive $[n,k/h]_{q^h}$-code, since $C$ has $q^k=\left(q^h\right)^{k/h}$ codewords, and $k/h$ need not be an integer. Weights and the minimum distance are 
defined as for linear codes. For $h=1$ we obtain the usual linear codes, and in this section we assume $h\ge 2$. 
A code is called \emph{$\Delta$-divisible} if the weights of all its codewords are divisible by $\Delta$. 
Additive codes over $\F_4$ are of particular interest since the stabilizer construction of quantum codes mentioned in the previous section leads to additive codes over $\F_4$. 

For the geometric description we proceed as in Subsection~\ref{subsec_geometric_description}. Each coordinate $i\in\{1,\dots,n\}$ of an additive code $C$ gives 
an $\F_q$-linear map $\pi_i\colon C\to\F_{q^h}\cong \F_q^h$, so that we can consider the subspace $S_i$ of the dual space of $C\cong\F_q^k$ that is spanned by 
the $h$ coordinate functionals of $\pi_i$. Assume that all $S_i$ have dimension exactly $h$ and that the $S_i$ span the dual space, so that $k$ is also the dimension of $C$. Then we obtain a multiset $\cS=\{S_1,\dots,S_n\}$ of $n$ $h$-dimensional 
subspaces in $\PG(k-1,q)$, which spans $\PG(k-1,q)$. We call these subspaces \emph{$(h-1)$-spaces}, i.e., lines for $h=2$, and in the following we refer to the 
elements of $\cS$ as the generators of $C$. A codeword $c\in C$ corresponds to a hyperplane $H$ of $\PG(k-1,q)$, 
up to scalar multiples, and $c_i=0$ if and only if $S_i\subseteq H$. Consequently,
\begin{equation}
  \wt(c)=n-\#\left\{1\le i\le n\,:\, S_i\subseteq H\right\},
\end{equation}
so that the minimum distance of $C$ is $n-s$, where $s$ is the maximum number of elements of $\cS$ that are contained in a hyperplane.
For $q=2$ and $h=2$ the code $C$ is an additive quaternary code, and $\cS$ is a spanning multiset of $n$ lines in $\PG(k-1,2)$ such that each hyperplane contains at most $s$ of 
them.

If we replace each element of $\cS$ by the set of its points, we obtain a multiset $\cM$ of points in $\PG(k-1,q)$ with $\#\cM=n\cdot[h]_q$, i.e., a linear $\left[n[h]_q,k\right]_q$-code $C'$.  
\begin{nlemma}
  \label{lemma_additive_divisible}
  Let $C$ be an additive $[n,k/h]_{q^h}$-code with minimum distance $n-s$ whose generators are $h$-spaces. Then $C'$ is a $q^{h-1}$-divisible $\left[n[h]_q,k,q^{h-1}(n-s)\right]_q$-code 
  with maximum weight at most $q^{h-1}\cdot n$. 
\end{nlemma}
\begin{proof}
  Let $H$ be a hyperplane of $\PG(k-1,q)$ and $i$ be the number of generators contained in $H$. Each of the other $n-i$ generators intersects $H$ in a $(h-2)$-space, i.e., 
  in $[h-1]_q$ points, so that exactly $[h]_q-[h-1]_q=q^{h-1}$ of its points are not contained in $H$. Thus, the codeword of $C'$ that corresponds to $H$ has weight $q^{h-1}(n-i)$. 
  As $0\le i\le s$, all weights are divisible by $q^{h-1}$, the minimum weight is $q^{h-1}(n-s)$, and the maximum weight is at most $q^{h-1}n$. 
\end{proof}
For $q=2$ and $h=2$ the lemma was used in this form in \cite[Lemma 2.3]{kurz2025optimal}, see also \cite{kurz2020additiveGriesmer}. It transfers 
every bound for linear codes to additive codes, and divisibility is the additional piece of information that is not captured by the parameters $[n[h]_q,k,q^{h-1}(n-s)]_q$ alone.
\begin{nexample}
  Let $\cS$ be a line spread of $\PG(3,2)$, i.e., a set of $5$ lines that partition the $15$ points. Every plane of $\PG(3,2)$ contains exactly one of the lines, and meets each of the 
  other four lines in a point. Hence, $\cS$ yields an additive $[5,2]_4$-code, i.e., with $16$ codewords, and minimum distance $4$. The code $C'$ is the binary simplex code $[15,4,8]_2$ of 
  Example~\ref{ex_simplex_code}, which is $8$-divisible, and in particular $2$-divisible.
\end{nexample}

Lemma~\ref{lemma_additive_divisible} is the basis for the following upper bounds on the maximum length $n$ of an additive code over $\F_4$ for given dimension and minimum distance. 
Let $n_{k}(s)$ denote the maximum $n$ such that a spanning multiset of $n$ lines in $\PG(2k-1,2)$ exists with at most $s$ lines in each hyperplane, i.e., the maximum length of an additive 
$[n,k]_4$-code with minimum distance $n-s$. The \emph{Griesmer bound} for $C'$ gives the largest $n$ such that 
$g\!\left(2k,2(n-s)\right)\le 3n$, where $g(k,d)=\sum_{i=0}^{k-1}\lceil d/2^i\rceil$. The \emph{weak coding bound} is the largest $n$ such that a binary linear 
$[3n,2k,2(n-s)]_2$-code exists, and the \emph{strong coding bound} is the largest $n$ such that a $2$-divisible binary $[3n,2k,2(n-s)]_2$-code with maximum weight at most $2n$ exists. By 
Lemma~\ref{lemma_additive_divisible} we have
\begin{equation}
  n_k(s)\le \text{strong coding bound}\le \text{weak coding bound}\le \text{Griesmer bound}.
\end{equation}
For $k=4$, i.e., additive quaternary codes of dimension $4$, the Griesmer bound for $s=8$ is $30$, so that no $[93,8,46]_2$-code and thus no additive $[31,4,23]_4$-code exists. 
Since a $[87,8,42]_2$-code does not exist while a $[84,8,40]_2$-code does, the weak coding bound is $28$, which is strictly smaller than the Griesmer bound. 
Using the divisibility and the bound on the maximum weight, the strong coding bound can be sharper than the weak coding bound, e.g., it gives $n_4(15)\le 55$, 
$n_4(28)\le 108$, and $n_4(29)\le 113$, see \cite{kurz2025optimal} for the evaluation of these bounds for $k=4$, where also further exclusions of the form 
``no $2$-divisible code with prescribed maximum weight exists'' are obtained by exhaustive computer enumeration using \texttt{LinCode} \cite{bouyukliev2020computer,kurz2019lincode}. 
The resulting values $n_4(s)$ are determined in \cite{kurz2025optimal} for all $s$, with the Griesmer bound being attained for all $s\ge 30$.

Divisibility can also be imposed on the additive code itself. A systematic study of the possible lengths of $\Delta$-divisible additive codes, i.e., of multisets of $(h-1)$-spaces over $\F_q$, is given in \cite{kurz2026divisibleadditive}. There, also divisible multisets of subspaces of mixed dimensions are considered, 
and partitions of $8$-divisible binary sets of points into points, lines, and planes are classified for small cardinalities, which parallels the treatment of vector space partitions in Chapter~\ref{sec_vector_space_partitions}. 
The results on the Griesmer bound for linear codes in Subsection~\ref{subsec_optimal_codes} carry over to additive codes only partially, but 
a Griesmer type bound for additive codes can always be attained if the minimum distance is sufficiently large \cite{kurz2020additiveGriesmer}. Additive codes of dimension $2.5$ are related to generalized ovals and multispreads, see \cite{krotov2026generalized}.
\begin{question}{Research problem}
  Determine the maximum length $n$ of a $\Delta$-divisible additive $[n,k/h]_{q^h}$-code for small parameters, e.g., for $q=2$ and $h=2$, and compare it with the maximum length of a $\Delta$-divisible 
  linear code over the same field.
\end{question}

\chapter{Partial spreads}
\label{sec_partial_spreads}
A \emph{partial $t$-spread} $\cT$ in $\PG(v-1,q)$ is a set of $t$-dimensional subspaces such that the points of $\PG(v-1,q)$ are covered 
at most once.\footnote{Note that we use the algebraic dimension, while authors in papers with a geometric background speak of partial 
$(t-1)$-spreads.} In other words, the non-zero vectors in $\F_q^v$ are covered at most once by the non-zero vectors of the $t$-dimensional 
subspaces, i.e., the elements of the partial $t$-spread.  
Using the notion of vector space partitions, see Section~\ref{sec_vector_space_partitions}, a partial $t$-spread is a vector space partition of type $t^{m_t}1^{m_1}$. 
The $m_1$ uncovered points are also called \emph{holes}. By $A_q(v,2t;t)$ we denote the maximum value of $m_t$.\footnote{The more 
general notation $A_q(v,2t-2w;t)$ denotes the maximum cardinality of a collection of $t$-dimensional subspaces, 
whose pairwise intersections have a dimension of at most $w$, see e.g.\ Subsection~\ref{subsec_subspace_codes}.}

If we replace the elements of a partial $t$-spread by their $[t]_q$ points, we obtain a set of points in $\PG(v-1,q)$ with cardinality at most $[v]_q$, so that
\begin{equation}
  \label{trivial_partial_spread_bound}
  \#\cT\le A_q(v,2t;t)\le \left\lfloor\frac{[v]_q}{[t]_q}\right\rfloor.
\end{equation} 
Observe that $[v]_q$ is divisible by $[t]_q$ iff $v$ is divisible by $t$. If $\cT$ is a partial $t$-spread in $\PG(v-1,q)$ attaining Inequality~(\ref{trivial_partial_spread_bound}) 
with equality, then we speak of a \emph{$t$-spread}. Those perfect packings of the points indeed exist for all positive integers $t$ and $v$ where $t$ divides $v$. To this 
end we can consider the set of all points in $\PG\!\left(v/t-1,q^t\right)$ and concatenate the corresponding linear codes with a $t$-dimensional simplex code over $\F_q$, 
see Subsection~\ref{subsec_subfield_constructions} for more details and e.g.\ \cite[Example 1]{honold2018partial} for a concrete example. The 
$[v/t]_{q^t}=[v]_q/[t]_q$ points in $\PG(v/t-1,q^t)$ and the corresponding $t$-dimensional simplex codes form the spread elements. Spreads arising by the sketched construction 
are also called \emph{Desarguesian spreads}.    

In order to construct large partial $t$-spreads we need:
\begin{nlemma}{(\cite{beutelspacher1975partial},\cite[Lemma 1.3]{govaerts2005small}}
  \label{lemma_mrd_vsp}
  If $\pi_a$ is an $a$-space in $\PG(a+b-1,q)$, where $a\ge b\ge 1$, then it is possible to partition the points of $\PG(a+b-1,q)\backslash \pi_a$  
  by a set of $q^a$ $b$-spaces.
\end{nlemma}
\begin{proof}
  Embed $\PG(a+b-1,q)$ in $\PG(2a-1,q)$ and take an $a$-spread $\cS$ in $\PG(2a-1,q)$ containing $\pi_a$. The elements of 
  $\cS\backslash \left\{\pi_a\right\}$ intersect $\PG(a+b-1,q)$ in a $b$-spread of $\PG(a+b-1,q)\backslash\pi_a$.
\end{proof}
If $v\ge 2t$, then by choosing $a=v-t$ and $b=t$ we can recursively construct partial $t$-spreads using Lemma~\ref{lemma_mrd_vsp}. If $t\le v<2t$, then 
we can choose an arbitrary $t$-space. Note that we end up with $t$-spreads if $v$ is divisible by $t$. Otherwise we have:
\begin{nproposition}(\cite{beutelspacher1975partial})
  \label{prop_partial_spread_lower_bound}
  If $v=tk+s$, where $t\ge 2$ and $1\le s\le t-1$, then we have
  \begin{equation}
    \label{ie_partial_spread_lower_bound}
    A_q(v,2t;t)\ge 1+\sum_{i=1}^{k-1} q^{v-it}=1+\sum_{i=1}^{k-1}q^{it+s}=\frac{q^v-q^{t+s}}{q^t-1}+1.
  \end{equation}  
\end{nproposition}
The same lower bound can be also obtained from the Echelon--Ferrers construction.

In \cite[Theorem 4.1]{beutelspacher1975partial}, see also \cite{hong1972general} for $q=2$, it was shown that the lower bound in Inequality~(\ref{ie_partial_spread_lower_bound})  
is attained with equality if $s=1$. In his original proof Beutelspacher considered the set of holes $N$ and the average number of holes per hyperplane, which is less than the 
total number of holes divided by $q$. An important insight was the relation $\#N\equiv \#(H\cap N) \pmod{q^{t-1}}$ for each hyperplane $H\in\cH$. In \cite[Corollary 2.6]{kurzspreads} 
the case $s=2$ was completely resolved for $q=2$. The original proof is based on a case analysis on possible vector space partitions in subspaces of codimension $2$. In 
\cite{kurz2017packing} it was observed that it suffices to study the number of holes in subspaces of codimension $2$. A major breakthrough was obtained by N{\u{a}}stase and Sissokho:
\begin{ntheorem}{(\cite[Theorem 5]{nastase2016maximum})}
   \label{thm_nastase_sissokho}
   If $v=kt+s$ with $0<s<t$ and $t>[s]_q$, then $A_q(v,2t;t)=\frac{q^v-q^{t+s}}{q^t-1}+1$, i.e., Inequality~(\ref{ie_partial_spread_lower_bound}) is tight.
\end{ntheorem} 
Ignoring the technical details one might say that a main idea was the study of the number of holes in subspaces of larger codimension by a clever inductive approach. All these 
results were obtained without using the notion of $q^r$-divisible (multi-)sets of points in $\PG(v-1,q)$. In retro perspective there is now an easy explanation. The set of 
holes of a partial $t$-spread in $\PG(v-1,q)$ is $q^{t-1}$-divisible, see Lemma~\ref{lem:union_subspaces} and Lemma~\ref{lemma_t_complement}. As shown in 
Section~\ref{sec_lengths_of_divisible_codes} the easy averaging argument used by Beutelspacher and the inheritance of divisibility properties to subspaces is sufficient to 
completely characterize the possible cardinalities of $q^{t-1}$-divisible multisets over $\F_q$, see Theorem~\ref{thm_characterization_div}. The property that the set of holes 
actually is a set and not just a multiset is not necessary and indeed Example~\ref{ex_gen_partial_spread_asymptotic_bound} slightly generalizes Theorem~\ref{thm_nastase_sissokho}  
to a wider context and gives a proof that is reduced to a single technical computation.    
For $q=2$ and $t=2$ see Figure~\ref{fig_partial_spread_holes}: each plane of $\PG(3,2)$ contains one line of a spread and a point of each of the other lines, so that the three holes meet every plane in an odd number of points.

\begin{figure}[htbp]
  \centering
  \begin{tikzpicture}[
      pt/.style={circle,draw=black,fill=white,inner sep=1.8pt},
      hp/.style={circle,draw=blue!70!black,fill=blue!70!black,inner sep=1.8pt},
      hole/.style={circle,draw=orange!90!black,fill=white,inner sep=1.8pt,thick},
      holeH/.style={circle,draw=orange!90!black,fill=blue!70!black,inner sep=1.8pt,thick},
      line/.style={draw=black!50,rounded corners=3pt}]
    \node[hp] at (0.00,0.00) {};
    \node[hp] at (0.00,0.70) {};
    \node[hp] at (0.00,1.40) {};
    \node[pt] at (1.50,0.00) {};
    \node[hp] at (1.50,0.70) {};
    \node[pt] at (1.50,1.40) {};
    \node[hp] at (3.00,0.00) {};
    \node[pt] at (3.00,0.70) {};
    \node[pt] at (3.00,1.40) {};
    \node[pt] at (4.50,0.00) {};
    \node[pt] at (4.50,0.70) {};
    \node[hp] at (4.50,1.40) {};
    \node[hole] at (6.00,0.00) {};
    \node[holeH] at (6.00,0.70) {};
    \node[hole] at (6.00,1.40) {};
    \foreach \c in {1,2,3,4} {\draw[line] (\c*1.5-1.5-0.3,-0.3) rectangle (\c*1.5-1.5+0.3,1.7);}
    \draw[line,draw=orange!90!black,dashed] (6-0.3,-0.3) rectangle (6+0.3,1.7);
    \foreach \c in {1,2,3,4} {\node[font=\footnotesize] at (\c*1.5-1.5,2.0) {$L_\c$};}
    \node[font=\footnotesize,orange!90!black] at (6,2.0) {holes};
  \end{tikzpicture}
  \caption{A partial line spread in $\PG(3,2)$ with $3$ holes. The points of a plane $H$ are filled.}
  \label{fig_partial_spread_holes}
\end{figure}

Additionally using the set property, i.e., considering possible cardinalities of $q^{t-1}$-divisible sets of points over $\F_q$, allows to explain another classical result from a 
different point of view. For a long time the best upper bound for partial spreads was given by Drake and Freeman:
\begin{ntheorem}{(\cite[Corollary 8]{nets_and_spreads})}
  \label{thm_partial_spread_4}
  If $v=kt+s$ with $0<s<t$, then 
  $$
    A_q(v,2t;t)\le \sum_{i=0}^{k-1} q^{it+s} -\left\lfloor\theta\right\rfloor-1
    =q^s\cdot \frac{q^{kt}-1}{q^t-1}-\left\lfloor\theta\right\rfloor-1
    =\frac{q^{v}-q^s}{q^t-1}-\left\lfloor\theta\right\rfloor-1,
  $$
  where $2\theta=\sqrt{1+4q^t(q^t-q^s)}-(2q^t-2q^s+1)$.
\end{ntheorem}
\begin{nexample}
  If we apply Theorem~\ref{thm_partial_spread_4} with $q=5$, $v=16$, $t=6$, and $s=4$, we obtain $\theta\approx 308.81090$ and $A_5(16,12;6)\le 9765941$.
\end{nexample}
The proof of Theorem~\ref{thm_partial_spread_4} is based on the work of Bose and Bush \cite{bose1952orthogonal} and uses nets, see Subsection~\ref{subsec_nets}, as crucial objects.  
A quadratic polynomial plays an important role. Starting from the {\lq\lq}quadratic condition{\rq\rq} in Lemma~\ref{lemma_hyperplane_types_arithmetic_progression} and its 
implication in Corollary~\ref{cor_nonexistence_arithmetic_progression}.(a) we consider the condition $\tau_q(u,\Delta,m)<0$, where $\Delta=q^{t-1}$ is the divisibility, $u$ depends 
linearly on the cardinality of the set of holes, and $m$ is a free parameter. Noting that $\tau_q(u,\Delta,m)$ is a quadratic polynomial in $m$, we can minimize $\tau_q(u,\Delta,m)$ 
in order to obtain a suitable choice for $m$. Instead of the total number of holes we can also consider the number of holes in a suitable subspace, which gives us a parameter $y$ as 
a degree of freedom. Referring to \cite{honold2018partial} or \cite{kurz2017packing} for details and explanations, we state:
\newcommand{\uu}{\lambda}
\begin{ntheorem}{(\cite[Theorem 10]{honold2018partial},\cite[Theorem 2.10]{kurz2017packing})}
  \label{theorem_parametric_ps_bound_2}
  For integers $s\ge 1$, $k\ge 2$, $y\ge \max\{s,2\}$, $z\ge 0$ with $\uu=q^{y}$, $y\le t$,  
  $t=[s]_q+1-z>s$, $v=kt+s$, and  $l=\frac{q^{v-t}-q^s}{q^t-1}$, we have 
  \begin{equation}
   A_q(v,2t;t)\le 
       lq^t+\left\lceil \uu -\frac{1}{2}-\frac{1}{2}
    \sqrt{1+4\uu\left(\uu-(z+y-1)(q-1)-1\right)} \right\rceil.
  \end{equation}   
\end{ntheorem}
Using Theorem~\ref{theorem_parametric_ps_bound_2} with $q=5$, $t=6$, $v=15$, $s=3$, $z=26$, and $y=5$ gives $A_5(15,12;6)\le 1953186$. Choosing $y=t$ we obtain 
Theorem~\ref{thm_partial_spread_4}. Theorem~\ref{theorem_parametric_ps_bound_2} also covers \cite[Theorems 6,7]{nastase2016maximumII} and yields improvements in a few 
instances, e.g.\ $A_3(15,12; 6)\le 19695$. Compared to Theorem~\ref{thm_partial_spread_4} we have e.g.\ improvements from $A_2(15,12; 6)\le 516$, $A_2(17,14;7)\le 1028$, and 
$A_9(18,16; 8)\le 3486784442$ to $A_2(15,12; 6)\le 515$, $A_2(17,14; 7)\le 1026$, and $A_9(18,16; 8)\le 3486784420$, respectively.   

Complementing Theorem~\ref{thm_nastase_sissokho} for smaller values of $t$ there is another parametric upper bound:
\begin{ntheorem}{(\cite[Corollary 7]{honold2018partial})}
  \label{theorem_parametric_ps_bound_3}
  For integers $s\ge 1$, $k\ge 2$, and $u,z\ge 0$ with $t=[s]_q+1-z+u>s$ we have
  $A_q(v,2t;t)\le lq^t+1+z(q-1)$, where $l=\frac{q^{v-t}-q^s}{q^t-1}$ and $v=kt+s$.   
\end{ntheorem}
Choosing $z=0$ implies Theorem~\ref{thm_nastase_sissokho}.

While explicit parametric upper bounds like Theorem~\ref{theorem_parametric_ps_bound_2} and Theorem~\ref{theorem_parametric_ps_bound_3} are certainly useful, they are in principle 
implied by the following observation:
\begin{center}
  Projective $q^{t-1}$-divisible codes of length $$n=[v]_q-A_q(v,2t;t)\cdot[t]_q$$ and dimension at most $v$ do exist over $\F_q$. 
\end{center}
As a refinement of the sharpened rounding from Definition~\ref{def:divisible_gauss_bracket} we introduce:
\begin{ndefinition}
  \label{def:divisible_gauss_bracket_point_multiplicity}
  For $a\in\Z$ and $b\in\Z\setminus\{0\}$ let $\llfloor a/b \rrfloor_{q^r\!\!,\lambda}$ be the maximal $n\in\Z$ such that there exists a $q^r$-divisible 
  multiset of points in $\PG(v-1,q)$ for suitably large $v$ with maximum point multiplicity at most $\lambda$ and cardinality $a-nb$. If no such multiset 
  exists for any $n$, we set $\llfloor a/b \rrfloor_{q^r\!\!,\lambda} = -\infty$.
\end{ndefinition}

With this our observation can be reformulated as:
\begin{nlemma}
  \label{lemma_partial_spread_div_bound}
  Let $\cU$ be a set of $k$-spaces in $\PG(v-1,q)$, where $1\le k\le v$, with pairwise trivial intersection. Then, we have
  \begin{equation}
    \label{ie_partial_spread_div_bound}
    \#\cU\le \llfloor [v]_q/[k]_q\rrfloor_{q^{k-1}\!\!,1}\text{.}
  \end{equation}   
\end{nlemma}
\begin{table}[htp]
  \begin{center}
    \begin{tabular}{llcc}
      \hline
      $\Delta$ & $n$ & bounds & reference \\ 
      \hline
      $2^3$ & 52      & $129 \le A_2(11,8;4)\le 132$            & Corollary~\ref{cor_implication_fourth_mac_williams} with $t=3$ \\
      $2^4$ & 131     & $257\le A_2(13,10;5)\le 259$            & \cite{kurz2020no131}\\
      $2^4$ & 200     & $1025 \le A_2(16,12;6)\le 1032$         & Corollary~\ref{cor_implication_fourth_mac_williams} with $t=6$ \\
      $2^5$ & 850     & $2049\le A_2(17,12;6)\le 2066$          & Corollary~\ref{cor_implication_fourth_mac_williams} with $t=13$ \\
      $3^3$ & 493     & $2188\le A_3(11,8;4)\le 2201$           & Corollary~\ref{cor_implication_fourth_mac_williams} with $t=12$ \\
      $3^4$ & 1586    & $6562\le A_3(13,10;5)\le 6574$          & Corollary~\ref{cor_implication_fourth_mac_williams} with $t=13$ \\
      $3^4$ & 4396    & $19684\le A_3(14,10;5)\le 19727$        & Corollary~\ref{cor_implication_fourth_mac_williams} with $t=36$ \\     
      $3^5$ & 14236   & $59050\le A_3(16,12;6)\le 59090$        & Corollary~\ref{cor_implication_fourth_mac_williams} with $t=39$ \\     
      $3^5$ & 39797   & $177148\le A_3(17,12;6)\le 177280$      & Corollary~\ref{cor_implication_fourth_mac_williams} with $t=109$ \\
      $3^6$ & 43760   & $177148\le A_3(18,14;7)\le 177187$      & Corollary~\ref{cor_implication_fourth_mac_williams} with $t=40$ \\
      $4^4$ & 10592   & $65537\le A_4(13,10;5)\le 65568$        & Corollary~\ref{cor_implication_fourth_mac_williams} with $t=31$ \\
      $4^4$ & 10250   & $262145\le A_4(15,12;6)\le 262174$      & Corollary~\ref{cor_implication_fourth_mac_williams} with $t=30$ \\     
      $4^5$ & 648716  & $4194305\le A_4(17,12;6)\le 4194852$    & Corollary~\ref{cor_implication_fourth_mac_williams} with $t=475$ \\ 
      $4^6$ & 693632  & $4194305\le A_4(18,14;7)\le 4194432$    & Corollary~\ref{cor_implication_fourth_mac_williams} with $t=127$ \\
      $5^1$ & 33      & $78126\le A_5(12,10;5)\le 78132$        & Corollary~\ref{cor_implication_fourth_mac_williams} with $t=5$ \\ 
      $5^4$ & 230551  & $1953126\le A_5(14,10;5)\le 1953454$    & Corollary~\ref{cor_implication_fourth_mac_williams} with $t=295$ \\
      $7^4$ & 3232754 & $40353608\le A_7(14,10;5)\le 40354853$  & Corollary~\ref{cor_implication_fourth_mac_williams} with $t=1154$ \\
      $8^3$ & 144568  & $2097153\le A_8(11,8;4)\le 2097416$     & Corollary~\ref{cor_implication_fourth_mac_williams} with $t=247$ \\ 
      $8^2$ & 1759    & $2097153\le A_8(12,10;5)\le 2097177$    & Corollary~\ref{cor_implication_fourth_mac_williams} with $t=24$ \\ 
      $8^2$ & 1539    & $16777217\le A_8(14,12;6)\le 16777237$  & Corollary~\ref{cor_implication_fourth_mac_williams} with $t=21$ \\ 
      $9^2$ & 3559    & $59050\le A_9(8,6;3)\le 59090$          & Corollary~\ref{cor_implication_fourth_mac_williams} with $t=39$ \\
      $9^4$ & 2679394 & $43046722\le A_9(13,10;5)\le 43047086$  & Corollary~\ref{cor_implication_fourth_mac_williams} with $t=363$ \\
      \hline
    \end{tabular}
    \caption{Sporadic upper bounds for partial spreads}
    \label{table_sporadic_upper_bounds_partial_spreads}
  \end{center}
\end{table}
The construction from Proposition~\ref{prop_partial_spread_lower_bound} and the non-existence of a $8$-divisible set of $52$ points over $\F_2$, see e.g.\ Lemma~\ref{lemma_no_8_div_52},  
imply
\begin{equation}
 2^4\cdot \frac{2^{4k-1}-2^3}{2^4-1}+1 \le A_2(4k+3,8;4)\le 2^4\cdot \frac{2^{4k-1}-2^3}{2^4-1}+4 
\end{equation}
for all $k\ge 2$. In general lower and upper bounds, if obtained by non-existence results of projective $q^{t-1}$-divisible codes, for $A_q(v,2t;t)$ come 
in parametric series for $v=kt+s$ with $k\in\N_{\ge 2}$, see \cite{honold2018partial} for details or \cite{kurz2017packing} for more examples. In 
Table~\ref{table_sporadic_upper_bounds_partial_spreads} we list the known upper (and corresponding lower) bounds that do not follow from Theorem~\ref{theorem_parametric_ps_bound_2} 
or Theorem~\ref{theorem_parametric_ps_bound_3} directly. Here $\Delta$ is the divisibility constant and $n$ the cardinality of the non-existent set of points 
that leads to the stated upper bound for a partial spread. When we mention the application of Corollary~\ref{cor_implication_fourth_mac_williams} as reference, 
then typically also the {\lq\lq}linear{\rq\rq} and the {\lq\lq}quadratic{\rq\rq} condition introduced in Section~\ref{sec_nonexistence_projective_q_r} are involved, 
see Example~\ref{example_A_8_12_10_5} for exemplarily details. As a measurement for our state of knowledge we also state the corresponding lower bound for the 
partial spread obtained via Proposition~\ref{prop_partial_spread_lower_bound}.    
\begin{nexample}
  \label{example_A_8_12_10_5}
  In order to show the upper bound $A_8(12,10;5)\le 2097177$, we actually have to show the non-existence of a $8^4$-divisible set of $177887$ points over $\F_8$. 
  Assuming its existence, Lemma~\ref{lemma_average} yields the existence of a $8^3$-divisible set of points over $\F_8$ with a cardinality contained in 
  $\{18143-i\cdot 8^4 \,:\, i\in \N_0\}\cap\N_0=\{18143, 14047, 9951, 5855, 1759\}$. Since $8^3$-divisible sets of $8^4$ points over $\F_8$ exist, it suffices to 
  exclude the existence of a $8^3$-divisible set of $18143$ points over $\F_8$. Assuming its existence, Lemma~\ref{lemma_average} yields the existence of a $8^2$-divisible 
  set of points over $\F_8$ with a cardinality contained in $\{1759-i\cdot 8^3 \,:\, i\in \N_0\}\cap\N_0=\{1759, 1247, 735, 223\}$. Again it is sufficient to 
  exclude cardinality $1759$. As mentioned in Table~\ref{table_sporadic_upper_bounds_partial_spreads} we can apply Corollary~\ref{cor_implication_fourth_mac_williams} with $t=24$ in order to 
  conclude the non-existence of a $8^2$-divisible set of $1759$ points over $\F_8$. 
\end{nexample}
\begin{nexercise}
  \label{exercise_details_leading_to_sporadic_partial_spread_bounds}
  Verify all details leading to the upper bounds $A_8(14,12;6)\le 16777237$ and $A_5(12,10;5)\le 78132$.
\end{nexercise}

A few more remarks on Lemma~\ref{lemma_partial_spread_div_bound} and Table~\ref{table_sporadic_upper_bounds_partial_spreads} 
are in order. So far it occurs that upper bounds for partial spreads that are based on non-existence results of projective 
$q^{t-1}$-divisible codes certified by the linear programming method using the first three MacWilliams equations only, can be obtained 
via Theorem~\ref{theorem_parametric_ps_bound_2} or Theorem~\ref{theorem_parametric_ps_bound_3} directly. However, this is not a proven fact 
at all. The explicit parametric conditions introduced in Section~\ref{sec_nonexistence_projective_q_r} are quite handy for automatic 
computations cf.~Exercise~\ref{ex_implement_exclusion_lemmas}. Since the numbers grow very quickly and the linear programming method 
reveals its full power only if applied recursively, efficient algorithms are indeed an issue. For Table~\ref{table_sporadic_upper_bounds_partial_spreads} 
we remark that we have applied the mentioned tools for $v\le 19$ if $q\le 4$, for $v\le 16$ if $q=5$, and for $v\le 14$ if $7\le q\le 9$ only.   

\begin{table}[htp]
  \begin{center}
    \begin{tabular}{llcc}
      \hline
      $q$ & $\Delta$ & $n$ & bounds \\ 
      \hline
      2 & $2^4$ & 232, 263 & $513\le A_2(14,10;5)\le 521$ \\
      2 & $2^5$ & 322, 385 & $513\le A_2(15,12;6)\le 515$ \\
      2 & $2^5$ & 913, 976 & $2049\le A_2(17,12;6)\le 2066$ \\
      3 & $3^3$ & 244      & $730\le A_3(10,8;4)\le 732$ \\ 
      \hline
    \end{tabular}
    \caption{Open cases with implications for partial spread bounds}
    \label{table_open_cases_upper_bounds_partial_spreads}
  \end{center}
\end{table}
It should be mentioned that the existence of a projective $q^{t-1}$-divisible code $C$ over $\F_q$ (of suitable length) does not imply the existence of a 
partial spread with matching parameters. In other words, Lemma~\ref{lemma_partial_spread_div_bound} is just a necessary condition. Since all currently 
known upper bounds for partial spreads are implied by Lemma~\ref{lemma_partial_spread_div_bound}, we list the open cases of relevant 
lengths of projective $q^{t-1}$-divisible codes, i.e., where their existence is currently undecided, in Table~\ref{table_open_cases_upper_bounds_partial_spreads}.   

We remark that the non-existence of a binary projective $2^4$-divisible code of length $130$ would imply the non-existence of a binary projective $2^5$-divisible code of length $322$, 
so that $A_2(15,12;6)\le 514$ would follow. In the other direction we remark that binary projective $2^r$-divisible codes of length $n$ with $(r,n)\in\{(3,67),(4,162),(5,519)\}$ indeed 
exist, so that the upper bounds $A_2(11,8;4)\le 132$, $A_2(13,10;5)\le 259$, and $A_2(16,12;6)\le 1032$ might be attained with equality.

\medskip

\begin{question}{Research problem}
  Improve the lower bound $129\le A_2(11,8;4)\le 132$.  
\end{question} 

\section{Realizations of $q^r$-divisible sets of points as partial spreads}
\label{subsec_realizations_partial_spreads}
As discussed, the non-existence of $q^r$-divisible sets of cardinality $n$ over $\F_q$ sometimes has implications to upper bounds for partial spreads. However, not all 
cardinalities are directly relevant to that extent. E.g.\ the non-existence of a binary projective triply--even code of length $59$, shown in \cite{honold2019lengths}, has 
no such implication. More precisely, there is no set of pairwise disjoint $4$-spaces in $\PG(v-1,2)$ such that there are exactly $59$ holes. To this end, we observe 
that $[v]_2 \mod [4]_2 \in \{0,1,3,7\}$ while $59\mod [4]_2=14$. An implication for the existence of vector space partitions of a certain type is stated in 
Section~\ref{sec_vector_space_partitions}. 
\begin{ndefinition}
  A $q^r$-divisible set $\cM$ of $n$ points over $\F_q$ is said to be \emph{realizable as a partial spread} if a partial $(r+1)$-spread over $\F_q$ exists whose set of holes 
  is equivalent to $\cM$ (eventually embedded in an ambient space of dimension larger than $\dim(\cM)$). We use the same notation for codes using their correspondence to multisets 
  of points.
\end{ndefinition}  
With this, we e.g.\ have $A_2(11,8;4)< 132$ if none of the projective binary triply--even codes of length $67$ are realizable as a partial spread. 

\begin{trailer}{The only partial spread better than Beutelspacher's construction}
So far, the only known cases in which the construction of Proposition~\ref{prop_partial_spread_lower_bound} has been surpassed are derived from the
sporadic example of a partial $3$-spread of cardinality $34$ in $\PG(7,2)$ \cite{spreadsk3}, which has $17$ holes and can be used to show 
$A_2(3m+2,6;3)\geq(2^{3m+2}-18)/7$, which exceeds the lower bound of Proposition~\ref{prop_partial_spread_lower_bound} by one. A first step towards 
the understanding of the sporadic example is the classification of all $2^2$-divisible sets of points with cardinality $17$ in $\PG(k-1,2)$. It turns out 
that there are exactly $3$ isomorphism types, one configuration $\mathcal{H}_k$ for each dimension $k\in\{6,7,8\}$. Generating matrices for the corresponding
doubly-even codes are given by 
\begin{equation}
  \label{eq:17gen}
  \footnotesize
\left(\begin{array}{c}
10000110010101110\\
01000010111011100\\
00100100000011000\\
00010111001110100\\
00001001100111110\\
00000011100111011\\
\end{array}\right),
\;
\left(\begin{array}{c}
10000011110100110\\
01000001111111000\\
00100010000110000\\
00010010000101000\\
00001001001000100\\
00000101001000010\\
00000010101011111\\
\end{array}\right),
\;
\left(\begin{array}{c}
10000000111011110\\
01000000010110000\\
00100000011100000\\
00010000001110000\\
00001001100000010\\
00000101000001010\\
00000011000000110\\
00000001111011101\\
\end{array}\right).
\normalsize
\end{equation}
\end{trailer}

While the classification, so far, is based on computer calculations, see e.g.\ \cite{doran2011codes} and \url{http://www.rlmiller.org/de_codes}, one can 
easily see that there are exactly three solutions of the MacWilliams equations.
\begin{nexercise}
  Let $C$ be a projective $2^2$-divisible $[17,k]_2$-code. Conclude $k\in\{6,7,8\}$ from the MacWilliams equations and determine the unique weight enumerator 
  in each case.  
\end{nexercise}
The set of holes of the partial $3$-spread in \cite{spreadsk3} corresponds to $\mathcal{H}_7$. A geometric description using coordinates, of this configuration is given in
\cite[p.~84]{lambert2013random}. We have computationally checked that indeed all three hole configurations can be realized by a partial $3$-spread
of cardinality $34$ in $\PG(7,2)$.\footnote{$624$ non-isomorphic examples can be downloaded at \url{http://subspacecodes.uni-bayreuth.de}. Several thousand non-isomorphic 
examples have been found so far, some of them admitting an automorphism group of order $8$.} All three $\mathcal{H}_i$ are special instances of  
parametric geometrical constructions, see \cite[Subsection 6.1]{honold2018partial} for the details.   

\begin{trailer}{The Hill cap}While not all automorphisms of the set of holes may extend to the entire partial spread, $q^{t-1}$-divisible sets of points with a large automorphism group may 
have some chances to be realized as a partial spread with at least some automorphisms. To this end, we want to mention another interesting and rather small example.  
Over the ternary field the smallest open case for partial spreads is given by $244\le A_3(8,6;3)\le 248$. A putative partial plane spread in $\PG(7,3)$ of size $248$ would 
have a $3^2$-divisible set $\cH$ of holes of cardinality $56$. Such a point set is unique up to isomorphism and has dimension $\dim(\cH)=6$. It corresponds to an
distance-optimal two-weight code with weight enumerator $W_C(x)=x^0+616x^{36}+112x^{45}$. The set $\cH$ was first described by Raymond Hill \cite{hill1978caps} and 
is known as the \emph{Hill cap}. A generator matrix is e.g.\ given by
\begin{equation*}
  \footnotesize
  \setlength{\arraycolsep}{0.05em}
  \left(
  \begin{array}{cccccccccccccccccccccccccccccccccccccccccccccccccccccccc}
    1&0&0&0&0&0&2&2&1&1&0&1&0&0&1&1&0&2&0&2&1&1&1&1&0&0&1&0&1&2&0&1&0&2&1&2&1&1&1&1&1&2&2&0&0&1&2&0&0&2&0&1&2&2&1&1\\
    0&1&0&0&0&0&1&1&1&0&1&2&1&0&1&0&1&1&2&1&1&2&0&0&1&0&2&1&1&2&2&2&1&1&1&2&1&0&0&0&0&2&1&2&0&2&2&2&0&0&2&2&2&0&1&0\\
    0&0&1&0&0&0&2&2&2&2&0&2&2&1&0&2&0&0&1&1&2&0&0&1&0&1&1&2&0&0&2&0&2&0&2&0&0&2&1&1&1&2&2&1&2&1&1&2&2&2&0&0&1&1&1&2\\
    0&0&0&1&0&0&1&0&1&1&2&2&2&2&0&2&2&1&0&2&0&0&2&2&1&0&0&1&0&1&0&1&0&0&2&2&2&2&1&0&0&2&2&2&1&1&2&1&2&2&2&2&1&2&0&0\\
    0&0&0&0&1&0&2&0&1&2&1&0&2&2&1&1&2&1&1&2&0&0&1&0&2&1&1&0&2&2&1&1&1&2&1&0&0&0&0&2&1&2&0&2&2&2&0&2&1&2&2&0&1&0&0&1\\
    0&0&0&0&0&1&1&2&2&0&2&0&0&2&2&0&1&0&1&2&1&2&2&0&0&2&0&1&1&0&2&0&1&2&1&2&2&2&2&2&1&2&0&0&2&1&0&0&2&0&2&1&1&2&2&2
  \end{array}\right).
\end{equation*}
The automorphism group of $\cH$ has order $40320$.
\end{trailer} 
\begin{question}{Research problem}
  Improve the lower bound  $244\le A_3(8,6;3)\le 248$.  
\end{question}

If we do not restrict ourselves to partial spreads of the maximum possible size, we can mention another example. 
In \cite{honold2019classification} the $14445$ isomorphism types of partial $3$-spreads in $\PG(6,2)$ of size $16$, i.e.\  one less 
than the maximum, were classified. In this context, the five non-isomorphic $2^2$-divisible sets of $15$ points over $\F_2$ were classified 
without the help of computer enumerations. Four of them consist of the union of a plane and an affine solid, while the fifth example is obtained 
by applying the cone construction to a $6$-dimensional projective base, see Example~\ref{example_cone_constructions}. The interesting point is 
that, again, all divisible sets can be realized.

A maximal partial $t$-spread in $\PG(n-1,q)$ is a set of $t$-dimensional subspaces that are pairwise disjoint which cannot be extended by another 
$t$-dimensional subspace without violating the property of pairwise disjointness. The corresponding complement $\cH$ is a set of holes/points 
that is $q^{t-1}$-divisible, does not contain any $t$-space in its support, and satisfies $\#\cH\equiv [n]_q\pmod{[t]_q}$. As an example let us 
consider maximal partial $4$-spreads in $\PG(7,2)$. Up to equivalence there are exactly $9316$ point sets that can be covered by $8$ pairwise disjoint 
solids in $\PG(7,2)$. In $181$~cases those point sets cannot be extended by an additional solid, i.e., a corresponding set of $8$ solids forms a maximal 
partial $4$-spread in $\PG(7,2)$. Note that there may be several configurations of $t$-spaces resulting in the same set of covered points. We have 
computationally checked that the minimum size of a maximal partial solid spread in $\PG(7,2)$ is indeed $8$, see Table~\ref{table_covered_point_sets_by_solids_in_pg_7_2}. 

\begin{table}[htp]
  \begin{center}
    \begin{tabular}{rrrrrrrrrrrr}
      \hline
      $r$                   & 1 & 2 & 3 & 4 &  5 &   6 &    7 &    8 &    9 &   10 &  11 \\   
      \hline 
      \# point sets         & 1 & 1 & 1 & 3 & 22 & 341 & 3726 & 9316 & 5442 & 1336 & 303 \\ 
      \# maximal point sets & 0 & 0 & 0 & 0 &  0 &   0 &    0 &  181 & 1343 &  317 &  58 \\ 
      \hline
    \end{tabular}   
    
    \medskip    
    
    \begin{tabular}{rrrrrrr}
      \hline
      $r$                   & 12 & 13 & 14 & 15 & 16 & 17 \\   
      \hline 
      \# point sets         & 42 &  6 &  1 &  1 &  1 &  1 \\ 
      \# maximal point sets &  5 &  3 &  0 &  0 &  0 &  1 \\   
      \hline
    \end{tabular}
    \caption{Non-isomorphic point sets covered by $r$ solids in $\PG(7,2)$.}
    \label{table_covered_point_sets_by_solids_in_pg_7_2}
  \end{center}
\end{table} 

The three different hole configurations of maximal partial $4$-spreads of size $13$ in $\PG(7,2)$ are given by: 
$$
\left(\begin{smallmatrix}      
000000000000000000000000000011111111111111111111111111111111\\
000000000000111111111111111100000000000000001111111111111111\\
000011111111000000001111111100000000111111110000000011111111\\
111100001111000011110000111100001111000011110000111100001111\\
001100110011000000110011001111110011001100111111001100110011\\
010101010101111101010101010111110101010101010000010101010101\\
111100001111001111000110101000111010001110010011100101011100\\
000011111111010101011100100101010110101011000101001110011010\\
\end{smallmatrix}\right)
$$
$$
\left(\begin{smallmatrix}      
000000000000000000000000000011111111111111111111111111111111\\
000000000000111111111111111100000000000000001111111111111111\\
000011111111000000001111111100000000111111110000000011111111\\
111100001111000011110000111100001111000011110000111100001111\\
001100010111001100110111011111110011011100010011001100010001\\
010001110011111101110001001100110100000101101100011101110110\\
110100101111010111010010110011011001001110000100100100111110\\
001011111101011101101000100101000101111010101001000110011011\\
\end{smallmatrix}\right)
$$
$$
\left(\begin{smallmatrix}      
000000000000000000000000000011111111111111111111111111111111\\
000000000000111111111111111100000000000000001111111111111111\\
000011111111000000001111111100000000111111110000000011111111\\
111100001111000011110000111100001111000011110000111100001111\\
001100110011111100110011001111110011001100110000001100110011\\
010101010101111101010101010100000101010101011111010101010101\\
111100001111001111000110101000111010001110010011100101011100\\
011111111000010100101100111001010001101010110101010010011101\\
\end{smallmatrix}\right)
$$
The corresponding automorphism groups of those projective $[60,8,\{24,32\}]_2$-codes have orders $576$, $14400$, and $96$, respectively. 
We remark that there are $12$ non-isomorphic projective $[60,8,\{24,32\}]_2$-codes. 
Three of the $12$ sets of holes, with automorphism groups of orders $120$, $288$, and $4320$ are the disjoint union of four solids, 
i.e., realizing partial spreads can be completed to solid spreads. In one further case, with an automorphism group 
of order $720$, the set of holes contains a solid but not two disjoint solids in its support. The remaining $8$ 
sets of holes might in principle be realized by maximal partial $4$-spreads. However, as already mentioned, only the stated 
three hole configurations can be realized by maximal partial $4$-spreads in $\PG(7,2)$. There are two non-isomorphic 
projective $[45,8,\{16,24\}]_2$-codes: a disjoint union of three solids and the concatenation of a Baer solid with a line over $\F_2$, see Example~\ref{example_baer} for a generalization. 
The latter case cannot be realized in $\PG(7,2)$ as a maximal partial $4$-spread. Indeed both cases can be obtained by concatenating a
$[15,4,\{8,12\}]_4$-code with a two-dimensional simplex code over $\F_2$. From the $12$ $[60,8,\{24,32\}]_2$-codes only $7$ can be obtained 
via the concatenation of a $[20,4,\{12,16\}]_4$-code with a two-dimensional simplex code over $\F_2$. 

\begin{nexample}
  \label{example_baer}
  For an arbitrary prime power $q$ and an integer $t\ge 3$ let $\cB$ be the set of points in $\PG(t-1,q^2)$ whose points can be written 
  with coordinates in $\GF(q)<\GF(q^2)$, i.e., a Baer-type construction. So, we have $\#\cB=[t]_q$. By concatenating with a 
  two-dimensional simplex code over $\F_q$ we obtain a $q^{t-1}$-divisible projective two-weight code over $\F_q$ with dimension $k=2t$ and length $n=(q+1)[t]_q$, 
  i.e., having the same parameters as the disjoint union of $q+1$ $t$-spaces but not containing a $t$-space in its support.
\end{nexample}

In Table~\ref{table_maximimal_partial_spreads} we have listed the possible sizes of maximal partial $t$-spreads in $\PG(n-1,q)$ for small parameters. 
For $(t,n)=(2,4)$ and $q\in\{7,8,9\}$ the listed sizes may be incomplete below the smallest listed size, cf.\ \cite[Section 7.3]{de2014intersection}, \cite[Section 4]{de2015small}, 
and possibly in the interval $26$--$29$ for $q=8$.  

\begin{table}[htp]
  \begin{center}
    \begin{tabular}{cccl}
      \hline
      $t$ & ambient space & cardinalities & references \\ 
      \hline 
      2 & $\PG(3,2)$ & 5              & \cite{computation_of_partial_spreads} \\
      2 & $\PG(4,2)$ & 5,7,9          & \cite{shaw2000subsets,computation_of_partial_spreads} \\
      2 & $\PG(5,2)$ & 13,15--19,21   & \cite{govaerts2005small,iurlo2015new}\\
      3 & $\PG(5,2)$ & 5,9            & \\
      3 & $\PG(6,2)$ & 9--17       & \cite{honold2019classification} \\
      4 & $\PG(7,2)$ & 8--13,17       & \\
      2 & $\PG(3,3)$ & 7,10           & \cite{computation_of_partial_spreads} \\
      2 & $\PG(3,4)$ & 11--14,17      & \cite{jungnickel2003maximal,computation_of_partial_spreads} \\
      2 & $\PG(3,5)$ & 13--22,26      & \cite{heden2000maximal} \\
      2 & $\PG(3,7)$ & \textit{23--45, 50}     & \cite{blokhuis2003blocking,heden1991greedy,heden2001maximal,iurlo2010new} \\
      2 & $\PG(3,8)$ & \textit{25,30--58,65}   & \cite{barat2004minimal,lurlo2016computer,soicher2018classifying} \\
      2 & $\PG(3,9)$ & \textit{36--74,82}      & \cite{faina2010maximal,heden2008non,lurlo2016computer} \\
      \hline 
    \end{tabular}
    \caption{Possible cardinalities of maximal partial $t$-spreads in $\PG(n-1,q)$.}
    \label{table_maximimal_partial_spreads}
  \end{center}
\end{table} 

In \cite[Lemma 4.15]{alavi2015triple} it was shown that the minimum size of a maximal partial $t$-spread in $\PG(2t-1,2)$ is at least $5$, which is met with equality for $t\in\{2,3\}$. 
In general, the minimum size of a maximal partial $t$-spread in $\PG(2t-1,q)$ is at least $2q-1$, see \cite[Theorem 3.6]{de2015small}, and at least $2q$ if $t=2$, see \cite{glynn1982lower}.

In Remark~\ref{remark_kernel_subfield_subcodes} we have observed that for non-prime field sizes $q$ the kernel of the incidence matrix 
between points and $k$-spaces yields further conditions on the multiset of points associated to a multiset of $k$-spaces that are not 
captured by the $q^{k-1}$-divisibility. 
\begin{nexample}
  \label{example_excluded_hole_set}
  In \cite{heden2008non} two non-isomorphic $9$-divisible sets of $60$ points in $\PG(3,9)$ were stated and characterized. None of these 
  two point sets contains a full line. Using a result of \cite{blokhuis1993size} the authors showed that both point sets cannot be realized as 
  the set of holes of a partial spread, see \cite[Theorems 1 and 2]{heden2008non}.  
\end{nexample}
\begin{nexercise}
  Compare \cite[Lemma 2.1]{blokhuis1993size} with the implications of the kernel approach, cf.~Remark~\ref{remark_kernel_subfield_subcodes}, for the 
  two $9$-divisible point sets of cardinality $60$ in $\PG(3,9)$ from \cite{heden2008non}.  
\end{nexercise}
\begin{question}{Research problem}Find an example of a $p^{t-1}$ divisible set of points over $\F_p$ that cannot be realized 
as a partial $t$-spread and does not admit a rather trivial justification.
\end{question}

\chapter{Vector space partitions}
\label{sec_vector_space_partitions}
A \emph{vector space partition} $\cV$ of $\PG(v-1,q)$ is a set of subspaces with the property that every point $P$ of $\PG(v-1,q)$, or every 
non-zero vector in $\F_q^v$, is contained in a unique member of $\cV$. If $\cV$  contains $m_d$ subspaces of dimension $d$, then $\cV$ is of 
type $k^{m_k}\dots 1^{m_1}$, where we may leave out some of the cases with $m_d=0$. If there is at least one dimension $d>1$ with $m_d>0$ and 
$m_v=0$, then $\cV$ is called non-trivial. By $\#\cV$ we denote the number $\sum_{i=1}^v m_i$ of elements of the vector space partition.

The relation between vector space partitions and divisible sets can be directly read off from Lemma~\ref{lem:union_subspaces} (noting that the 
$1$-spaces indeed form a set):
\begin{nlemma}
  \label{lemma_connection vsp}
  Let $\cV$ be a vector space partition of type $t^{m_t}\dots s^{m_s}1^{m_1}$ of $\PG(v-1,q)$, where $v>t\ge s\ge 2$ Then, the $1$-dimensional 
  elements of $\cV$ form a $q^{s-1}$-divisible set of cardinality $m_1$ in $\PG(v-1,q)$.
\end{nlemma}
Since there is no $2^1$-divisible set of $2$ points over $\F_2$ there is e.g.\ no vector space partition of type $4^{16}3^1 2^2 1^2$ of $\PG(7,2)$. For a potential 
vector space partition of type $4^{17}3^{35}2^2 1^5$ of $\PG(8,2)$ we cannot apply the argument directly since a $2^1$-divisible set of $5$ points over $\F_2$ indeed exists. 
However, if we replace the two lines by their three points each, we would end up with a $2^2$-divisible set of $11$ points over $\F_2$ which does not exist.
\begin{nlemma}
  \label{lemma_connection vsp_full}
  Let $\cV$ be a vector space partition of type $t_k^{m_k}\dots t_1^{m_1}$ of $\PG(v-1,q)$, where $v>t_k>\dots>t_1>0$. Then, for each index $1\le s<k$ the 
  $\sum_{i=1}^s m_i\cdot [t_i]_q$ points contained in the elements of dimension at most $t_s$ in $\cV$ elements of $\cV$ form a $q^{t_{s+1}-1}$-divisible set in $\PG(v-1,q)$.
\end{nlemma}
Note that the values of the $m_i$, for $i>s$, and the $t_i$, for $i>s+1$, as well as the dimension $v$ of the ambient space are irrelevant.
\begin{nexercise}
  \label{exercise_vsp_234_v_8}
  Show that no vector space partition of type $4^a 3^b 2^c$ of $\PG(7,2)$ exists if 
  \begin{eqnarray*}
    (a,b,c)&\in&\big\{(0,36,1),(1,33,3),(4,27,2),(5,24,4),(7,21,1),(8,18,3),\\ && (11,12,2),(12,9,4),(14,6,1),(15,3,3)\big\}.
  \end{eqnarray*}
\end{nexercise}

As an example for a construction we remark that each partial $t$-spread of size $n$ in $\PG(v-1,q)$ gives rise to a vector space partition of $\PG(v-1,q)$ of type 
$t^n 1^{m_1}$, where $m_1=[v]_q-n[t]_q$, by complementing the set of $t$-spaces of the partial spread with its set of holes. If $t$ divides $v$, then 
$t$-spreads in $\PG(v-1,q)$ directly give a vector space partition of $\PG(v-1,q)$ of type $t^{m_t}$, where $m_t=[v]_q/[t]_q$. Also Lemma~\ref{lemma_mrd_vsp} 
gives a vector space partition. 

\begin{trailer}{The packing and the dimension condition}Counting points gives the necessary condition 
\begin{equation}
  \label{eq_vsp_counting_points}
  \sum_{1\le i\le v} m_i\cdot [i]_q=[v]_q
\end{equation}
for the existence of a vector space partition of type $v^{m_v}\dots 1^{m_1}$ in $\PG(v-1,q)$. Since an $a$-space and a disjoint $b$-space 
span an $(a+b)$-space, we also have
\begin{equation}
  \label{eq_vsp_dimension}
  m_i\cdot m_j=0
\end{equation}
for all $1\le i<j\le v$ with $i+j>v$ and $m_i\le 1$ for all $i>v/2$.
\end{trailer}
\begin{ndefinition}
  \label{def_composed_vsp}
   If a vector space partition $\cV$ of $\PG(v-1,q)$ arises from a vector space partition $\cV_1$ of $\PG(v-1,q)$ where one $a$-space of $\cV_1$ is replaced 
   by a vector space partition $\cV_2$ of $\PG(a-1,q)$ with $\#\cV_2\neq 1$, then we say that $\cV$ is \emph{reducible} and \emph{composed of} $\cV_1$ and 
   $\cV_2$.  
\end{ndefinition}
\begin{nexercise}$\,$\\[-5mm]
  \begin{enumerate}
    \item[(a)] Show the existence of a vector space partition of type $4^{m_4} 2^{m_2}$ of $\PG(7,2)$, where 
               $m_4=17-i$ and $m_2=5i$, for all $0\le i\le 17$.
    \item[(b)] Show the existence of a vector space partition of type $3^{33} 2^8$ of $\PG(7,2)$. \textit{Hint:} Construct a vector space 
               partition of type $5^1 3^{32}$ of $\PG(7,2)$ first.
  \end{enumerate}                           
\end{nexercise}

If we only focus on the occurring dimensions in a type of a vector space partition, then the following general existence result was shown   
using the Frobenius number:
\begin{ntheorem}{(\cite[Theorem 2]{beutelspacher1978partitions})}
  Let $T=\left\{t_1<t_2<\dots<t_k\right\}$ be a set of positive integers with $d:=\gcd(T)$. If $v$ is an integer with
  \begin{equation}
     v> 2t_1\left\lceil\frac{t_k}{dk}\right\rceil+t_2+\dots+t_k,
  \end{equation}
  then a vector space partition of $\PG(v-1,q)$ of a type satisfying $\left\{i\,:\,m_i>0\right\}=T$ exists iff $\gcd(T)$ divides $v$.  
\end{ntheorem}
\begin{trailer}{Types of vector space partitions in $\PG(v-1,q)$ for $v\le 5$}
\vspace*{-4mm}
\begin{nexercise}
  \label{exercise_vsp_classification_small}
  Show that for $v\le 4$ conditions (\ref{eq_vsp_counting_points}) and (\ref{eq_vsp_dimension}) are sufficient to characterize all possible 
  types of a vector space partition in $\PG(v-1,q)$. More precisely:
  \begin{itemize}
    \item the possible vector space partitions of $\PG(1,q)$ are given by $1^{q+1}$;
    \item the possible vector space partitions of $\PG(2,q)$ are given by $2^1 1^{q^2}$ and $1^{q^2+q+1}$;
    \item the possible vector space partitions of $\PG(3,q)$ are given by $2^{q^2+1-j}1^{(q+1)j}$, where $0\le j\le q^2+1$, and $3^1 1^{q^3}$.  
  \end{itemize}
\end{nexercise}
For vector space partitions of type $2^{m_2}1^{m_1}$ in $\PG(4,q)$ conditions (\ref{eq_vsp_counting_points}) and (\ref{eq_vsp_dimension}) only imply 
$m_2=q^3+q-j$ and $m_1=1+(q+1)j$ for $0\le j\le q^3+q$. Lemma~\ref{lemma_connection vsp} (or $A_q(5,4;2)=q^3+1)$) yields $j\ge q-1$.
\begin{nexercise}
  \label{exercise_vsp_classification_dim_5}
  Show that the conditions (\ref{eq_vsp_counting_points}), (\ref{eq_vsp_dimension}) and Lemma~\ref{lemma_connection vsp} are sufficient to characterize all possible 
  types of a vector space partition in $\PG(4,q)$. More precisely, the possible vector space partitions of $\PG(4,q)$ are given by $4^1 1^{q^4}$, $3^1 2^{q^3-j}1^{(q+1)j}$ 
  for $0\le j\le q^3$, and $2^{q^3+1-j}1^{q^2+(q+1)j}$, where $0\le j\le q^3+1$.
\end{nexercise}
\vspace*{-3mm}
\end{trailer}
There is little hope to classify all feasible types of vector space partitions of $\PG(v-1,q)$ unless the parameters are relatively small. Already the determination 
of the minimum possible $m_1$ such that a vector space partition of type $t^{m_t}1^{m_1}$ of $\PG(v-1,q)$ exists, i.e., the determination of $A_q(v,2t;t)$, is a 
really hard problem if $v$ and $t$ get large. More precisely, already the exact value of $A_q(8,6;3)$ is unknown if $q>2$. Nevertheless, the mentioned classification 
is an ongoing, very hard, major project, see e.g.\ \cite{el2009partitions,heden2009length,heden2012survey,heden2013supertail,lehmann2012some,seelinger2012partitions}. 
Currently all feasible types of vector space partitions of $\PG(v-1,2)$ with $v\le 7$ are characterized \cite{el2009partitions}. The feasible types of vector space 
partitions of $\PG(7,2)$ that do not contain elements of dimension $1$ are classified in \cite{el2010partitions}. For the characterization of all feasible vector space
partitions of $\PG(8-1,2)$ see \cite{kurz2022vectorspacepartitions8}. Here we want to focus on non-existence results.    

\begin{trailer}{Using classification results for divisible codes}In $\PG(5,2)$ the only infeasible type of a vector space partition that is not excluded by conditions (\ref{eq_vsp_counting_points}), (\ref{eq_vsp_dimension}) or 
Lemma~\ref{lemma_connection vsp_full} is $3^7 2^3 1^5$, see e.g.~\cite{el2009partitions} for constructions in the other cases. Here, $2^1$-divisible sets of $5$ or 
$2^2$-divisible sets of $14$ points indeed exist over $\F_2$. However, in the latter case the $14$ points always form two disjoint planes, see Lemma~\ref{lemma_disjoint_union_two_subspaces} in Section~\ref{sec_classification_results}.  
Since two lines contained in the same plane have to intersect non-trivially, type $3^7 2^3 1^5$ is infeasible. The same argument also excludes the existence 
of a vector space partition of type $4^1 3^{14}2^3 1^5$ of $\PG(6,2)$, cf.~\cite[Proposition 6.4]{el2009partitions}, and can easily be generalized to:
\begin{nlemma}
  \label{lemma_vsp_disjoint_subspaces}
  Let $\cV$ be a vector space partition of type $t_k^{m_k}\dots t_1^{m_1}$ of $\PG(v-1,q)$, where $v>t_k>\dots>t_1>0$. If $1\le s<k$ is an index with $l:=\sum_{i=1}^s m_i\cdot [t_i]_q 
  /[t_{s+1}]_q\in \N_{>0}$, $l\le q$, and every $q^{t_{s+1}-1}$-divisible set of $l[t_{s+1}]_q$ is the disjoint union of $l$ $t_{s+1}$-spaces, then we have
  \begin{equation}
    \sum_{1\le i\le s\,:\, 2t_i>t_{s+1}} m_i\le l.
  \end{equation}
\end{nlemma}
Another example is the exclusion of a vector space partition of type $3^{26} 2^4 1^{10}$ of $\PG(5,3)$ since $3^2$-divisible sets of $26$ points over $\F_3$ can be partitioned 
into two disjoint planes, see Lemma~\ref{lemma_disjoint_union_two_subspaces} in Section~\ref{sec_classification_results}. 

Also other classification or characterization results for $q^r$-divisible sets of points can be used to exclude the existence of certain types of vector space partitions. In 
Exercise~\ref{exercise_empty_hyperplane} we see that for each $q^r$-divisible set $\cM$ of $q^{r+1}$ points over $\F_q$, where $r\in\N_{>0}$, there exists an empty hyperplane $H\in \cH$, i.e., 
$\cM(H)=0$. So, in particular $\supp(\cM)$ does not contain a line. As a consequence, there is no vector space partition of type $3^{17} 2^1 1^5$ of $\PG(6,2)$, cf.\  
\cite[Proposition 6.5]{el2009partitions} . By replacing an arbitrary line in a vector space partition of type $4^1 3^{13} 2^7 1^0$ of $\PG(6,2)$ by its three points we obtain  
a vector space partition of type $4^1 3^{13} 2^6 1^3$, cf.\ Definition~\ref{def_composed_vsp}. Since each $q^r$-divisible set of $[r+1]_q$ points is the characteristic 
function of an $(r+1)$-space, see Lemma~\ref{lemma_div_implies_subspace}, here the mapping also works in the other direction, i.e., a vector space partition of type 
$4^1 3^{13} 2^6 1^3$ implies the existence of a vector space partition of type 
$4^1 3^{13} 2^7 1^0$.   
\end{trailer}
\begin{nexercise}
  \label{exercise_vsp_4_16_3_1_2_1_1_5}
  Show that no vector space partition of type $4^{16} 3^1 2^1 1^5$ of $\PG(7,2)$ exists. \textit{Hint:} e.g.\ Use Lemma~\ref{lemma_div_implies_subspace} or \cite[Lemma 2]{heden2013supertail}.
\end{nexercise}
\begin{question}{Research problem}Do vector space partitions of type $4^4 3^{135}1^{18}$ or $4^3 3^{137}1^{19}$ exist in $\PG(9,2)$?
\end{question}

The previously mentioned non-existence results, except the non-existence of type $4^1 3^{13} 2^7$ that we will prove later on, and suitable constructions give the full characterization 
of vector space partitions of $\PG(v-1,2)$ for all $v\le 7$, see \cite{el2009partitions} for the details. For the vector space partitions of $\PG(7,2)$ without $1$-dimensional elements 
conditions (\ref{eq_vsp_counting_points}), (\ref{eq_vsp_dimension}) and Lemma~\ref{lemma_connection vsp_full} are sufficient except for the type $4^{13}3^6 2^6$, 
cf.\ Exercise~\ref{exercise_vsp_234_v_8}, Example~\ref{ex_no_vsp_of_type_4_13_3_6_2_6}, and \cite{el2010partitions}. 

\begin{trailer}{The {\lq\lq}tail condition{\rq\rq}}Another, very explicit, necessary  criterion for the existence of vector space partitions is the so-called 
\emph{tail condition}: 

\begin{ntheorem} {(\cite[Theorem 1]{heden2009length})}
  \label{thm_length_of_tail}
  Let $\cV$ be a vector space partition of type $t^{m_t}\dots {d_2}^{m_{d_2}}{d_1}^{m_{d_1}}$ of $\PG(v-1,q)$, 
  where $m_{d_2},m_{d_1}>0$ and $n_1=m_{d_1}$, $n_2=m_{d_2}$.
  \begin{enumerate}
    \item[(i)]   if $q^{d_2-d_1}$ does not divide $n_1$ and if $d_2<2d_1$, then $n_1\ge q^{d_1}+1$;
    \item[(ii)]  if $q^{d_2-d_1}$ does not divide $n_1$ and if $d_2\ge 2d_1$, then $n_1>2q^{d_2-d_1}$ or $d_1$ divides $d_2$ and 
                 $n_1=\left(q^{d_2}-1\right)/\left(q^{d_1}-1\right)$;
    \item[(iii)] if $q^{d_2-d_1}$ divides $n_1$ and $d_2<2d_1$, then $n_1\ge q^{d_2}-q^{d_1}+q^{d_2-d_1}$;
    \item[(iv)]  if $q^{d_2-d_1}$ divides $n_1$ and $d_2\ge 2d_1$, then $n_1\ge q^{d_2}$.
  \end{enumerate}   
\end{ntheorem}
We remark that the proof is based on so-called \emph{mixed perfect codes}, see e.g.\ \cite{heden2009length,herzog1972group} for details. In this context we would like to mention 
\cite{blokhuis1989heden}, which translates similar results obtained via mixed perfect codes into geometry. 
\end{trailer}   

The tail of a vector space partition consists of the elements of the smallest occurring dimension. This notion was generalized to the so-called \emph{supertail} 
containing all elements of the vector space partition with a dimension below a certain bound. For details and results on e.g.\ the minimum possible cardinality 
or the minimum possible number of covered points by the supertail we refer to \cite{heden2013supertail,nuastase2017structure,nuastase2018complete}. In some cases 
the structure of the minimum tails or supertails can be completely characterized, see e.g.\ \cite{heden2009length,nuastase2018complete} and  
cf.\ Exercise~\ref{exercise_vsp_4_16_3_1_2_1_1_5}. For a few  observations from the point of view of divisible codes we refer to Subsection~\ref{subsec_partitions_q_r_divisible_sets}.  

Theorem~\ref{thm_length_of_tail} was slightly improved and reformulated in \cite{kurz2018heden}.
\begin{ndefinition}(\cite[Definition 4]{kurz2018heden})
  \label{def_divisible_sets_of_points_generalization}
  Let $\cN$ be a set of pairwise disjoint $k$-subspaces in $\PG(v-1,q)$. If there exists a positive integer $r$ such that 
  \begin{equation}
    \#\left\{N\in\cN\,:\, N\le H\right\} \equiv \#\cN \pmod{q^r}
  \end{equation}
  for every hyperplane $H\in\cH$, then we call $\cN$ \emph{$q^r$-divisible}.  
\end{ndefinition}
\begin{nexercise}
  Let $\cV$ be a vector space partition of type $t^{m_t}\dots {d_2}^{m_{d_2}}{d_1}^{m_{d_1}}$ of $\PG(v-1,q)$, 
  where $m_{d_2},m_{d_1}>0$ and $n_1=m_{d_1}$, $n_2=m_{d_2}$. Show that the set $\cN$ of the $d_1$-dimensional elements of 
  $\cV$ is $q^{d_2-d_1}$-divisible.
\end{nexercise}
\begin{ntheorem}(\cite[Theorem 12]{kurz2018heden})
  \label{thm_length_of_tail_improved}
  For a non-empty $q^r$-divisible set $\cN$ of $k$-subspaces in $\PG(v-1,q)$ the following bounds on $n=\#\cN$ are tight.
  \begin{enumerate}
    \item[(i)] We have $n\ge q^k+1$ and if $r\ge k$ then either $k$ divides $r$ and $n\ge \frac{q^{k+r}-1}{q^k-1}$ or $n\ge\frac{q^{(a+2)k}-1}{q^k-1}$, 
          where $r=ak+b$ with $0<b<k$ and $a,b\in\mathbb{N}$.
    \item[(ii)] Let $q^r$ divide $n$. If $r<k$ then $n\ge q^{k+r}-q^k+q^r$ and $n\ge q^{k+r}$ otherwise. 
  \end{enumerate}  
\end{ntheorem} 
For (i) the lower bounds are attained by $k$-spreads. For (ii) the second lower bound is attained by the construction of Lemma~\ref{lemma_mrd_vsp}. 
In the other case the two-weight codes constructed in \cite[Theorem 4]{bierbrauer1997family} attain the lower bound. 
\begin{ncorollary}
  Let $\cV$ be a vector space partition of type ${d_l}^{u_l}\ldots {d_2}^{u_2}{d_1}^{u_1}$ of $\PG(v-1,q)$, where
  $u_1,u_2>0$ and $d_l>\dots>d_2>d_1\ge 1$.
  \begin{enumerate}
    \item[(i)] We have $u_1\ge q^{d_1}+1$ and if $d_2\ge 2d_1$ then either $d_1$ divides $d_2$ and $u_1\ge \frac{q^{d_2}-1}{q^{d_1}-1}$ or 
          $u_1\ge\frac{q^{(a+1)d_1}-1}{q^{d_1}-1}$, where $d_2=ad_1+b$ with $0<b<d_1$ and $a,b\in\mathbb{N}$.
    \item[(ii)] Let $q^{d_2-d_1}$ divide $u_1$. If $d_2<2d_1$ then $u_1\ge q^{d_2}-q^{d_1}+q^{d_2-d_1}$ and $u_1\ge q^{d_2}$ otherwise. 
  \end{enumerate}
\end{ncorollary}
\begin{nexample}
  \label{example_2_div_binary_lines}
  Let $\cN$ be a non-empty $2^1$-divisible set of lines in $\PG(v-1,2)$. From Theorem~\ref{thm_length_of_tail_improved}.(i) we conclude $\#\cN\ge 5$ and 
  Theorem~\ref{thm_length_of_tail_improved}.(ii) gives $\#\cN\ge 6$ if $\#\cN\equiv 0\pmod 2$. A $2$-spread of $\PG(3,2)$ and a vector space 
  partition of type $3^1 2^8$ of $\PG(4,2)$ give examples for $\#\cN\in\{5,8\}$. For $n\in\{6,7,9\}$ there exist projective $4^{1/2}$-divisible 
  codes of length $n$ over $\F_4$. Concatenation with a two-dimensional simplex code gives examples with $\#\cN\in\{6,7,9\}$. By combining these 
  examples we can attain all $\#\cN\ge 5$. 
\end{nexample} 
\begin{nexercise}
  \label{exercise_4_div_binary_lines}
  Show that non-empty $2^2$-divisible sets of lines over $\F_2$ exist iff $$\#\cN\in\{5,10,15,16,17\}\cup\N_{\ge 20}.$$ \textit{Hint:} For the constructive direction 
  consider projective $4^1$-divisible codes of suitable lengths over $\F_4$. For the non-existence results consider the possible lengths of projective $2^3$-divisible 
  binary codes. 
\end{nexercise}
\begin{trailer}{A generalization of $q^r$-divisible sets of points}Choosing $k=1$ in Definition~\ref{def_divisible_sets_of_points_generalization} we end up 
with $q^r$-divisible sets of points, so that we have some kind of a generalization for $k>1$. So, we can again ask for the sets of possible 
cardinalities depending on $k$, $q$ and $r$. Since $k$-spreads and the construction of Lemma~\ref{lemma_mrd_vsp} give examples with coprime cardinalities a 
finite Frobenius-type number and only finitely many non-feasible cardinalities exist.
\end{trailer}
\begin{question}{Research problem}Characterize the possible cardinalities of $q^r$-divisible sets of $k$-spaces over $\F_q$ for some small parameters 
$q$, $r$, and $k$, i.e.\ continue Example~\ref{example_2_div_binary_lines} and Exercise~\ref{exercise_4_div_binary_lines}.
\end{question}

Now we show the non-existence of a vector space partition of type $4^1 3^{13}2^7$ of $\PG(6,2)$. Note that a $2$-divisible set of seven lines $\cN$ over $\F_2$ indeed 
exists. However, we can deduce some information on those sets. To this end let $\cM$ be a corresponding spanning set of $21$ points in $\PG(k-1,2)$. Since 
$\#\cN>\#\left\{N\in\cN\,:\, N\le H\right\}\equiv \#\cN\equiv 1\pmod 2$ for every hyperplane $H\in\cH$, we have $\cM(H)\in\{9,13,17\}$, i.e., $\cM$ is 
$4$-divisible. However $\cM(H)=17$ is impossible, since removing five lines would give a $2$-divisible set of $2$ points over $\F_2$, which does not exist. 
With this, the corresponding standard equations are given by
\begin{eqnarray*}
  a_9+a_{13} &=& [k]_2\\
  9a_9+13a_{13} &=& 21\cdot[k-1]_2\\
  36a_9+78a_{13} &=& 210\cdot[k-2]_2
\end{eqnarray*} 
and have the unique solution $k=6$, $a_9=42$, $a_{13}=21$.\footnote{According to \cite{bouyukliev2006projective} there are exactly 
$2$ two-weight codes with these parameters having automorphism groups of order $336$ and $1008$. The latter two-weight 
code is given as Example SU2 in \cite{calderbank1986geometry}, i.e., the union of three non-intersecting planes in $\mathbb{F}_2^6$, 
with automorphism group $\operatorname{GL}\times S_3$. The other example has automorphism group $\operatorname{GL}\ltimes \mathbb{Z}_2$ 
and arises by concatenating the $[7,3,\{4,6\}]_4$ code, see \cite[Example RT1]{calderbank1986geometry}, with a two-dimensional binary simplex code. 
In both cases there are $B_3=28$ lines and the $21$ points admit several possibilities as a partition of $7$ lines.} Now consider the hyperplane $H$ 
that contains all seven lines. Since $H$ intersects solids in at least $[3]_2$ and planes in at least $[2]_2$ points, the intersection of $H$ 
with a vector space partition of type $4^1 3^{13} 2^7$ consists of at least
\begin{equation}
  1\cdot[3]_2+13\cdot[2]_2+7\cdot[2]_2=67>63=[6]_2
\end{equation}
points, which is a contradiction. We remark that a vector space partition of type $3^{70}2^7$ of $\PG(8,2)$ indeed exists, see \cite{el2008partitions}. 
A similar proof can be found in \cite[Proposition 6.2]{el2009partitions}. We can also write down equations similar to the standard equations directly 
for the elements of $\cN$ or work with the counts of different types of vector space partitions in hyperplanes, see e.g.\ \cite{heden2012survey,lehmann2012some} and the 
subsequent example.
\begin{nexercise}
  \label{exercise_2_div_set_of_six_lines}
  Show that a $2$-divisible set $\cN$ of six lines over $\F_2$ has dimension $\dim(\cN)=6$.
\end{nexercise}
\begin{nexample}
  \label{ex_no_vsp_of_type_4_13_3_6_2_6}  
  Assume that $\cV$ is a vector space partition of type $4^{13}3^62^6$ of $\PG(7,2)$. Let $4^r 3^s 2^t 1^u$ be the type of the intersection of $\cV$ with a hyperplane 
  $H\in \cH$, so that $r+s+t+u=13+6+6=25$ and $r[4]_2+s[3]_2+t[2]_2+u[1]_2=[7]_2$. Since two solids in $H\cong \PG(6,2)$ have to intersect non-trivially, we have $r\in\{0,1\}$. 
  Since there is no $2$-divisible set of $n\le 2$ points over $\F_2$, we have $u=0$ or $u\ge 3$. This gives the following possible types for $H$:
  \begin{itemize}
    \item[(a)] $4^1 3^{13}2^51^6$;\,\,\,
    (b) $4^1 3^{12}2^81^4$;\,\,\,
    (c) $4^0 3^{16}2^3 1^6$;\,\,\,
    (d) $4^0 3^{15}2^6 1^4$;\,\,\,
    (e) $4^0 3^{13}2^{12}1^0$.
  \end{itemize}  
  Let us denote their corresponding counts by $a$, $b$, $c$, $d$, and $e$, respectively. Counting the number of hyperplanes gives $a+b+c+d+e=[8]_2=255$. Counting the number 
  of solid-hyperplane incidences gives $a+b=13\cdot [4]_2=195$, so that $c+d+e=60$. From Exercise~\ref{exercise_2_div_set_of_six_lines} we know that the six lines span a 
  $6$-dimensional subspace, so that $e=[2]_2=3$,\footnote{In the original proof of \cite[Theorem 7]{el2010partitions} the estimation $e\le 7$ was used.} i.e., $c+d=57$. Counting 
  pairs of planes gives ${3\choose 2}c+d={6\choose 2}\cdot [2]_2$, i.e., $3c+d=45$, so that $c$ has to be negative.
\end{nexample}


\begin{trailer}{Generalizations of vector space partitions}The notion of a vector space partition can be generalized in several directions. A $\lambda$-fold 
vector space partition of $\PG(v-1,q)$ is a (multi-) set of subspaces such that every point $P\in\cP$ is covered exactly $\lambda$ times, see e.g.\ \cite{el2011lambda}. Here 
non-existence results for $q^r$-divisible multisets of points over $\F_q$ with point multiplicity at most $\lambda$ can be utilized, 
cf.\ Subsection~\ref{subsec_subspace_packings_coverings}. Another variant considers set of subspaces such that every $t$-subspace is covered exactly once, see 
\cite{heinlein2019generalized}. Also here divisible codes can be used for non-existence results for those vector space $t$-partitions. We remark that the upper bound $A_2(8,6;4)<289$ 
for constant-dimension codes is also implied by a non-existence result of certain vector space $2$-partitions \cite{heinlein2019generalized}. Vector space partitions of affine 
spaces have been considered in \cite{bamberg2022affine}. Another variant are multispreads \cite{krotov2022multifold}.    
\end{trailer}

\section{Partitions of $q^r$-divisible sets of points}
\label{subsec_partitions_q_r_divisible_sets}
Following up the idea of the tail, see e.g.\ Theorem~\ref{thm_length_of_tail}, and the supertail of a vector space partition in the context of divisible codes, we say that a 
set of points $\cM$ over $\F_q$ \emph{admits a partition}, or is \emph{partitionable}, \emph{of type} $k^{m_k}\dots 1^{m_1}$ if there exists a set $\cS$ of $m_i$ $i$-subspaces 
for $1\le i\le k$, such that $\cM=\sum_{S\in\cS}\chi_S$, i.e., the set of points of the elements of $\cS$ coincides with $\cM$. We are mainly interested in $q^r$-divisible partitionable 
sets of points where $r\ge k$. In this context, the non-existence of a vector space partition of type $4^1 3^{14} 2^3 1^5$ of $\PG(6,2)$ follows from the non-existence of a 
$2^2$-divisible set of points with partition type $2^3 1^5$, i.e., in general no vector space partition over $\F_2$ can end with $2^3 1^5$. The classification of $q^r$-divisible 
partition types of the form $1^{m_1}$ over $\F_q$ corresponds to the classification of the possible lengths of $q^r$-divisible sets of points over $\F_q$, see 
Section~\ref{sec_lengths_projective_q_r}. 

Let us consider $2^2$-divisible sets of points of partition type $2^{m_2}1^{m_1}$ over $\F_2$ for a moment. In Example~\ref{example_2_div_binary_lines} we have 
shown that type $2^{m_2}1^0$ is feasible iff $m_2\ge 5$ (or the trivial case $m_2=0$). Since there are no $2^1$-divisible sets of cardinality $1$ or $2$ over $\F_2$, the 
types $2^{m_2}1^1$ and $2^{m_2}1^2$ are infeasible in general. 
\begin{nexercise}
  Let $\cM$ be a $q^r$-divisible multiset of points over $\F_q$. Show that if a $k$-space $S$ is completely contained in $\supp(\cM)$, then $\cM-\chi_S$ is 
  $q^{\min\{r,k-1\}}$-divisible.
\end{nexercise}
\begin{nexercise}
  Let $0\le j\le 5$. Show that $2^2$-divisible set of points over $\F_2$ of partition type $2^{m_2}1^{3j}$ exist iff $m_2\ge 5-j$. 
\end{nexercise}
Using Lemma~\ref{lemma_picture_q_2_r_2} we can easily conclude that type $2^{m_2}1^4$ is impossible for $m_2\in\{0,2,3\}$ while type $2^11^4$ is e.g.\ attained by a
vector space partition of type $2^11^4$ of $\PG(2,2)$, so that we have constructions for all $m_2\ge 6$. For $m_2\in\{4,5\}$ it remains to be checked if the $2^2$-divisible 
sets of $16$ or $19$ points can contain sufficiently many disjoint lines. Of course this amounts to a finite computation.
\begin{nexercise}
  Show that a $2^2$-divisible set of points over $\F_2$ of partition type $2^{m_2}1^{m_1}$ exist for all $m_2\in\N_0$, $m_1\in\N_{\ge 29}$. 
  \textit{Hint:} Use Lemma~\ref{lemma_picture_q_2_r_2} and Example~\ref{example_2_div_binary_lines}.
\end{nexercise}   
\begin{question}{Research problem}Complete the classification of the possible parameters $\left(m_2,m_1\right)$ of a $2^2$-divisible set of points over $\F_2$ of partition 
type $2^{m_2}1^{m_1}$.
\end{question}
Of course, also other parameters are of interest and the general classification problem is widely open. Also the question of the representation of such results arises. Taking 
Lemma~\ref{lemma_picture_q_2_r_2} as given, we may summarize the presented knowledge on non-existence results of $2^2$-divisible sets of points over $\F_2$ of partition type $2^{m_2}1^{m_1}$  
by the forbidden types $2^1 1^5$ and $2^3 1^5$. For $3^2$-divisible sets of points over $\F_3$ of partition type $2^{m_2}1^{m_1}$ we mention that the forbidden pattern $2^4 1^{10}$ 
is implied by the forbidden pattern $2^3 1^{14}$. 

\chapter{Classification results for $q^r$-divisible codes}
\label{sec_classification_results}
Sets of points $\cM$ where each hyperplane has the same multiplicity can be easily classified using the standard equations:
\begin{nexercise}
  \label{exercise_one_weight}
  Let $\cM$ be a spanning set of points in $\PG(k-1,q)$, where $k\ge 2$, such that $\cM(H)=c\in\N_{>0}$ for every hyperplane $H\in \cH$. Show that $c=[k-1]_q$, $\#\cM=[k]_q$, $\cM$ is 
  $q^{k-1}$-divisible, and $\cM=\chi_{\cP}$, i.e., $\cM$ is the full $k$-space.   
\end{nexercise} 
As a direct implication we obtain:
\begin{nlemma}
  \label{lemma_div_implies_subspace}
  Let $\cM$ be a $q^r$-divisible set of $[r+1]_q$ points, where $r\in \N_{>0}$. Then $\cM=\chi_S$ for some $(r+1)$-space $S$, i.e., the corresponding points form an $(r+1)$-space.
\end{nlemma}
If we consider multisets of points in Exercise~\ref{exercise_one_weight} instead of sets of points, then we end up with $\lambda$-fold $k$-spaces, i.e., $\cM=\lambda\cdot\chi_S$, see 
\cite{bonisoli1984every}. So Lemma~\ref{lemma_div_implies_subspace} also applies to multisets of points. Point sets with two different hyperplane multiplicities have a very rich 
diversity, see Subsection~\ref{subsec_two_weight_codes}. However, we can generalize Lemma~\ref{lemma_div_implies_subspace} in a different direction.
\begin{nexercise}
  \label{exercise_holes_classification_two_times_r_flat}
  Let $\cM$ be a $q^r$-divisible set of $2[r+1]_q$ points over $\F_q$, where $r\in \N_{>0}$ and $(q,r)\neq (2,1)$. Show that the standard equations have  
  a unique solution corresponding to the disjoint union of two $(r+1)$-spaces, so that especially $\dim(\cM)=2r+2$ and there are $a_{2[r]_q}=\left(q^{r+1}-1\right)\cdot [r+1]_q$ 
  hyperplanes of multiplicity $2[r]_q$ and $a_{[r]_q+[r+1]_q}=2[r+1]_q$ hyperplanes of multiplicity $[r]_q+[r+1]_q$. For $q=2$ and $r\ge 2$ it may help to show first that no hyperplane has multiplicity $2^r-2$.
\end{nexercise}
We remark that over $\F_2$ a $5$-dimensional projective base gives a spanning $2$-divisible set of $6$ points in $\PG(4,2)$. Given a set of points $\cM$ as in 
Exercise~\ref{exercise_holes_classification_two_times_r_flat}, we observe $\cM(H)\in\left\{2[r]_q,[r]_q+[r+1]_q\right\}$, i.e., there are just two different hyperplane 
multiplicities. If $\cM(S)>2[r]_q$ for an $(r+1)$-space $S$, then Equation~(\ref{eq_subspace_multiplicity}) yields 
\begin{eqnarray*}
  \cM(S) &=&\frac{1}{q^{v-s-1}}\cdot\left(\sum_{H\in\cH\,:\, S\le H} \cM(H)\,-\,[v-s-1]_q\cdot \#\cM\right) \\  
  &=& \frac{1}{q^r} \cdot\big( [r+1]_q\cdot\left([r]_q+[r+1]_q\right)-[r]_q\cdot 2[r+1]_q \big)=[r+1]_q, 
\end{eqnarray*}
i.e., $S\subseteq \supp(\cM)$ so that applying Lemma~\ref{lemma_div_implies_subspace} to $\cM-\chi_S$ gives that $\cM$ is the disjoint union of two $(r+1)$-spaces. For 
$r=1$ the existence of a line $L$ with $\cM(L)>2$ can be deduced from $B_3>0$, which is satisfied for a projective $q$-divisible $[2q+2,k]_q$-code with $(q,k)\neq (2,5)$, 
so that:
\begin{nlemma}
  \label{lemma_disjoint_union_two_lines}
  For $q\ge 3$ every $q$-divisible set of $2q+2$ points over $\F_q$ is the disjoint union of two lines. 
\end{nlemma}
The large number of hyperplanes of multiplicity $2[r]_q$ concluded in Exercise~\ref{exercise_holes_classification_two_times_r_flat} can be used for an induction argument:
\begin{nexercise}
  \label{exercise_two_disjoint_subspaces_induction}
  Let $\cM$ be a $q^r$-divisible set of $2[r+1]_q$ points over $\F_q$, where $r\in\N_{\ge 2}$, such that for each hyperplane $H\in \cH$ with $\cM(H)=2[r]_q$ the restricted 
  point set $\cM|_H$ is the disjoint union of two $r$-spaces.
  \begin{enumerate}
    \item[(a)] Show that each $s$-space $S$ with $1\le s\le r-1$ and $S\subseteq \supp(\cM)$ is contained in an $(s+1)$-space $S'$ with $S'\subseteq \supp(\cM)$.
    \item[(b)] Show that each $(r-1)$-space $F$ with $F\subseteq\supp(\cM)$ is contained in two different $r$-spaces $R_1$ and $R_2$ with $R_1,R_2\subseteq \supp(\cM)$.
    \item[(c)] Show that the $(r+1)$-dimensional space $X:=\left\langle R_1,R_2\right\rangle$ satisfies $\cM(H)\neq 2[r]_q$ for each hyperplane $H\in\cH$ containing $X$.
    \item[(d)] Show that $\cM(X)=[r+1]_q$ and that $\cM$ is the disjoint union of two $(r+1)$-spaces.
  \end{enumerate}  
\end{nexercise} 
A quick computer enumeration reveals that each $2^2$-divisible set of $14$ points over $\F_2$ is indeed the disjoint union of two planes,\footnote{A computer-free proof can 
roughly run as follows. First show that a $2$-divisible set of $6$ points over $\F_2$ is either the disjoint union of two lines or a $5$-dimensional projective base that does not contain 
a full line. Let $\cM$ be a $2^2$-divisible set of $14$ points over $\F_2$. From the MacWilliams equations for the corresponding code we conclude $B_3>0$ so that there exists a 
line $L$ with $L\subseteq\supp(\cM)$. For this line $L$ we can proceed as in Exercise~\ref{exercise_two_disjoint_subspaces_induction} since every hyperplane $H$ containing $L$ 
with multiplicity $\cM(H)=6$ is the disjoint union of two lines.} so that we obtain:
\begin{nlemma}
  \label{lemma_disjoint_union_two_subspaces}
  Let $\cM$ be a $q^r$-divisible set of $2[r+1]_q$ points over $\F_q$, where $r\in \N_{>0}$ and $(q,r)\neq (2,1)$. Then, $\cM$ is the disjoint union of two $(r+1)$-spaces.
\end{nlemma}
Using results on blocking sets in $\PG(2,q)$, actually a much stronger classification result for minihypers of a certain type was proven in \cite{govaerts2003particular}.
\begin{ntheorem}{(\cite[Theorem 13]{govaerts2003particular})}
  \label{thm_union_subspaces_blocking_set}
  Let $\cM$ be a $q^r$-divisible multiset of cardinality $\delta[r+1]_q$ over $\F_q$. If $q>2$ and $1\le \delta<\varepsilon$, where $q+\varepsilon$ is the 
  size of the smallest non-trivial blocking sets in $\PG(2,q)$, then there exist $(r+1)$-spaces $S_1,\dots,S_\delta$ such that
  $$
    \cM=\sum_{i=1}^\delta \chi_{S_i},  
  $$ 
  i.e., $\cM$ is the sum of $(r+1)$-spaces.
\end{ntheorem} 
\begin{ntheorem}
  \label{thm_blocking_set_cardinality_non_trivial_pg_2_q}
  If $q+\varepsilon$ is the size of the smallest non-trivial blocking sets in $\PG(2,q)$, then
  \begin{enumerate}
    \item[(a)] $\varepsilon=(q+3)/2$ if $q>2$ is prime \cite{blokhuis1994size};
    \item[(b)] $\varepsilon=\sqrt{q}+1$ if $q$ is square \cite{bruen1970baer};
    \item[(c)] $\varepsilon\ge c_p q^{2/3}+1$, where $c_2=c_3=2^{-1/3}$ and $c_p=1$ for $p>3$, if $q=p^h$ with $h>2$ and $h\equiv 1\pmod 2$ \cite{blokhuis1999lacunary}.
  \end{enumerate}  
\end{ntheorem}
Note that Theorem~\ref{thm_union_subspaces_blocking_set} does not apply to $q=2$ and Lemma~\ref{lemma_disjoint_union_two_subspaces} applies to $q=2$ for $r\ge 2$ only. 
Moreover, Lemma~\ref{lemma_disjoint_union_two_subspaces} is tight in the sense that $2^2$-divisible sets of $21$ points over $\F_2$ that are not the union of three planes 
indeed exist. From e.g.\ \cite{projective_divisible_binary_codes} we know that the number of non-isomorphic $2^2$-divisible sets of $21$ points over $\F_2$ is given by $2$, $7$, $9$, and $6$ for dimensions 
$6$, $7$, $8$, and $9$, respectively. So, there is even a projective $[21,6,\{8,12\}]_2$ two-weight code which is not given by \cite[Example SU2]{calderbank1986geometry}, 
as there is just one isomorphism type, see:
\begin{nexercise}
  Let $\cM$ be the set of points of three pairwise disjoint $r$-spaces. Show that $2r\le \dim(\cM)\le 3r$ and that there is a unique isomorphism type for each possible dimension.
\end{nexercise}
The second two-weight code has a nice geometric description. By the Klein correspondence there exist two disjoint planes $\pi_1$, $\pi_2$ in the Klein quadric $\mathsf{Q}^+(5,q)$. 
If $\cK$ is the set of points of the Klein quadric in $\PG(5,2)$, then $\cK-\chi_{\pi_1}-\chi_{\pi_2}$ is a $2^2$-divisible set of $21$ points.
\begin{nexercise}
  \label{exercise_klein_quadric}
  Show that the points of the Klein quadric form a $2^2$-divisible set $\cK$ of $35$ points over $\F_2$. If $\pi_1$ and $\pi_2$ are two disjoint planes contained in the 
  support of $\cK$, then $\cK':=\cK-\chi_{\pi_1}-\chi_{\pi_2}$ is a $2^2$-divisible set of $21$ points that can be partitioned into $7$ lines and only attains two different 
  hyperplane multiplicities.
\end{nexercise}  

From e.g.\ \cite{projective_divisible_binary_codes} we also know that the number of non-isomorphic $2^3$-divisible sets of $45$ points over $\F_2$ is given by 
$2$, $1$, $1$, $1$, and $1$ for dimensions $8\le k\le  12$. Thus, beside the examples that arise as the disjoint union of three solids, there is a unique other projective $[45,8,\{16,24\}]_2$ 
two-weight code, which is e.g.\ described in \cite[Theorem 4.1]{haemers1999binary}.\footnote{The residual $[21,7]_2$-codes correspond to the construction directly following 
Example~\ref{example_complement_of_parabolic_quadric}.} So, Lemma~\ref{lemma_disjoint_union_two_subspaces} is also tight for $q=2$, $r=3$. However, the 
number of cases, which are not given as the union of three disjoint $(r+1)$-spaces, seem to decrease. And indeed, enumerating all projective $2^4$-divisible codes of length $93$ over 
$\F_2$, with \texttt{LinCode} \cite{bouyukliev2020computer}, yields that all examples arise as the disjoint union of three $5$-spaces.    
\begin{nexercise}
  Show that the weights of a projective $q^r$-divisible code of length $\delta[r+1]_q$ over $\F_q$ are contained in $\left\{iq^r\,:\, 1\le i\le \delta\right\}$ for all $r,\delta\in\N_{>0}$ with $\delta<q^r$ and $\delta[r]_q\le rq^r$ (e.g.\ $r\ge 2$ and $\delta\le q$, or $q=2$, $\delta=3$, and $r\ge 3$).
  \textit{Hint:} Use Theorem~\ref{thm_exclusion_q_r} for the restriction to a hyperplane.
\end{nexercise}  
\begin{nexercise}
  Show by induction that every $2^r$-divisible set of $3[r+1]_2$ points over $\F_2$ is the disjoint union of three $(r+1)$-spaces for all $r\ge 4$.
\end{nexercise}
\begin{nconjecture}
  There exists a function $f\colon \N_{>0}\to\N_{>0}$ such that every $2^r$-divisible set of $f(r)\cdot [r+1]_2$ points over $\F_2$ is the disjoint union of $f(r)$ $(r+1)$-spaces 
  and $\lim_{r\to\infty} f(r)=\infty$.
\end{nconjecture}
A few remarks on the case $q=8$ are also contained in Section~\ref{sec_extendability_results}.

Let $n$ be the cardinality of a $q^r$-divisible set of points over $\F_q$, where $r,n\in\N_{>0}$. So far we have studied the isomorphism types for $n=\delta[r+1]_q$ over $\F_q$, 
where $r$ and $\delta$ are positive integers, which includes the smallest possible cardinality attained at $\delta=1$. From Theorem~\ref{thm_exclusion_q_r} we know 
that for $n\le rq^{r+1}$ all possible values of $n$ can be written as $a[r+1]_q+bq^{r+1}$ with $a,b\in\N_0$. So, the next interesting case is cardinality $n=q^{r+1}$, which will 
be treated in the subsequent subsection. For $b\ge 1\wedge a+b\ge 2$ the situation seems to be more complicated. For $q=2$ the cases $(a,b)=(1,1)$ and $(0,2)$ correspond to 
$2^r$-divisible sets of $2^{r+2}-1$ or $2^{r+2}$ points over $\F_2$. Examples that are not the union of subspaces and affine subspaces are obtained in 
Example~\ref{example_cone_constructions} via the cone construction.
 
\begin{question}{Research problem}Classify all $2^r$-divisible sets of $2^{r+2}-1$ or $2^{r+2}$ points over $\F_2$.
\end{question}
The case of $r=2$ and cardinality $15$ is solved in \cite{honold2019classification}.

\begin{nexercise}
  Consider the $\left[2^{k-1} + l \left(2^k -1\right) , k, (2l + 1)2^{k-2}\right]_2$ code $C$, where $k\ge 2$ and  $l\ge 0$ are integers. Let  
  $\cM$ be the corresponding multiset of points. Show that $\cM(P)\in \{l,l+1\}$ for all $P\in\cP$ and that the $2^{k-1}-1$ points 
  with multiplicity $l$ form a hyperplane in $\F_2^k$. 
\end{nexercise}

\section{The (generalized) cylinder conjecture}
\label{subsec_cylinder_conjecture}
Applying the cone construction with a base $\cB$ of arbitrary $q$ points and an $r$-space as center $X$ gives a $q^r$-divisible set of $q^{r+1}$ points 
over $\F_q$, see (\ref{cone_construction_1}). For $r=1$ these sets consist of $q$ affine lines meeting in a common point, which is not part of the point set, so that 
one can speak of a \emph{cylinder}. For general $r\ge 1$ we also speak of cylinders, or more precisely \emph{$r$-cylinders} in these cases. As an abbreviation, we say that 
the cylinder conjecture is true for $(v,r,q)$ if each $q^r$-divisible set $\cM$ of $q^{r+1}$ points in $\PG(v-1,q)$ with $\dim(\cM)=v$ is a cylinder.  
The origin of the cylinder conjecture was the idea of classifying all sets of $p^2$ points in $\AG(3,p)$ determining few directions, see \cite{ball2008graph}, 
and is a continuation of similar results in $\AG(2,p)$ starting in \cite{lovasz1981remarks,redei1970luckenhafte}. The assumption on the number of directions was weakened 
to the property that every hyperplane contains $0\pmod p$ of the points in \cite{de2019cylinder}. There the authors proved the cylinder conjecture for $(4,1,2)$ and 
$(4,1,3)$. A relaxed version of the cylinder conjecture for $(4,1,p)$ was proven for all primes $p\le 13$, see \cite{de2019cylinder} for the details.
For $(4,1,p)$, i.e., for $p$-divisible sets of $p^2$ points in $\AG(3,p)$, it was shown in \cite{kiss2027cylinder} that each such set is an integer linear combination of cylinders. 
Equivalently, it is the sum of the characteristic function of a plane and an integer linear combination of differences of the characteristic functions of parallel lines.

Our first observation is that the standard equations can be used to deduce the existence of a hyperplane with multiplicity zero.
\begin{nexercise}
  \label{exercise_empty_hyperplane}
  Let $\cM$ be a $q^r$-divisible set of $q^{r+1}$ points in $\PG(v-1,q)$ with $\dim(\cM)=v$ and $r\in \N_{>0}$. Use the standard equations to show $a_0\ge \frac{q^{v-r-1}-1}{q-1}\ge 1$.
\end{nexercise}    
In other words, it makes no difference if we consider sets of points in $\AG(v-1,q)$ or $\PG(v-1,q)$. However, the (or at least a) assumption on the maximum point 
multiplicity is essential, since $q$ arbitrary points of multiplicity $q$ each also form a $q$-divisible multiset of cardinality $q^2$ that is not a cylinder unless the $q$ points 
form an affine line. If $v\le r+1$ then no set of $q^{r+1}$ points does exist in $\PG(v-1,q)$ at all. For dimension $v=r+2$ the existence of the empty hyperplane forces $\cM$ to be contained in an affine $(r+2)$-space, which has exactly $q^{r+1}$ points, so that 
$\cM$ is this affine $(r+2)$-space, so that the cylinder conjecture is true for $(r+2,r,q)$.

In \cite[Corollary 20]{kurz2020generalization} it was shown that the cylinder conjecture is true for $(v,r,q)$ iff it is true for $(v-r+1,1,q)$, i.e., it suffices to study the case $r=1$. 
Dimension $v=4$ is indeed the smallest case where things start to get non-trivial. In \cite{kurz2020generalization} the cylinder conjecture was shown to be true for 
$(4,1,q)$ for all $q\le 7$ and some partial results for $q=8$ were obtained. If the field size is not a prime and $v>r+3$ is chosen suitably, then cylinders over subfields certify that 
the cylinder conjecture is wrong for $(v,r,q)$, see~\cite{kurz2020generalization} for the details.  

To sum up, the classification of $q^r$-divisible sets of $q^{r+1}$ points is quite a challenge, while there is a precise conjecture for field sizes that are prime.
\begin{nconjecture}
  The cylinder conjecture is true for all $(v,r,p)$, where $p$ is a prime.
\end{nconjecture}
A manuscript claiming a proof of the case $(4,1,p)$ for all primes $p$, and thus by \cite[Corollary 20]{kurz2020generalization} of the whole conjecture, has recently been posted on arXiv \cite{sziklai2026strong}.
The case $(4,1,p)$ is also the subject of \cite{kiss2027cylinder}, see above, where the possible sets are described as integer linear combinations of cylinders.
We remark that e.g.\ the cylinder conjecture is wrong for $(4,1,8)$. Abbreviating the elements 
$c_0+c_1x+c_2x^2\in \F_2[x]/\left(x^3+x+1\right)$ as $c_0+2c_1+4c_2$, a generator matrix is given by:
$$ 
  \begin{pmatrix}
  1111111111111111111111111111111111111111111111111111111000001000\\
  0000000000000111111222222333333444444555555666666777777111110100\\
  1233345556777123457045666012345123555015557024567012456444570010\\
  1001261467356476012037146574610666012023561651000101163134000001\\
  \end{pmatrix}.
$$
\begin{question}{Research problem}Can this specific counter example be explained from a geometric point of view and generalized to other field sizes?\end{question}
\begin{nexercise}
  Let $\cM$ be an $8$-divisible set of $64$ points in $\PG(3,8)$ that is not a cylinder. Show $a_0=29$, $a_8=528$, and $a_{16}=28$ for the spectrum. Moreover, the 
  total number $b_i$ of $i$-lines is given by $b_0=1753$, $b_1=1536$, $b_2=1344$, $b_4=112$ and in a $16$-plane the distribution has to be $b_0=13$, $b_2=48$, $b_4=12$.
\end{nexercise}

\chapter{Extendability results}
\label{sec_extendability_results}
$t$-spreads in $\PG(st-1,q)$ exists for all $s\in\N_{\ge 2}$, $t \in\N_{\ge 1}$, see Section~\ref{sec_partial_spreads}. If a partial $t$-spread in $\PG(st-1,q)$ has cardinality 
$[st]_q/[t]_q-\delta$ then we say that it has \emph{deficiency} $\delta$. The corresponding set of holes, i.e., the set of $\delta[t]_q$ uncovered points, is $q^{t-1}$-divisible,   
so that the results stated in Section~\ref{sec_classification_results} can be used to show the extendability to a $t$-spread. More precisely, if $\delta$ is small enough such that 
every $q^{t-1}$-divisible set of $\delta[t]_q$ points is the disjoint union of $\delta$ $t$-spaces, then each partial $t$-spread in $\PG(st-1,q)$ with deficiency $\delta$ can be 
extended to a $t$-spread. For the other direction, e.g.\ the existence of a maximal partial line spread of size $45$ in $\PG(3,7)$, see \cite{heden2001maximal}, shows the existence 
of a $7$-divisible set of $40$ points over $\F_7$ that is not the disjoint union of five lines.

\medskip

\begin{warning}{The non-existence of maximal partial $t$-spreads does not necessarily imply classification results for $\mathbf{q^{t-1}}$-divisible sets of points.} 
Note that in general the non-existence of a maximal partial $t$-spread in $\PG(st-1,q)$ of deficiency $\delta$ does not imply that every $q^{t-1}$-divisible set of 
$\delta[t]_q$ points over $\F_q$ contains a $t$-space in its support, see Example~\ref{example_excluded_hole_set}. 
\end{warning}
\begin{nexercise}
  Consider the non-existence proof of maximal partial line spreads of deficiency $5$ and $6$ in $\PG(3,8)$ given in \cite{barat2004minimal}. Do the details of the proof 
  imply that each $8$-divisible set of $9\delta$ points over $\F_8$ is the disjoint union of $\delta$ lines for all $\delta\le 6$? (A maximal partial line spread of 
  deficiency $7$ is indeed known.) 
\end{nexercise}

In principle we can ask the question of extendability also for partial $t$-spreads in $\PG(v-1,q)$ where the dimension $v$ of the ambient space is not divisible by $t$. More generally,
we consider a vector space partition $\cV$ of type $t^{m_t}\dots s^{m_s}1^{m_1}$ of $\PG(v-1,q)$, see Section~\ref{sec_vector_space_partitions}. Due to Lemma~\ref{lemma_connection vsp} 
the set of holes $\cN$, i.e., the set of $1$-dimensional elements, is $q^{s-1}$-divisible. We call $\cV$ \emph{$k$-extendable} if the support of $\cN$ contains a full $k$-space and 
\emph{extendable} if it is $k$-extendable for some $k\ge 2$. As an example we refer to a hypothetical vector space partition $\cV$ of type $4^1 3^{13} 2^6 1^3$ in $\PG(6,2)$ discussed in 
Section~\ref{sec_vector_space_partitions}. It would be $2$-extendable. However, the non-existence of a vector space partition of type $4^1 3^{13} 2^7$ implies the non-existence 
of $\cV$. So, the question arises for which cardinalities $n$ every $q^r$-divisible set of $n$ points over $\F_q$ contains a $k$-space in its support. Certainly, the most restricted and 
interesting case is $k=r+1$. In this context we mention Example~\ref{example_cone_constructions} and the construction of $2^r$-divisible sets of $2^{r+2}-1$ points over $\F_2$ not 
containing an $(r+1)$-space in its support. So in principle, maximal partial $(r+1)$-spreads in $\PG(2r+2,2)$ with size one less than the maximum possible cardinality $A_2(2r+3,2r+2;r+1)$, 
see \cite[Theorem 4.1]{beutelspacher1975partial} and Proposition~\ref{prop_partial_spread_lower_bound}, may exist. They do indeed exist for $r=2$ as shown in 
\cite{honold2019classification}. What about $r>2$ or general field sizes $q>2$?    

As partial spreads are just a special case of constant-dimension codes, see Subsection~\ref{subsec_subspace_codes}, one may wonder whether results on 
the structure of divisible multisets of points can be used to deduce extendability results for constant-dimension codes. To our knowledge, the first extendability results for 
constant-dimension codes, that are not partial spreads, was shown in \cite{nakic2016extendability}.
\begin{ntheorem}{(\cite[Theorem 4.2]{nakic2016extendability})}
  Let $\cC$ be a set of $\qbin{v}{t}{q}/\qbin{k}{t}{q}-\delta$ $k$-spaces in $\PG(v-1,q)$ such that every $t$-space is contained in at most one element of $\cC$, where $1< t<k<v$. 
  If $v-i\equiv 0\pmod{k-i}$ for $i=0,1,\dots,t-1$ and $\delta\le (q+1)/2$, then $\cC$ can be extended by $\delta$ $k$-spaces without destroying the property on the covering of 
  the $t$-spaces. 
\end{ntheorem}
\begin{ncorollary}{(\cite[Corollary 4.3]{nakic2016extendability})}
  \label{cor_nakic_storme}
  Let $1< t<k<v$ be integers with $v-i\equiv 0\pmod{k-i}$ for $i=0,1,\dots,t-1$. Then, either $A_q(v,2k-2t+2;k)=\qbin{v}{t}{q}/\qbin{k}{t}{q}$ or 
  $A_q(v,2k-2t+2;k)<\qbin{v}{t}{q}/\qbin{k}{t}{q}-(q+1)/2$. 
\end{ncorollary}
\begin{trailer}{A $2$-analog of the Fano plane}Let $\cC$ be a set of planes in $\PG(6,2)$ such that every line is covered at most once. What is the maximum size $A_2(7,4;3)$ of $\cC$? If every line would be covered exactly 
once, then we would have $\#\cC=\qbin{7}{2}{2}/\qbin{3}{2}{2}=381$ and $\cC$ would be called a \emph{$2$-analog of the Fano plane}.\footnote{A Fano plane is a configuration of seven 
$3$-element subsets $\cB$ of a $7$-set $V$ such that every $2$-subset of $V$ is contained in exactly one element $B\in\cB$.} We show $A_2(7,4;3)=381$ 
or $A_2(7,4;3)\le 379$ directly. Corollary~\ref{cor_nakic_storme} does not apply as stated, since it requires $v\equiv 0\pmod k$, which is violated for $v=7$ and $k=3$.\footnote{For $t=2$ the proof in \cite{nakic2016extendability} uses $v\equiv 0\pmod k$ only to guarantee that $\qbin{v}{2}{q}/\qbin{k}{2}{q}$ is an integer.} So, assume $\#\cC=380$ for a moment. Double-counting lines yields that exactly seven lines $L_1,\dots,L_7$ of $\PG(6,2)$ are uncovered by the elements of $\cC$. 
From Lemma~\ref{lem:union_subspaces} we know that the multiset of points $\cM$ of all points of the elements of $\cC$ is $2^2$-divisible. Let $\cC_P:=\{C\in\cC\,:\,P\le C\}$ denote 
the elements of $\cC$ that contain an arbitrary but fixed point $P\in\cP$. Moding out $P$ from $\cC_P$ yields a (partial) line spread in $\PG(5,2)$, so that $\#\cC_P\le [6]_2/[2]_2=21$. 
Thus, the $21$-complement $\overline{\cM}:=\cM^{\complement_{21}}$ of $\cM$ is a $2^2$-divisible multiset of points with cardinality $7$ in $\PG(6,2)$ by Lemma~\ref{lemma_t_complement}. 
Moreover, $\overline{\cM}=\chi_\pi$ for a plane $\pi$, see e.g.\ \cite{bonisoli1984every}, so that $\cC\cup\{\pi\}$ covers each point of $\PG(6,2)$ exactly $21$ times. 
In principle, an element $C\in\cC$ with $\dim(\pi\cap C)\ge 2$ might exist. However, the seven uncovered lines partition $3\cdot\overline{\cM}$, i.e.\ $3\cdot\overline{\cM}=
\sum_{i=1}^7 \chi_{L_i}$, so that $\dim(\pi\cap C)\le 1$ for all $C\in\cC$ and $\cC\cup\{\pi\}$ is a $2$-analog of the Fano plane. Thus, $\#\cC=380$ implies $A_2(7,4;3)=381$, i.e., $A_2(7,4;3)=381$ or $A_2(7,4;3)\le 379$ since $A_2(7,4;3)\le 381$. We will slightly tighten the {\lq\lq}gap{\rq\rq} result in a moment.

The currently best lower bound is $A_2(7,4;3)\ge 333$ \cite{paper333} and if $A_2(7,4;3)=381$ then a matching code can have an automorphism group of order at most 
$2$ \cite{kiermaier2018order}.
\end{trailer}

In the following paragraph we want to generalize the idea of using classification results for divisible multisets of points to show that either $A_2(7,4;3)=381$ or 
$A_2(7,4;3)\le 378$. After that example we give a general problem statement in Definition~\ref{def_delta_subsapces_additional_conditions}. Note that we actually have not 
used the information on the line covering of $3\cdot\overline{\cM}$ for its classification. 

Let $\cM$ be a $k$-dimensional $2^2$-divisible multiset of cardinality $14$ in $\PG(k-1,2)$ such that $14$ lines 
$L_1,\dots,L_{14}$ exist with $3\cM=\sum_{i=1}^{14} \chi_{L_i}$. Here the latter condition will be essential, since e.g.\ $2\cdot \cB$ for a $6$-dimensional projective 
base $\cB$ over $\F_2$ is a $2^2$-divisible multiset of points in $\PG(5,2)$ with cardinality $14$ not fitting our subsequent classification result, see 
Lemma~\ref{lemma_two_planes_extra_condition}. Counting points gives that each hyperplane $H\in\cH$ contains exactly 
$(3\cdot\cM(H)-14)/2$ out of the $14$ lines, so that $\cM(H)\in\{6,10\}$. With this, the standard equations give $a_6=3\cdot 2^{k-2}+1$, $a_{10}=2^{k-2}-2$, and 
$2^{6-k}-1=\sum_{i\ge 2} {i\choose 2}\lambda_i$. Now let $H_6$ be an arbitrary hyperplane with multiplicity $6$. Since $\cM|_H$ is $2$-divisible and 
$2$ out of the $14$ lines are contained in $H$, there exists a line $L$ in the support of $\cM|_H$. Since $\cM|_H-\chi_L$ is a $2$-divisible multiset of cardinality $3$ over $\F_2$, 
$\cM|_H$ is the sum of two lines $L$, $L'$, i.e.\ $\cM|_H=\chi_L+\chi_{L'}$. Now let $P$ be an arbitrary point with positive multiplicity. Since $P$ is contained in 
$2^{k-1}-1$ hyperplanes, $P$ is contained in a hyperplane $H$ of multiplicity $6$, so that a line $L$ with $P\le L$ and $L\subseteq\supp(\cM)$ exists. Since $L$ is 
contained in $2^{k-2}-1$ hyperplanes there are at least
$$
  2^{k-1}-1-\left(2^{k-2}-1\right)-a_{10}=2
$$
hyperplanes of multiplicity $6$ that contain $P$ but not $L$, so that there exists another line $L'\neq L$ with $P\le L'$ and $L'\subseteq\supp(\cM)$. Now let $E=E(P)$ be the 
plane spanned by $L$ and $L'$. If $\cM(P)=1$, then $E$ cannot be contained in a hyperplane of multiplicity $6$ due to their classification as the sum of two lines. Thus, every hyperplane 
$H$ through $E$ has multiplicity $\cM(H)=10$, so that counting points gives $\cM(E)=6+2^{6-k}$. Since $2^{k-2}-2=a_{10}\ge 0$ and $2^{6-k}-1=\sum_{i\ge 2} {i\choose 2}\lambda_i\ge 0$,  
we have $3\le k\le 6$, so that $\lambda_2+3\lambda_3+6\lambda_4=2^{6-k}-1\le 7$ and $\lambda_i=0$ for $i\ge 5$. If $\lambda_1=0$, then $k=3$ and $\lambda_2=7$, i.e., $\cM=2\cdot\chi_\pi$ 
for some plane $\pi$. If $k=6$, then $\lambda_1=14$, and we can apply Lemma~\ref{lemma_disjoint_union_two_subspaces} to deduce that $\cM$ is the sum of two planes. If $\lambda_1>0$, then $k\neq 3$, since for $k=3$ we have $a_{10}=0$, so that each of the seven lines of $\PG(2,2)$ has multiplicity $6$ and counting the points on the three lines through a point $P$ gives $14+2\cdot\cM(P)=18$, i.e., $\cM(P)=2$ for all points $P$. Thus, for $\lambda_1>0$
we can choose a point $P$ with multiplicity $\cM(P)=1$ and construct the plane $E(P)$ as described above. For $k=5$ we conclude $\cM(E(P))=8$, $\lambda_2=1$, and $\lambda_i=0$ for $i\ge 3$, 
so that $E(P)\subseteq\supp(\cM)$. For $k=4$ we conclude $\cM(E(P))=10$ and $\lambda_2+3\lambda_3+6\lambda_4=3$, so that $\lambda_4=0$. If $\lambda_3=1$, then $\cM(E(P))\le 3+6=9$, so that $\lambda_2=3$, $\lambda_3=0$, and $E(P)\subseteq\supp(\cM)$. 
So, in both remaining cases $k\in\{4,5\}$ the plane $E(P)$ is contained in the support of $\cM$ and $\cM-\chi_{E(P)}$ is a $2^2$-divisible multiset of points of cardinality $7$ over $\F_2$. 
Thus, we have:
\begin{nlemma}
  \label{lemma_two_planes_extra_condition}
  Let $\cM$ be a $2^2$-divisible multiset of points of cardinality $14$ over $\F_2$. If there exist $14$ lines $L_1,\dots,L_{14}$ such that $3\cdot \cM=\sum_{i=1}^{14}\chi_{L_i}$, then 
  $\cM$ is the sum of two planes.
\end{nlemma}  
\begin{nexercise}
  Use Lemma~\ref{lemma_two_planes_extra_condition} to show that either $A_2(7,4;3)=381$ or $A_2(7,4;3)\le 378$. 
\end{nexercise}

\begin{ndefinition}
  \label{def_delta_subsapces_additional_conditions}
  For integers $1\le t\le r$ we denote by $m_q(r,t)$ the smallest number $\delta$ such that there exists a $q^r$-divisible multiset of points $\cM$ over $\F_q$ with cardinality 
  $\delta[r+1]_q$ that is not the union of $\delta$ $(r+1)$-spaces but where for each $1\le j\le t$ there exist $\delta\qbin{r+1}{j}{q}$ $j$-spaces $S_1^j$, $S_2^j$, \dots such that
  \begin{equation}
    \qbin{r}{j-1}{q}\cdot \cM=\sum_{i=1}^{\delta\qbin{r+1}{j}{q}} \chi_{S_i^j}.
  \end{equation}
\end{ndefinition}
In Lemma~\ref{lemma_two_planes_extra_condition} we have shown $m_2(2,2)\ge 3$ and the two-weight code in Exercise~\ref{exercise_klein_quadric}, which can be partitioned into seven lines, 
yields $m_2(2,2)\le 3$, so that $m_2(2,2)=3$. 

\begin{nexercise}
  Show that $m_2(3,2)\ge 3$.
\end{nexercise}
\begin{nexercise}
  Show that the $[45,8,\{16,24\}]_2$ two-weight code described in \cite[Theorem 4.1]{haemers1999binary} can be partitioned into $15$ lines. 
\end{nexercise}
So, we have $m_2(3,2)=3$.

\medskip
 
\begin{trailer}{An application for partial MRD codes}For two $m\times n$-matrices $A$ and $B$ the \emph{rank distance} 
is given by the rank $\operatorname{rk}(A-B)$ of their difference. A set $\cM$ of such matrices over $\F_q$ with minimum rank distance $d$ is called a \emph{maximum rank distance} 
(MRD) code if it has the maximum possible size $q^{\max\{m,n\}\cdot (\min\{m,n\}-d+1)}$ and those codes indeed exist for all parameters. Considering the row spans of 
the matrices $(I|M)$ for all $M\in\cM$, where a $m\times m$ identity matrix was put in front of $M$, gives a constant-dimension code $\cC$, called \emph{lifted MRD code}, of  
$m$-spaces such that there exists an $n$-space $S$ which is disjoint to the elements of $\cC$. In the 
remaining part we choose the specific parameters $q=2$, $m=n=4$, and $d=3$. Here we have $\#\cC=256$, every line with trivial intersection with the special solid $S$ is covered 
exactly once by the elements of $\cC$, and each point $P$ not contained in $S$ is contained in exactly $16$ elements from $\cC$. Now assume that $\cC$ satisfies the same conditions 
as before but does not have the maximum possible cardinality, so that we speak of a lifted partial MRD code. If $\cM$ is the multisets of points of the elements of $\cC$, then 
we can apply our result $m_2(3,2)\ge 3$ to the $16$-complement of $\cM+16\cdot \chi_S$ in order to conclude that for $\#\cC\in\{256-1,256-2\}$ an extension to a lifted MRD code 
exists, cf.\ \cite[Proposition 6]{heinlein2018binary}. Clearly, also the analog statement for the partial MRD code holds. 
\end{trailer}
\begin{nexercise} 
  Let $\cM$ be an $m\times n$ rank-metric code, where $m\le n$, over $\F_q$ with minimum rank distance $d$ whose cardinality is $\delta$ less than that of an MRD code with the same parameters. 
  Show that if $\delta<m_q(m-1,m+1-d)$, then $\cM$ is extendable to an MRD code. 
\end{nexercise}

\begin{question}{Research problem}Does there exist a set of $256-3=253$ $4\times 4$-matrices over $\F_2$ with minimum rank distance $3$ that is maximal, i.e., where no further 
matrix can be added without decreasing the minimum rank distance.
\end{question}
\enlargethispage{2\baselineskip}
\begin{backgroundinformation}{Covering radius}The covering radius $\rho(C)$ of a code $C$ in $V$ is the smallest integer $r$ such that every element of $V$ has a distance of at most 
$r$ to an element of $C$. If the covering radius is larger or equal to the minimum distance of $C$, then $C$ is extendable. For rank-metric codes results on the covering radius can 
be found in \cite{byrne2017covering}.   
\end{backgroundinformation}
Another application of $m_2(2,2)=3$ is that each set of $93\cdot 3-2$ planes in $\PG(5,2)$ such that every line is covered at most thrice can be completed to a 
corresponding $2$-(subspace) design with $\lambda=3$ over $\F_2$. 

\chapter{Dimensions of divisible codes}
\label{sec_dimensions}
The dimension $k$ of a linear code of length $n$ over $\F_q$ can be at most $n$, which is attained by an identity matrix as generator matrix. An even binary linear code of length $n$ 
can have dimension at most $n-1$, which is attained by the codes consisting of all even weight codewords, i.e., projective bases. Higher divisibility implies tighter bounds. 
A doubly-even binary linear code of length $n$ has dimension at most
\begin{equation}
  \begin{array}{lcl}
    4 \left\lfloor n/8\right\rfloor &:& \rem(n,8)\in \{0,1,2,3\},\\
    4 \left\lfloor n/8\right\rfloor+1 &:& \rem(n,8)\in \{4,5\},\\ 
    4 \left\lfloor n/8\right\rfloor+\rem(n,8)-4 &:& \rem(n,8)\in \{6,7\},
  \end{array}
\end{equation}
where $\rem(n,a)$ denotes the remainder of $n$ divided by $a$. Equality can indeed be attained, see e.g.\ \cite[Section VIII]{gaborit1996mass}. In general the 
dimension is upper bounded by $n/2$ and equality is attained for self-dual codes only, where especially $n$ is divisible by $8$. For $2^3$-divisible binary linear 
codes with effective length $n$ the dimension is at most $5n/16$, see \cite{liu2010binary} for the details.

When using the linear programming method to exclude the existence of certain lengths of divisible codes, upper bounds on the dimension are certainly useful. With 
respect to lower bounds we remark that a single codeword of weight $\Delta$ always generates a $1$-dimensional code. If the maximum point (or column) multiplicity 
is at most $\gamma$ then $k$ clearly has to be at least as large so that $\gamma[k]_q\ge n$, where $n$ is the effective length. In \cite{honold2019lengths} a general tool 
was used to compute an upper bound on the minimal dimension of a projective binary linear code $C$ of length $n$. To this end let $\cM$ be the corresponding set of points 
in $V:=\PG(k-1,2)$. Let $Q$ be a point not contained in $\cM$, i.e., $\cM(Q)=0$. Consider the projection of $\cM$ modulo $Q$, that is the multiset image of $\cM$ under the 
map $V \to V/Q$, $\bv \mapsto (\bv+Q) / Q$.  The resulting multiset $\cM'$ in $\PG(V/Q) \cong \PG(k-2,2)$ arises by identifying points of $\cM$ on the same 
line through $Q$.  The corresponding linear code $C'$ is a subcode of $C$ of effective length $n$ and dimension $k-1$.  If $C$ is $\Delta$-divisible, so is $C'$. The 
assumed minimality of $k$ implies that $C'$ is not projective. Equivalently, there is a secant through $Q$, that is a line whose remaining two points are contained in $\cM$. 
So each of the $[k]_2-n$ points of $V$ not contained in $\cM$ lies on a secant.  Since $\cM$ admits at most $\binom{\#\cM}{2} = {{n}\choose {2}}$ secants, covering at 
most ${n\choose 2}$ different points not in $\cM$, we get
\begin{equation}
  [k]_2-n\le {n\choose 2}\quad\Longleftrightarrow\quad   2^k \le \frac{n^2+n+2}{2}.
\end{equation}
In \cite{honold2019lengths} this inequality was used to conclude the non-existence of a projective $2^3$-divisible binary linear code of length $59$ from the non-existence 
of projective $2^3$-divisible $[59,\le 10]_2$-codes.  

\begin{trailer}{The divisible code bound}Let $q=p^f$ and $v_p$ be the \emph{$p$-adic valuation on $\Z$}, i.e., $v_p(x)$ is the exponent of the highest power of $p$ dividing 
$x$, with $v_p(0) = \infty$.
\begin{ntheorem}{(\cite{ward1992bound}, \cite[Theorem 7]{ward2001divisiblejcta}, \cite[Theorem 6]{ward2001divisible})}\label{thm_div_code_bound} Let $C$ be 
  an $[n,k]_q$-code whose non-zero word weights lie in the sequence $(b-m+1)\Delta$, \dots, $b\Delta$ of $m$ consecutive multiples of $\Delta$. Then
  \begin{equation}
    kv_p(q)\le m\big(v_p(\Delta)+v_p(q)\big) +v_p\!\left({b\choose m}\right).
  \end{equation}   
\end{ntheorem}
\end{trailer}
\begin{nexample}
  Let $C$ be a projective $2^3$-divisible binary linear code of length $59$. Since the residual code of each codeword $\bc$ is a projective $2^2$-divisible binary linear code 
  of length $59-\wt(\bc)$, the non-zero weights of $C$ are contained in $\{8,16,24,32,40\}$. Theorem~\ref{thm_div_code_bound} with $b=5$, $m=5$, and $\Delta=8$ gives 
  $k\le 20$.  
\end{nexample}
Improvements of the divisible code bound can be found in \cite{phd_liu,liu2006weights}.
\begin{nexercise}
  Show that each $3^r$-divisible $\left[3^{r+1} ,k\right]_3$-code satisfies $k\le 3r+3$ for all $r\in \N_{>0}$. Additionally assume that there is no codeword of weight $3^r$ to deduce 
  $k\le 2r+3$. 
\end{nexercise}
In \cite{kurz2020generalization} it was shown that each projective $3^r$-divisible $\left[3^{r+1} ,k\right]_3$-code satisfies $k\le r+3$ for all $r\in\N_{>0}$.

\begin{question}{Research problem}Improve the divisible code bound for projective $q^r$-divisible codes over $\F_q$.\end{question}

\begin{nexercise}
  Let $C$ be a binary $2^r$-divisible code of effective length $n$ that is spanned by codewords of weight $2^r$, where $r\in\N_{\ge 2}$. 
  Show that the dimension $k$ of $C$ satisfies $k\le n\cdot (r+2)/2^{r+1}$ and that equality can be attained if $n$ is divisible by $2^{r+1}$.\\ 
  \textit{Hint:} Use the classification in Section~\ref{sec_improved_lp_method}.  
\end{nexercise}

\begin{nconjecture}
  Let $r$ be an integer and $\eta_q(r,k)$ be the minimum possible length $n$ of a $q^r$-divisible $[n,k]_q$-code. For $r\ge 2$ we have $\eta_2(r,k)=2^{r-k+1}\cdot [k]_2$ 
  for $k\le r+1$ and $\eta_2(r,k)=\eta_2(r,k-r-2)+2^{r+1}$ for $k\ge r+2$. 
\end{nconjecture}

\chapter{Enhancing the linear programming method with additional constraints}
\label{sec_improved_lp_method}
Here we want to continue our discussion of the linear programming method from Subsection~\ref{subsec_lp_method} and discuss a few additional conditions. First we note that the 
number of even-weight codewords of an $[n,k]_2$-code can just take one of two possible values, i.e.,
\begin{equation}
  \label{eq_even_weight_subcode}
  \sum_{i=0}^{\left\lfloor n/2\right\rfloor} A_{2i} \in\left\{2^{k-1},2^k\right\}.
\end{equation}
\begin{nexercise}
  \label{exercise_even_weight_subcode}
  Let $C$ be an $[n,k]_2$-code. Show that the set of codewords of even weight forms a subcode of dimension at least $k-1$.
\end{nexercise}
\begin{nexample}
  We can use Equation~(\ref{eq_even_weight_subcode}) in order to e.g.\ show that each $[\le 16,4,7]_2$ code contains at least one codeword of weight $8$, cf.~\cite[Lemma 3.1]{kurz202146}.  
  Assume that $C$ is an $[n,4,7]_2$ code with $n\le 16$ and $A_8=0$. From the first two MacWilliams equations we conclude
  $$
    A_7+A_9+\sum_{i\ge 10} A_i = 2^4-1=15\quad\text{and}\quad 
    7A_7+9A_9+\sum_{i\ge 10} iA_i = 2^3n =8n, 
  $$  
  so that
  $$
    2A_9+3A_{10}+\sum_{i\ge 11} (i-7)A_i = 8n-105.
  $$
  Thus, the number of even weight codewords is at most $8n/3-34$. Since at least half of the codewords have to be of even weight, we obtain 
  $n\ge \left\lceil 15.75\right\rceil=16$. In the remaining case $n=16$ we use the linear programming method with the first four MacWilliams identities, 
  $A_8=0$, $B_1=0$, and the fact that there are exactly $8$ even weight codewords to conclude $A_{11}+\sum_{i\ge 13} A_i <1$, i.e., $A_{11}=0$ and $A_i=0$ for all $i\ge 13$. 
  With this and rounding to integers we obtain the bounds $5\le B_2\le 6$, which then gives the unique solution $A_7=8$, $A_9=0$, $A_{10}=6$, and $A_{12}=1$. Computing 
  the full dual weight distribution unveils $B_{15}=-2$, which is negative. 
\end{nexample}
The subcode in Exercise~\ref{exercise_even_weight_subcode} is also called \emph{even weight subcode} and its dimension equals $k$ iff $C$ is even itself. We have the 
following generalization, see \cite[Section IV]{brouwer1993linear}:
\begin{nproposition}
  \label{prop_num_4_div_codewords}
  Let $C$ be an even $[n,k]_2$-code and $t$ be the maximum dimension of a doubly-even subcode. Then, for the set $D$ of codewords of $C$ whose Hamming weight is 
  divisible by $4$ we have
  \begin{equation}
    |D|=\sum_{i=0}^{\left\lfloor n/4\right\rfloor} A_{4i}\in\left\{ 2^{k-1}-2^t,2^{k-1}, 2^{k-1}+2^{t-1},2^k\right\}\!.
  \end{equation}
\end{nproposition}
In the context of linear codes with maximum possible minimum distance it suffices to consider even codes, so that Proposition~\ref{prop_num_4_div_codewords} gives an extra 
condition for the linear programming method. In the context of (binary) divisible codes we commonly have even higher divisibility constants and \cite[Theorem 2]{brouwer1993linear}
states that the number of codewords with weight divisible by $2^a$ of a $2^{a-1}$-divisible binary linear code $C$ is at least $|C|/2^a$. This bound was e.g.\ used in \cite{brouwer1993linear} 
in order to show the non-existence of a $[124,9,60]_2$-code. We have the following refinement and generalization of Proposition~\ref{prop_num_4_div_codewords}:   
\begin{nproposition}{(\cite[Proposition 5]{dodunekov1999some}, see also \cite{simonis1994restrictions})}
  \label{prop_div_one_more}
  Let $C$ be an $[n,k,d]_2$-code with all weights divisible by $\Delta:=2^a$ and let $\left(A_i\right)_{i=0,1,\dots,n}$ be the weight distribution of $C$. Put
  \begin{eqnarray*}
    \alpha&:=&\min\{k-a-1,a+1\},\\
    \beta&:=&\lfloor(k-a+1)/2\rfloor,\text{ and}\\ 
    \delta&:=&\min\{2\Delta i\,\mid\,A_{2\Delta i }\neq 0\wedge i>0\}.
  \end{eqnarray*}
  Then the integer 
  $$
    T:=\sum_{i=0}^{\lfloor n/(2\Delta)\rfloor} A_{2\Delta i}
  $$  
  satisfies the following conditions.
  \begin{enumerate}
    \item[(i)] \label{div_one_more_case1}
          $T$ is divisible by $2^{\lfloor(k-1)/(a+1)\rfloor}$.
    \item[(ii)] \label{div_one_more_case2}
          If $T<2^{k-a}$, then
          $$
            T=2^{k-a}-2^{k-a-t}
          $$
          for some integer $t$ satisfying $1\le t\le \max\{\alpha,\beta\}$. Moreover, if $t>\beta$, then $C$ has an $[n,k-a-2,\delta]_2$-subcode and if $t\le \beta$, it has an 
          $[n,k-a-t,\delta]_2$-subcode.
    \item[(iii)] \label{div_one_more_case3}
          If $T>2^k-2^{k-a}$, then
          $$
            T=2^k-2^{k-a}+2^{k-a-t}
          $$        
          for some integer $t$ satisfying $0\le t\le \max\{\alpha,\beta\}$. Moreover, if $a=1$, then $C$ has an $[n,k-t,\delta]_2$-subcode. If $a>1$, then $C$ has an 
          $[n,k-1,\delta]_2$-subcode unless $t=a+1\le k-a-1$, in which case it has an $[n,k-2,\delta]_2$-subcode.
  \end{enumerate}
\end{nproposition}
\begin{nexample}
  \label{ex_no_32_10_8_16_24_2_code}
  \begin{sloppypar}
  An implication of Proposition~\ref{prop_div_one_more} is that no projective $[32,10,\{8,16,24\}]_2$-code exists, see \cite{kiermaier2020strongly} for the context and an application.  
  From the first three MacWilliams equations we compute $A_8 = 61$, $A_{16} = 899$, and $A_{24} = 63$. Applying Proposition~\ref{prop_div_one_more} with $a=3$ 
  gives $\Delta = 8$, $\alpha = 4$, $\beta= 4$, $\delta= 16$, and $T = 900$. As required by Part (i), $T$ is divisible by $4$. However, Part (iii) gives $t=5$, which 
  contradicts $0 \le t \le \max\{\alpha, \beta\}=4$, so that such a code cannot exist.
  \end{sloppypar}
\end{nexample}

The general idea behind Proposition~\ref{prop_num_4_div_codewords} is to consider $\sum_{i\in I} A_i$, for some subset $I\subseteq \{1,\dots, n\}$, as weights of codewords 
in (generalized) Reed-Muller codes, see \cite{phd_guritman,simonis1994restrictions} for the details. It is well known that the occurring weights of generalized Reed-Muller codes have 
some gaps, e.g.:
\begin{nproposition}{(\cite{mceliece1969quadratic}, see also \cite{li2019weight})}
Let $C$ be a second order $q$-ary generalized Reed-Muller code of length $q^k$. Then, all non-zero weights of $C$ are of the form
\begin{equation}
  q^k-q^{k-1}-\nu q^{k-1-j},
\end{equation}
where $\nu\in\{0,\pm 1,\pm(q-1)\}$ and $0\le j\le \left\lfloor k/2\right\rfloor$.
\end{nproposition}
For more such {\lq\lq}gap{\rq\rq} results we refer to e.g.\ \cite{phd_guritman}. Results similar to Proposition~\ref{prop_div_one_more} 
for field sizes $q\in\{3,4\}$ were used in e.g.\ \cite{phd_guritman,guritman2001degree,guritman2000nonexistence,guritman2002restrictions}.  

\begin{trailer}{$\Delta$-divisible codes spanned by codewords of weight $\Delta$}The characterization of indecomposable self-orthogonal binary codes which are spanned by codewords 
of weight $4$ from \cite[Theorem 6.5]{pless1975classification} was generalized in \cite[Theorem 1]{kiermaier2020classification}:
\begin{ntheorem}
	\label{thm_delta_div_spaned_by_min}
	Let $\Delta$ be a positive integer and let $a$ be the largest integer such that $q^a$ divides $\Delta$.
	Let $C$ be a $q$-ary $\Delta$-divisible linear code that is spanned by codewords of weight $\Delta$.
	Then $C$ is isomorphic to the direct sum of codes of the following form, possibly extended by zero positions:
	\begin{enumerate}
		\item[(i)] The $\frac{\Delta}{q^{k-1}}$-fold repetition of the $q$-ary simplex code of dimension $k\in\{1,\ldots,a+1\}$.
	\end{enumerate}
	In the binary case $q=2$ additionally:
	\begin{enumerate}
		\item[(ii)] The $\frac{\Delta}{2^{k-2}}$-fold repetition of the binary first order Reed-Muller code of dimension $k\in\{3,\ldots,a+2\}$.
		\item[(iii)] For $a \geq 1$: The $\frac{\Delta}{2}$-fold repetition of the binary parity check code of dimension $k \geq 4$.
	\end{enumerate}
	Up to the order, the choice of the codes is uniquely determined by $C$.
\end{ntheorem}    
\end{trailer}
We remark that if $C=C_1\oplus C_2\oplus\dots\oplus C_l$ is the direct sum of $l$ linear codes $C_i$, then we have
$$
  W_C(x)=W_{C_1}(x) W_{C_2}(x)\dots W_{C_l}(x) 
$$
for the weight enumerator. We have $A_\Delta(C_i)=q^k-1$, $A_\Delta(C_i)=[k]_2-1$, and $A_\Delta(C_i)={{k+1}\choose 2}$, in cases (i), (ii), and (iii) of 
Theorem~\ref{thm_delta_div_spaned_by_min}, respectively. This can of course be used to compute $A_\Delta(C)$.
\begin{nexercise}
  Let $C$ be a binary linear code with non-zero weights in $\{8,16,24\}$ that is spanned by codewords of weight $8$. Then, we have
  \begin{eqnarray*}
    A_8&\in& \{0,1,2,3,4,5,6,7,8,9,10,11,13,14,15,\\&&16,17,18,19,21,22,23,25,29,30,31,33,37,45\}.
  \end{eqnarray*}
\end{nexercise}
Note that the non-existence result in Example~\ref{ex_no_32_10_8_16_24_2_code} is a direct implication.
\begin{nexercise}
  Let $a\in\N_{\ge 3}$, $\Delta=2^a$, and $C$ be a (projective) $\Delta$-divisible $[4\Delta,k]_2$-code. Show $k\le 2a+4$, cf.\ \cite[Theorem 4]{liu2010binary}. 
\end{nexercise}


\end{document}